\documentclass[12pt]{article}

\RequirePackage{amsthm,amsmath,amsfonts,amssymb}
\usepackage{graphicx} 

\usepackage{float}

\usepackage{pgfplots}
\usepackage[utf8]{inputenc}
\usepackage{natbib}
\usepackage{paralist}
\usepackage{enumitem}
\usepackage{authblk}
\usepackage{url}
\usepackage[colorlinks,
linkcolor=red,
anchorcolor=blue,
citecolor=blue
]{hyperref}

\usepackage{titlesec}

\titleformat{\paragraph}
  {\normalfont\normalsize\itshape}
  {\theparagraph}
  {1em}
  {}

\usepackage{diagbox}
\usepackage{makecell, multirow}
\allowdisplaybreaks

\usepackage{xcolor}

\definecolor{truegray}{RGB}{153,153,153}
\definecolor{oracleblack}{RGB}{0,0,0}
\definecolor{bayesorange}{RGB}{242,115,26}
\definecolor{mlred}{RGB}{217,26,26}
\definecolor{uvblue}{RGB}{0,77,217}
\definecolor{hpurple}{RGB}{140,0,191}
\definecolor{klbrown}{RGB}{140,89,13}

\usepackage[top=1.5in, left=1in, right=1in, bottom=1in]{geometry}  
\usepackage{xspace,exscale,relsize}
\usepackage{fancybox,shadow}
\usepackage{graphicx}
\usepackage{caption}
\usepackage{subcaption} 
\usepackage{xcolor}
\usepackage{longtable}
\usepackage[title]{appendix}
\usepackage{caption}
\usepackage{extarrows}
\usepackage{comment}
\usepackage{booktabs} 
\usepackage{makecell, multirow}

\usepackage{mathrsfs}
\makeatletter
\let\over=\@@over \let\overwithdelims=\@@overwithdelims
\let\atop=\@@atop \let\atopwithdelims=\@@atopwithdelims
\let\above=\@@above \let\abovewithdelims=\@@abovewithdelims
\makeatother
\usepackage{rotating}

	\usepackage{ifpdf}
	
	\usepackage{psfrag}
	
	\usepackage{prettyref}
	
	\usepackage{pgfplots}
	
	\usepackage{tikz}
	\usetikzlibrary{arrows}
	\tikzstyle{int}=[draw, fill=blue!20, minimum size=2em]
	\tikzstyle{dot}=[circle, draw, fill=blue!20, minimum size=2em]
	\tikzstyle{init} = [pin edge={to-,thin,black}]

	\usepackage[all]{xy}

	\usepackage[ruled,vlined]{algorithm2e}

	\usepackage{mathtools}

	     \def\EE{\mathbb{E}}

	     \def\PP{\mathbb{P}}

	\def\11{\mathbbm{1}}

	\def\calC{{\cal  C}}

	\def\calJ{{\cal  J}}

	\def\calP{{\cal  P}}

	\def\calW{{\cal  W}}

	\def\simind{\stackrel{\mbox{\scriptsize{ind}}}{\sim}}
         \def\simiid{\stackrel{\mbox{\scriptsize{iid}}}{\sim}}

	\DeclareMathOperator{\argmin}{argmin}

	\DeclareMathOperator{\E}{E}

	\DeclareMathOperator{\Var}{Var}

	\def\P{{\mathbb{P}}}
	
	\def\E{{\mathbb{E}}}

	\def\TV{{d_{\rm TV}}}

	\graphicspath {{fig/}}
	
	\usepackage{ifthen}
	\newboolean{aos}
	\setboolean{aos}{FALSE}

	\ifx\eqref\undefined
	\newcommand{\eqref}[1]{~(\ref{#1})}
	\fi
	\ifx\mod\undefined
	\def\mod{\mathop{\rm mod}}
	\fi
	
	\usepackage{bm}

	\def\argmin{\mathop{\rm argmin}}

	\def\exp{\mathop{\rm exp}}

	\def\EE{\Expect}

	\def\Var{\mathrm{Var}}

	\def\PP{\mathbb{P}}

	\newcommand{\abs}[1]{\left| #1 \right|}

	\makeatletter
	\def\bbordermatrix#1{\begingroup \m@th
		\@tempdima 4.75\p@
		\setbox\z@\vbox{%
			\def\cr{\crcr\noalign{\kern2\p@\global\let\cr\endline}}%
			\ialign{$##$\hfil\kern2\p@\kern\@tempdima&\thinspace\hfil$##$\hfil
				&&\quad\hfil$##$\hfil\crcr
				\omit\strut\hfil\crcr\noalign{\kern-\baselineskip}%
				#1\crcr\omit\strut\cr}}%
		\setbox\tw@\vbox{\unvcopy\z@\global\setbox\@ne\lastbox}%
		\setbox\tw@\hbox{\unhbox\@ne\unskip\global\setbox\@ne\lastbox}%
		\setbox\tw@\hbox{$\kern\wd\@ne\kern-\@tempdima\left[\kern-\wd\@ne
			\global\setbox\@ne\vbox{\box\@ne\kern2\p@}%
			\vcenter{\kern-\ht\@ne\unvbox\z@\kern-\baselineskip}\,\right]$}%
		\null\;\vbox{\kern\ht\@ne\box\tw@}\endgroup}
	\makeatother
	
	\newcommand{\Poi}{\mathrm{Poi}}

	\newcommand{\reals}{\mathbb{R}}
	
	\newcommand{\integers}{\mathbb{Z}}
	
	\newcommand{\Expect}{\mathbb{E}}

	\newcommand{\iid}{iid\xspace}

	\newcommand{\pth}[1]{\left( #1 \right)}
	\newcommand{\qth}[1]{\left[ #1 \right]}
	\newcommand{\sth}[1]{\left\{ #1 \right\}}

	\newcommand{\iiddistr}{{\stackrel{\text{\iid}}{\sim}}}

	\newcommand{\indc}[1]{{I\left({#1}\right)}}

	\definecolor{myblue}{rgb}{.8, .8, 1}
	\definecolor{mathblue}{rgb}{0.2472, 0.24, 0.6} 
	\definecolor{mathred}{rgb}{0.6, 0.24, 0.442893}
	\definecolor{mathyellow}{rgb}{0.6, 0.547014, 0.24}

	\newrefformat{eq}{(\ref{#1})}
	\newrefformat{thm}{Theorem~\ref{#1}}
	\newrefformat{th}{Theorem~\ref{#1}}
	\newrefformat{chap}{Chapter~\ref{#1}}
	\newrefformat{sec}{Section~\ref{#1}}
	\newrefformat{seca}{Section~\ref{#1}}
	\newrefformat{asmp}{Assumption~\ref{#1}}
	\newrefformat{fig}{Fig.~\ref{#1}}
	\newrefformat{tab}{Table~\ref{#1}}
	\newrefformat{rmk}{Remark~\ref{#1}}
	\newrefformat{clm}{Claim~\ref{#1}}
	\newrefformat{def}{Definition~\ref{#1}}
	\newrefformat{cor}{Corollary~\ref{#1}}
	\newrefformat{lmm}{Lemma~\ref{#1}}
	\newrefformat{prop}{Proposition~\ref{#1}}
	\newrefformat{pr}{Proposition~\ref{#1}}
	\newrefformat{property}{Property~\ref{#1}}
	\newrefformat{app}{Appendix~\ref{#1}}
	\newrefformat{apx}{Appendix~\ref{#1}}
	\newrefformat{ex}{Example~\ref{#1}}
	\newrefformat{exer}{Exercise~\ref{#1}}
	\newrefformat{soln}{Solution~\ref{#1}}
	
	\def\unifto{\mathop{{\mskip 3mu plus 2mu minus 1mu%
				\setbox0=\hbox{$\mathchar"3221$}%
				\raise.6ex\copy0\kern-\wd0%
				\lower0.5ex\hbox{$\mathchar"3221$}}\mskip 3mu plus 2mu minus 1mu}}

	\ifx\lesssim\undefined
	\def\simleq{{{\mskip 3mu plus 2mu minus 1mu%
				\setbox0=\hbox{$\mathchar"013C$}%
				\raise.2ex\copy0\kern-\wd0%
				\lower0.9ex\hbox{$\mathchar"0218$}}\mskip 3mu plus 2mu minus 1mu}}
	\else
	\def\simleq{\lesssim}
	\fi
	
	\ifx\gtrsim\undefined
	\def\simgeq{{{\mskip 3mu plus 2mu minus 1mu%
				\setbox0=\hbox{$\mathchar"013E$}%
				\raise.2ex\copy0\kern-\wd0%
				\lower0.9ex\hbox{$\mathchar"0218$}}\mskip 3mu plus 2mu minus 1mu}}
	\else
	\def\simgeq{\gtrsim}
	\fi

		\newif\ifmapx
		{\catcode`/=0 \catcode`\\=12/gdef/mkillslash\#1{#1}}
		\edef\jobnametmp{\expandafter\string\csname ic_apx\endcsname}
		\edef\jobnameapx{\expandafter\mkillslash\jobnametmp}
		\edef\jobnameexpand{\jobname}
		\ifx\jobnameexpand\jobnameapx
		\mapxtrue
		\else
		\mapxfalse
		\fi

		\renewcommand{\hat}{\widehat}
		\renewcommand{\tilde}{\widetilde}

		\newcommand{\hatG}{{\hat G}}
		\newcommand{\Gstar}{{G^\ast}}

\theoremstyle{plain}

\newtheorem{theorem}{Theorem}[section]
\newtheorem{lemma}{Lemma}[section]
\newtheorem{corollary}{Corollary}[section]
\newtheorem{proposition}{Proposition}[section]
\newtheorem{myassumption}{Assumption}[section]

\theoremstyle{definition}

\newtheorem{definition}{Definition}[section]

		\numberwithin{figure}{section}

\newcommand{\subexpo}{{\sf SubE}}

\title{Poisson empirical Bayes estimation of sums of random variables via minimum-distance methods}
\author[1]{Stefano Favaro}
\author[2]{Sandra Fortini}
\author[3]{Soham Jana}
\affil[1]{University of Turin and Collegio Carlo Alberto, Turin, IT}
\affil[2]{Bocconi University, Milan, IT}
\affil[3]{University of Notre Dame, Notre Dame, Indiana, US}

\date{}

\begin{document}

\maketitle

\begingroup

\renewcommand{\thefootnote}{}
\footnotetext{\texttt{stefano.favaro@unito.it}}
\footnotetext{\texttt{sandra.fortini@unibocconi.it}}
\footnotetext{\texttt{soham.jana@nd.edu}}
\endgroup

\begin{abstract}
The estimation of sums of functions of observable and unobservable variables is a long-standing problem in statistics, with applications in many domains. We consider this problem in Poisson mixture models, where empirical Bayes provides a natural framework but nonparametric theory remains limited. We develop a nonparametric empirical Bayes methodology based on regular and coarsened minimum-distance estimation of the unknown mixing distribution. For a broad class of such sums, we establish large-sample guarantees showing that the resulting plug-in estimates asymptotically merge with the oracle Bayes estimate. In particular, when the mixing distribution has finite support, we obtain a nearly parametric convergence rate, up to a logarithmic factor. We then provide a finite-sample analysis of two representative sums: the total intensity among units whose observed count does not exceed a fixed threshold, and the number of units whose observed count exceeds their latent intensity. For the total-intensity sum, we establish a lower bound on the minimax regret and derive upper bounds for both regular and coarsened minimum-distance procedures under compact-support and subexponential assumptions on the mixing distribution. These upper bounds match the minimax rate up to logarithmic factors, with the coarsened procedure in the compact-support setting matching the rate exactly when the coarsening level is fixed. The above-average sum displays a markedly different behavior: a lower bound on the minimax regret shows that bounded regret is in general impossible to achieve. We identify additional stability conditions under which bounded regret can be recovered up to logarithmic factors, and show that these conditions are automatically satisfied when the mixing distribution has finite support. Numerical experiments on both synthetic and real data illustrate the performance of the proposed methodology.
\end{abstract}

\tableofcontents


\section{Introduction}\label{sec:intro}
\subsection{Problem setting}

Consider $n\geq 1$ data points $X_{1:n}=(X_1,\ldots,X_n)$, modeled as conditionally independent Poisson random variables with unknown means $\theta_1,\ldots,\theta_n$, with the $\theta_i$'s viewed either as unobservable random variables or as fixed constants. Given a utility function $u:\integers_+\times\mathbb R_+\rightarrow\mathbb R$, a broad class of statistical problems can be formulated as the estimation of sums of random variables
\begin{equation}\label{eq}
S_n=\sum_{i=1}^n u(X_i,\theta_i);
\end{equation}
see, for instance, \citet{Rob(88),Rob(00),Zha(05)} and references therein. As an illustration, we recall the motorist example of \citet{Rob(77)}. For each individual $i$ in a pool of $n$ motorists, let $X_i$ denote the number of accidents in the current year, $Y_i$ the unobservable number of accidents in the following year, and $\theta_i$ the corresponding accident intensity, for $i=1,\ldots,n$. Different choices of the utility function $u$ in \eqref{eq} give rise to different sums $S_{n}$. For instance, if $u(x,\theta)=\theta I(x=\kappa)$, with $\kappa\in\integers_+$, then $S_n$ is the total intensity of motorists with $\kappa$ accidents in the current year; if $u(x,\theta)=I(x>\theta)$, then $S_n$ is the number of motorists whose number of accidents in the current year exceeds their intensity; and if $u(x,\theta)=\theta I(x\leq\kappa)$, with $\kappa\in\integers_+$, then $S_n$ is the expected total number of accidents in the following year among motorists with no more than $\kappa$ accidents in the current year.

The estimation of sums of random variables arises in a broad range of statistical problems. Early work in this area was largely motivated by treatment-effect estimation in clinical trials \citep{Robz(88),Rob(89),Rob(91)}, but the framework has since found applications in many other settings; see \citet{Zha(05)} for an extensive review. These include species sampling problems, such as the estimation of population richness, the missing mass and the number of unseen species \citep{Goo(53),Goo(56),Efr(76),Bun(93),Mao(02)}; data-confidentiality problems involving measures of disclosure risk and risk of re-identification \citep{Bet(90),Ski(02),Roc(19)}; and network-data problems, such as the estimation of node degrees from source-destination data \citep{Var(96),Teb(98),Zha(05)}.

\subsection{Background on empirical Bayes}

\citet{Rob(00)} and \citet{Zha(05)} developed a theory for the efficient estimation of the sums $S_n$ in \eqref{eq} within a mixture-modeling framework. They consider the $(X_i,\theta_i)$'s to be i.i.d. as
\begin{align}\label{eq:model}
X_i\mid\theta_i & \quad\simind\quad \text{Poisson}(\cdot\mid\theta_i),\qquad i=1,\ldots,n,\\
\notag\theta_i & \quad\simiid\quad \Gstar,
\end{align}
where $\text{Poisson}(\cdot\mid\theta)$ is the Poisson kernel with mean $\theta\in\Theta\subseteq\reals_+$, and $\Gstar$ is the unknown mixing distribution on $\Theta$. Under the Poisson mixture model \eqref{eq:model}, the probability mass function of each $X_i$ is
\begin{align}\label{eq:mixture-density}
f_{\Gstar}(x)=\int_{\Theta}\text{Poisson}(x\mid\theta)\Gstar(d\theta),\qquad x\in\integers_+.
\end{align}
Within this framework, empirical Bayes (EB) provides a natural approach to estimating $S_n$ \citep{Rob(56)}. Interpreting $\Gstar$ as a prior distribution, the optimal estimate of $S_n$ under squared-error loss is the posterior mean
\begin{equation}\label{eq:post_exp}
\hat{S}_{n}(\Gstar)=\E_{\Gstar}[S_n\mid X_{1:n}]=\sum_{i=1}^{n}\E_{\Gstar}[u(X_i,\theta_i)\mid X_i],
\end{equation}
that is, the oracle Bayes estimate. The EB approach consists of estimating $\Gstar$ from the data and plugging the resulting estimate into \eqref{eq:post_exp}, a strategy referred to as $g$-modeling \citep{Efr(14)}.

Parametric $g$-modeling assumes that the mixing distribution $\Gstar$ belongs to a parametric family $\{G_\tau:\tau\in\mathcal T\}$, where $\mathcal T$ is a Euclidean space. Estimating $\tau$ by $\hat\tau_n$ yields the plug-in estimate $G_{\hat\tau_n}$ and, upon substitution into \eqref{eq:post_exp}, a corresponding estimate of $S_n$. Under suitable regularity assumptions, this estimate is asymptotically efficient whenever $\hat\tau_n$ is itself asymptotically efficient for $\tau$ \citep[Theorem 2.2]{Zha(05)}. While computationally attractive, parametric $g$-modeling relies on a correctly specified model for the mixing distribution and may therefore be sensitive to misspecification. Nonparametric $g$-modeling removes this restriction by replacing $\Gstar$ in \eqref{eq:post_exp} with a nonparametric estimate $\hat G_n$. A recent development in this direction is the quasi-Bayes approach of \citet{Fav(26)}, which estimates the mixing distribution recursively and establishes frequentist large-sample guarantees for the resulting plug-in estimates of $S_n$. Despite this recent development, nonparametric $g$-modeling for sums of random variables remains largely unexplored. A classical alternative to nonparametric $g$-modeling is the ``$u,v$'' method of \citet{Rob(88)}, which is closely related to the so-called $f$-modeling strategy \citep{Efr(14)}. Although \citet[Theorem 2.5]{Zha(05)} shows that ``$u,v$'' estimates of $S_n$ are asymptotically efficient, their construction applies only to utility functions for which the corresponding ``$u,v$'' equation admits a suitable explicit solution.

\subsection{Preview of our contributions}

Our first contribution is a large-sample theory for nonparametric $g$-modeling of the sums $S_n$ in \eqref{eq} for a general utility function $u$. The theory of \citet{Zha(05)} is primarily concerned with asymptotic efficiency and information bounds, characterizing when regular $\sqrt n$-consistent estimation of $S_n$ is possible and identifying settings in which such estimation cannot be achieved. Here we take a different perspective and study plug-in estimates obtained by replacing the mixing distribution $\Gstar$ in \eqref{eq:post_exp} with nonparametric minimum-distance estimates, including the nonparametric maximum likelihood estimate as a special case. We show that, under general conditions, both regular and coarsened minimum-distance procedures yield estimates of $S_{n}$ that asymptotically merge with the oracle Bayes estimate, including in settings where regular $\sqrt n$-consistent estimation is unavailable. No explicit convergence rate is available in this general setting, reflecting the ill-posedness of recovering the mixing distribution from the Poisson mixture. When the mixing distribution has finite support, however, we obtain a nearly parametric convergence rate, up to a logarithmic factor. These results apply to a broad class of utility functions and do not require the existence of a ``$u,v$'' representation.

Our second contribution is a finite-sample analysis of two representative instances of $S_{n}$: the total-intensity and above-average sums, obtained from \eqref{eq} by taking, respectively, $u(x,\theta)=\theta I(x\leq\kappa)$, with $\kappa\in\integers_+$, and $u(x,\theta)=I(x>\theta)$. For both sums, we measure the regret of an estimator of $S_n$ as its mean-squared discrepancy from the oracle Bayes estimate at the $\sqrt n$-scale. For the total-intensity sum, we first establish a lower bound on the minimax regret and then derive corresponding upper bounds for both regular and coarsened minimum-distance procedures under compact-support and subexponential assumptions on the mixing distribution $\Gstar$. These upper bounds match the minimax rate up to logarithmic factors, with the coarsened procedure in the compact-support setting matching the rate exactly when the coarsening level is fixed. The above-average sum displays a markedly different behavior from the total-intensity sum, as a lower bound on the minimax regret shows that bounded regret is in general impossible to achieve in this estimation problem. This impossibility result motivates the introduction of additional stability conditions. Under these conditions, we derive upper bounds on the regret for both regular and coarsened minimum-distance procedures, showing that bounded regret can be recovered up to logarithmic factors. We further show that, when $\Gstar$ is supported on a finite set, the required stability conditions are automatically satisfied.

We conclude our study by empirically validating the proposed methodology on synthetic and real datasets, focusing on the total-intensity and above-average sums. Both regular and coarsened minimum-distance procedures are compared with the closed-form parametric estimates of \citet[Sections 2.3 and 2.4]{Zha(05)} and, whenever available, with nonparametric ``$u,v$'' estimates. The results show that the proposed procedures perform competitively with these alternatives, with suitable smoothing used for the above-average sum, while remaining computationally efficient and accommodating a broader class of sums through the choice of the utility function $u$.

\subsection{Organization of the paper}

The paper is structured as follows. Section \ref{sec1} introduces regular and coarsened minimum-distance estimates of the mixing distribution $\Gstar$ and establishes finite-sample guarantees for the corresponding mixture probability mass functions. Section \ref{sec2} develops the large-sample theory for the resulting plug-in estimates of $S_{n}$. Sections \ref{sec3} and \ref{sec4} focus on the total-intensity and above-average sums, respectively, providing a finite-sample regret analysis. Section \ref{sec5} presents numerical illustrations. Section \ref{sec6} concludes with a discussion on directions for future research. Proofs and additional numerical illustrations are deferred to the supplementary material.


\section{Regular and coarsened minimum-distance estimation}\label{sec1}

\subsection{Minimum-distance estimation of the mixing distribution}

In this section, we introduce the minimum-distance approach to estimating the mixing distribution $\Gstar$ in the Poisson mixture model \eqref{eq:model}. Recall that the oracle Bayes estimate of $S_n$ is the posterior mean $\hat S_n(\Gstar)$ in \eqref{eq:post_exp}. Nonparametric $g$-modeling replaces the unknown mixing distribution $\Gstar$ with an estimate $\hat G$, yielding the EB estimate
\begin{displaymath}
\hat S_n(\hat G)=\sum_{i=1}^n \E_{\hat G}[u(X_i,\theta_i)\mid X_i].
\end{displaymath}
We estimate $\Gstar$ by minimizing, over a suitable class of mixing distributions $G$, a statistical distance between the empirical distribution $\hat p_m$ of a training sample $(X'_1,\ldots,X'_m)$ and the corresponding Poisson mixture probability mass function $f_G$, with the training sample assumed to be independent of $(X_1,\ldots,X_n)$.

\subsubsection{Regular minimum-distance estimation}

To define the regular minimum-distance estimate of $\Gstar$, let $d$ denote a statistical distance between probability mass functions on $\integers_+$. Throughout the paper, we consider statistical distances satisfying the following assumption.

\begin{myassumption}\label{asmp:EB_t-ell} 
There exist maps $t:\calP(\integers_+)\to \mathbb R$ and $\ell:\mathbb R^2\to\mathbb R$ such that, for any two probability mass functions $q_1,q_2\in\calP(\integers_+)$,
\begin{displaymath}
d(q_1\|q_2)=t(q_1)+\sum_{y\geq0}\ell(q_1(y),q_2(y)),
\end{displaymath}
where, for every $a>0$, the map $b\mapsto \ell(a,b)$ is strictly decreasing and strictly convex, and $\ell(0,b)=0$ for all $b\geq0$.
\end{myassumption}

It is easy to verify that \prettyref{asmp:EB_t-ell} is satisfied, for instance, by the Kullback--Leibler divergence, the squared Hellinger distance, and the chi-squared distance.

Let $X'_1,\dots,X'_m$ be an independent training sample from the Poisson mixture model \eqref{eq:model}, also independent of $X_1,\dots,X_n$. Denote its empirical distribution by
\begin{displaymath}
\hat p_m(x)=\frac1m\sum_{i=1}^m I(X'_i=x),\qquad x\in\integers_+.
\end{displaymath}
Let $\mathcal P(\Theta)$ denote the space of probability measures on $\Theta$, endowed with the weak topology. Given a suitable class $\mathcal C\subseteq\mathcal P(\Theta)$ of mixing distributions, we define a regular minimum-$d$-distance estimate of $\Gstar$ as any solution
\begin{equation}\label{eq:optim-genr}
\widehat G_{d,m}\in\argmin_{G\in\mathcal C}d\!\left(\hat p_m\|f_{G}\right).
\end{equation}
The corresponding EB estimate of $S_n$ is then $\hat S_n(\widehat G_{d,m})$; see \citet[Section 2.1]{jana2025optimal} and references therein.

\subsubsection{Coarsened minimum-distance estimation}

For some choices of the utility function $u$, the oracle Bayes estimate $\hat S_n(\Gstar)$ depends on the mixture probability mass function $f_{\Gstar}$ only through finitely many of its values. This occurs, for instance, for the total-intensity sum, corresponding to $u(x,\theta)=\theta I(x\leq\kappa)$ for some $\kappa\in\integers_+$. In this case,
\begin{displaymath}
\hat{S}_{n}(\Gstar)=\E_{\Gstar}[S_{n}\mid X_{1:n}]=n\sum_{x=0}^{\kappa}(x+1)\frac{f_{\Gstar}(x+1)}{f_{\Gstar}(x)}\hat{p}_{n}(x).
\end{displaymath}
Hence, $\hat S_n(\Gstar)$ depends on $f_{\Gstar}$ only through $f_{\Gstar}(0),\ldots,f_{\Gstar}(\kappa+1)$, so that accurate estimation of these values is sufficient for estimating the oracle Bayes estimate. This observation motivates a coarsened version of minimum-distance estimation, in which the distance is computed after aggregating the tail probabilities beyond a prescribed threshold.

\begin{definition}\label{def:coarsening}
For any probability mass function $p\in\calP(\integers_+)$ and any integer $T\geq0$, define the $T$-coarsening of $p$ by
\begin{equation}\label{eq:coarsening}
\tilde p^{(T)}(y)=\begin{cases}
p(y) & 0\leq y\leq T,\\[0.2cm]
\sum_{j\geq T+1}p(j) & y=T+1.
\end{cases}
\end{equation}
\end{definition}

This coarsening leads to a natural generalization of the regular minimum-$d$-distance estimate \eqref{eq:optim-genr}. Given a class $\mathcal C$ of mixing distributions on $\Theta\subset\reals_+$, the $T$-coarsened minimum-$d$-distance estimate of $\Gstar$ is defined as any solution
\begin{equation}\label{eq:optim-coarse}
\widetilde G_{d,m}^{(T)}\in\argmin_{G\in\mathcal C}d\!\left(\tilde{\hat p}_m^{(T)}\|\tilde f_G^{(T)}\right),
\end{equation}
where $\tilde{\hat p}_m^{(T)}$ and $\tilde f_G^{(T)}$ denote the $T$-coarsenings of $\hat p_m$ and $f_G$, respectively. The corresponding EB estimate of $S_{n}$ is then $\hat S_n(\widetilde G_{d,m}^{(T)})$. By choosing $T$ according to the sum of interest, the coarsened procedure focuses the minimum-distance estimation problem on the components of $f_{\Gstar}$ that are directly relevant for estimating $S_n$.

\subsubsection{Existence of minimum-distance estimates}

For $h<\infty$, let $\mathcal P[0,h]$ denote the class of probability distributions supported on $[0,h]$. For $s>0$, let $\subexpo(s)$ denote the class of $s$-subexponential distributions on $\reals_+$, namely $\subexpo(s)=\left\{G:G((t,\infty))\leq 2e^{-t/s},\ \forall t>0\right\}$. The next theorem gives sufficient conditions for the existence of both the
regular and coarsened minimum-$d$-distance estimates. In particular, these conditions cover the classes $\mathcal P[0,h]$ and $\subexpo(s)$ introduced above.

\begin{theorem}\label{thm:exist-genr}
Under Assumption \ref{asmp:EB_t-ell}, the regular minimum-$d$-distance estimate \eqref{eq:optim-genr} and the $T$-coarsened minimum-$d$-distance estimate \eqref{eq:optim-coarse} exist when the class $\mathcal C$ of mixing distributions on $\Theta$ is a convex and weakly closed subset of $\mathcal P(\Theta)$ and $\mathcal C\subseteq\subexpo(s)$ for some $s>0$.
\end{theorem}

The proof is given in Section \ref{app:existence}. The result applies, in particular, when $\mathcal C$ is either $\mathcal P[0,h]$ or $\subexpo(s)$, or a convex and weakly closed subset thereof. Indeed, for $s$ sufficiently large, $\mathcal P[0,h]\subset\subexpo(s)$, and both classes are convex and weakly closed; the latter property follows from the
Portmanteau theorem \citep[Theorem 29.1]{billingsley2017probability}. The additional assumptions introduced later in the paper, such as Assumptions~\ref{asmp:regular} and \ref{asmp:aa-local-mass}, naturally lead to convex and weakly closed subclasses of $\mathcal P[0,h]$ or $\subexpo(s)$. The corresponding minimum-distance estimates therefore continue to exist after imposing these additional constraints.

\subsection{Finite-sample guarantees for minimum-distance estimates of $f_{G}$}\label{sec11}

For the regular minimum-distance estimate $\hat G_{d,m}$ in \eqref{eq:optim-genr}, \citet[Section 2.2]{jana2025optimal} provides finite-sample guarantees for the corresponding mixture probability mass function $f_{\hat G_{d,m}}$ in squared Hellinger distance. In addition to Assumption \ref{asmp:EB_t-ell}, these guarantees require the following condition on the distance $d$.

\begin{myassumption}\label{asmp:dist_sandwich}
There exist absolute constants $c_1,c_2>0$ such that, for any probability mass functions $q_1,q_2$ on $\integers_+$,
\begin{align}\label{eq:EB_Assupmtion_2}
c_1H^2(q_1,q_2)\leq d(q_1\|q_2)\leq c_2\chi^2(q_1\|q_2).
\end{align}
\end{myassumption}

We first recall the regular minimum-distance bound of \citet[Theorem 2]{jana2025optimal}.

\begin{lemma}[Regular minimum-distance estimation]\label{lem:EB_density-H2}
Let $d$ satisfy Assumptions \ref{asmp:EB_t-ell}--\ref{asmp:dist_sandwich}, and let $\Gstar$ denote the data-generating mixing distribution. There exist constants $C_1,C_2$ such that the following holds for the minimum-$d$-distance estimate $\hat G_{d,m}$ in \eqref{eq:optim-genr}, where the optimization is performed over an appropriate class given below:
\begin{enumerate}[label=(\alph*)~]
\item \label{pt:bounded-regular-dens}  Let $\mathcal C\subset \calP[0,\hat h]$ be a convex and weakly closed class. If $\Gstar\in\calP[0,h]\cap\mathcal C$ with $h\leq\hat h$ and $\hat G_{d,m}$ is optimized over the class $\mathcal C$, then for any $m\geq3$
\begin{displaymath}
\EE\qth{H^2(f_{\Gstar},f_{{\hat G_{d,m}}})}\leq {C_1\over m}\min\sth{\max\{1,\hat h\}{\log m\over \log\log m},\hat h+\log m};
\end{displaymath}
\item \label{pt:subexpo-regular-dens} Let $\mathcal C\subset \subexpo(s)$ be a convex and weakly closed class. If $\Gstar\in\mathcal C$ and $\hat G_{d,m}$ is optimized over the class $\mathcal C$, then for any $m\geq2$
\begin{displaymath}
\EE\qth{H^2(f_{\Gstar},f_{{\hat G_{d,m}}})}\leq {C_2\max\{1,s\}\over m}{\log m}.
\end{displaymath} 
\end{enumerate}
\end{lemma}

The proof follows the argument of \citet[Theorem 2]{jana2025optimal}; details are given in Section \ref{prooflemma21}.

We next establish analogous finite-sample guarantees for the $T$-coarsened minimum-$d$-distance estimate $\widetilde G_{d,m}^{(T)}$ in \eqref{eq:optim-coarse}, now in terms of the squared Hellinger distance between the corresponding coarsened mixture probability mass functions.

\begin{lemma}[Coarsened minimum-distance estimation]
\label{lem:EB_density-H2_coars}
Let $d$ satisfy Assumptions \ref{asmp:EB_t-ell}--\ref{asmp:dist_sandwich}, and let $\Gstar$ denote the data-generating mixing distribution. Consider the $T$-coarsened minimum-$d$-distance estimate $\widetilde G_{d,m}^{(T)}$ of $\Gstar$ obtained from a training sample of size $m$, with $T\geq0$.
\begin{enumerate}[label=(\alph*)~]
\item \label{pt:bounded-coarse} Let $\mathcal C\subset \calP[0,\hat h]$ be a convex and weakly closed class. If $\Gstar\in\calP[0,h]\cap\mathcal C$ with $h\leq \hat h$ and $\widetilde G_{d,m}^{(T)}$ is optimized over $\mathcal C$, then for every $T\geq 0$ and $m\geq 3$,
\begin{align}
\mathbb E\Big[H^2\big(\widetilde f_{\Gstar}^{(T)},\,\widetilde f_{\widetilde
G^{(T)}_{d,m}}^{(T)}\big)\Big]
&\le \frac{1}{m}\min\Big\{\frac{4c_2}{c_1}(T+2),\
\Big(\frac{12c_2}{c_1}+4\Big)
(\hat h e^2+2)\frac{2\log
m}{\log\log m}\Big\},
\label{eq:final-bounded}
\end{align}
\item Let $\mathcal C\subset \subexpo(s)$ be a convex and weakly closed class. If $\Gstar\in\mathcal C$ and $\widetilde G_{d,m}^{(T)}$ is optimized over $\mathcal C$, then for every $T\geq 0$ and $m\geq \exp\left(\frac 1 2 \log(1+\frac{1}{2s})\right)$
\begin{align}
\mathbb E\Big[H^2\big(\widetilde f_{\Gstar}^{(T)},\,\widetilde f_{\widetilde
G^{(T)}_{d,m}}^{(T)}\big)\Big]
&\le \frac{1}{m}\min\Big\{\frac{4c_2}{c_1}(T+2),\
\Big(\frac{16c_2}{c_1}+6\Big)\frac{2\log m}{\log(1+1/(2s))}\Big\}.
\label{eq:final-subexp}
\end{align}
\end{enumerate}
\end{lemma}

The proof is given in Section \ref{prooflemma22}. We conclude this subsection by specializing Lemmas \ref{lem:EB_density-H2} and \ref{lem:EB_density-H2_coars} to discrete mixing distributions supported on a finite set of known locations, with unknown probability masses. Specifically, let
\begin{displaymath}
\Theta_J:=\{\theta_1,\dots,\theta_J\}\subset [0,h],\qquad 0\le \theta_1<\cdots<\theta_J\le h,
\end{displaymath}
and 
\begin{displaymath}
\mathcal P(\Theta_J):=\left\{G=\sum_{j=1}^J w_j\delta_{\theta_j}:w_j\ge 0,\ \sum_{j=1}^J w_j=1\right\}.
\end{displaymath}
This setting is relevant in practice, where the support of the mixing distribution is often discretized over a known grid and the corresponding probability masses are estimated. Let
$\Gstar=\sum_{j=1}^J w_j^\ast\delta_{\theta_j}\in\mathcal P(\Theta_J)$, where the $w_j^\ast$'s are the unknown probability masses to be estimated. The following corollary is an immediate consequence of Lemmas \ref{lem:EB_density-H2}--\ref{lem:EB_density-H2_coars}.

\begin{corollary}[Regular and coarsened minimum-distance estimation on $\mathcal P(\Theta_J)$]\label{corr:hell-risk-discrete-feasible}
Let $\mathcal C\subseteq\mathcal P(\Theta_J)$ be a convex and weakly closed class, and let $\Gstar\in\mathcal C$. Let $\hat G_{d,m}$ and $\widetilde G_{d,m}^{(T)}$ be defined as in \eqref{eq:optim-genr} and \eqref{eq:optim-coarse}, respectively. Under Assumptions \ref{asmp:EB_t-ell}--\ref{asmp:dist_sandwich}, $\hat G_{d,m}$ satisfies the guarantees of \prettyref{lem:EB_density-H2} and $\widetilde G_{d,m}^{(T)}$ satisfies the guarantees of \prettyref{lem:EB_density-H2_coars}.
\end{corollary}


\section{Large-sample guarantees for general utility functions}\label{sec2}

In this section, we study the asymptotic behavior of empirical Bayes plug-in estimators of $S_n$ under the Poisson mixture model for a general utility function $u$. Our analysis is motivated by the nonregular regime considered by \citet{Zha(05)}, in which regular $\sqrt n$-consistent estimation is not available. We consider minimum-distance plug-in estimators constructed from an independent training sample whose size $m_n$ is allowed to depend on $n$. We establish consistency under general conditions and derive an explicit convergence-rate upper bound when the mixing distribution has finite support. For later use, we introduce the following condition:
 \begin{equation}\label{eq:condition2}
 \sum_{x\in \integers_+} f_{\Gstar}(x)\sup_{G\in\mathcal C} \int_\Theta |u(x,\theta)|^2G(d\theta\mid x)<\infty, 
 \end{equation}
where
\begin{displaymath}
       G(d\theta\mid x):=\frac{\text{Poisson}(x\mid\theta)G(d\theta)}{f_G(x)}, \quad\quad\quad x\in\mathbb Z_+.
\end{displaymath}
The next theorem establishes large-sample guarantees for the empirical Bayes plug-in estimator of $S_n$ for a general utility function $u$ satisfying suitable regularity conditions. Throughout this section, $\stackrel{\PP}{\rightarrow}$ denotes convergence in probability.

\begin{theorem}\label{th:consist1}
Let $\Gstar$ denote the data-generating mixing distribution, and let $d$ satisfy Assumptions \ref{asmp:EB_t-ell} and \ref{asmp:dist_sandwich}. The following asymptotic results hold for the minimum-$d$-distance estimator $\widehat G_{d,m_n}$ defined in \eqref{eq:optim-genr}, based on an independent training sample of size $m_n$ and with the optimization performed over a class $\mathcal C$ specified below.
\begin{enumerate}[label=(\alph*)~]
\item Suppose that $\mathcal C$ is a convex and weakly closed subset of $\subexpo(s)$, $\Gstar\in\mathcal C$, $m_n\to\infty$ as $n\to\infty$, $\theta\mapsto u(x,\theta)$ is $G$-almost surely continuous for every $x$ and every $G\in\mathcal C$, and condition \eqref{eq:condition2} holds. Then,
\begin{displaymath}
\frac{1}{n}\bigl(S_n-\widehat S_n(\widehat G_{d,m_n})\bigr)\stackrel{\PP}{\rightarrow}0.
\end{displaymath}
\item Suppose that $\mathcal C$ is a convex and weakly closed subset of $\mathcal P(\Theta_J)$ such that $G(\{0\})\leq a_0$ for some $a_0<1$ and every $G\in \mathcal C$, $\Gstar\in\mathcal C$, $m_n\asymp n$, and there exist constants $c$ and $M$ such that $u(x,\theta_j)=c$ for every $j=1,\ldots,J$ and every $x>M$. Then,
\begin{displaymath}
\frac{1}{n}\bigl(S_n-\widehat S_n(\widehat G_{d,m_n})\bigr)=O_{\mathbb P}(n^{-1/2}(\log n)^{1/2}).
\end{displaymath}
\end{enumerate}
\end{theorem}

The proof is given in Section~\ref{app:proofconst1}. For $s$ sufficiently large, $\mathcal P(\Theta_J)\subset\subexpo(s)$, so part~(a) applies in particular when $\mathcal C$ is a convex and weakly closed subset of $\mathcal P(\Theta_J)$. Moreover, the assumptions in part~(b) imply condition~\eqref{eq:condition2}. Equivalently, the conclusions in parts~(a) and~(b) can be expressed in terms of the discrepancy between the empirical Bayes plug-in estimator and the oracle Bayes estimator as
\begin{displaymath}
\frac{1}{n}\bigl(\widehat S_n(\Gstar)-\widehat S_n(\widehat G_{d,m_n})\bigr)=o_{\mathbb P}(1),
\end{displaymath}
and
\begin{displaymath}
\frac{1}{n}\bigl(\widehat S_n(\Gstar)-\widehat S_n(\widehat G_{d,m_n})\bigr)=O_{\mathbb P}(n^{-1/2}(\log n)^{1/2}),
\end{displaymath}
respectively. To see this, define
\begin{displaymath}
u_G(x):=\EE_G[u(X_i,\theta_i)\mid X_i=x].
\end{displaymath}
Then the random variables $u(X_i,\theta_i)-u_{\Gstar}(X_i)$, $i=1,\ldots,n$, are independent and identically distributed with mean zero. By the central limit theorem,
\begin{displaymath}
n^{-1/2}(\log n)^{-1/2} \bigl(S_n-\widehat S_n(\Gstar)\bigr)=o_{\PP}(1).
\end{displaymath}
The absence of an explicit rate in Theorem \ref{th:consist1}(a) reflects the generality of the utility function. For a general $u$, the plug-in error depends on the convergence of the mixing distribution $\widehat G_{d,m_n}$ to $\Gstar$, rather than merely on the convergence of the corresponding mixture probability mass functions. While $f_{\widehat G_{d,m_n}}$ may converge to $f_{\Gstar}$ at a quantifiable rate, the inverse problem $f_G\mapsto G$ is ill-posed, and no comparable rate for the convergence of $\widehat G_{d,m_n}$ is available in general. In the finite-support setting of part~(b), this limitation disappears: the support points are fixed and finite, and the map from the mixing weights to finitely many values of the mixture probability mass function is invertible. This allows the convergence rate for the mixture probability mass function to be transferred to the mixing weights, yielding the sharper rate in part~(b).

We next extend the asymptotic analysis to the coarsened minimum-distance estimator in \eqref{eq:optim-coarse}, allowing the coarsening level $T_n$ to depend on $n$.

\begin{theorem}\label{th:consist2}
Let $\Gstar$ denote the data-generating mixing distribution, and let $d$ satisfy Assumptions \ref{asmp:EB_t-ell} and \ref{asmp:dist_sandwich}. The following asymptotic results hold for the coarsened minimum-$d$-distance estimator $\widetilde G^{(T_n)}_{d,m_n}$, based on an independent training sample of size $m_n\to\infty$ and with the optimization performed over a class $\mathcal C$ specified below.
\begin{enumerate}[label=(\alph*)~]
    \item Suppose that $\mathcal C$ is a convex and weakly closed subset of $\subexpo(s)$,  $\Gstar\in\mathcal C$, $T_n\ge  \max\left\{1,\frac{2}{\log(1+1/(2s))}\right\}\log m_n$ for $n$ sufficiently large, $\theta\mapsto u(x,\theta)$ is $G$-almost surely continuous for every $x$ and every $G\in\mathcal C$, and condition \eqref{eq:condition2} holds. Then, as $n\to\infty$,
\begin{displaymath}
   \frac{ S_n-\widehat S_n(\widetilde G^{(T_n)}_{d,m_n})}{n}\stackrel{\PP}{\rightarrow}0.
\end{displaymath}
     \item  Suppose that $\mathcal C$ is a convex and weakly closed subset of $\mathcal P(\Theta_J)$  such that $G(\{0\})\leq a_0$ for some $a_0<1$ and every $G\in \mathcal C$, $T_n\rightarrow\infty$,
    $m_n\asymp n$, and there exist constants $c$ and $M$ such that $u(x,\theta_j)=c$ for every $j=1,\dots,J$ and every $x>M$. Then, as
$n\to\infty$,
\begin{displaymath}
   \frac{ S_n-\widehat S_n(\widetilde G^{(T_n)}_{d,m_n})}{n} =O_{\PP}(n^{-1/2}(\log n)^{1/2}).
\end{displaymath}
\end{enumerate}
\end{theorem}

The proof is given in Section \ref{app:proofconsist2}. For $s$ sufficiently large, $\mathcal P(\Theta_J)\subset\subexpo(s)$, so part~(a) applies in particular when $\mathcal C$ is a convex and weakly closed subset of $\mathcal P(\Theta_J)$. As in Theorem \ref{th:consist1}, the conclusions in parts~(a) and~(b) can equivalently be expressed as
\begin{displaymath}
\frac{\widehat S_n(\Gstar)-\widehat S_n(\widetilde G^{(T_n)}_{d,m_n})}{n}=o_{\PP}(1),
\end{displaymath}
and
\begin{displaymath}
\frac{\widehat S_n(\Gstar)-\widehat S_n(\widetilde G^{(T_n)}_{d,m_n})}{n}=O_{\PP}(n^{-1/2}(\log n)^{1/2}),
\end{displaymath}
respectively. 


 \section{Finite-sample guarantees for the total-intensity sum}\label{sec3}

In this section, we consider the utility function $u(x,\theta)=\theta I(x\leq\kappa)$, for some choice of $\kappa\in\integers_+$ \citep{Robz(88),Zha(05)}. This choice of the utility $u$ leads to a sum of random variables referred to as the total-intensity sum, i.e.,
\begin{displaymath}
S_{n,\kappa}=\sum_{i=1}^{n}\theta_{i}I(X_{i}\leq\kappa).
\end{displaymath}
From \eqref{eq:post_exp}, the oracle Bayes estimate of $S_{n,\kappa}$, which assumes knowledge of the data-generating mixing distribution $\Gstar$, is
\begin{equation}\label{eq:bayes_totin}
\hat{S}_{n,\kappa}(\Gstar)=\E_{\Gstar}[S_{n,\kappa}\mid X_{1:n}]=\sum_{i=1}^{n}\E_{\Gstar}[\theta_{i}\mid X_{i}]I(X_{i}\leq\kappa).
\end{equation}
We study the EB estimates of $S_{n,\kappa}$ obtained by plugging the regular and $T$-coarsened minimum-$d$-distance estimates of $\Gstar$ into \eqref{eq:bayes_totin}, namely $\hat G_{d,m}$ in \eqref{eq:optim-genr} and $\widetilde G_{d,m}^{(T)}$ in \eqref{eq:optim-coarse}, respectively.

\subsection{Minimax lower bound}

We first establish a minimax lower bound for the estimation of the oracle Bayes estimate $\hat S_{n,\kappa}(G)$. The bound provides a benchmark against which the finite-sample performance of the minimum-distance procedures introduced below can be assessed.

\begin{theorem}\label{thm:total_intensity_lb}
{Let $\mathcal C$ be a convex and weakly closed class of mixing distributions including $\operatorname{Unif}[a,b]$ and $G_{j,\varepsilon}$ for $j=0,1$, where 
\[
G_{0,\varepsilon}  =(1-c\varepsilon)\operatorname{Unif}[a,b]    +c\varepsilon\,\operatorname{Unif}[1-\varepsilon,1],\qquad
G_{1,\varepsilon}  =(1-c\varepsilon)\operatorname{Unif}[a,b]    +c\varepsilon\,\operatorname{Unif}[1,1+\varepsilon],
\]
for every $\varepsilon\in(0,\varepsilon_0]$, and some $c>0$, $a$, $b$ and $\varepsilon_0>0$ such that $(1+\varepsilon_0)^2<a<b\leq 3$, and $c\varepsilon_0\leq 1/2$. } 
Let 
\[
m_G^{(\kappa)}(x)=\EE_{G}[\theta_1I(X_1\leq\kappa)\mid X_1=x].
\] If $\varepsilon_0$ is small enough,  then the following results hold:
\begin{itemize}
    \item[(a)] There exists a constant $C_3>0$ such that 
\[
 \EE_{G_{j,\varepsilon}}
 \left[m_{G_{1,\varepsilon}}^{(\kappa)}(X_1)-m_{G_{0,\varepsilon}}^{(\kappa)}(X_1)\right]
 \geq C_3\varepsilon^2\quad\quad (j=0,1).
\]
    \item[(b)]  For any estimator $\widehat S$ of $
S_{n,\kappa}=\sum_{i=1}^{n}\theta_{i}I(X_{i}\leq\kappa) $ based on $(X_1,\dots,X_n,X'_1,\dots,X'_m)$, 
\[
\inf_{\widehat S}\sup_{G\in\mathcal C}
\EE_G\left[  \left(    \frac{\widehat S-\widehat S_{n,\kappa}(G)}{\sqrt n}
  \right)^2\right]\geq C_4\frac{n}{m+n},
\]
for $n$ large enough, where $C_4$ is a constant.
In particular, if $m=n$, 
\[
\inf_{\widehat S}\sup_{G\in\mathcal C}\EE_G\left[  \left(
    \frac{\widehat S-\widehat S_{n,\kappa}(G)}{\sqrt n}  \right)^2
\right]\geq \frac{C_4}{2} .
\]
The bound remains valid for estimators that only depend on $(X'_1,\dots,X'_m)$. 
\end{itemize}
\end{theorem}

The proof is given in Section \ref{proof41_lb}. {In particular, if $\mathcal C=\mathcal P[0,\hat h]$ with $\hat h>1$ or $\mathcal C=\subexpo(s)$, then the assumption of Theorem \ref{thm:total_intensity_lb} holds.} Theorem \ref{thm:total_intensity_lb} shows that the minimax risk is bounded from below by a constant multiple of $n/(m+n)$. In particular, when the training sample size is of the same order as the sample size, $m\asymp n$, the minimax risk is bounded away from zero. Thus, under the normalization by $\sqrt n$ used above, no estimator can achieve a uniformly vanishing risk over $\mathcal C$.

\subsection{Finite-sample upper bounds}

{Throughout this section, we assume that the distance $d$ satisfies Assumptions \ref{asmp:EB_t-ell}--\ref{asmp:dist_sandwich}. We also impose the following nondegeneracy condition.

\begin{myassumption}\label{asmp:regular}
There exists a constant $b_0>0$ such that $\min_{0\leq x\leq \kappa} f_G(x)\geq b_0$ for all $G\in\mathcal C$.
\end{myassumption}

Assumption \ref{asmp:regular} is a technical nondegeneracy condition on the mixture probability mass functions induced by the elements of $\mathcal C$. It ensures that the probabilities appearing in the oracle Bayes estimate of $S_{n,\kappa}$ are bounded away from zero uniformly over $\mathcal C$. This condition is particularly useful for controlling ratios of the form $f_G(x+1)/f_G(x)$, which naturally arise in the analysis of the total-intensity sum. {If $\mathcal C\subset \subexpo(s)$ for some $s>0$ and, for some $\epsilon>0$ and $c_0<1$, $G([0,\epsilon))\leq c_0$ for every $G\in\mathcal C$, then $\mathcal C$ satisfies Assumption \ref{asmp:regular}. See Lemma \ref{lmm:mixture-lower-bound} for the proof.}
}

The next theorem provides finite-sample upper bounds for the EB estimate $\hat{S}_{n,\kappa}(\hat G_{d,m})$ of $S_{n,\kappa}$, where $\hat G_{d,m}$ is the regular minimum-$d$-distance estimate of $\Gstar$ in \eqref{eq:optim-genr}. The result relies on the finite-sample guarantees in Lemma \ref{lem:EB_density-H2}.

\begin{theorem}\label{thm:sum_regular}
Let $\Gstar$ denote the data-generating mixing distribution. Consider the regular minimum-$d$-distance estimate $\hat G_{d,m}$ of $\Gstar$ obtained from an independent training sample of size $m$.
\begin{enumerate}[label=(\alph*)]
\item \label{pt:bounded-regular} Let $\mathcal C\subset \calP[0,\hat h]$ be a convex and weakly closed class of mixing distributions satisfying Assumption \ref{asmp:regular}. If $\Gstar\in \calP[0,h]\cap\mathcal C$ with $h\leq \hat h$, and $\hatG_{d,m}$ is optimized over $\mathcal C$, then there is a constant $C_{\hat h}$ such that
			\[
			\mathbb{E}\left[
			\left(
			\frac{\hat S_{n,\kappa}(\hat G_{d,m}) - \hat{S}_{n,\kappa}(\Gstar)}{\sqrt{n}}
			\right)^2
			\right]
			\le n{C_{\hat h}(\hat h^2+h\kappa)\log m\over m\log\log m}.
			\]

            \item \label{pt:subexpo-regular} Let $\mathcal C\subset \subexpo(s)$ be a convex and weakly closed class of mixing distributions satisfying Assumption \ref{asmp:regular}. If $\Gstar\in\mathcal C$ and $\hatG_{d,m}$ is optimized over $\mathcal C$, then there is a constant $C_s$ such that
			\[
			\mathbb{E}\left[
			\left(
			\frac{\hat S_{n,\kappa}(\hat G_{d,m}) - \hat{S}_{n,\kappa}(\Gstar)}{\sqrt{n}}
			\right)^2
			\right]
			\le C_s\frac{n(\log m)^3}{m}.
			\]
	\end{enumerate}
\end{theorem}

The proof is given in Section \ref{proof41}. We next turn to the $T$-coarsened minimum-distance procedure. The following theorem provides finite-sample upper bounds for the EB estimate $\hat S_{n,\kappa}(\widetilde G_{d,m}^{(T)})$, relying on the corresponding finite-sample guarantees in Lemma \ref{lem:EB_density-H2_coars}.

\begin{theorem}\label{thm:sum_coarse}
	Let $\Gstar$ denote the data-generating mixing distribution. Consider the $T$-coarsened minimum-$d$-distance estimate $\widetilde G_{d,m}^{(T)}$ of $\Gstar$ obtained from an independent  training sample of size $m$, with $T\geq\kappa+1$.
	\begin{enumerate}[label=(\alph*)]
		
        \item \label{pt:bounded-coarse-sum} Let $\mathcal C\subset \calP[0,\hat h]$ be a convex and weakly closed class of mixing distributions satisfying Assumption \ref{asmp:regular}. If $\Gstar\in\calP[0,h]\cap\mathcal C$ with $h\leq \hat h$ and $\widetilde G_{d,m}^{(T)}$ is optimized over $\mathcal C$, then there is a constant $C_{\hat h}$ such that for any $m\geq 3$,
		\[
		\mathbb{E}\left[
		\left(
		\frac{\hat S_{n,\kappa}(\widetilde G_{d,m}^{(T)}) - \hat{S}_{n,\kappa}(\Gstar)}{\sqrt{n}}
		\right)^2
		\right]
		\le C_{\hat h}\frac {n(\hat h^2+\hat h\kappa)}m\min\sth{T,{\log m\over \log\log m}}.
		\]
        
        \item \label{pt:subexpo-coarse-sum} Let $\mathcal C\subset \subexpo(s)$ be a convex and weakly closed class of mixing distributions satisfying Assumption \ref{asmp:regular}. If $\Gstar\in\mathcal C$ and $\widetilde G_{d,m}^{(T)}$ is optimized over $\mathcal C$, then there is a constant $C_s$ such that for any $m\geq \max\{2,\exp\left(\frac 1 2 \log(1+\frac{1}{2s})\right)\}$
		\[
		\mathbb{E}\left[
		\left(
		\frac{\hat S_{n,\kappa}(\widetilde G_{d,m}^{(T)}) - \hat{S}_{n,\kappa}(\Gstar)}{\sqrt{n}}
		\right)^2
		\right]
		\le C_s\frac{n\{(\log m)^2+\kappa\log m\} \min\{T,\log m\}}{m}.
		\]

	\end{enumerate}
\end{theorem}

The proof is given in Section \ref{proof:totalcoarse}. We conclude by comparing these upper bounds with the minimax lower bound in Theorem \ref{thm:total_intensity_lb}. When $m\asymp n$ and $\kappa$ and the parameters defining the mixing class are regarded as fixed, the lower bound is of constant order. For classes to which both the lower and upper bounds apply, the upper bounds in Theorems \ref{thm:sum_regular} and \ref{thm:sum_coarse} match this benchmark up to logarithmic factors. More specifically, in the compact-support setting, the coarsened procedure with a fixed coarsening level $T\geq\kappa+1$ attains a constant-order risk when $m\asymp n$, and hence matches the order of the minimax lower bound. The regular procedure differs from the minimax benchmark only by a factor of order $\log n/\log\log n$. These results show that minimum-distance
estimation yields near-minimax-optimal estimation of the oracle Bayes estimate of the total-intensity sum, while coarsening can remove the logarithmic penalty in the compact-support setting.

\subsection{Discrete mixing distributions}

\begin{corollary}[Regular and coarsened minimum-distance estimation on \(\mathcal P(\Theta_J)\)]\label{corr:tot-sum-risk-discrete-feasible}
Let $\hat G_{d,m}$ and $\widetilde G_{d,m}^{(T)}$ be defined as in \eqref{eq:optim-genr} and \eqref{eq:optim-coarse}, respectively, with the optimization performed over $\mathcal P(\Theta_J)$ and $T\geq\kappa+1$. Then, under $\hat h=h$, $\hat G_{d,m}$ satisfies the guarantees of \prettyref{thm:sum_regular}\ref{pt:bounded-regular} and $\widetilde G_{d,m}^{(T)}$ satisfies the guarantees of \prettyref{thm:sum_coarse}\ref{pt:bounded-coarse-sum}.
\end{corollary}

For fixed $\Theta_J$ and $T$, the preceding corollary shows that the coarsened minimum-distance procedure achieves the parametric rate $n/m$, while the regular minimum-distance procedure achieves this rate up to a factor of order $\log m/\log\log m$. In particular, when $m\asymp n$, the coarsened procedure has uniformly bounded risk at the $\sqrt n$ scale. Thus, in the discrete setting, coarsened minimum-distance estimation of the total-intensity sum retains the parametric-rate behavior established above, while the regular procedure incurs only a logarithmic loss.


\section{Finite-sample guarantees for the above-average sum}\label{sec4}

In this section, we consider the utility function $u(x,\theta)=I(x>\theta)$; see, e.g., \citet[Examples 2.1 and 2.2]{Zha(05)}. This choice of the utility $u$ leads to a sum of random variables referred to as the above-average sum, i.e.,
\begin{displaymath}
S_{n,>}=\sum_{i=1}^{n}I(X_i>\theta_i).
\end{displaymath}
From \eqref{eq:post_exp}, the oracle Bayes estimate of $S_{n,>}$, which assumes knowledge of the data-generating mixing distribution $\Gstar$, is
\begin{equation}\label{eq:bayes_above}
\hat S_{n,>}(\Gstar)=\E_{\Gstar}[S_{n,>}\mid X_{1:n}]=\sum_{i=1}^{n}\P_{\Gstar}(\theta_i<X_i\mid X_i).
\end{equation}
We study the EB estimates of $S_{n,>}$ obtained by plugging the regular and $T$-coarsened minimum-$d$-distance estimates of $\Gstar$ into \eqref{eq:bayes_above}, namely $\hat G_{d,m}$ in \eqref{eq:optim-genr} and $\widetilde G_{d,m}^{(T)}$ in \eqref{eq:optim-coarse}, respectively. Throughout this section, we assume that the distance $d$ satisfies Assumptions \ref{asmp:EB_t-ell}--\ref{asmp:dist_sandwich}.

\subsection{Minimax lower bound}

We first establish a minimax lower bound for the estimation of the oracle Bayes estimate $\hat S_{n,>}(G)$. In contrast to the total-intensity sum, the resulting lower bound shows that, without additional structural assumptions, estimation at the parametric scale is in general impossible.

\begin{theorem}\label{above_average_lb}
    Let $\mathcal C$ be a convex and weakly closed class of mixing distributions including $\operatorname{Unif}[a,b]$ and $G_{j,\varepsilon}$ for $j=0,1$, {defined as in Theorem \ref{thm:total_intensity_lb}.} 
    Let 
\[
m_G^{(a)}(x)=\EE_{G}[I(X_1>\theta_1)\mid X_1=x].
\]  If $\varepsilon_0$ is small enough, then the following results hold:
\begin{itemize}
    \item[(a)] There exists a constant $C_5>0$ such that 
\[
 \EE_{G_{j,\varepsilon}}\left[m_{G_{0,\varepsilon}}^{(a)}(X_1)-m_{G_{1,\varepsilon}}^{(a)}(X_1)\right]\geq C_5\varepsilon\quad (j=0,1).
\]
    \item[(b)] For any estimator $\widehat S$ of $S_{n,>}=\sum_{i=1}^{n}I(X_{i}>\theta_i) $ based on $(X_1,\dots,X_n,X'_1,\dots,X'_m)$, 
\[
\inf_{\widehat S}\sup_{G\in\mathcal C}\EE_G\left[  \left(    \frac{\widehat S-\widehat S_{n,>}(G)}{\sqrt n}\right)^2\right]\geq C_6\frac{n}{\sqrt{m+n}},
\]
for $n$ large enough, where $C_6$ is a constant.
In particular, if $m=n$, 
\[
\inf_{\widehat S}\sup_{G\in\mathcal C}\EE_G\left[  \left(\frac{\widehat S-\widehat S_{n,>}(G)}{\sqrt n}  \right)^2\right]\geq \frac{C_6}{2}\sqrt n.
\]
The bound remains valid for estimators that only depend on $(X'_1,\dots,X'_m)$. 
\end{itemize}
\end{theorem}

The proof is given in Section \ref{proof51_lb}. Theorem \ref{above_average_lb} reveals a markedly different behavior from that of the total-intensity sum. When $m\asymp n$, the minimax risk is bounded from below by a constant multiple of $\sqrt n$, and therefore diverges under the $\sqrt n$ normalization. Thus, over the class considered in Theorem \ref{above_average_lb}, no estimator can achieve a uniformly bounded risk at the parametric scale. This motivates the additional structural assumptions introduced above, under which we next derive finite-sample upper bounds.

\subsection{Finite-sample upper bounds}
{Throughout this section, we assume that the distance $d$ satisfies Assumptions \ref{asmp:EB_t-ell}--\ref{asmp:dist_sandwich}.  We also impose additional regularity conditions tailored to the above-average sum. Since several of these conditions are used for both regular and coarsened minimum-distance estimation, we collect them here and invoke them as needed.}

\begin{definition}[Conditional above-average probability]\label{def_ratio}
For any probability measure $G$ on $\reals_+$, define $m_G(x):=\frac{A_x(G)}{B_x(G)}$, where
\begin{displaymath}
A_x(G)=\int_{\reals_+} e^{-\theta}\theta^x I(\theta<x)\,dG(\theta),\qquad B_x(G)=\int_{\reals_+} e^{-\theta}\theta^x\,dG(\theta)=x!\,f_G(x).
\end{displaymath}
Whenever $B_x(G)>0$, the ratio $m_G(x)=A_x(G)/B_x(G)$ has its usual conditional-probability interpretation. If $B_x(G)=0$, we define $m_G(x)$ by the deterministic value of the event $\{\theta<x\}$ whenever this event is $G$-almost surely constant. Thus, we set $m_G(x)=1$ if $G([0,x))=1$ and $m_G(x)=0$ if $G([x,\infty))=1$. This convention is used only outside the relevant finite window, where the posterior event is deterministic under the support constraint.
\end{definition}

\begin{definition}[Smoothing function]\label{def_smooth}
Fix $\delta\in(0,1]$. Define $\phi_\delta:\mathbb R\to[0,1]$ by
\begin{displaymath}
\phi_\delta(u):=
\begin{cases}
1, & u\ge \delta,\\[1mm]
\frac{u+\delta}{2\delta}, & |u|<\delta,\\[2mm]
0, & u\le -\delta.
\end{cases}
\end{displaymath}
\end{definition}

Then $0\le \phi_\delta\le 1$, $\phi_\delta$ is $(2\delta)^{-1}$-Lipschitz, and $\phi_\delta(x-\theta)=I(\theta<x)\quad\text{whenever }|x-\theta|\ge \delta$. Hence, for every $x\in\integers_+$ and every $\theta\in\reals_+$,
\begin{equation}\label{eq:phi-pointwise-error-main}
\big|\phi_\delta(x-\theta)-I(\theta<x)\big|\le I(|\theta-x|<\delta).
\end{equation}
Combining Definitions \ref{def_ratio}--\ref{def_smooth}, we define the smoothed conditional above-average probability $m_{G,\delta}(x):=\frac{A_{x,\delta}(G)}{B_x(G)}$, 
where
\begin{displaymath}
A_{x,\delta}(G):=\int_{\reals_+} e^{-\theta}\theta^x\,\phi_\delta(x-\theta)\,dG(\theta).
\end{displaymath}
The assumptions below are imposed on a finite window of integer values. In the compact-support case, this window is $\mathcal W=\{0,\ldots,\hat K\}$, with $\hat K=\lfloor \hat h\rfloor$, whereas in the subexponential case it is $\mathcal W=\{0,\ldots,K\}$, where $K$ is a truncation level. We write $\mathcal W_+=\mathcal W\cap\{1,2,\ldots\}$. Sums over $\mathcal W_+$ are understood to be empty when $\mathcal W_+=\varnothing$.

Let $\mathcal C$ denote the feasible class of mixing distributions used in the regular or coarsened minimum-distance optimization. Throughout the finite-sample analysis, we assume that $\Gstar\in\mathcal C$ and that the corresponding minimum-distance estimate belongs to $\mathcal C$ almost surely.

\begin{myassumption}[Local mass bound at integer points]\label{asmp:aa-local-mass}
There exists $L>0$ such that $G([x-\delta,x+\delta]\cap\mathrm{supp}(G))\leq L\delta$ for all $G\in\mathcal C$, all $\delta\in(0,1]$, and all $x\in\mathcal W$.
\end{myassumption}

The next assumption is specific to regular minimum-distance estimation. It requires that the smoothed numerator $A_{x,\delta}(G)$ can be approximated by a linear functional of the full mixture probability mass function $f_G$.

\begin{myassumption}[Self-normalized representation: regular case]\label{asmp:aa-self-normal-regular}
For each $x\in\mathcal W_+$ and each $\delta\in(0,1]$, there exist coefficients 
$\alpha_{x,\delta}(t)$, $t\in\integers_+$, and a remainder bound 
$\varepsilon_{x,\delta}\geq0$ such that, for all $G\in\mathcal C$,
\begin{equation}\label{eq:A-rep-regular-aa}
\left|A_{x,\delta}(G)-\sum_{t\geq0}\alpha_{x,\delta}(t)f_G(t)\right|\leq\varepsilon_{x,\delta},
\end{equation}
and
\begin{equation}\label{eq:A-rep-eps-regular-aa}
\varepsilon_{x,\delta}\leq C_\varepsilon |\calW_+| \delta,\qquad\delta\in(0,1],
\end{equation}
uniformly over $x\in\mathcal W_+$, with constants depending only on the structural parameters of the class $\mathcal C$.
Moreover, for all $G_1,G_2\in\mathcal C$,
\begin{equation}\label{eq:A-selfnorm-regular-aa}
\sum_{t\geq0}\alpha_{x,\delta}(t)^2\left(f_{G_1}(t)+f_{G_2}(t)\right)\leq C_A\left(A_{x,\delta}(G_1)+A_{x,\delta}(G_2)\right).
\end{equation}
\end{myassumption}

For the coarsened minimum-distance estimation, the analogue of Assumption \ref{asmp:aa-self-normal-regular} is imposed with respect to the coarsened mixture probability mass function. In particular, recall that $\tilde f_G^{(T)}$ denotes the $T$-coarsening of $f_G$.

\begin{myassumption}[Self-normalized representation: coarsened case]\label{asmp:aa-self-normal-coarse}
For each $x\in\mathcal W_+$ and each $\delta\in(0,1]$, there exist coefficients $\alpha_{x,\delta}^{(T)}(t)$, $t=0,\ldots,T+1$, and a remainder bound $\varepsilon_{x,\delta}^{(T)}\geq0$ such that, for all $G\in\mathcal C$,
\begin{equation}\label{eq:A-rep-coarse-aa}
\left|A_{x,\delta}(G)-\sum_{t=0}^{T+1}\alpha_{x,\delta}^{(T)}(t)\tilde f_G^{(T)}(t)\right|\leq\varepsilon_{x,\delta}^{(T)},
\end{equation}
and
\begin{equation}\label{eq:A-rep-eps-coarse-aa}
\varepsilon_{x,\delta}^{(T)}\leq C_\varepsilon h(T,|\calW_+|) \delta,\qquad\delta\in(0,1],
\end{equation}
uniformly over $x\in\mathcal W_+$, where $C_\varepsilon$ depends only on the structural parameters of the class $\mathcal C$, and $h(T,|\mathcal W_+|)>0$ satisfies $h(T,|\mathcal W_+|)\leq|\mathcal W_+|$. Moreover, for all $G_1,G_2\in\mathcal C$,
\begin{equation}\label{eq:A-selfnorm-coarse-aa}
\sum_{t=0}^{T+1}\left(\alpha_{x,\delta}^{(T)}(t)\right)^2\left(\tilde f_{G_1}^{(T)}(t)+\tilde f_{G_2}^{(T)}(t)\right)\leq C_A\left(A_{x,\delta}(G_1)+A_{x,\delta}(G_2)\right),
\end{equation}
with $C_A$ depending only on the structural parameters of the class $\mathcal C$.
\end{myassumption}

The next theorem provides finite-sample upper bounds for the EB estimate $\hat S_{n,>}(\hat G_{d,m})$ of $S_{n,>}$, where $\hat G_{d,m}$ is the regular minimum-$d$-distance estimate of $\Gstar$ in \eqref{eq:optim-genr}. The result relies on the finite-sample guarantees in Lemma \ref{lem:EB_density-H2}.

\begin{theorem}\label{thm:main-above-average}
Let $\Gstar$ denote the data-generating mixing distribution. Consider the regular minimum-$d$-distance estimate $\hat G_{d,m}$ of $\Gstar$ obtained from an independent training sample of size $m$.
\begin{enumerate}[label=(\alph*)]
\item Let $\hat h>0$ and $\hat K=\lfloor\hat h\rfloor$. Let $\mathcal C\subseteq\mathcal P([0,\hat h])$ be a convex and weakly closed class such that \prettyref{asmp:aa-local-mass} and
\prettyref{asmp:aa-self-normal-regular} hold on the window $\{0,\ldots,\hat K\}$. Assume that $\Gstar\in\mathcal P([0,h])\cap\mathcal C$ for some $h\leq\hat h$. Then there exist constants $C_0,C_1$, depending only on $(\hat h,L)$ and the constants in \prettyref{asmp:aa-self-normal-regular}, such that, for every $\delta\in(0,1]$ and every $m\geq3$,
\begin{align}\label{eq:main-compact-result}
&\EE\!\left[\left(\frac{\hat S_{n,>}(\hat G_{d,m})-\hat S_{n,>}(\Gstar)}{\sqrt n}\right)^2\right]\\
&\notag\quad\le nC_0\delta+nC_1 C_d(\hat h,c_1,c_2)\,\frac{1}{m}\min\!\left\{\max\{1,\hat h\}\frac{\log m}{\log\log m},
\hat h+\log m\right\}.
\end{align}
\item Let $K\geq0$ be an integer. Let $\mathcal C\subseteq\subexpo(s)$ be a convex and weakly closed class such that \prettyref{asmp:aa-local-mass} and \prettyref{asmp:aa-self-normal-regular}
hold on the window $\{0,\ldots,K\}$. Assume that $\Gstar\in\mathcal C$. Then there exist constants $C_0,C_1$, depending only on $L$ and the constants in \prettyref{asmp:aa-self-normal-regular}, such that, for every $\delta\in(0,1]$ and every $m\geq2$,
\begin{equation}\label{eq:main-sube-result}
\EE\!\left[\left(\frac{\hat S_{n,>}(\hat G_{d,m})-\hat S_{n,>}(\Gstar)}{\sqrt n}\right)^2\right]\le nC_0K\delta+nC_1 C_d(s,c_1,c_2)\frac{\log m}{m}+nC_s e^{-c_sK},
\end{equation}
where one may take $c_s:=\log\left(1+\frac{1}{2s}\right)$ and $C_s:=3$.
\end{enumerate}
\end{theorem}

The proof is given in Section \ref{proof_teo_above_regular}. Before turning to the coarsened procedure, we discuss the implications of Theorem \ref{thm:main-above-average}. In the compact-support setting of Theorem~\ref{thm:main-above-average}, choosing $\delta=1/n$ gives
\begin{align*}
&\EE\!\left[\left(\frac{\hat{S}_{n,>}(\hat{G}_{d,m})-\hat{S}_{n,>}(\Gstar)}{\sqrt{n}}\right)^2\right]\\
&\quad\le C_0+C_1\,C_d(\hat h,c_1,c_2)\,\frac{n}{m}\min\!\left\{\max\{1,\hat h\}\frac{\log m}{\log\log m},\hat h+\log m\right\}.
\end{align*}
In particular, if $m\asymp n$, then
\begin{displaymath}
\EE\!\left[\left(\frac{\hat{S}_{n,>}(\hat{G}_{d,m})-\hat{S}_{n,>}(\Gstar)}{\sqrt{n}}\right)^2\right]\le C_0+C_1\,C_d(\hat h,c_1,c_2)\,\min\!\left\{\max\{1,\hat h\}\frac{\log n}{\log\log n},\hat h+\log n\right\}.
\end{displaymath}
Hence, for fixed $\hat h$,
\begin{displaymath}
\EE\!\left[\left(\frac{\hat{S}_{n,>}(\hat{G}_{d,m})-\hat{S}_{n,>}(\Gstar)}{\sqrt{n}}\right)^2\right]=O\!\left(1+\frac{\log n}{\log\log n}\right).
\end{displaymath}
Moreover, a uniformly bounded upper bound follows by taking $n\delta=O(1)$ and $\frac{nR_{\hat h}(m)}{m}=O(1)$, where we defined
\begin{displaymath}
R_{\hat h}(m):=\min\!\left\{\max\{1,\hat h\}\frac{\log m}{\log\log m},\hat h+\log m\right\}.
\end{displaymath}
For example, by taking $\delta=1/n$, it suffices that $m\gtrsim nR_{\hat h}(m)$; in particular, for any fixed $\hat h$, this is implied by
\begin{displaymath}
m\gtrsim n\frac{\log m}{\log\log m}.
\end{displaymath}
Finally, consider the subexponential case in Theorem~\ref{thm:main-above-average}. Let
\begin{displaymath}
K_n:=\left\lceil\frac{\log n}{c_s}\right\rceil.
\end{displaymath}
Taking $\delta=1/n$ gives
\begin{displaymath}
\EE\!\left[\left(\frac{\hat S_{n,>}(\hat G_{d,m})-\hat S_{n,>}(\Gstar)}{\sqrt n}\right)^2\right]\le C_0K_n+C_1C_d(s,c_1,c_2)\frac{n\log m}{m}+C_s.
\end{displaymath}
Hence, if $m\asymp n$, the bound is $O(\log n)$. On the other hand, taking $\delta=1/(nK_n)$ gives a uniformly bounded first term, and therefore $m\gtrsim n\log n$ yields an $O(1)$ bound. 

We next turn to the $T$-coarsened minimum-distance procedure. The following theorem provides finite-sample upper bounds for the EB estimate $\hat S_{n,>}(\widetilde G_{d,m}^{(T)})$, relying on the corresponding finite-sample guarantees in Lemma \ref{lem:EB_density-H2_coars}.

\begin{theorem}\label{thm:main-above-average-coarsened}
Let $\Gstar$ denote the data-generating mixing distribution. Consider the $T$-coarsened minimum-$d$-distance estimate $\widetilde G_{d,m}^{(T)}$ of $\Gstar$ obtained from an independent training sample of size $m$.
\begin{enumerate}[label=(\alph*)]
\item Let $\hat h>0$, \(\hat K=\lfloor\hat h\rfloor\) and $T\geq \hat  K$. Let $\mathcal C\subseteq\mathcal P([0,\hat h])$ be a convex and weakly closed class such that \prettyref{asmp:aa-local-mass} and \prettyref{asmp:aa-self-normal-coarse} hold on the window $\{0,\ldots,\hat K\}$. Assume that $\Gstar\in \mathcal P([0,h])\cap\mathcal C $ for some $h\leq \hat h$. Then there exist constants \(C_0,C_1\), depending only on $(\hat h,T,L)$ and the constants in \prettyref{asmp:aa-local-mass} and \prettyref{asmp:aa-self-normal-coarse}, such that, for every \(\delta\in(0,1]\) and every $m\ge3$,
\begin{equation}\label{eq:main-coarse-compact-result}
\EE\!\left[\left(\frac{\hat{S}_{n,>}(\tilde {G}^{(T)}_{d,m})-\hat{S}_{n,>}(\Gstar)}{\sqrt{n}}\right)^2\right]\le nC_0\delta+nC_1 C_d^{\mathrm{coarse}}(c_1,c_2)\frac{T+2}{m},
\end{equation}
where $C_d^{\mathrm{coarse}}(c_1,c_2)=\dfrac{4c_2}{c_1}$.
\item Let $K\geq 0$ and $T\geq K$. Let $\mathcal C\subseteq\subexpo(s)$ be a convex and weakly closed class such that \prettyref{asmp:aa-local-mass} and
\prettyref{asmp:aa-self-normal-coarse} hold on the window $\{0,\ldots,K\}$. Assume that $\Gstar\in\mathcal C $. Then there exist constants \(C_0,C_1\), depending only on $(s,L)$ and the constants in \prettyref{asmp:aa-self-normal-coarse}, such that, for every \(\delta\in(0,1]\) and every $m\geq \max\{2,\exp\left(\frac 1 2 \log(1+\frac{1}{2s})\right)\}$,
\begin{align}\label{eq:main-coarse-sube-result}
&\EE\!\left[\left(\frac{\hat{S}_{n,>}(\tilde {G}^{(T)}_{d,m})-\hat{S}_{n,>}(\Gstar)}{\sqrt{n}}\right)^2\right]\\
&\notag\quad\le nKC_0\delta+nC_s e^{-c_sK}+n C_1C_d^{\mathrm{coarse}}(c_1,c_2)\min\{\frac{T+2}{m},{\log m\over m}\},
\end{align}
where one may take $c_s=\log\left(1+\frac1{2s}\right)$ and $C_s=3$. 
\end{enumerate}
\end{theorem}

The proof is given in Section \ref{proof:main-above-average-coarsened}. We next discuss the implications of Theorem \ref{thm:main-above-average-coarsened}. For the compact-support case, taking \(T=\hat K\) and \(\delta=1/n\) gives
\begin{displaymath}
\EE\!\left[\left(\frac{\hat{S}_{n,>}(\tilde {G}^{(T)}_{d,m})-\hat{S}_{n,>}(\Gstar)}{\sqrt{n}}\right)^2\right]\le C_0+C_1 C_d^{\mathrm{coarse}}(c_1,c_2)\frac{n(\hat K+2)}{m}.
\end{displaymath}
Thus, for fixed \(\hat h\), the bound is \(O(1)\) when \(m\asymp n\). Finally, consider the subexponential case in Theorem~\ref{thm:main-above-average-coarsened}. Taking
$K=K_n=\left\lceil \frac{\log n}{c_s}\right\rceil$, $T=K_n$, and $\delta=1/(nK_n)$ gives
\begin{displaymath}
\EE\!\left[\left(\frac{\hat{S}_{n,>}(\tilde {G}^{(T)}_{d,m})-\hat{S}_{n,>}(\Gstar)}{\sqrt{n}}\right)^2\right]\le C_0+C_1C_d^{\mathrm{coarse}}(c_1,c_2)\frac{n\log n}{m}+C_s.
\end{displaymath}
If the localized constants can be chosen uniformly in \(K\), then \(m\asymp n\) gives an \(O(\log n)\) bound, whereas \(m\gtrsim n\log n\) gives an \(O(1)\) bound.

These upper bounds should be contrasted with the minimax lower bound in Theorem \ref{above_average_lb}. Without the additional structural conditions used in Theorems \ref{thm:main-above-average} and \ref{thm:main-above-average-coarsened}, the risk is bounded from below by a constant multiple of $\frac{n}{\sqrt{m+n}}$. Thus, when $m\asymp n$, the minimax risk diverges at least at order $\sqrt n$, showing that uniformly bounded risk at the parametric scale is impossible over the general class considered in Theorem \ref{above_average_lb}. More generally, this lower bound becomes of constant order only when the training sample is of order at least $n^2$. Hence, without additional structure, attaining the parametric scale may require a training sample substantially larger than the sample used to define the target sum, further illustrating the intrinsic difficulty of estimating the above-average sum.

Under the additional structural conditions imposed in Theorems \ref{thm:main-above-average} and \ref{thm:main-above-average-coarsened}, substantially sharper guarantees become possible. In particular, in the compact-support setting, the coarsened procedure achieves an $O(1)$ risk for fixed $\hat h$ already when $m\asymp n$. Thus, these structural conditions fundamentally change the statistical difficulty of the problem, allowing estimation at the parametric scale with a training sample of the same order as the sample used to define the target sum.

\subsection{Discrete mixing distribution}

When the mixing distribution is supported on a finite set of known locations, the structural conditions required above can be verified directly, leading to the following finite-sample guarantees.

\begin{corollary}[Regular and coarsened minimum-distance estimation on \(\mathcal P(\Theta_J)\)]\label{corr:above-average-discrete-feasible}
Set $K_J:=\lfloor \theta_J\rfloor$ and $T_J:=\max\{K_J,J-1\}$. Let $\mathcal C\subseteq\mathcal P(\Theta_J)$ be a convex and weakly closed class such that $\min_{0\le x\le K_J} f_G(x)\ge b_J$ for every $G\in\mathcal C$, where $b_J>0$. Assume that $d$ satisfies Assumptions \ref{asmp:EB_t-ell}--\ref{asmp:dist_sandwich}, and let $\Gstar\in \mathcal C$. Consider the regular minimum-$d$-distance estimate $\hat G_{d,m}$ of $\Gstar$ and the $T$-coarsened minimum-$d$-distance estimate $\widetilde G_{d,m}^{(T)}$ of $\Gstar$ obtained from a training sample of size $m$.
\begin{enumerate}[label=(\alph*)]
\item There exists a constant $C_{\Theta_J}<\infty$, depending only on $(\Theta_J,b_J,c_1,c_2)$, such that for every $m\geq3$,
\begin{equation}\label{eq:above-average-discrete-regular}
\EE\!\left[\left(\frac{\hat S_{n,>}(\hat G_{d,m})-\hat S_{n,>}(\Gstar)}{\sqrt n}\right)^2\right]\le C_{\Theta_J}\,\frac{n\log m}{m}.
\end{equation}
\item Let \(T\ge T_J\). Then there exists a constant \(C_{{\Theta_J},T}<\infty\), depending only on
\(({\Theta_J},T,b_J,c_1,c_2)\), such that for every $m\geq 3$,
\begin{equation}\label{eq:above-average-discrete-coarsened}
\EE\!\left[\left(\frac{\hat S_{n,>}(\widetilde G_{d,m}^{(T)})-\hat S_{n,>}(\Gstar)}
{\sqrt n}\right)^2\right]\le C_{{\Theta_J},T}\,\frac{n(T+2)}{m}.
\end{equation}
\end{enumerate}
\end{corollary}

The proof is given in Section \ref{subsec:proof-above-average-discrete}. For fixed $\Theta_J$ and $T$, Corollary \ref{corr:above-average-discrete-feasible} shows that the coarsened minimum-distance procedure achieves the parametric rate $n/m$, while the regular minimum-distance procedure achieves the same rate up to a logarithmic factor. In particular, when $m\asymp n$, the coarsened procedure has uniformly bounded risk at the $\sqrt n$ scale. Thus, finite support provides sufficient structure to recover parametric-rate estimation of the above-average sum, in sharp contrast with the general setting of Theorem \ref{above_average_lb}.


\section{Numerical experiments}\label{sec5}
We illustrate the proposed minimum-distance procedures through numerical experiments, comparing their empirical performance with that of benchmark
parametric Bayes, parametric maximum-likelihood, and, whenever available, nonparametric ``$u,v$'' estimates \citep[Section 2]{Zha(05)}. 

\subsection{Preliminaries}

To mirror the theoretical setting, we distinguish between an independent training sample, used to estimate the mixing distribution $\Gstar$, and a main sample, on which the target sums of random variables are defined. Let $X'_1,\ldots,X'_m$ denote the training sample, generated according to
\begin{align*}
X'_j\mid\theta'_j &\quad\simind\quad\operatorname{Poisson}(\cdot\mid\theta'_j),\qquad j=1,\ldots,m,\\
\theta'_j &\quad\simiid\quad \Gstar.
\end{align*}
Independently of the training sample, let $(X_1,Y_1,\theta_1),\ldots,(X_n,Y_n,\theta_n)$ denote the main sample, with
\begin{align}\label{model_distr_poi}
X_i\mid\theta_i &\quad\simind\quad\operatorname{Poisson}(\cdot\mid\theta_i),\qquad i=1,\ldots,n,\\
\notag Y_i\mid\theta_i &\quad\simind\quad\operatorname{Poisson}(\cdot\mid\theta_i),\\
\notag\theta_i &\quad\simiid\quad \Gstar,
\end{align}
where $X_i$ and $Y_i$ are conditionally independent given $\theta_i$. Thus, both samples are generated from the same $\Gstar$, but only $X'_{1:m}=(X'_1,\ldots,X'_m)$ is used to estimate $\Gstar$. In practice, sample splitting is not necessary for implementing the proposed procedures: one may use the full sample both to estimate the mixing distribution and to construct the corresponding empirical Bayes estimate. As a robustness check, we also considered this implementation and found no appreciable differences in empirical performance. We therefore report only the sample-splitting results.

Given the observed main sample $X_{1:n}=(X_{1},\ldots,X_{n})$, we focus on the total-intensity sum $S_{n,\kappa}$, with $\kappa\in\integers_+$, and the above-average sum $S_{n,>}$. The regular minimum-$d$-distance estimate $\hat G_{d,m}$ of the mixing distribution $\Gstar$ is computed from the training sample $X'_{1:m}$. EB estimates of the target sums are then obtained by plugging the estimate $\hat G_{d,m}$ into the Bayes rules evaluated on the main sample $X_{1:n}$. In particular, under \eqref{model_distr_poi}, EB estimates of $S_{n,\kappa}$ and $S_{n,>}$ are
\begin{equation}\label{eq:rmd_total}
\hat{S}_{n,\kappa}(\hat G_{d,m})=\sum_{i=1}^n\E_{\hat G_{d,m}}[\theta_i\mid X_i]I(X_i\leq\kappa)=
\sum_{i=1}^{n}I(X_i\leq\kappa)\int_{\Theta}\theta\,\hat G_{d,m}(d\theta\mid X_i)
\end{equation}
and
\begin{equation}\label{eq:rmd_above}
\hat{S}_{n,>}(\hat G_{d,m})=\sum_{i=1}^n\E_{\hat G_{d,m}}\!\left[I(X_i>\theta_i)\mid X_i\right]=\sum_{i=1}^n\int_{\Theta}I(X_i>\theta)\,\hat G_{d,m}(d\theta\mid X_i),
\end{equation}
respectively. The $T$-coarsened estimates of the total-intensity and above-average sums are defined analogously. Specifically, the $T$-coarsened minimum-$d$-distance estimate $\widetilde G_{d,m}^{(T)}$ of the mixing distribution $\Gstar$ is computed from the training sample $X'_{1:m}$ and is then substituted for $\hat G_{d,m}$ in \eqref{eq:rmd_total} and \eqref{eq:rmd_above}.

\subsection{Benchmark procedures}

\subsubsection{Parametric Bayes and maximum-likelihood estimates}

To construct parametric competitors of our minimum-$d$-distance estimates, we follow \citet[Section 2.3]{Zha(05)} and assume that the mixing distribution $\Gstar$ in \eqref{model_distr_poi} is an Exponential distribution with rate $\tau>0$. The next proposition recalls the corresponding parametric Bayes and maximum-likelihood estimates of $S_{n,\kappa}$ and $S_{n,>}$, as well as those of their predictive counterparts.

\begin{proposition}\label{prop:poisson}
Let $m,n\geq1$. Let $X'_1,\ldots,X'_m$ be an independent training sample, with $X'_j\mid\theta'_j\simind \operatorname{Poisson}(\cdot\mid\theta'_j)$ and $\theta'_j\simiid \Gstar_\tau$ for $j=1,\ldots,m$, and let $(X_{1},Y_{1},\theta_{1}),\ldots,(X_{n},Y_{n},\theta_{n})$ be the main sample distributed as in \eqref{model_distr_poi}, independently of $X'_{1:m}=(X'_1,\ldots,X'_m)$. Assume that $\Gstar_\tau$ is an Exponential distribution with rate $\tau$. Then the oracle Bayes estimates of $S_{n,\kappa}$ and
$S^{\text{\tiny{[Pred]}}}_{n,\kappa}$ are
\begin{equation}\label{eq:par_tot}
\hat{S}_{n,\kappa}(\Gstar_\tau)= \hat{S}^{\text{\tiny{[Pred]}}}_{n,\kappa}(\Gstar_\tau)=\sum_{i=1}^{n}\frac{1+X_i}{1+\tau}I(X_i\leq\kappa),
\end{equation}
whereas the oracle Bayes estimates of $S_{n,>}$ and $S^{\text{\tiny{[Pred]}}}_{n,>}$ are
\begin{equation}\label{eq:par_above}
\hat{S}_{n,>}(\Gstar_\tau)=n-\sum_{i=1}^{n}\frac{\Gamma(1+X_i,(1+\tau)X_i)}{\Gamma(1+X_i)},
\end{equation}
and
\begin{equation}\label{eq:par_above_pred}
\hat{S}^{\text{\tiny{[Pred]}}}_{n,>}(\Gstar_\tau)=\sum_{i=1}^{n}\left[1-(1+\tau)^{-X_i}\right],
\end{equation}
respectively, where $\Gamma(x)=\int_{\mathbb R_+}y^{x-1}e^{-y}\,dy$ and $\Gamma(x,z)=\int_z^\infty y^{x-1}e^{-y}\,dy$. The rate parameter $\tau$ is estimated from the training sample $X'_{1:m}$. In particular, the maximum-likelihood estimate of $\tau$ is
\begin{displaymath}
\hat{\tau}^{\text{\tiny{[ML]}}}_{m}=\frac{m}{\sum_{j=1}^{m}X'_j}.
\end{displaymath}
Accordingly, the resulting parametric maximum-likelihood estimates of $S_{n,\kappa}$ and $S_{n,>}$, together with those of their predictive counterparts, are obtained by replacing $\tau$ in \eqref{eq:par_tot}--\eqref{eq:par_above_pred} with its estimate $\hat{\tau}^{\text{\tiny{[ML]}}}_{m}$. If, instead, a Beta prior distribution with hyperparameters $(\alpha,\beta)$ is placed on $\tau/(1+\tau)$, then the corresponding Bayes estimate of $\tau$ based on the training sample $X'_{1:m}$ is
\begin{displaymath}
\hat{\tau}^{\text{\tiny{[B]}}}_{m}=\frac{m+\alpha}{\beta+\sum_{j=1}^{m}X'_j}.
\end{displaymath}
Accordingly, the resulting parametric Bayes estimates of $S_{n,\kappa}$ and $S_{n,>}$, together with those of their predictive counterparts, are obtained by replacing $\tau$ in \eqref{eq:par_tot}--\eqref{eq:par_above_pred} with its estimate $\hat{\tau}^{\text{\tiny{[B]}}}_{m}$.
\end{proposition}

\subsubsection{Nonparametric ``$u,v$'' estimates}

Unlike the $g$-modeling procedures considered above, the ``$u,v$'' method does not estimate the mixing distribution $\Gstar$ and therefore does not require an independent training sample. It is applied directly to the main sample $X_{1:n}$. The ``$u,v$'' method of \citet{Rob(88)} provides a nonparametric competitor to our minimum-$d$-distance procedures; see also \citet[Section 2.3]{Zha(05)} and references therein. For a generic sum of random variables $S_n$, the ``$u,v$'' method estimates $S_n$ by
\begin{displaymath}
\hat S_n^{\text{\tiny[``u,v'']}}=\sum_{i=1}^n v(X_i),
\end{displaymath}
provided there exists a function $v$ satisfying the ``$u,v$'' equation
\begin{equation}\label{int_eq1}
\sum_{x\geq0}\bigl(v(x)-u(x,\theta)\bigr)\frac{e^{-\theta}\theta^x}{x!}=0,\qquad \forall\,\theta\in\Theta.
\end{equation}
For the total-intensity sums $S_{n,\kappa}$ and $S^{\text{\tiny[Pred]}}_{n,\kappa}$, equation \eqref{int_eq1} admits the unique solution $v(x)=xI(x-1\leq\kappa)$ \citep{Rob(88),Rob(00)}. Hence, the corresponding ``$u,v$'' estimates are
\begin{equation}\label{uv_poiss}
\hat S_{n,\kappa}^{\text{\tiny[``u,v'']}}=\hat S_{n,\kappa}^{\text{\tiny[Pred,``u,v'']}}=\sum_{i=1}^{n}X_i I(X_i-1\leq\kappa).
\end{equation}
For the above-average sum $S_{n,>}$ and $S^{\text{\tiny[Pred]}}_{n,>}$, instead, the ``$u,v$'' equation \eqref{int_eq1} does not admit an explicit solution $v$, so that a corresponding explicit ``$u,v$'' estimate is not available \citep[Section 2.3]{Zha(05)}.

\subsection{Synthetic-data analysis}

\subsubsection{Weibull mixing distribution}

For $i=1,\ldots,100$, let $\mathbf X_i=X_{1:100i}$ denote a dataset of total size $m+n=100i$ generated from the Poisson mixture model with mixing distribution $\Gstar$ given by a Weibull distribution with scale $5$ and shape $3$. The datasets $\mathbf X_i$ are nested, so that their size increases progressively from $100$ to $10,000$ by adding $100$ new observations at each stage. For each $i$, we split $\mathbf X_i$ into two equally sized subsamples, with $m=n=50i$. The first $m$ observations form the training sample and are used to estimate the mixing distribution $\Gstar$, whereas the remaining $n$ observations form the main sample on which the target sums are evaluated. For notational consistency with the previous sections, we denote the training sample by $X'_{1:m}$ and relabel the second half as the main sample $X_{1:n}$. For each training sample $X'_{1:m}$, we compute the regular minimum-$d$-distance estimate $\widehat G_{d,m}$ of $\Gstar$ for two choices of the statistical distance $d$: the Hellinger (H) distance and the Kullback--Leibler (KL) divergence. We also compute the corresponding $T$-coarsened minimum-$d$-distance estimate $\widetilde G_{d,m}^{(T)}$. The coarsening level $T$ is chosen according to the distribution of the observations. Letting $q_X(p)=\inf\{t\in\integers_+:\mathbb P(X\leq t)\geq p\}$ and taking the maximum total sample size to be $10,000$, we consider the coarsening levels
\begin{displaymath}
T\in\{2,q_X(0.95),q_X(1-1/10,000)\}=\{2,9,18\}.
\end{displaymath}
This choice for the construction of $T$ keeps the coarsening levels fixed across sample sizes while allowing for different degrees of information loss, ranging from substantial coarsening at $T=2$ to a truncation that retains almost the entire Poisson mixture distribution at $T=18$.

The minimum-$d$-distance estimates of the mixing distribution $\Gstar$, both regular $\widehat G_{d,m}$ and $T$-coarsened $\widetilde G_{d,m}^{(T)}$, are computed from the training samples using the vertex direction method (VDM) described in \citet[Section 5]{jana2025optimal}. The algorithm exploits the fact that the minimum-distance estimate is a discrete distribution and iteratively updates both its support points and their associated masses. Starting from an initial discrete approximation, candidate support points are generated from the first-order optimality conditions for the minimum-$d$-distance criterion. In the Poisson mixture model, these conditions reduce to a polynomial equation whose roots provide candidate locations for the support. The masses are then re-optimized over the enlarged support, nearby support points are merged, and atoms carrying negligible mass are removed. The procedure is iterated until numerical convergence; we refer to \citet[Section 5]{jana2025optimal} for implementation details. See \citet{Sim(76)} for a similar VDM-type algorithm  and \citet{Koe(14)} for discretization-based algorithms.

The minimum-$d$-distance estimates $\widehat G_{d,m}$ and $\widetilde G_{d,m}^{(T)}$ are then plugged into the corresponding Bayes rules and evaluated on the independent main samples. In particular, $\widehat G_{d,m}$ yields the regular minimum-$d$-distance estimates of $S_{n,\kappa}$ and $S_{n,>}$ given in \eqref{eq:rmd_total} and \eqref{eq:rmd_above}, respectively, whereas replacing $\widehat G_{d,m}$ with $\widetilde G_{d,m}^{(T)}$ gives the corresponding $T$-coarsened estimates. Throughout the experiment, we set $\kappa=2$. We compare these estimates with three benchmarks: (i) the oracle Bayes estimates $\hat S_{n,\kappa}(\Gstar)$ and $\hat S_{n,>}(\Gstar)$, computed using the true Weibull mixing distribution; (ii) the parametric Bayes and maximum-likelihood estimates of Proposition~\ref{prop:poisson}, with the unknown parameter estimated from the same training sample $X'_{1:m}$ used by the minimum-distance procedures; and (iii), whenever available, the nonparametric ``$u,v$'' estimate, which is computed directly from the main sample $X_{1:n}$ and does not require a training sample.

Figures~\ref{fig_weib_s1_reg}--\ref{fig_weib_s1_coa} report the regular and $T$-coarsened minimum-$d$-distance estimates of $S_{n,2}$, whereas Figures~\ref{fig_weib_s3_reg}--\ref{fig_weib_s3_coa} report the corresponding estimates of $S_{n,>}$. All estimates, as well as the true values, are normalized by the main-sample size $n$. The figures also report the mean absolute deviation (MAD) from the true normalized sums across the $100$ sample sizes considered. More precisely, at stage $i=1,\ldots,100$, let $n_i=50i$ denote the size of the main sample, and let $\widehat S_{n_i,2}$ and $\widehat S_{n_i,>}$ denote generic estimates of the total-intensity and above-average sums, respectively. Denoting the corresponding true values by $S_{n_i,2}$ and $S_{n_i,>}$, we define
\begin{displaymath}
\mathrm{MAD}(\widehat S_{2}):=\frac{1}{100}\sum_{i=1}^{100}\left|\frac{\widehat S_{n_i,2}}{n_i}-\frac{S_{n_i,2}}{n_i}\right|,
\end{displaymath}
and
\begin{displaymath}
\mathrm{MAD}(\widehat S_{>}):=\frac{1}{100}\sum_{i=1}^{100}\left|\frac{\widehat S_{n_i,>}}{n_i}-\frac{S_{n_i,>}}{n_i}\right|.
\end{displaymath}
Thus, the MAD measures the average absolute error of the normalized estimates over increasing sample sizes.

\begin{figure}[H]
\centering
\includegraphics[width=1\textwidth]{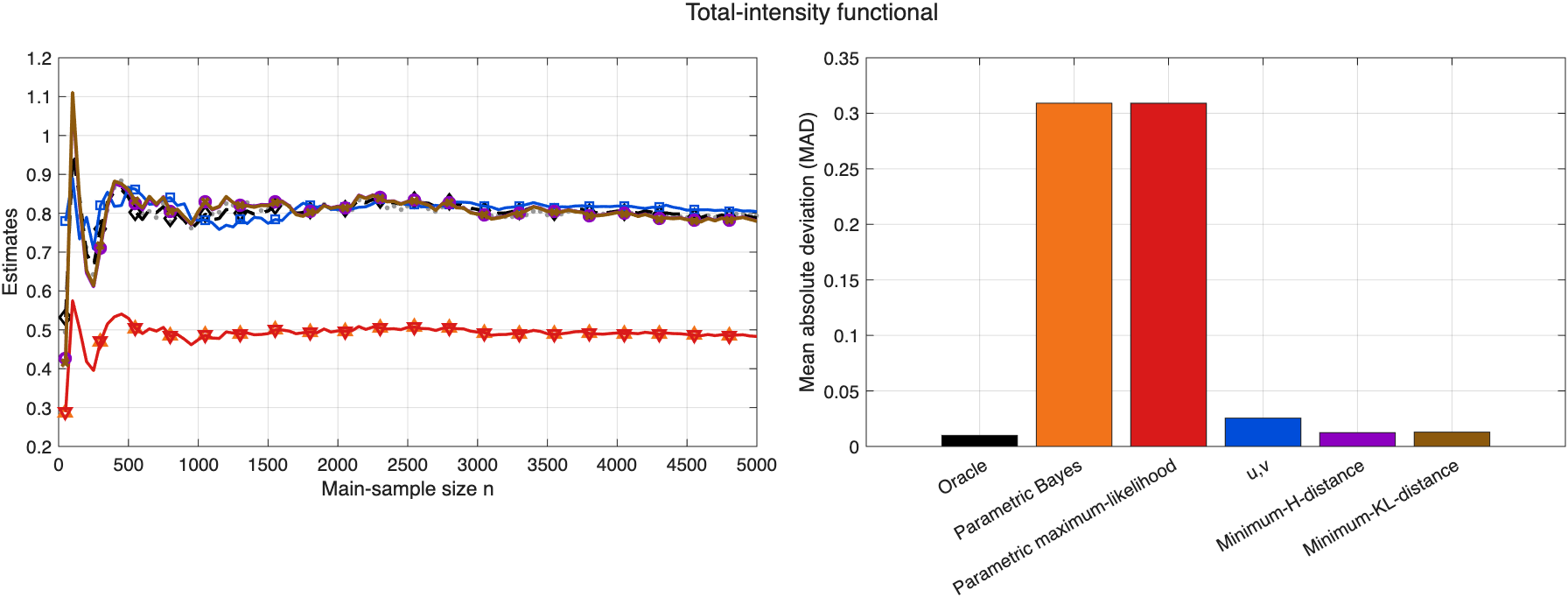}
\caption{\scriptsize{Total-intensity sum under the Weibull mixing distribution: regular minimum-$d$-distance estimation. Left: true normalized values of $S_{n,2}$ (gray $\bullet$) and normalized estimates as a function of the main-sample size $n$. The oracle estimate is shown by a black dashed line with $\diamond$, the parametric Bayes estimate by an orange line with $\triangle$, the parametric maximum-likelihood estimate by a red line with $\triangledown$, the $u,v$ estimate by a blue line with $\square$, and the minimum-H- and minimum-KL-distance estimates by purple $\circ$ and brown $\times$, respectively. Right: mean absolute deviation (MAD) of the normalized estimates from the true normalized values of  $S_{n,2}$.}}
\label{fig_weib_s1_reg}
\end{figure}

\begin{figure}[H]
\centering
\includegraphics[width=1\textwidth]{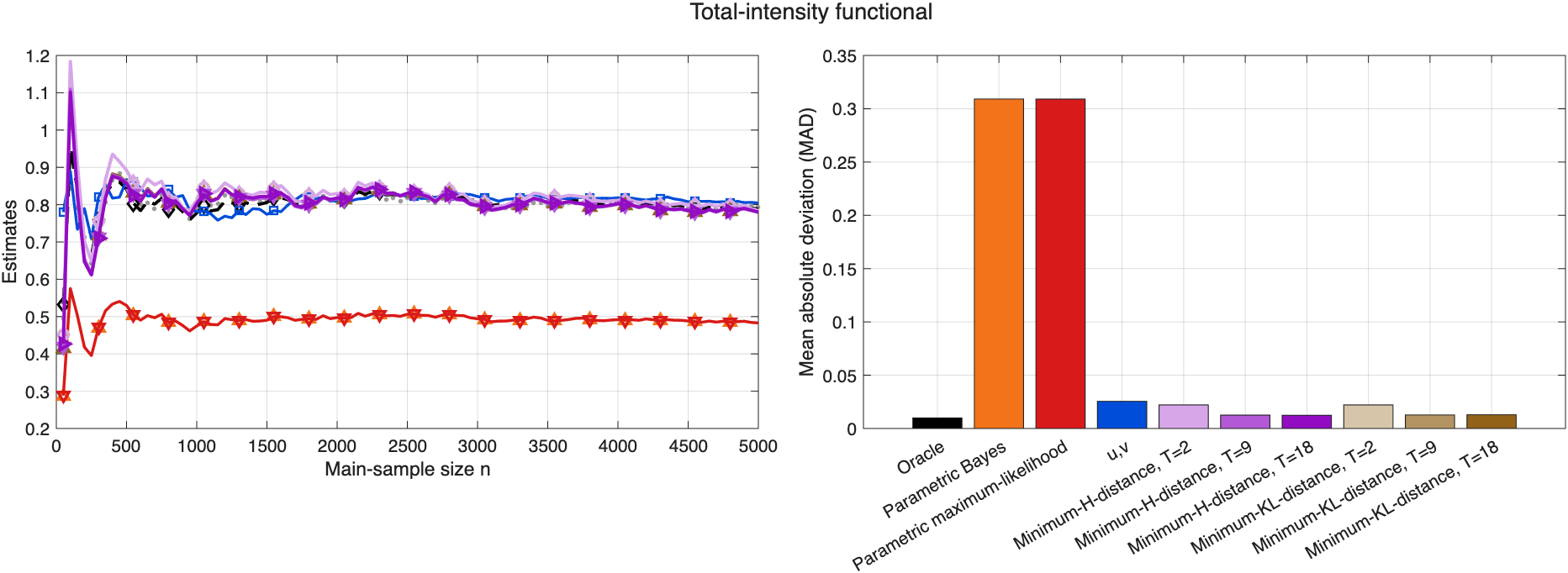}
\caption{\scriptsize{Total-intensity sum under the Weibull mixing distribution: $T$-coarsened minimum-$d$-distance estimation. Left: true normalized values of $S_{n,2}$ (gray $\bullet$) and normalized estimates as a function of the main-sample size $n$. The oracle estimate is shown by a black dashed line with $\diamond$, the parametric Bayes estimate by an orange line with $\triangle$, the parametric maximum-likelihood estimate by a red line with $\triangledown$, the $u,v$ estimate by a blue line with $\square$. The minimum-H-distance estimates for $T=2,9,18$ are shown by purple lines with $\diamond$, $\triangledown$, and $\triangleright$, respectively, whereas the corresponding minimum-KL-distance estimates are shown by brown lines with $\circ$, $\square$, and $\triangle$. Right: mean absolute deviation (MAD) of the normalized estimates from the true normalized values of $S_{n,2}$.}}
\label{fig_weib_s1_coa}
\end{figure}

\begin{figure}[H]
\centering
\includegraphics[width=1\textwidth]{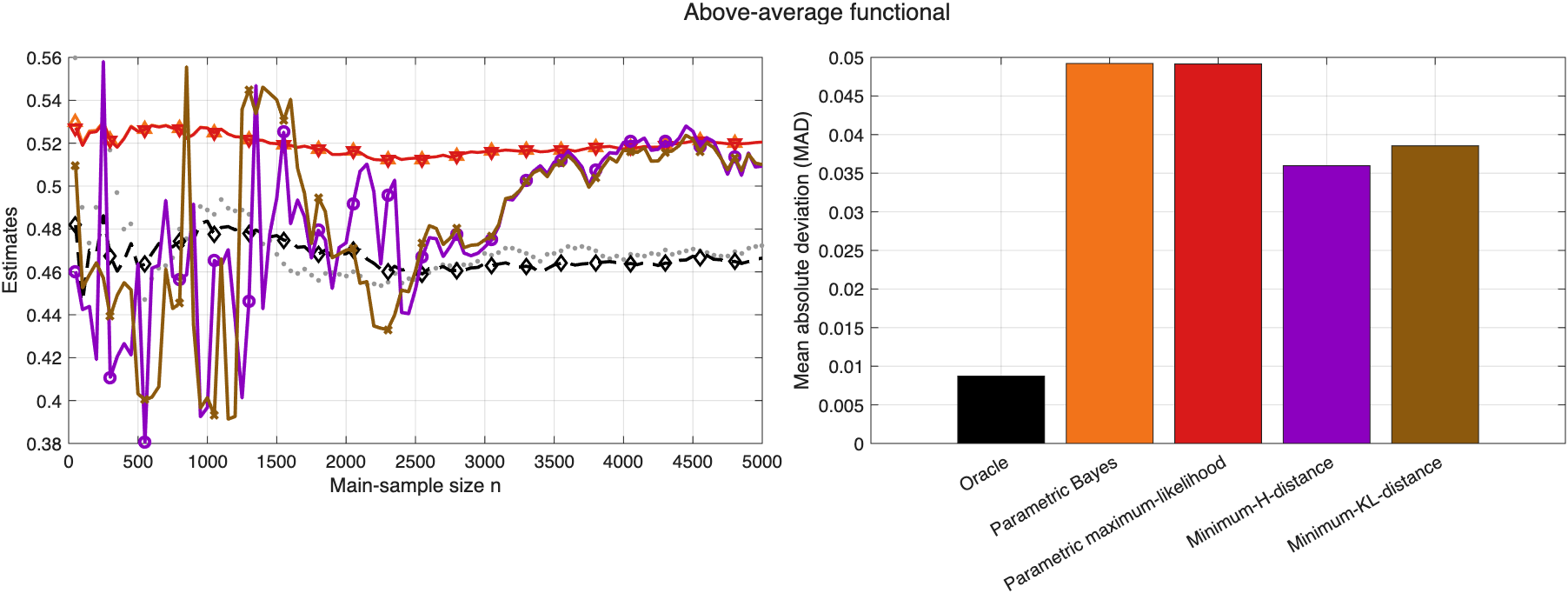}
\caption{\scriptsize{Above-average sum under the Weibull mixing distribution: regular minimum-$d$-distance estimation. Left: true normalized values of $S_{n,>}$ (gray $\bullet$) and normalized estimates as a function of the main-sample size $n$. The oracle estimate is shown by a black dashed line with $\diamond$, the parametric Bayes estimate by an
orange line with $\triangle$, the parametric maximum-likelihood estimate by a red line with $\triangledown$, and the minimum-H- and minimum-KL-distance estimates by purple $\circ$ and brown $\times$, respectively. Right: mean absolute deviation (MAD) of the normalized estimates from the true normalized values of $S_{n,>}$.}}
\label{fig_weib_s3_reg}
\end{figure}

\begin{figure}[H]
\centering
\includegraphics[width=1\textwidth]{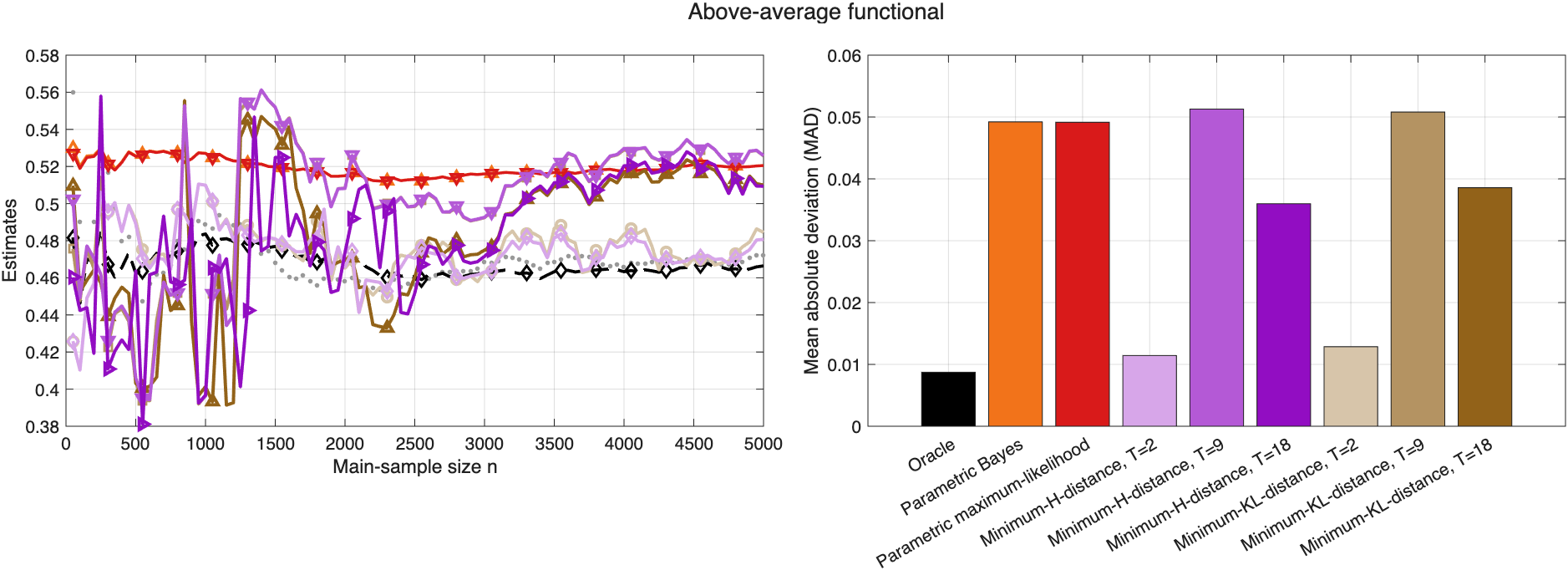}
\caption{\scriptsize{Above-average sum under the Weibull mixing distribution: $T$-coarsened minimum-$d$-distance estimation. Left: true normalized values of $S_{n,>}$ (gray $\bullet$) and normalized estimates as a function of the main-sample size $n$. The oracle estimate is shown by a black dashed line with $\diamond$, the parametric Bayes estimate by an orange line with $\triangle$, the parametric maximum-likelihood estimate by a red line with $\triangledown$. The minimum-H-distance estimates for $T=2,9,18$ are shown by purple lines with $\diamond$, $\triangledown$, and $\triangleright$, respectively, whereas the corresponding minimum-KL-distance estimates are shown by brown lines with $\circ$, $\square$, and $\triangle$. Right: mean absolute deviation (MAD) of the normalized estimates from the true normalized values of $S_{n,>}$.}}
\label{fig_weib_s3_coa}
\end{figure}

The poor performance observed for the estimation of the above-average sum $S_{n,>}$ in Figures \ref{fig_weib_s3_reg}--\ref{fig_weib_s3_coa} is not unexpected. In contrast to the estimation of the total-intensity sum $S_{n,2}$, the map $\Gstar \mapsto \P_{\Gstar}(\theta<X\mid X=x)$ involves the discontinuous function $\theta\mapsto I(\theta<x)$. Consequently, small perturbations of the estimated mixing distribution, and in particular small displacements of its support points around the observed $x$, may produce substantial changes in the corresponding plug-in estimate.  This instability is especially relevant for minimum-$d$-distance estimates, which recover the mixing distribution $\Gstar$ through a discrete approximation and may estimate the mixture distribution well without locating its individual support points with comparable accuracy. This phenomenon is consistent with the additional stability conditions required by our theoretical results for the above-average sum. To illustrate the role of this discontinuity, we also consider the smoothed above-average sum
\begin{displaymath}
S_{n,>,\delta}= \sum_{i=1}^n \phi_\delta(X_i-\theta_i),
\end{displaymath}
where $\phi_\delta$ is the Lipschitz approximation of the indicator used in our theoretical analysis, namely
\begin{displaymath}
\phi_\delta(t)=
\begin{cases}
0, & t\leq-\delta,\\
(t+\delta)/(2\delta), & -\delta<t<\delta,\\
1, & t\geq\delta.
\end{cases}
\end{displaymath}
For the regular minimum-$d$-distance estimate $\hat G_{d,m}$ of $\Gstar$, the corresponding EB estimate of $S_{n,>,\delta}$ is
\begin{equation}\label{eq:smootabove}
\hat{S}_{n,>,\delta}(\hat G_{d,m})=\sum_{i=1}^n\mathbb{E}_{\hat G_{d,m}}\left[\phi_\delta(X_i-\theta_i)\mid X_i\right]=\sum_{i=1}^n\int_{\Theta}\phi_\delta(X_i-\theta)\,\hat G_{d,m}(d\theta\mid X_i);
\end{equation}
the $T$-coarsened estimate of the smoothed above-average sum is defined analogously. Specifically, the $T$-coarsened minimum-$d$-distance estimate $\widetilde G_{d,m}^{(T)}$ is substituted for $\hat G_{d,m}$ in \eqref{eq:smootabove}.  Thus, the parameter $\delta$ controls the trade-off between approximation of the original above-average sum and stability with respect to perturbations of the estimated mixing distribution: as $\delta\rightarrow0$ the original above-average sum is recovered, whereas larger values of $\delta$ reduce the sensitivity to inaccuracies in the estimated locations of the mixing distribution. Accordingly, in the numerical comparison below we regard smoothing as a regularization device for estimating the original above-average sum $S_{n,>}$, rather than as defining a different target of interest. We therefore evaluate both the smoothed oracle and the smoothed minimum-distance estimates against the true value of $S_{n,>}$.

Figures~\ref{fig_weib_s3_reg_smooth}--\ref{fig_weib_s3_coa_smooth} report the regular and $T$-coarsened minimum-$d$-distance estimates of $S_{n,>,\delta}$, for increasing values of the smoothing parameter $\delta=0.25,0.5,1,2$. All estimates are normalized by the main-sample size $n$ and are compared with the normalized true value of $S_{n,>}$ and with the corresponding smoothed oracle estimates. The figures also report the mean absolute deviation (MAD) of the smoothed oracle and minimum-$d$-distance estimates from the true normalized $S_{n,>}$, computed over the sequence of increasing sample sizes as follows
\begin{displaymath}
\operatorname{MAD}(\widehat S_{>,\delta})=\frac{1}{100}\sum_{i=1}^{100}\left|\frac{\widehat S_{n_i,>,\delta}}{n_i}-\frac{S_{n_i,>}}{n_i}\right|.
\end{displaymath}

\begin{figure}[H]
\centering
\includegraphics[width=1\textwidth]{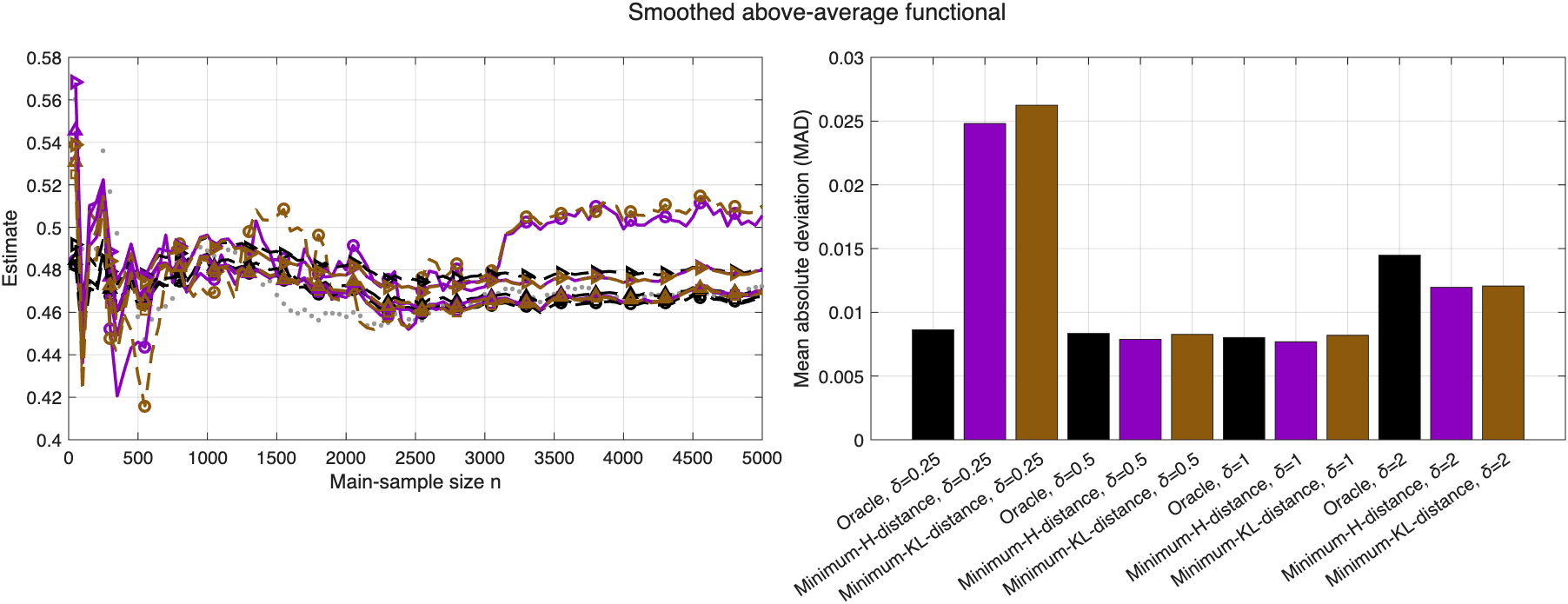}
\caption{\scriptsize{Above-average sum under the Weibull mixing distribution: regular minimum-$d$-distance estimation with smoothing. Left: true normalized values of $S_{n,>}$ (gray $\bullet$), smoothed oracle estimates (black dashed lines), and smoothed minimum-distance estimates as a function of the main-sample size $n$. The smoothed oracle, minimum-H-, and minimum-KL-distance estimates are shown for $\delta=0.25,0.5,1,2$, with $\circ$, $\square$, $\triangle$, and $\triangleright$, respectively; minimum-H- and minimum-KL-distance estimates are represented by purple solid and brown dashed lines. Right: mean absolute deviation (MAD) of the smoothed oracle and smoothed minimum-distance estimates from the true normalized values of $S_{n,>}$.}}
\label{fig_weib_s3_reg_smooth}
\end{figure}

\begin{figure}[H]
\centering
\includegraphics[width=1\textwidth]{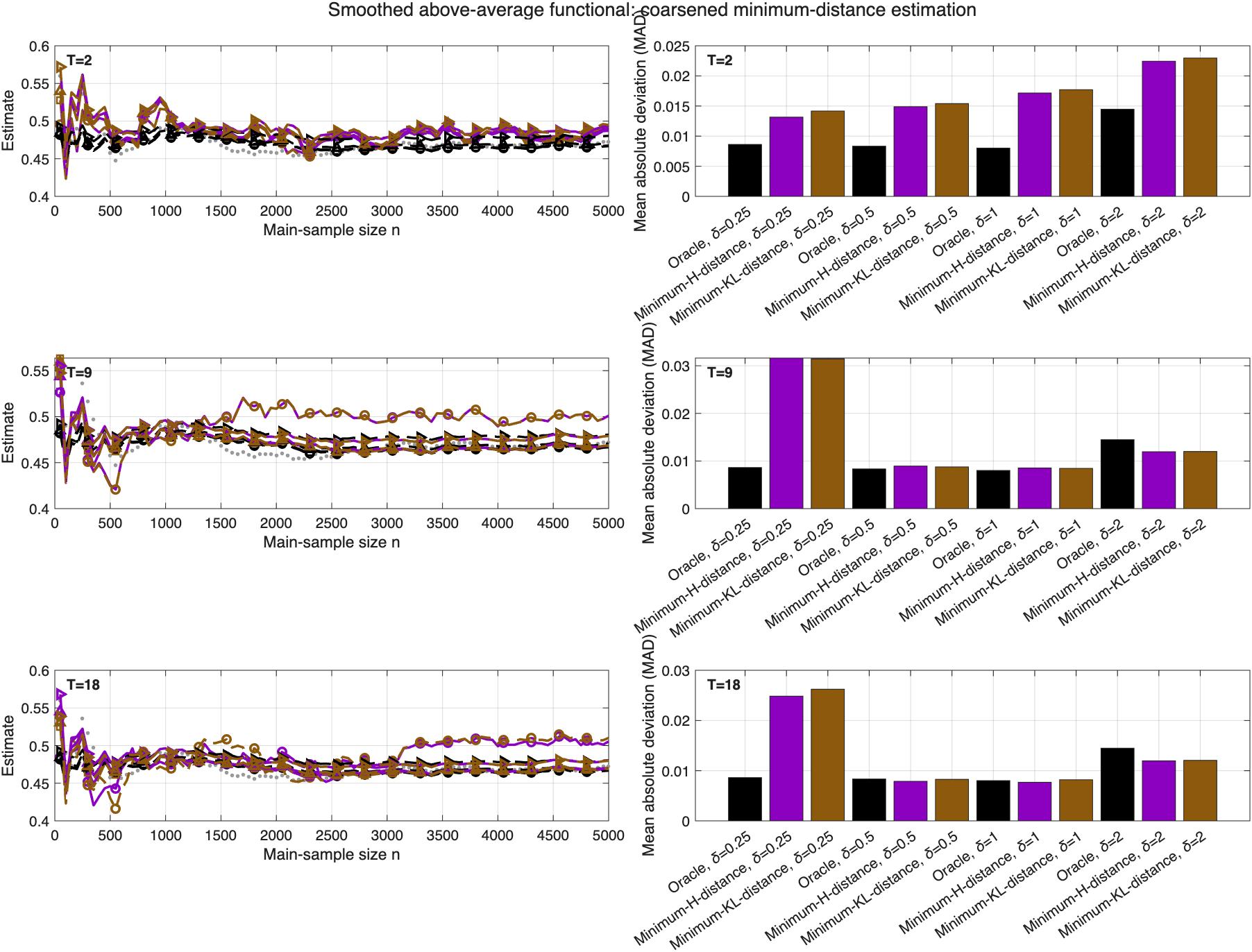}
\caption{\scriptsize{Above-average sum under the Weibull mixing distribution: $T$-coarsened minimum-$d$-distance estimation with smoothing. Each row corresponds to one coarsening level, $T=2,9,18$ from top to bottom. Left: true normalized values of $S_{n,>}$ (gray $\bullet$), smoothed oracle estimates (black dashed lines), and smoothed minimum-distance estimates as a function of the main-sample size $n$. The smoothed oracle, minimum-H-, and minimum-KL-distance estimates are shown for $\delta=0.25,0.5,1,2$, with $\circ$, $\square$, $\triangle$, and $\triangleright$, respectively; minimum-H- and minimum-KL-distance estimates are represented by purple solid and brown dashed lines. Right: mean absolute deviation (MAD) of the smoothed oracle and smoothed minimum-distance estimates from the true normalized values of $S_{n,>}$.}}
\label{fig_weib_s3_coa_smooth}
\end{figure}

\subsubsection{Uniform mixing distribution}

We repeat the same experimental design with the mixing distribution $\Gstar$ given by the Uniform distribution on $[0,10]$. Thus, for each $i=1,\ldots,100$, we generate a nested dataset of total size $m+n=100i$, split equally into a training sample of size $m=50i$ and a main sample of size $n=50i$. The regular and $T$-coarsened minimum-$d$-distance estimates are computed from the training samples for the Hellinger (H) distance and the Kullback--Leibler (KL) divergence, and are then plugged into the Bayes rules evaluated on the main samples. The coarsening levels are selected according to the same criterion as in the Weibull experiment. With $q_X(p)=\inf\{t\in\integers_+:\mathbb P(X\leq t)\geq p\}$  and maximum total sample size equal to $10,000$, we take $T\in\{2,q_X(0.95),q_X(1-1/10,000)\}=\{2,12,21\}$. In particular, these levels are kept fixed across sample sizes and reflect the tail behavior of the corresponding Poisson mixture distribution. Figures~\ref{fig_unif_s1_reg}--\ref{fig_unif_s1_coa} report the regular and $T$-coarsened minimum-$d$-distance estimates of $S_{n,2}$, whereas Figures~\ref{fig_unif_s3_reg}--\ref{fig_unif_s3_coa} report the corresponding estimates of $S_{n,>}$. All estimates, as well as the true values, are normalized by the main-sample size $n$. The figures also report the MADs, computed as in the Weibull experiment over the sequence of increasing sample sizes.

\begin{figure}[H]
\centering
\includegraphics[width=1\textwidth]{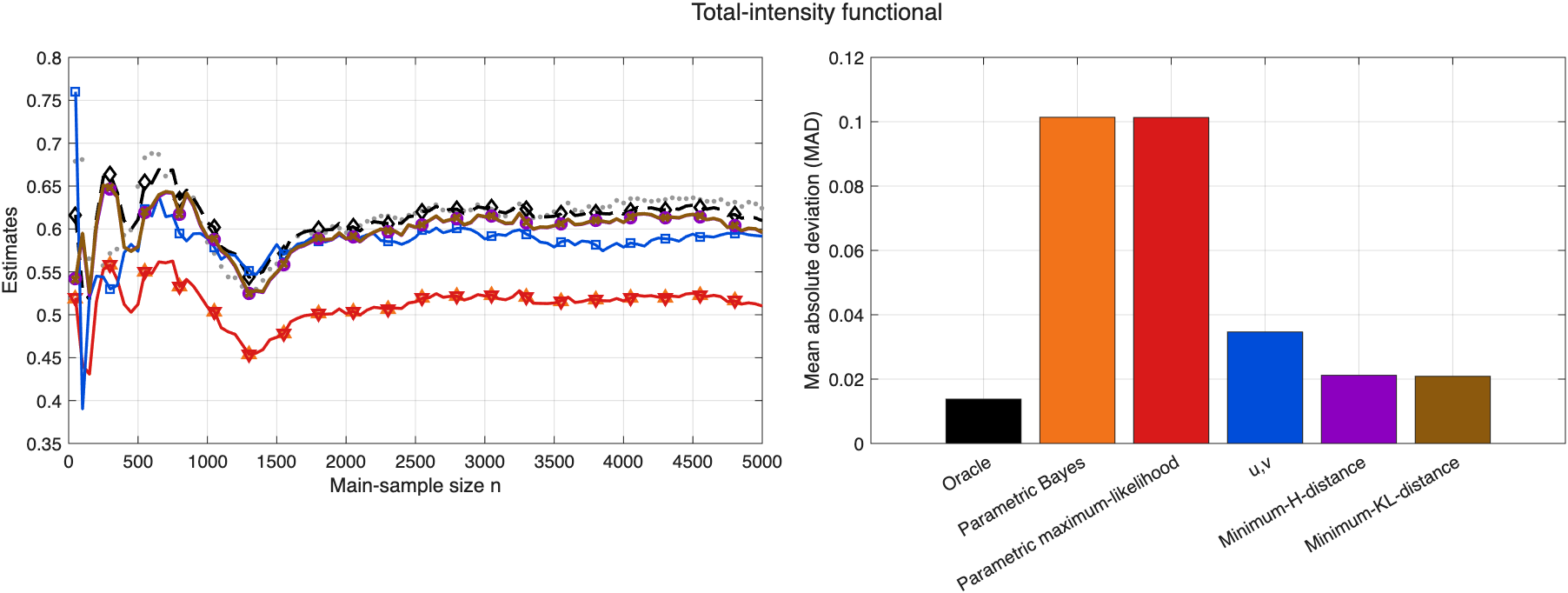}
\caption{\scriptsize{Total-intensity sum under the Uniform mixing distribution: regular minimum-$d$-distance estimation. Left: true normalized values of $S_{n,2}$ (gray $\bullet$) and normalized estimates as a function of the main-sample size $n$. The oracle estimate is shown by a black dashed line with $\diamond$, the parametric Bayes estimate by an orange line with $\triangle$, the parametric maximum-likelihood estimate by a red line with $\triangledown$, the $u,v$ estimate by a blue line with $\square$, and the minimum-H- and minimum-KL-distance estimates by purple $\circ$ and brown $\times$, respectively. Right: mean absolute deviation (MAD) of the normalized estimates from the true normalized values of  $S_{n,2}$.}}
\label{fig_unif_s1_reg}
\end{figure}

\begin{figure}[H]
\centering
\includegraphics[width=1\textwidth]{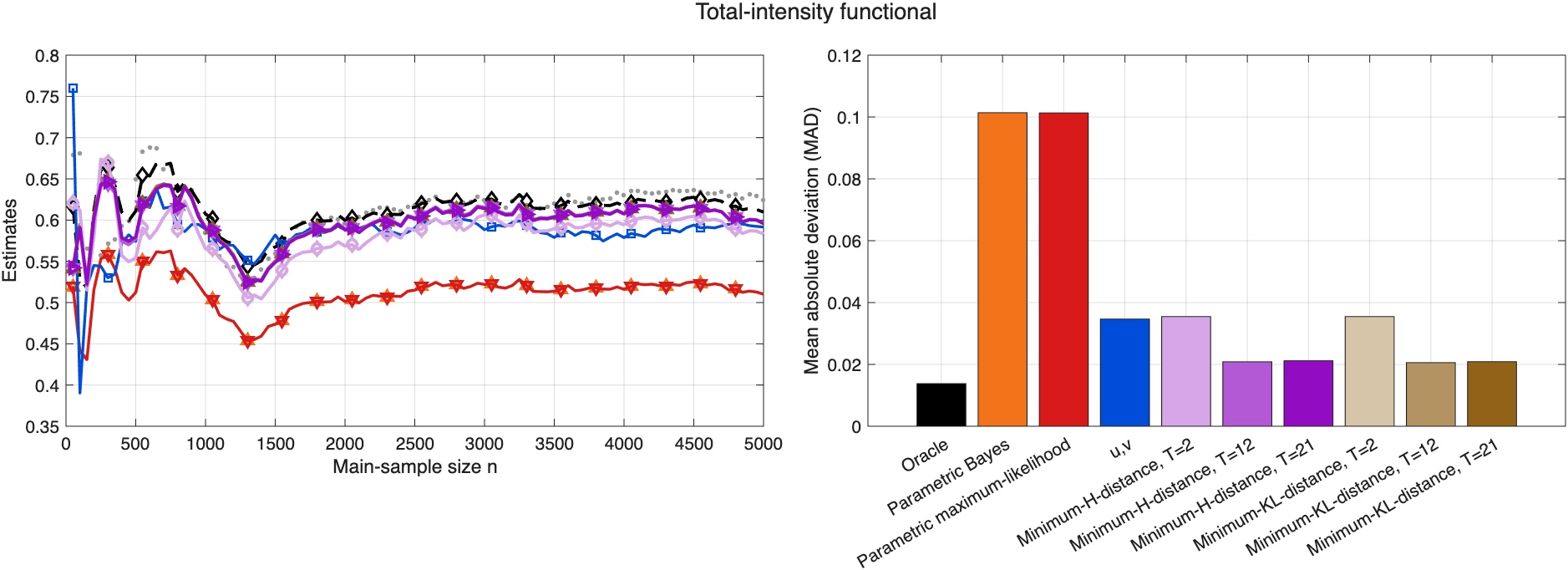}
\caption{\scriptsize{Total-intensity sum under the Uniform mixing distribution: $T$-coarsened minimum-$d$-distance estimation. Left: true normalized values of $S_{n,2}$ (gray $\bullet$) and normalized estimates as a function of the main-sample size $n$. The oracle estimate is shown by a black dashed line with $\diamond$, the parametric Bayes estimate by an orange line with $\triangle$, the parametric maximum-likelihood estimate by a red line with $\triangledown$, the $u,v$ estimate by a blue line with $\square$. The minimum-H-distance estimates for $T=2,12,21$ are shown by purple lines with $\diamond$, $\triangledown$, and $\triangleright$, respectively, whereas the corresponding minimum-KL-distance estimates are shown by brown lines with $\circ$, $\square$, and $\triangle$. Right: mean absolute deviation (MAD) of the normalized estimates from the true normalized values of $S_{n,2}$.}}
\label{fig_unif_s1_coa}
\end{figure}

\begin{figure}[H]
\centering
\includegraphics[width=1\textwidth]{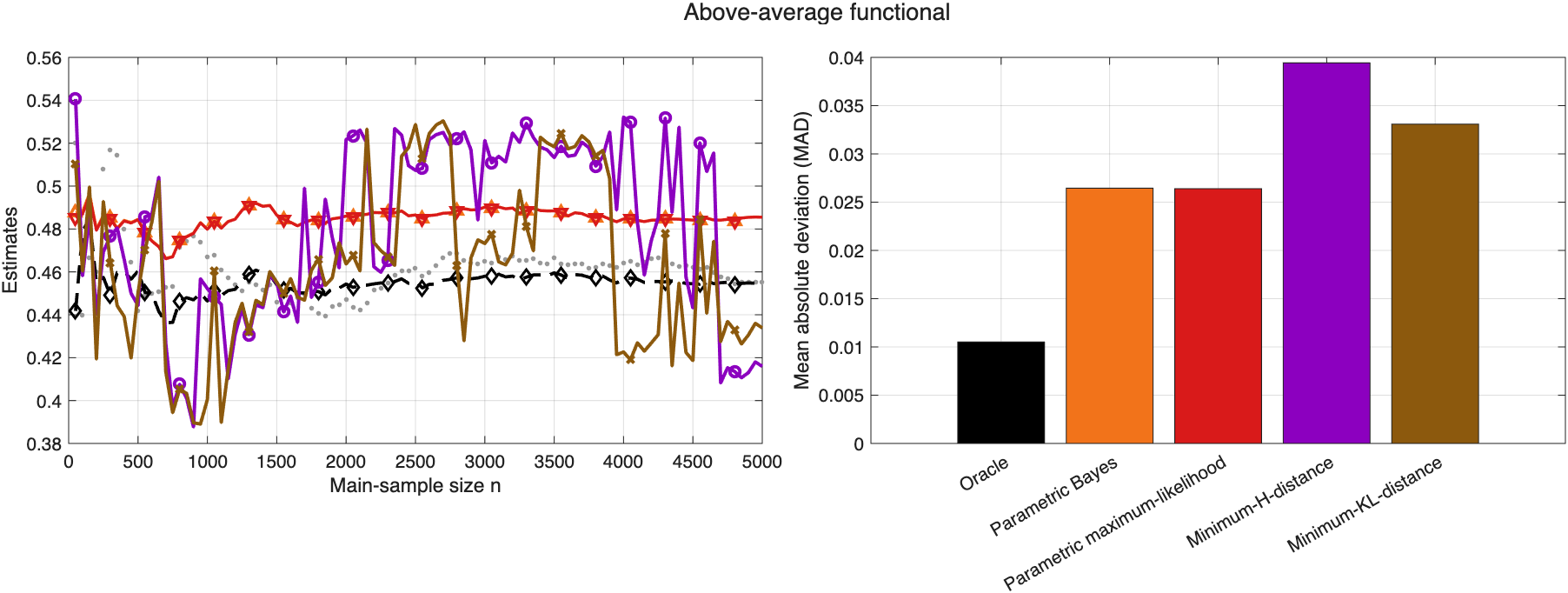}
\caption{\scriptsize{Above-average sum under the Uniform mixing distribution: regular minimum-$d$-distance estimation. Left: true normalized values of $S_{n,>}$ (gray $\bullet$) and normalized estimates as a function of the main-sample size $n$. The oracle estimate is shown by a black dashed line with $\diamond$, the parametric Bayes estimate by an
orange line with $\triangle$, the parametric maximum-likelihood estimate by a red line with $\triangledown$, and the minimum-H- and minimum-KL-distance estimates by purple $\circ$ and brown $\times$, respectively. Right: mean absolute deviation (MAD) of the normalized estimates from the true normalized values of $S_{n,>}$.}}
\label{fig_unif_s3_reg}
\end{figure}

\begin{figure}[H]
\centering
\includegraphics[width=1\textwidth]{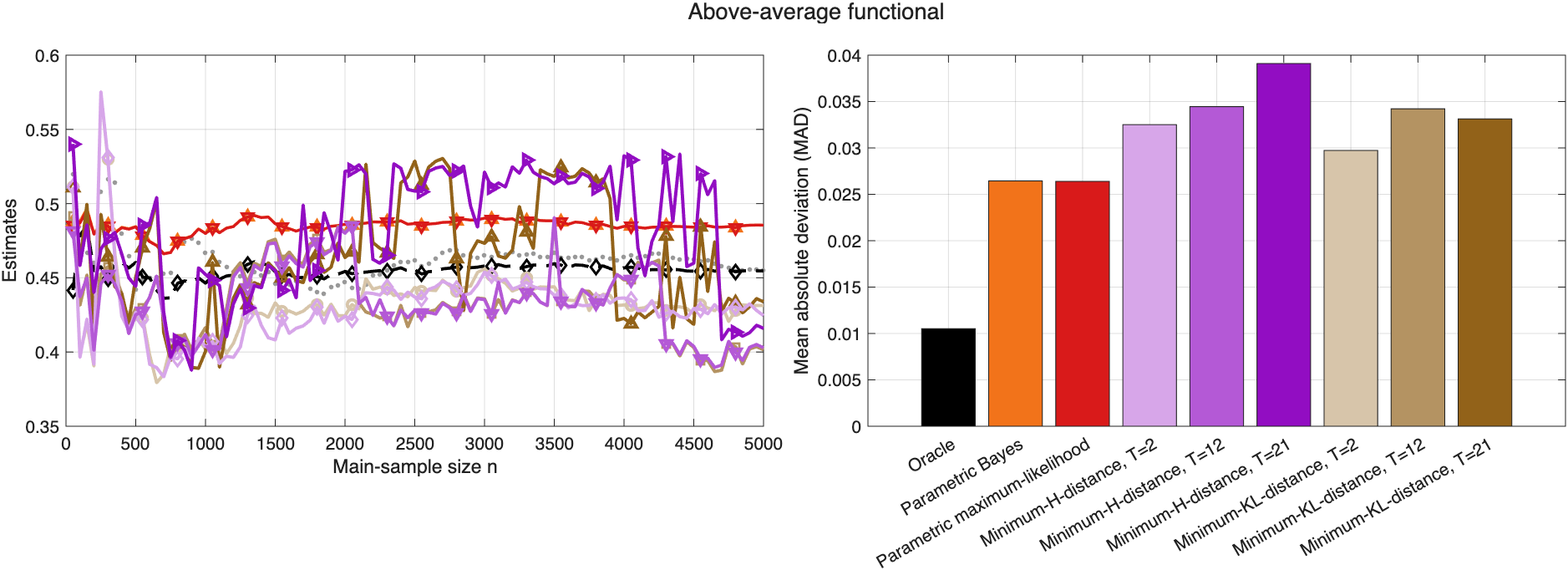}
\caption{\scriptsize{Above-average sum under the Uniform mixing distribution: $T$-coarsened minimum-$d$-distance estimation. Left: true normalized values of $S_{n,>}$ (gray $\bullet$) and normalized estimates as a function of the main-sample size $n$. The oracle estimate is shown by a black dashed line with $\diamond$, the parametric Bayes estimate by an orange line with $\triangle$, the parametric maximum-likelihood estimate by a red line with $\triangledown$. The minimum-H-distance estimates for $T=2,12,21$ are shown by purple lines with $\diamond$, $\triangledown$, and $\triangleright$, respectively, whereas the corresponding minimum-KL-distance estimates are shown by brown lines with $\circ$, $\square$, and $\triangle$. Right: mean absolute deviation (MAD) of the normalized estimates from the true normalized values of $S_{n,>}$.}}
\label{fig_unif_s3_coa}
\end{figure}

Figures~\ref{fig_unif_s3_reg_smooth}--\ref{fig_unif_s3_coa_smooth} report the regular and $T$-coarsened minimum-$d$-distance estimates of $S_{n,>,\delta}$ for the smoothing parameter $\delta=0.25,0.5,1,2$. As in the Weibull experiment, the estimates are normalized by the main-sample size $n$ and compared with the normalized true value of $S_{n,>}$ and with the corresponding smoothed oracle estimates. The figures also report the MADs of the smoothed oracle and minimum-$d$-distance estimates from the true normalized $S_{n,>}$.

\begin{figure}[H]
\centering
\includegraphics[width=1\textwidth]{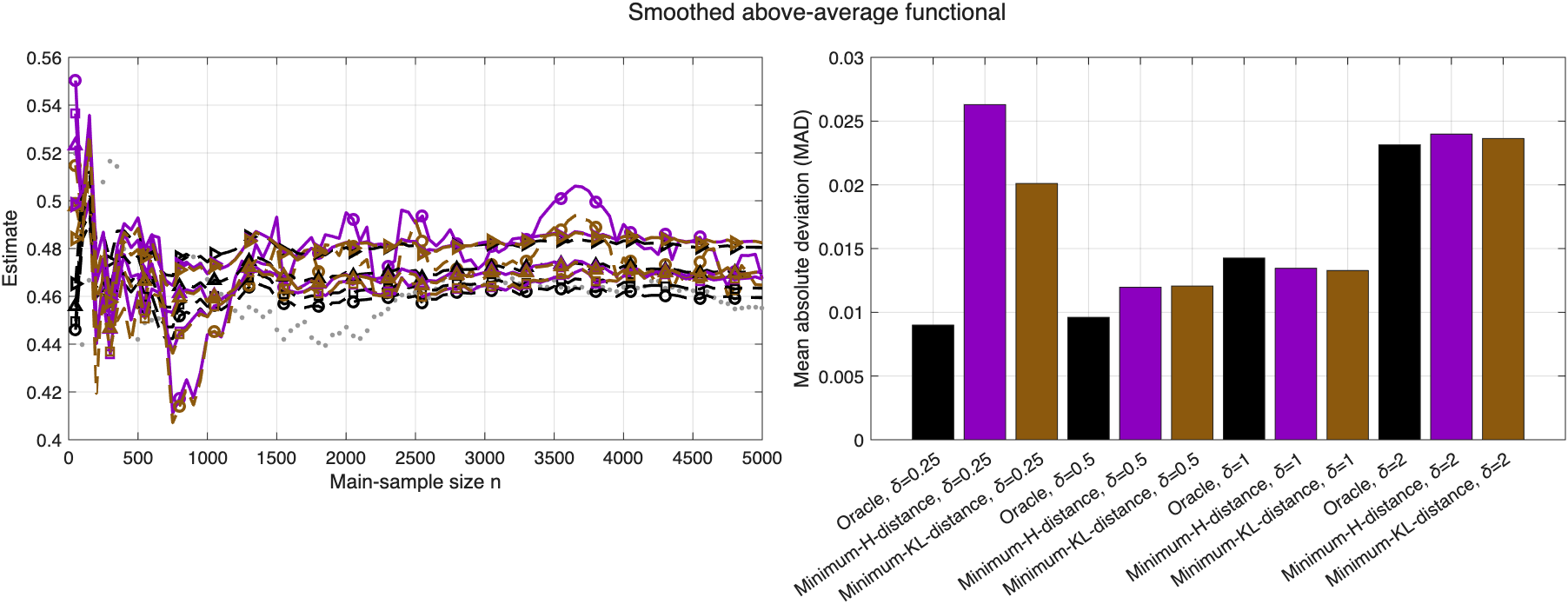}
\caption{\scriptsize{Above-average sum under the Uniform mixing distribution: regular minimum-$d$-distance estimation with smoothing. Left: true normalized values of $S_{n,>}$ (gray $\bullet$), smoothed oracle estimates (black dashed lines), and smoothed minimum-distance estimates as a function of the main-sample size $n$. The smoothed oracle, minimum-H-, and minimum-KL-distance estimates are shown for $\delta=0.25,0.5,1,2$, with $\circ$, $\square$, $\triangle$, and $\triangleright$, respectively; minimum-H- and minimum-KL-distance estimates are represented by purple solid and brown dashed lines. Right: mean absolute deviation (MAD) of the smoothed oracle and smoothed minimum-distance estimates from the true normalized values of  $S_{n,>}$.}}
\label{fig_unif_s3_reg_smooth}
\end{figure}

\begin{figure}[H]
\centering
\includegraphics[width=1\textwidth]{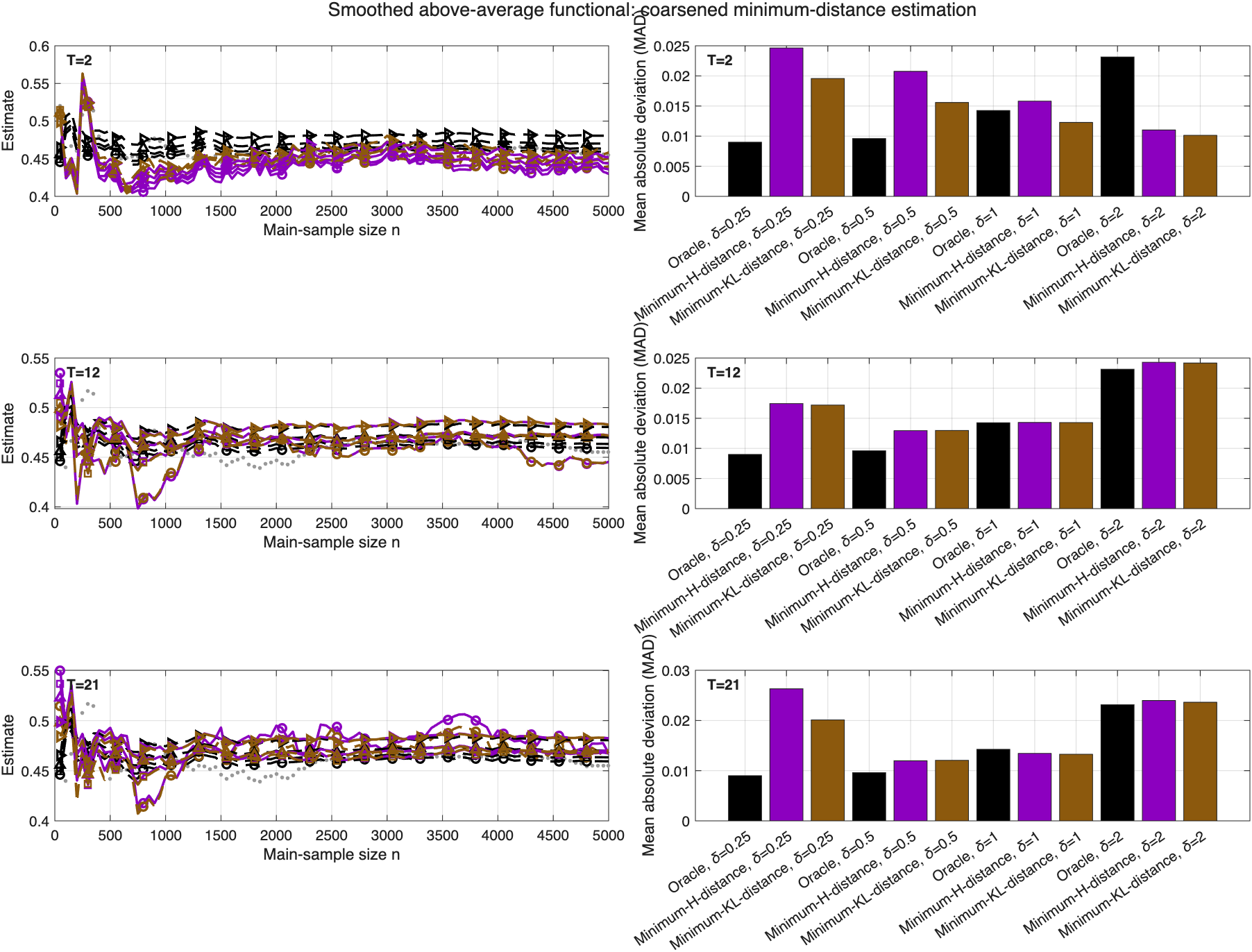}
\caption{\scriptsize{Above-average sum under the Uniform mixing distribution: $T$-coarsened minimum-$d$-distance estimation with smoothing. Each row corresponds to one coarsening level, $T=2,12,21$ from top to bottom. Left: true normalized values of $S_{n,>}$ (gray $\bullet$), smoothed oracle estimates (black dashed lines), and smoothed minimum-distance estimates as a function of the main-sample size $n$. The smoothed oracle, minimum-H-, and minimum-KL-distance estimates are shown for $\delta=0.25,0.5,1,2$, with $\circ$, $\square$, $\triangle$, and $\triangleright$, respectively; minimum-H- and minimum-KL-distance estimates are represented by purple solid and brown dashed lines. Right: mean absolute deviation (MAD) of the smoothed oracle and smoothed minimum-distance estimates from the true normalized values of $S_{n,>}$.}}
\label{fig_unif_s3_coa_smooth}
\end{figure}

\subsubsection{Half-Gaussian mixing distribution}

We further repeat the experiment with $\Gstar$ given by a half-Gaussian distribution, namely the distribution of the absolute value of a standard Gaussian random variable. For each $i=1,\ldots,100$, the nested dataset has total size $m+n=100i$ and is split equally into a training sample of size $m=50i$ and a main sample of size $n=50i$. Regular and $T$-coarsened minimum-$d$-distance estimates are computed from the training samples for the Hellinger (H) distance and the Kullback--Leibler (KL) divergence and then evaluated on the main samples through the plug-in Bayes rules. Using the same criterion for the coarsening levels, $T\in\{2,q_X(0.95),q_X(1-1/10,000)\}=\{2,3,8\}$. Again, these levels are kept fixed across sample sizes. Figures~\ref{fig_gauss_s1_reg}--\ref{fig_gauss_s1_coa} report the regular and $T$-coarsened minimum-$d$-distance estimates of $S_{n,2}$, whereas Figures~\ref{fig_gauss_s3_reg}--\ref{fig_gauss_s3_coa} report the corresponding estimates of $S_{n,>}$. All estimates, as well as the true values, are normalized by the main-sample size $n$. The figures also report the MADs, computed as in the Weibull experiment over the sequence of increasing sample sizes.

\begin{figure}[H]
\centering
\includegraphics[width=1\textwidth]{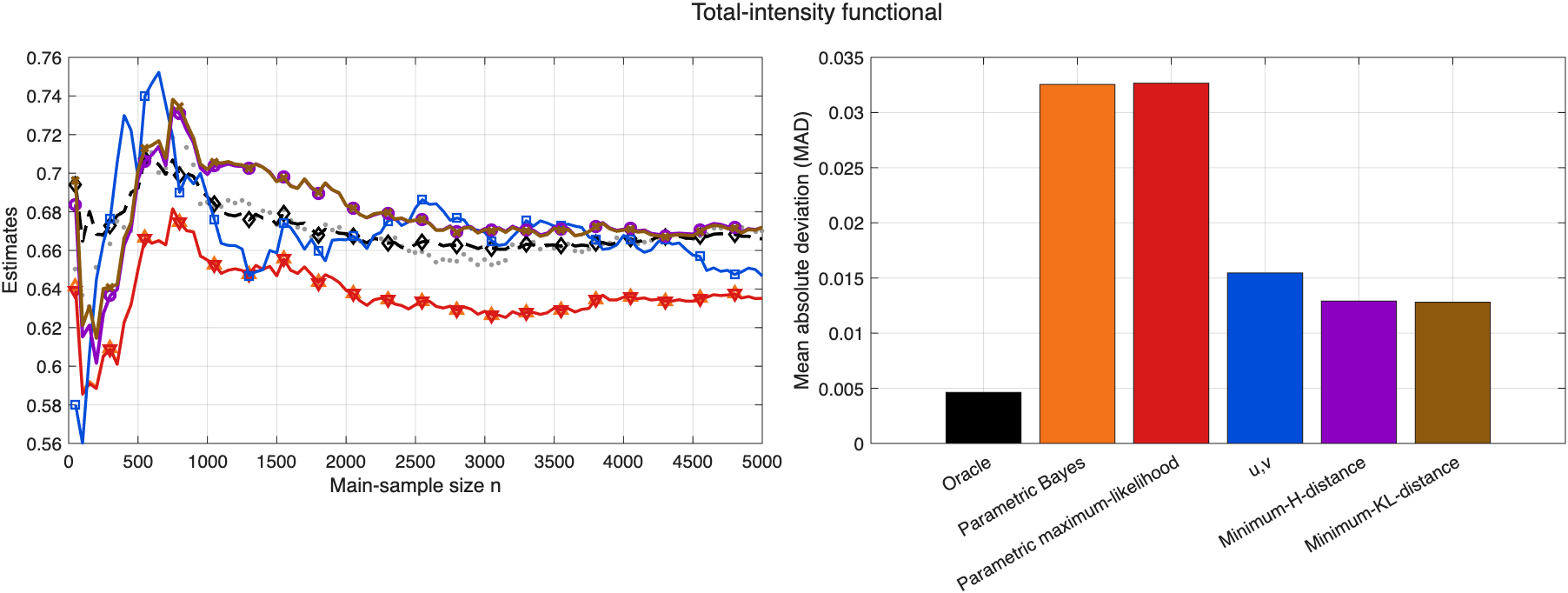}
\caption{\scriptsize{Total-intensity sum under the half-Gaussian mixing distribution: regular minimum-$d$-distance estimation. Left: true normalized  values of $S_{n,2}$ (gray $\bullet$) and normalized estimates as a function of the main-sample size $n$. The oracle estimate is shown by a black dashed line with $\diamond$, the parametric Bayes estimate by an orange line with $\triangle$, the parametric maximum-likelihood estimate by a red line with $\triangledown$, the $u,v$ estimate by a blue line with $\square$, and the minimum-H- and minimum-KL-distance estimates by purple $\circ$ and brown $\times$, respectively. Right: mean absolute deviation (MAD) of the normalized estimates from the true normalized values of  $S_{n,2}$.}}
\label{fig_gauss_s1_reg}
\end{figure}

\begin{figure}[H]
\centering
\includegraphics[width=1\textwidth]{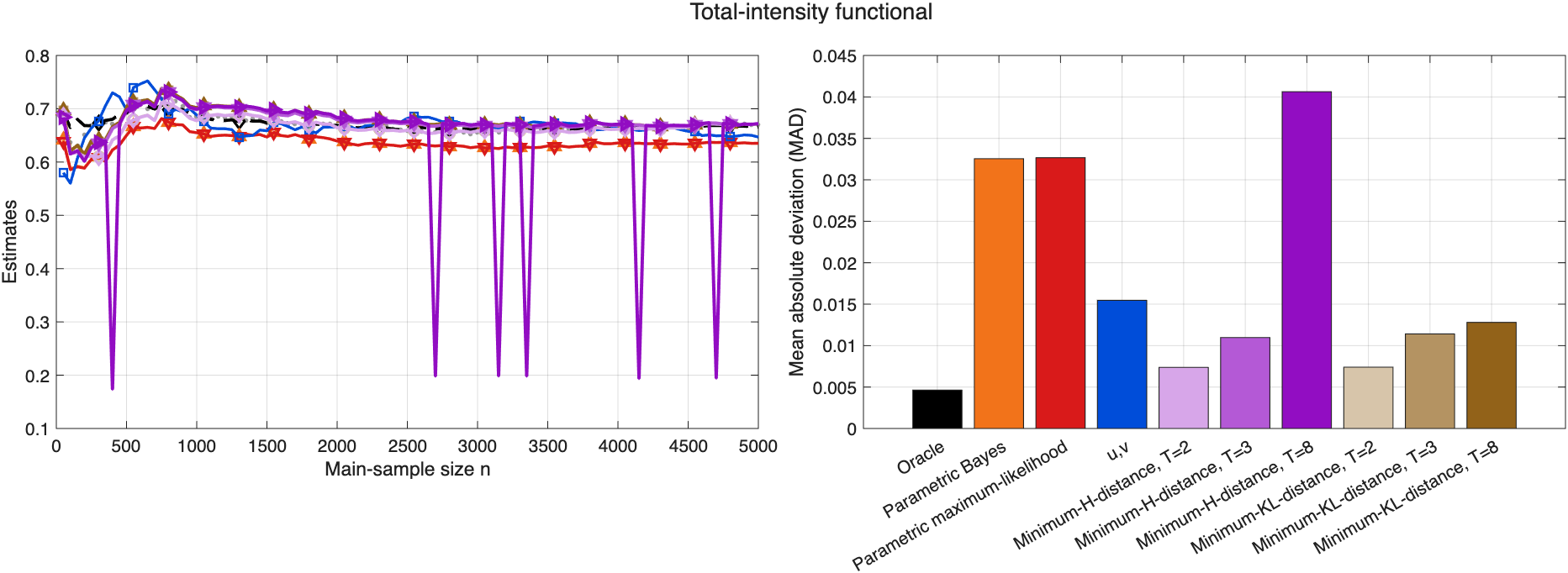}
\caption{\scriptsize{Total-intensity sum under the half-Gaussian mixing distribution: $T$-coarsened minimum-$d$-distance estimation. Left: true normalized values of $S_{n,2}$ (gray $\bullet$) and normalized estimates as a function of the main-sample size $n$. The oracle estimate is shown by a black dashed line with $\diamond$, the parametric Bayes estimate by an orange line with $\triangle$, the parametric maximum-likelihood estimate by a red line with $\triangledown$, the $u,v$ estimate by a blue line with $\square$. The minimum-H-distance estimates for $T=2,3,8$ are shown by purple lines with $\diamond$, $\triangledown$, and $\triangleright$, respectively, whereas the corresponding minimum-KL-distance estimates are shown by brown lines with $\circ$, $\square$, and $\triangle$. Right: mean absolute deviation (MAD) of the normalized estimates from the true normalized values of $S_{n,2}$.}}
\label{fig_gauss_s1_coa}
\end{figure}

\begin{figure}[H]
\centering
\includegraphics[width=1\textwidth]{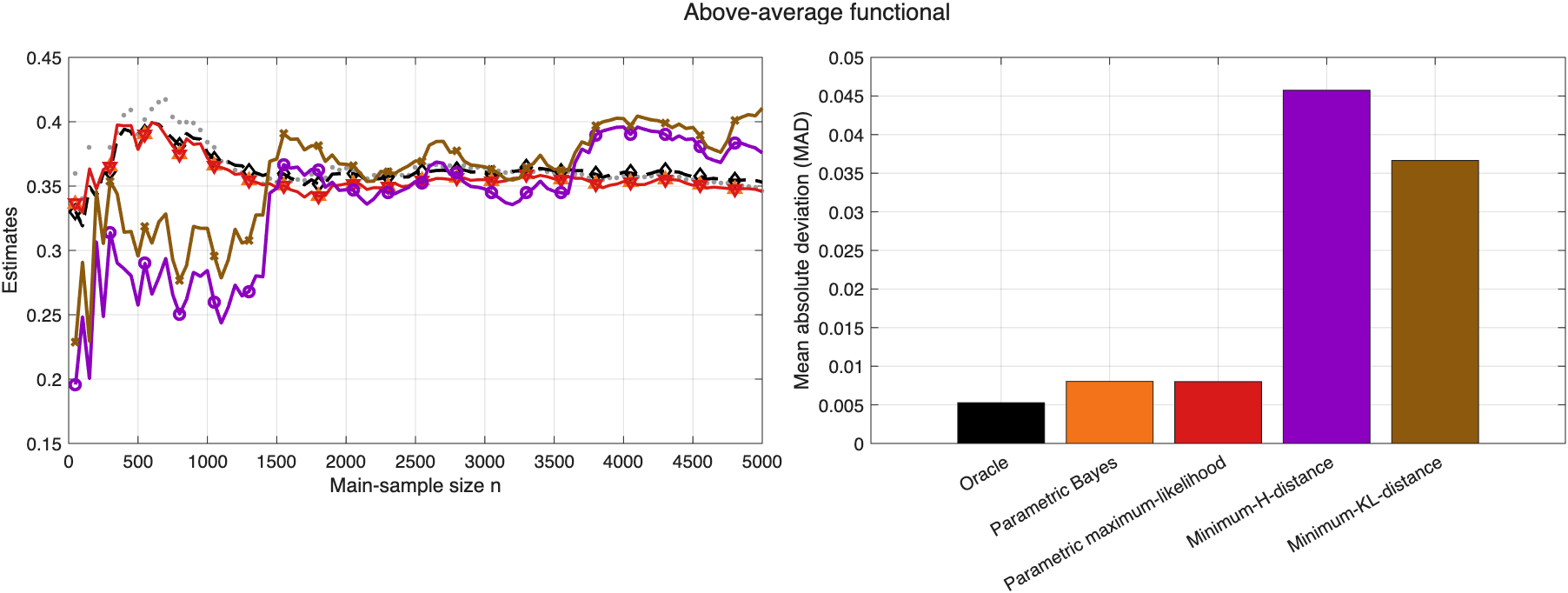}
\caption{\scriptsize{Above-average sum under the half-Gaussian mixing distribution: regular minimum-$d$-distance estimation. Left: true normalized values of $S_{n,>}$ (gray $\bullet$) and normalized estimates as a function of the main-sample size $n$. The oracle estimate is shown by a black dashed line with $\diamond$, the parametric Bayes estimate by an
orange line with $\triangle$, the parametric maximum-likelihood estimate by a red line with $\triangledown$, and the minimum-H- and minimum-KL-distance estimates by purple $\circ$ and brown $\times$, respectively. Right: mean absolute deviation (MAD) of the normalized estimates from the true normalized values of $S_{n,>}$.}}
\label{fig_gauss_s3_reg}
\end{figure}

\begin{figure}[H]
\centering
\includegraphics[width=1\textwidth]{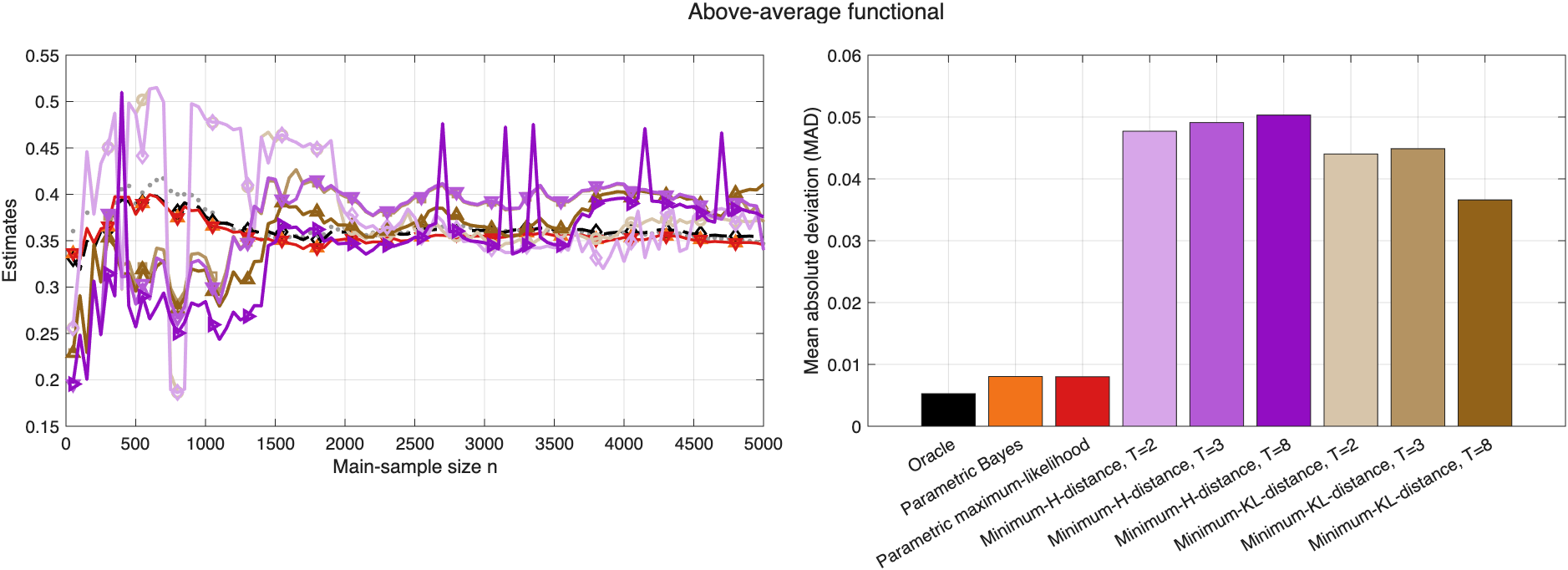}
\caption{\scriptsize{Above-average sum under the half-Gaussian mixing distribution: $T$-coarsened minimum-$d$-distance estimation. Left: true normalized values of $S_{n,>}$ (gray $\bullet$) and normalized estimates as a function of the main-sample size $n$. The oracle estimate is shown by a black dashed line with $\diamond$, the parametric Bayes estimate by an orange line with $\triangle$, the parametric maximum-likelihood estimate by a red line with $\triangledown$. The minimum-H-distance estimates for $T=2,3,8$ are shown by purple lines with $\diamond$, $\triangledown$, and $\triangleright$, respectively, whereas the corresponding minimum-KL-distance estimates are shown by brown lines with $\circ$, $\square$, and $\triangle$. Right: mean absolute deviation (MAD) of the normalized estimates from the true normalized values of $S_{n,>}$.}}
\label{fig_gauss_s3_coa}
\end{figure}

Figures~\ref{fig_gauss_s3_reg_smooth}--\ref{fig_gauss_s3_coa_smooth} report the regular and $T$-coarsened minimum-$d$-distance estimates of $S_{n,>,\delta}$ for the smoothing parameter $\delta=0.25,0.5,1,2$. As in the Weibull experiment, the estimates are normalized by the main-sample size $n$ and compared with the normalized true value of $S_{n,>}$ and with the corresponding smoothed oracle estimates. The figures also report the MADs of the smoothed oracle and minimum-$d$-distance estimates from the true normalized $S_{n,>}$.

\begin{figure}[H]
\centering
\includegraphics[width=1\textwidth]{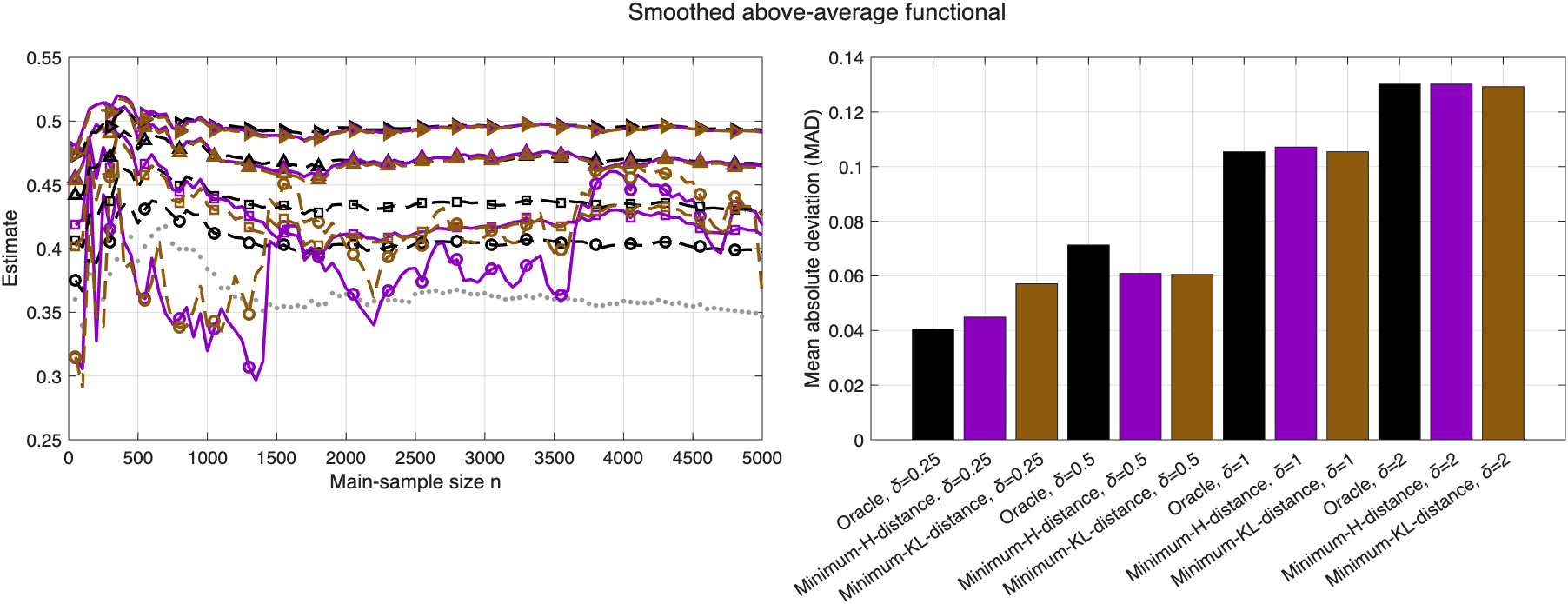}
\caption{\scriptsize{Above-average sum under the half-Gaussian mixing distribution: regular minimum-$d$-distance estimation with smoothing. Left: true normalized values of $S_{n,>}$ (gray $\bullet$), smoothed oracle estimates (black dashed lines), and smoothed minimum-distance estimates as a function of the main-sample size $n$. The smoothed oracle, minimum-H-, and minimum-KL-distance estimates are shown for $\delta=0.25,0.5,1,2$, with $\circ$, $\square$, $\triangle$, and $\triangleright$, respectively; minimum-H- and minimum-KL-distance estimates are represented by purple solid and brown dashed lines. Right: mean absolute deviation (MAD) of the smoothed oracle and smoothed minimum-distance estimates from the true normalized values of  $S_{n,>}$.}}
\label{fig_gauss_s3_reg_smooth}
\end{figure}

\begin{figure}[H]
\centering
\includegraphics[width=1\textwidth]{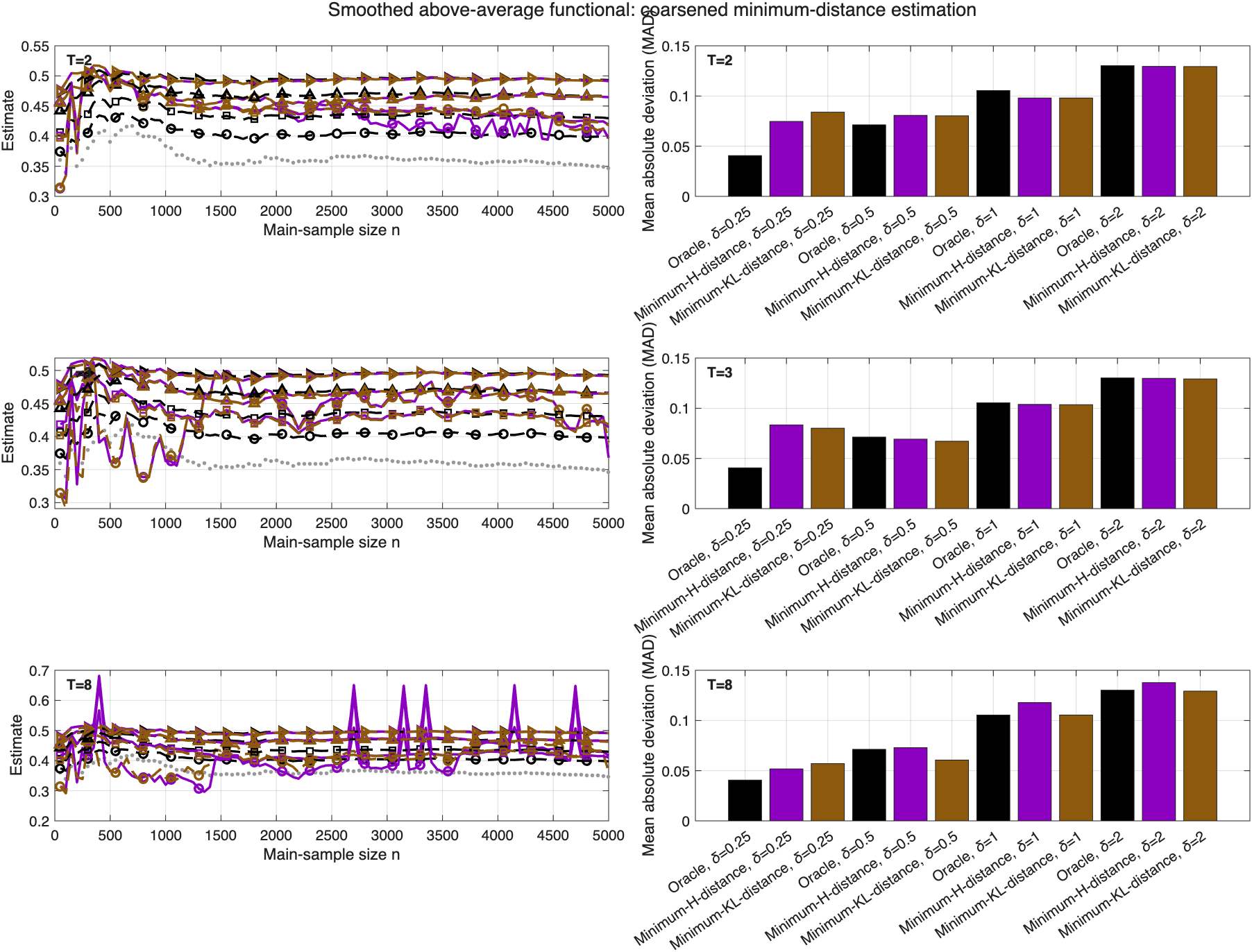}
\caption{\scriptsize{Above-average sum under the half-Gaussian mixing distribution: $T$-coarsened minimum-$d$-distance estimation with smoothing. Each row corresponds to one coarsening level, $T=2,3,8$ from top to bottom. Left: true normalized values of $S_{n,>}$ (gray $\bullet$), smoothed oracle estimates (black dashed lines), and smoothed minimum-distance estimates as a function of the main-sample size $n$. The smoothed oracle, minimum-H-, and minimum-KL-distance estimates are shown for $\delta=0.25,0.5,1,2$, with $\circ$, $\square$, $\triangle$, and $\triangleright$, respectively; minimum-H- and minimum-KL-distance estimates are represented by purple solid and brown dashed lines. Right: mean absolute deviation (MAD) of the smoothed oracle and smoothed minimum-distance estimates from the true normalized values of $S_{n,>}$.}}
\label{fig_gauss_s3_coa_smooth}
\end{figure}

\subsection{Real-data analysis}

We apply the proposed minimum-distance procedures to data on the total number of goals scored by National Hockey League players in the 2017--2018 and 2018--2019 seasons, available from \url{https://www.hockey-reference.com/}. We restrict attention to the $745$ players for whom goal totals are available in both seasons. For player $i$, let $X_i$ denote the total number of goals scored in the 2017--2018 season and let $Y_i$ denote the corresponding total in the 2018--2019 season. We analyze the data under the Poisson mixture model \eqref{model_distr_poi}, with $\theta_i$ representing the latent scoring intensity of player $i$. To conform with the training-sample and main-sample framework considered throughout the paper, we split the $745$ players into two disjoint groups. The first $m=372$ players form the training sample: their 2017--2018 goal totals, denoted by $X'_{1:m}=(X'_{1},\ldots,X'_{m})$, are used exclusively to estimate the mixing distribution $\Gstar$. The remaining $n=373$ players form the main sample. Their 2017--2018 goal totals $X_{1:n}=(X_{1},\ldots,X_{n})$ are used to construct the EB estimates of the target sums, whereas their 2018--2019 totals $Y_{1:n}=(Y_{1},\ldots,Y_{n})$ are not used at any stage of the estimation procedure and serve only for out-of-sample validation. Accordingly, the two-season structure of the data provides a natural assessment of the predictive performance of the proposed procedures.

Given the main sample $X_{1:n}$, we consider the predictive sums $S^{\text{\tiny{[Pred]}}}_{n,\kappa}$ and $S^{\text{\tiny{[Pred]}}}_{n,>}$, defined in terms of the corresponding 2018--2019 goal totals $Y_{1:n}$, where $\kappa\in\integers_+$ is a fixed threshold. Specifically,
\begin{displaymath}
S^{\text{\tiny{[Pred]}}}_{n,\kappa}=\sum_{i=1}^{n}Y_i I(X_i\leq\kappa)
\end{displaymath}
and
\begin{displaymath}
S^{\text{\tiny{[Pred]}}}_{n,>}=\sum_{i=1}^{n}I(Y_i<X_i).
\end{displaymath}
Thus, $S^{\text{\tiny{[Pred]}}}_{n,\kappa}$ is the total number of goals scored in the 2018--2019 season by players in the main sample whose 2017--2018 goal totals do not exceed $\kappa$, whereas $S^{\text{\tiny{[Pred]}}}_{n,>}$ is the number of main-sample players whose goal total decreases from the 2017--2018 to the 2018--2019 season; see \citet[Section 2.3]{Zha(05)}.

From the training sample $X'_{1:m}$, we compute the regular minimum-$d$-distance estimate $\widehat G_{d,m}$ and its $T$-coarsened  minimum-$d$-distance estimate $\widetilde G_{d,m}^{(T)}$, for both the Hellinger (H) distance and the Kullback--Leibler (KL) divergence. These estimates are then plugged into the Bayes rules and evaluated on the independent main sample $X_{1:n}$, in analogy with what is done for $S_{n,\kappa}$ and $S_{n,>}$. For a generic estimate $\widehat G$ of the mixing distribution, the resulting estimate of $S^{\text{\tiny{[Pred]}}}_{n,\kappa}$ is
\begin{displaymath}
\hat{S}^{\text{\tiny{[Pred]}}}_{n,\kappa}(\widehat G)=\sum_{i=1}^nI(X_i\leq\kappa)\int_{\Theta}\theta\,\widehat G(d\theta\mid X_i),
\end{displaymath}
where
\begin{displaymath}
\widehat G(d\theta\mid X_i)=\frac{e^{-\theta}\theta^{X_i}\widehat G(d\theta)}{\int_{\Theta}e^{-t}t^{X_i}\widehat G(dt)}.
\end{displaymath}
Similarly, the estimate of $S^{\text{\tiny{[Pred]}}}_{n,>}$ is
\begin{displaymath}
\hat{S}^{\text{\tiny{[Pred]}}}_{n,>}(\widehat G)=\sum_{i=1}^n\int_{\Theta}\left\{\sum_{y=0}^{X_i-1}e^{-\theta}\frac{\theta^y}{y!}\right\}
\widehat G(d\theta\mid X_i),
\end{displaymath}
where the inner sum is zero when $X_i=0$. Therefore, we evaluate these expressions with $\widehat G=\widehat G_{d,m}$ and $\widehat G=\widetilde G_{d,m}^{(T)}$ to obtain the regular and $T$-coarsened minimum-$d$-distance EB estimates, respectively.

We compute estimates of $S^{\text{\tiny{[Pred]}}}_{n,\kappa}$ for $\kappa\in\{2i:i=0,1,\ldots,\lfloor \max_{1\leq j\leq n}X_j/2\rfloor\}$, as well as an estimate of $S^{\text{\tiny{[Pred]}}}_{n,>}$. These estimates are compared with: (i) the parametric Bayes and maximum-likelihood estimates of Proposition~\ref{prop:poisson}, with the unknown parameter estimated from the same training sample $X'_{1:m}$ used by the minimum-distance procedures; and (ii), for $S^{\text{\tiny{[Pred]}}}_{n,\kappa}$, the nonparametric ``$u,v$'' estimate, which is computed directly from the main sample $X_{1:n}$ and does not require a training sample. As discussed above, an explicit ``$u,v$'' estimate is not available for $S^{\text{\tiny{[Pred]}}}_{n,>}$. For the $T$-coarsened minimum-$d$-distance procedures, we choose $T$ according to the same rule used in the synthetic-data experiments, which in the present dataset gives $T=50$. Figures~\ref{fig_hock_s1_reg}--\ref{fig_hock_s1_coa} report the regular and $T$-coarsened minimum-$d$-distance estimates of $S^{\text{\tiny{[Pred]}}}_{n,\kappa}$, respectively, as a function of the threshold $\kappa$. Figures~\ref{fig_hock_s3_reg}--\ref{fig_hock_s3_coa} report the regular and $T$-coarsened minimum-$d$-distance estimates of $S^{\text{\tiny{[Pred]}}}_{n,>}$. All estimates, as well as the true values, are normalized by the main-sample size $n$. The figures also report the absolute deviations (ADs) of the normalized estimates from the corresponding true normalized values.

\begin{figure}[H]
\centering
\includegraphics[width=1\textwidth]{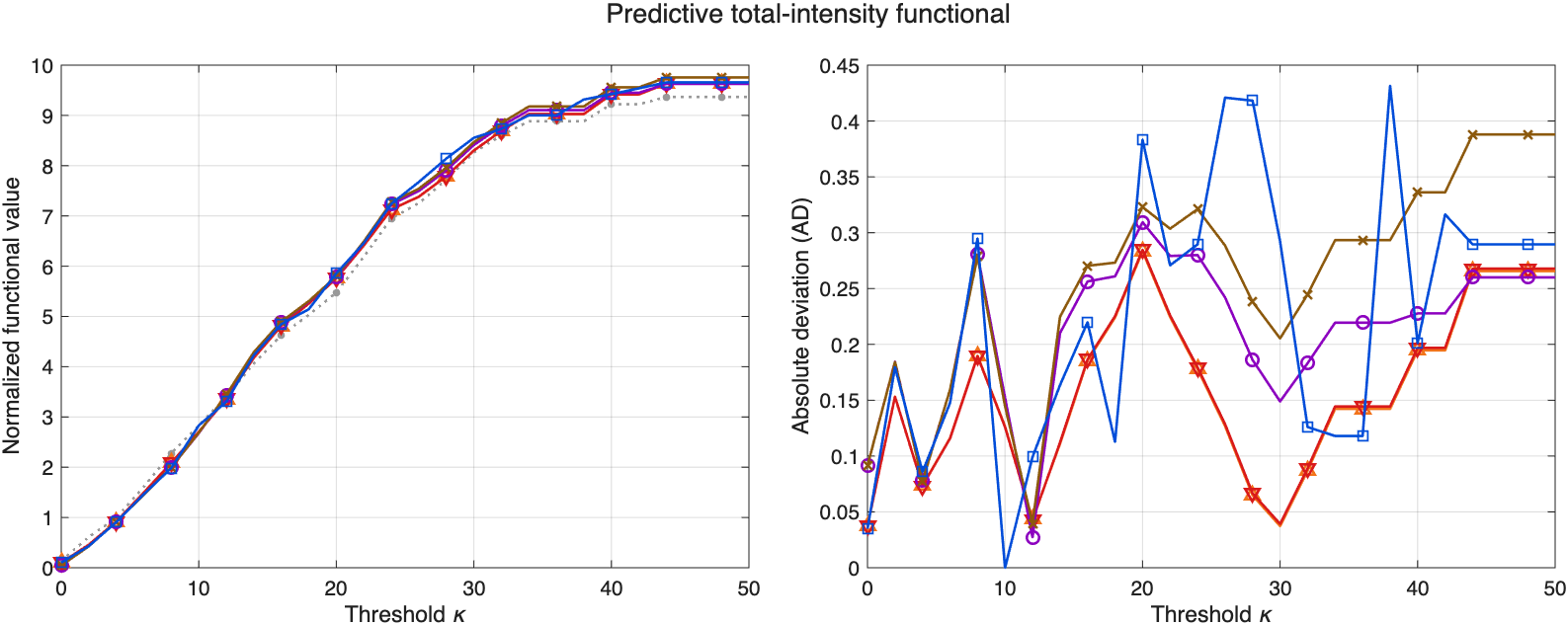}
\caption{\scriptsize{Predictive total-intensity sum for the NHL data: regular minimum-$d$-distance estimation. Left: true normalized sum values of $S^{\text{\tiny{[Pred]}}}_{n,\kappa}$ (gray $\bullet$) and normalized estimates as a function of the threshold $\kappa$. The parametric Bayes estimate is shown by an orange line with $\triangle$, the parametric maximum-likelihood estimate by a red line with $\triangledown$, the minimum-H- and minimum-KL-distance estimates by purple $\circ$ and brown $\times$, respectively, and the nonparametric ``$u,v$'' estimate by a blue line with $\square$. Right: absolute deviation (AD) of the normalized estimates from the true normalized values of $S^{\text{\tiny{[Pred]}}}_{n,\kappa}$.}}
\label{fig_hock_s1_reg}
\end{figure}

\begin{figure}[H]
\centering
\includegraphics[width=1\textwidth]{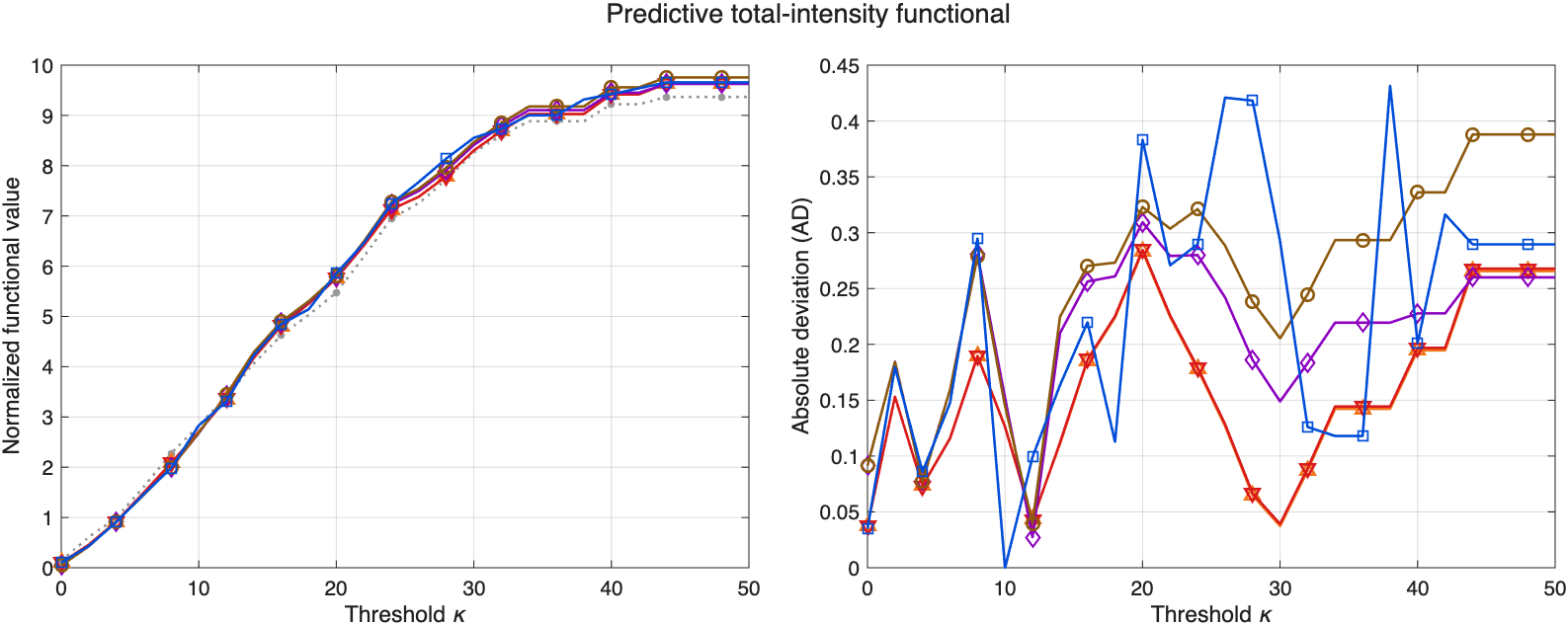}
\caption{\scriptsize{Predictive total-intensity sum for the NHL data: $T$-coarsened minimum-$d$-distance estimation. Left: true normalized values of $S^{\text{\tiny{[Pred]}}}_{n,\kappa}$ (gray $\bullet$) and normalized estimates as a function of the threshold $\kappa$. The parametric Bayes estimate is shown by an orange line with $\triangle$, the parametric maximum-likelihood estimate by a red line with $\triangledown$, the $T$-coarsened minimum-H- and minimum-KL-distance estimates by purple $\circ$ and brown $\times$, respectively, and the nonparametric ``$u,v$'' estimate by a blue line with $\square$. Right: absolute deviation (AD) of the normalized estimates from the true normalized values of $S^{\text{\tiny{[Pred]}}}_{n,\kappa}$.}}
\label{fig_hock_s1_coa}
\end{figure}

\begin{figure}[H]
\centering
\includegraphics[width=1\textwidth]{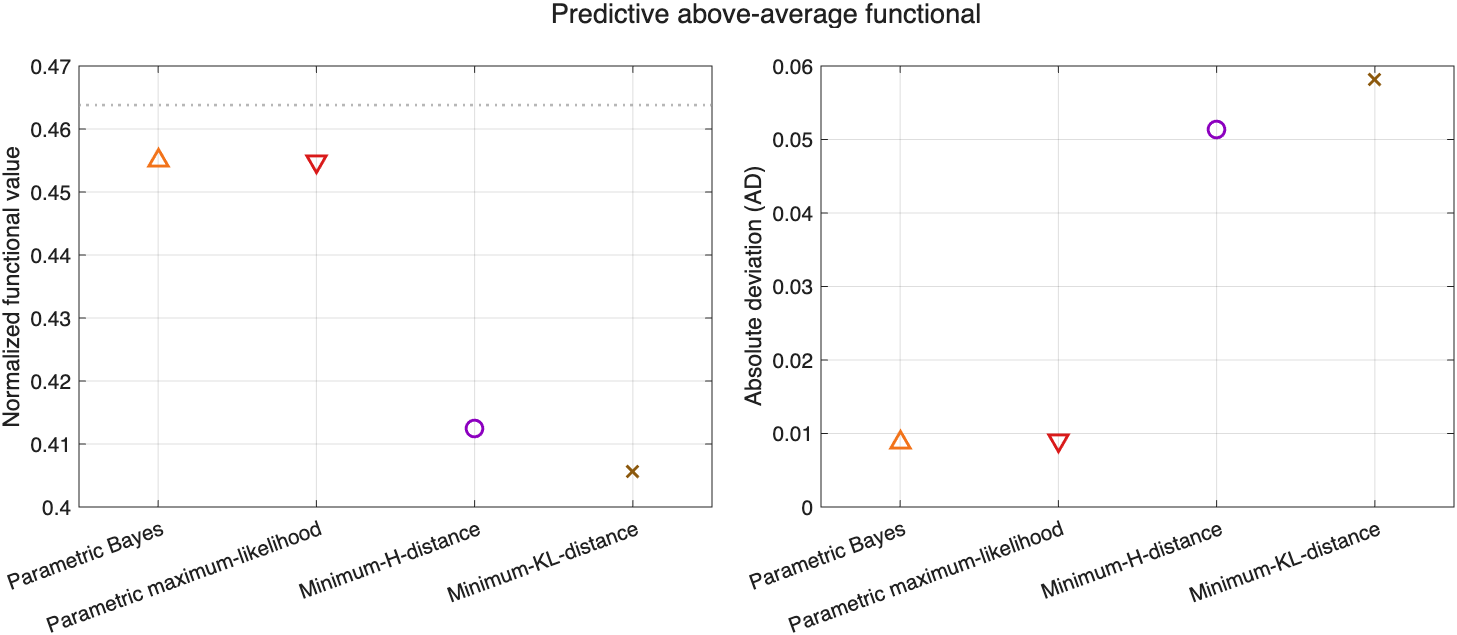}
\caption{\scriptsize{Predictive above-average sum for the NHL data: regular minimum-$d$-distance estimation. Left: true normalized value of $S^{\text{\tiny{[Pred]}}}_{n,>}$ and normalized parametric Bayes, parametric maximum-likelihood, minimum-H-distance, and minimum-KL-distance estimates. Right: absolute deviation (AD) of the normalized estimates from the true normalized value of $S^{\text{\tiny{[Pred]}}}_{n,>}$.}}
\label{fig_hock_s3_reg}
\end{figure}

\begin{figure}[H]
\centering
\includegraphics[width=1\textwidth]{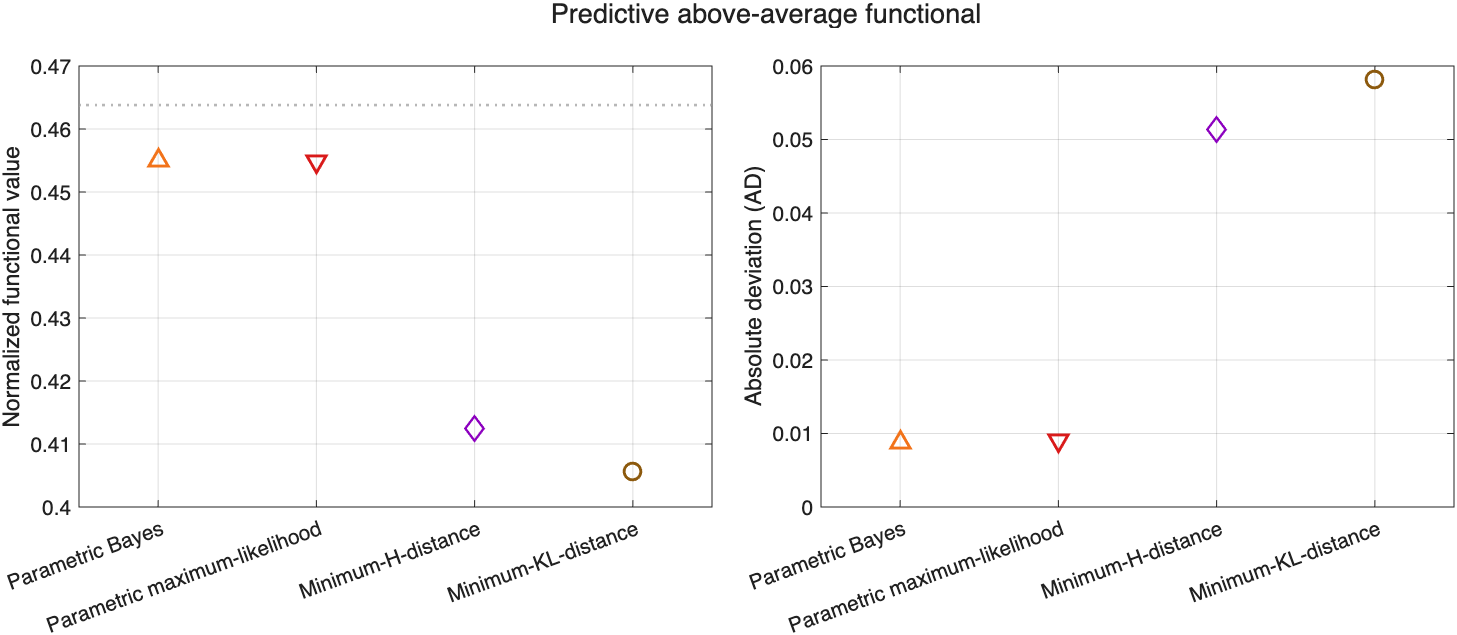}
\caption{\scriptsize{Predictive above-average sum for the NHL data: $T$-coarsened minimum-$d$-distance estimation. Left: true normalized value of $S^{\text{\tiny{[Pred]}}}_{n,>}$ and normalized parametric Bayes, parametric maximum-likelihood, $T$-coarsened minimum-H-distance, and $T$-coarsened minimum-KL-distance estimates. Right: absolute deviation (AD) of the normalized estimates from the true normalized value of $S^{\text{\tiny{[Pred]}}}_{n,>}$.}}
\label{fig_hock_s3_coa}
\end{figure}


\section{Concluding remarks}\label{sec6}

The nonparametric estimation of sums of functions of observable and unobservable random variables remains considerably less developed than empirical Bayes estimation of individual latent variables. The frequentist theory of \citet{Zha(05)} provides asymptotic information bounds and efficiency results, including for nonparametric mixture models, while recent quasi-Bayesian developments provides large-sample guarantees for recursive empirical Bayes procedures. A genuinely nonparametric Bayesian theory for this problem is, to the best of our knowledge, still unavailable, and finite-sample frequentist guarantees have received little attention. The present paper provides a first step toward filling this gap. Beyond establishing large-sample guarantees for a broad class of sums, we develop a non-asymptotic empirical Bayes theory for two representative sums, deriving finite-sample regret bounds relative to the oracle Bayes estimator and complementary minimax lower bounds. In particular, the markedly different behavior of the total-intensity and above-average sums further shows that the difficulty of estimating sums depends critically on the structure of the sum under consideration.

Several questions remain open. For the total-intensity sum, our upper and minimax lower regret bounds match up to logarithmic factors. It would be interesting to determine whether this gap is intrinsic or merely reflects the techniques used in the upper-bound analysis. In view of the sharp regret techniques developed by \citet{Pol(21)} for nonparametric empirical Bayes estimation in Poisson models, adapting similar arguments to the estimation of sums may lead to exact minimax regret rates and remove the remaining logarithmic factors.

The above-average sum raises a different set of questions. Here, the minimax lower bound shows that additional structure is necessary to obtain regret bounds at the parametric scale. Our upper bounds rely on suitable stability conditions that allow to capture such additional structure, but these conditions are not straightforward to verify analytically for general mixing distributions. The fact that they can be established directly for finitely supported mixing distributions provides only a first sufficient setting. One direction is therefore to identify broader and more easily verifiable classes of mixing distributions for which the required stability holds. A different approach is to incorporate stability directly into the estimation procedure, for instance through constrained or penalized minimum-distance estimation. Such a construction could potentially yield optimal regret guarantees without imposing difficult-to-check stability conditions on the resulting estimator, and may provide the basis for a more general non-asymptotic theory of empirical Bayes estimation for non-smooth sums.

More broadly, the contrast between the total-intensity and above-average sums suggests the need for a general finite-sample theory relating properties of the utility function $u$ to the difficulty of estimating sums $S_n=\sum_{i=1}^n u(X_i,\theta_i)$. While our large-sample results apply to a broad class of sums of random variables, the non-asymptotic analysis developed here is necessarily sum-specific. Identifying structural properties of $u$ that determine attainable minimax regret rates, and developing estimation procedures that adapt to these properties, would extend the present results beyond the two sums considered here and provide a more systematic non-asymptotic theory for empirical Bayes estimation of sums of random variables.


\clearpage

\appendix

\section{Proofs of Section \ref{sec1}}
	
	\subsection{Proof of \prettyref{thm:exist-genr}}
	\label{app:existence}

	\subsubsection{Regular procedure}
	\label{app:exist-regular}
	For any $\theta\in\reals_+$ let $\delta_\theta$ denote the Dirac measure at $\theta$. Suppose that the training sample used to compute
	\eqref{eq:optim-genr} is $X'_1,\dots,X'_m$ and define $X'_{(m)}=\max_{j=1}^m X'_j$. Denote the set
	\begin{align}
		\{y_1,\dots,y_k\}
		=
		\{0 \le y \le X'_{(m)} : \widehat p_m(y) > 0\},\quad
		y_1<y_2<\dots<y_k.
	\end{align} 
	%
    Define the set
	\begin{align}\label{eq:EB_mS-genr}
		S\triangleq\sth{f_{Q}(y): y\in \{y_1,\dots,y_k\}: Q\in \mathcal C},
	\end{align} 
	We claim that $S$ is convex and compact.
	The convexity follows from definition. For compactness, note that $S$ is bounded since $\sup_{\theta\geq 0} f_\theta(y) = e^{-y} y^y/y!$, so it suffices to check $S$ is closed. Let 
	$(f_1',\dots,f_{k}')\in\reals_+^{k}$ be the limiting point of $\sth{f_{Q_k}(y): y\in \{y_1,\dots,y_k\}}$ for some sequence 
	$\{Q_k\}$ in $\mathcal C$.  The sequence of distributions $\{Q_k\}_{k\geq 1}$ is tight since, for every $k\geq 1$, $Q_k\in \mathcal C\subseteq \subexpo(s)$, and $\subexpo(s)$ is a tight class. By Prokhorov's theorem, there is a subsequence $\{Q_{k_\ell}\}$ that converges weakly to some $Q'$, and since $\mathcal C$ os closed $Q'\in\mathcal C$. Then, by property of the weak convergence, as $\theta \mapsto f_\theta(y)$ is bounded and continuous, we get that $f_{Q_{k_\ell}}(y)\to f_{Q'}(y)$ for each $y\in \{y_1,\dots,y_K\}$. 
	Since $\theta \mapsto f_\theta(y)$ is continuous and bounded, we have $f_j'=f_{Q'}(y_j)$ for all $j=1,\dots,k$. In other words, $S$ is closed.
	
	Define $v:S \to \reals$ by $v(f_1,\dots,f_k)=\sum_{i=1}^k
	{\ell( \widehat p_m(y_i),f_i)}$.
	By \prettyref{asmp:EB_t-ell}, the value of the min-distance optimization can be written as
	\begin{equation}
		\min_{Q\in\mathcal C} d(\widehat p_m\|f_{Q}) = t(\widehat p_m) + \min_{(f_1,\ldots,f_k) \in S} v(f_1,\ldots,f_k),
		\label{eq:mindist-equiv}
	\end{equation}
	where the equality holds since, by assumption $\ell(0,b)\equiv 0$. Since $b\mapsto \ell(a,b)$ is strictly convex for $a>0$, then $v$ is strictly convex.
	Therefore, there exists a unique point $(f_1^*,\dots,f_k^*) \in S$ that achieves the minimum on the right side of \prettyref{eq:mindist-equiv}. 
	Thus, 
	the left side has a minimizer $\hatG\in\mathcal C$ that satisfies 	
	\begin{align}\label{eq:EB_m4}
		 f_{\hatG}(y_i)=f_i^*,
		\quad i=1,\dots,k.
	\end{align}

	\subsubsection{Coarsened procedure}
	For any $\theta\in\reals_+$ let $\delta_\theta$ denote the Dirac measure at $\theta$. Denote the set
	\begin{align}
		\{y_1,\dots,y_k\}
		=
		\{0 \le y \le T : \widehat p_m(y) > 0\},
		y_1<y_2<\dots<y_k,
		\quad
		\calJ_T
		=\sth{y_1,\dots,y_k,T+1}.
	\end{align} 
    Define
	\begin{align}\label{eq:EB_mS}
		S\triangleq\sth{\tilde f^{(T)}_{Q}(y): y\in \calJ_T: Q\in \mathcal C},
	\end{align} 
	where $\tilde f^{(T)}_Q(y) = \Expect_{\theta \sim Q}[\tilde f^{(T)}_\theta(y)]$ is the binned-probability mass defined in \eqref{eq:coarsening}.
	We claim that $S$ is convex and compact.
	The convexity follows from definition. For compactness, note that $S$ is bounded since $\sup_{\theta\geq 0} f_\theta(y) = e^{-y} y^y/y!$, so it suffices to check $S$ is closed. Let 
	$(f_1',\dots,f_{k}',f_{T+1}')\in\reals_+^{k+1}$ be the limiting point of $\sth{\tilde f^{(T)}_{Q_k}(y): y\in \calJ_T}$ for some sequence 
	$\{Q_k\}$ in $\mathcal C$.  The sequence of distributions $\{Q_k\}_{k\geq 1}$ is tight since, for every $k\geq 1$, $Q_k\in \mathcal C\subseteq \subexpo(s)$,  and $\subexpo(s)$ is a tight class. 	By Prokhorov's theorem, there is a subsequence $\{Q_{k_\ell}\}$ that converges weakly to some $Q'\in \mathcal C$.
	Since $\theta \mapsto f_\theta(y)$ is continuous and bounded, we have $f_j'=\tilde f^{(T)}_{Q'}(y_j)$ for all $j=1,\dots,k$ and $f_{T+1}'=\tilde f^{(T)}_{Q'}(T+1)$. This implies that, $S$ is closed.
	
	Define $v:S \to \reals$ by $v(f_1,\dots,f_k,f_{T+1})=\ell(\tilde p^{(T)}(T+1),f_{T+1})+
	\sum_{i=1}^k
	{\ell(\tilde p^{(T)}(y_i),f_i)}$.
	By \prettyref{asmp:EB_t-ell}, the value of the min-distance optimization can be written as
	\begin{equation}
		\min_{Q\in\mathcal C} d(\widetilde{\widehat p}^{(T)}_m\|\tilde f^{(T)}_{Q}) = t(\widetilde{\widehat p}^{(T)}_m) + \min_{(f_1,\ldots,f_k,f_{T+1}) \in S} v(f_1,\ldots,f_k,f_{T+1}),
		\label{eq:mindist-equiv-coarse}
	\end{equation}
	where the last equality holds by assumption $\ell(0,b)\equiv 0$. Since $b\mapsto \ell(a,b)$ is strictly convex for $a>0$, then $v$ is strictly convex.
	Therefore, there exists a unique point $(f_1^*,\dots,f_k^*,f_{T+1}^*) \in S$ that achieves the minimum on the right side of \prettyref{eq:mindist-equiv-coarse}. 
	Thus, 
	the left side has a minimizer $\hatG\in\mathcal C$ that satisfies 	
	\begin{align}
		\tilde f^{(T)}_{\hatG}(y_i)=f_i^*,
		\quad i=1,\dots,k,\quad 
		\tilde f^{(T)}_{\hatG}(T+1)=f_{T+1}^*.
	\end{align}


\subsection{Proof of \prettyref{lem:EB_density-H2}}\label{prooflemma21}

The proof is similar to the proof of \cite[Theorem 2]{jana2025optimal}. In particular, we outline below the key steps of the proof, which will be used later for proving our new some results. We first note the following lemmas about tails bounds on Poisson mixtures from \cite{jana2025optimal}.

\begin{lemma}\label{lmm:bounded_prior_properties}
	Let $h>0$ and $\Gstar\in\calP([0,h])$. If $\sth{Z_i}_{i=1}^m\iiddistr f_{\Gstar}$, then the following are satisfied given  any $t\geq 0,a\geq 1, m\geq 3$ and $K=\min\sth{{a(he^2+2)\log m\over \log\log m},he^2+a\log m}$.
	$$ \PP\qth{Z_1>K+t}\leq {2e^{-t}\over m^a},\quad \EE[Z_{(m)}^4]\leq 5K^4+ 8.
	$$
\end{lemma}

\begin{lemma}\label{lmm:subexpo_properties}
	Given any $s>0$ and $\Gstar\in\subexpo(s)$, the following are satisfied.
	\begin{enumerate}[label=(\roman*)~]
		\item \label{pt:i} If $\theta\sim {\Gstar}$, then $\EE[\theta^4]\leq 12s^4$.
		\item \label{pt:ii} If $\sth{Z_i}_{i=1}^m\iiddistr f_{\Gstar}$, then for any $K\geq 1$
		\begin{align}
			\label{eq:EB_prop_mixture}
			\PP\qth{Z_1\geq K}
			\leq 3 e^{-K\log\pth{1+\frac 1{2s}}},
			\quad
			\EE\qth{Z_{(m)}^4}\leq {64(\log m)^4+90\over \pth{\log\pth{1+\frac 1{2s}}}^4}.
		\end{align}
	\end{enumerate}
\end{lemma}

 We next prove a general inequality.
\begin{lemma}\label{lmm:hel-bound-genr}
	Suppose that $Z_1,\dots,Z_m\simiid f$ for some distribution $f\in \calP(\integers)$ and fix $J\geq 1$. Let $\hat q_m$ be the empirical distribution of $Z_1,\dots,Z_m$, and $\hat f \in \calP(\integers)$ satisfies
	$
	d(\hat q_m\|\hat f)\leq d(\hat q_m\|f).
	$ 
	Then we have 
	$$\EE_f\!\left[H^2(\hat f,f)\right]\le\frac{4c_2}{c_1}\,\frac{J}{m}+\Big(\frac{4c_2}{c_1}+2m\Big)\PP_{Z\sim  f}[Z\geq J].$$
\end{lemma}
\begin{proof}

Using $d(\hat q_m\|\hat f)\le d(\hat q_m\|f)$, the triangle inequality for the Hellinger metric, and \prettyref{asmp:dist_sandwich}, we get
\begin{align}
	H^2(\hat f,f)
	&\le\Big(H(\hat q_m,\hat f)+H(\hat q_m,f)\Big)^2
	\le 2H^2(\hat q_m,\hat f)+2H^2(\hat q_m,f) \nonumber\\
	&
	\le\frac{2}{c_1}d(\hat q_m\|\hat f)+\frac{2}{c_1}d(\hat q_m\|f)
	\le\frac{4}{c_1}d(\hat q_m\|f).
	\label{eq:hell-by-optimality-feasible}
\end{align}
Let $Z_{(m)}$ denote the maximum of $Z_1,\dots,Z_m$.
Using \eqref{eq:hell-by-optimality-feasible} and the fact that $H^2\le 2$, we obtain
\begin{align}
	\EE_f\!\left[H^2(\hat f,f)\right]
	&=\EE_f\!\left[H^2(\hat f,f)I(Z_{(m)}<J)\right]+\EE_f\!\left[H^2(\hat f,f)I(Z_{(m)}\ge J)\right]
	\nonumber\\
	&\le\frac{4}{c_1}\EE_f\!\left[d(\hat q_m\|f)I(X'_{(m)}<J)\right]+2\,\PP_{f}(Z_{(m)}\ge J).
	\label{eq:split-Ymax-feasible}
\end{align}
By the upper bound in \prettyref{asmp:dist_sandwich}, we bound the first term in the right hand side of \prettyref{eq:split-Ymax-feasible} as
\begin{align}
	\label{eq:app0}
	\frac{1}{c_2}\EE_f\!\left[d(\hat q_m\|f)I(Z_{(m)}<J)\right]\le\EE_f\!\left[\chi^2(\hat q_m\|f)I(Z_{(m)}<J)\right].
\end{align}
The right hand side of the above display can be further bounded as
\begin{align}
	\EE_f\!\left[\chi^2(\hat q_m\|f)I(Z_{(m)}<J)\right]
	&=
	\sum_{y\ge 0}\frac{\EE_f[(\hat q_m(y)-f(y))^2I(Z_{(m)}<J)]}{f(y)}
	\nonumber\\
	&=\sum_{y<J}\frac{\EE_f[(\hat q_m(y)-f(y))^2I(Z_{(m)}<J)]}{f(y)}+\sum_{y\ge J}\frac{\EE_f[(\hat q_m(y)-f(y))^2I(Z_{(m)}<J)]}{f(y)}.
	\label{eq:chi-split-feasible}
\end{align}
On the event $\{Z_{(m)}<J\}$ we have $\hat q_m(y)=0$ for all $y\ge J$, so we can bound the second sum as
\[
\sum_{y\ge J}\frac{f(y)^2\,\PP_f(Z_{(m)}<J)}{f(y)}=\PP_{Z\sim f}[Z\geq J]\,\PP_G(Z_{(m)}<J)\le\PP_{Z\sim f}[Z\geq J].
\]
For the first sum in \prettyref{eq:chi-split-feasible}, we use the bound
\[
\sum_{y<J}\frac{\EE_f[(\hat q_m(y)-f(y))^2I(Z_{(m)}<J)]}{f(y)}\le\sum_{y<J}\frac{\EE_f[(\hat q_m(y)-f(y))^2]}{f(y)}.
\]
Using the following results
\[
\EE_f[\hat q_m(y)]=f(y),\qquad\Var_f(\hat p_m(y))=\frac{f(y)(1-f(y))}{m},
\]
we get,
\[
\sum_{y<J}\frac{\EE_f[(\hat p_m(y)-f(y))^2]}{f(y)}=\frac1m\sum_{\substack{y<J\\ f(y)>0}}(1-f(y))\le\frac{J}{m}.
\]
Combining these bounds with \eqref{eq:chi-split-feasible} and \eqref{eq:app0} yields
\[
\EE_f\!\left[d(\hat q_m\|f)I(Z_{(m)}<J)\right]\le c_2\Big(\frac{J}{m}+\PP_{Z\sim f}[Z\geq J]\Big).
\]
In view of the last result and the fact that we can bound $\PP_f(Z_{(m)}\ge J)$ from above using the union bound $\PP_f(Z_{(m)}\ge J)\le m\,\PP_{Z\sim f}[Z\geq J]$, we continue \eqref{eq:split-Ymax-feasible} to get the desired inequality
\[
\EE_f\!\left[H^2(\hat f,f)\right]\le\frac{4c_2}{c_1}\Big(\frac{J}{m}+\PP_{Z\sim f}[Z\geq J]\Big)+2m\,\PP_{Z\sim f}[Z\geq J].
\]
	
\end{proof}

Based on Lemmas \ref{lmm:bounded_prior_properties}--\ref{lmm:hel-bound-genr}, we are in the position of proving Lemma \ref{lem:EB_density-H2}. To apply Lemma \ref{lmm:hel-bound-genr}, we use $f_{\Gstar}$ in place of $f$ and $f_{\hat G_{d,m}}$ in the place of $\hat f$, as described in the result statement. To simplify notations, for any integer $J\ge 1$, define
\[
\Xi_J(\Gstar):=\sum_{y\ge J} f_{\Gstar}(y),\qquad Y\sim f_{\Gstar}.
\]
Hence, we can rewrite the bound in \prettyref{lmm:hel-bound-genr} as  
\begin{align}\label{eq:Xi-bound}
	\EE_{\Gstar}\!\left[H^2( f_{\hatG_{d,m}},f_{\Gstar})\right]\le\frac{4c_2}{c_1}\,\frac{J}{m}+\Big(\frac{4c_2}{c_1}+2m\Big)\Xi_J(\Gstar).
\end{align}
In the case of bounded prior, Lemma \ref{lmm:bounded_prior_properties}, with
$$
J_m=C_1\min\sth{\max\{1,\hat h\}{\log m\over \log\log m},\hat h+\log m}
$$
leads to the tail bound $\Xi_{J_m}(\Gstar)\leq m^{-2}$ for a suitably large constant $C_1>0$. This implies the desired result in the case of bounded priors. 
In the case of subexponential priors, applying Lemma \ref{lmm:subexpo_properties}    with 
\[
J_m:=\max\!\left\{1,\left\lceil \frac{2\log m}{\log(1+1/2s)}\right\rceil\right\},
\]
we  get $
\Xi_{J_m}(\Gstar)\le \bar C_s m^{-2}$, where $\bar C_s$ is a constant only depending on $s$. 
Substituting the above in \eqref{eq:Xi-bound} we get
\[
\EE_{\Gstar}\!\left[H^2(f_{\hat G_{\mathcal C}},f_{\Gstar})\right]\le\frac{4c_2}{c_1}\frac{J_m}{m}+\Big(\frac{4c_2}{c_1}+2m\Big) C_s m^{-2},
\]
for a constant $C_s$ only depending on $s$. Since $J_m\asymp \log m$ with constants depending only on $s$, the right-hand side is bounded by
\[
C_d(s,c_1,c_2)\,\frac{\log m}{m}
\]
for some finite constant $C_d(s,c_1,c_2)$. This completes our proof.

\subsection{Proof of Lemma \ref{lem:EB_density-H2_coars}}\label{prooflemma22}

Before proving the Lemma, we state a preliminary lemma.

\begin{lemma}\label{lem:E0}
Let $d$ satisfy Assumptions \ref{asmp:EB_t-ell} and \ref{asmp:dist_sandwich}, and let $\Gstar$ denote the mixing measure of the data-generating distribution. 
Let  $\widetilde G^{(T)}_{d,m}$ be the $T$-coarsened minimum-$d$-distance estimate of $\Gstar$ defined in~(\ref{eq:optim-coarse}), where the optimization is performed over an appropriate class as in
Lemma~\ref{lem:EB_density-H2}. Then, 
\begin{equation}
\label{eq:E0-main}
\mathbb E\Big[H^2\big(\widetilde f_{\Gstar}^{(T)},\,\widetilde f_{\widehat
G^{(T)}_{d,m}}^{(T)}\big)\Big]
\;\le\;\left\{
\begin{array}{ll}
\frac{4c_2}{c_1}\,\frac{J}{m}+\Big(\frac{4c_2}{c_1}+2m\Big)\,\mathbb P_{X\sim f_{\Gstar}}[X\ge J]& 1\le J\le T+1\\
\frac{4c_2}{c_1}\,\frac{J}{m} & J=T+2.
\end{array}
\right.
\end{equation}
\end{lemma}
\begin{proof}
Let $(X'_1,\dots,X'_m)$ denote the training sample, and for $i=1,\dots,m$, let
\[
Z_i := X_i'\, I(X_i'\le T) + (T+1) \,I(X_i'\ge T+1).
\]
Since $X_1',\dots,X_m'$ are i.i.d.\ with common density $f_{\Gstar}$, then the variables $Z_1,\dots,Z_m$ are i.i.d.\ with common density $f:=\widetilde f_{\Gstar}^{(T)}$. Let $\widehat q_m$ denote the empirical distribution of $Z_1,\dots,Z_m$. By definition $\widehat q_m=\widetilde p_m^{(T)}$, the $T$-coarsening of the empirical distribution of the training sample. Let $\widehat f:=\widetilde f_{\widetilde G^{(T)}_{d,m}}^{(T)}$. Then, by definition 
\[
d\big(\widehat q_m\,\|\,\widehat f\big)\;\le\; d\big(\widehat q_m\,\|\,f\big).
\]
Thus $(f,\widehat f,\widehat q_m)$ satisfy the hypotheses of Lemma~\ref{lmm:hel-bound-genr}. It follows that, for every integer $J\ge 1$,
\begin{equation}
\label{eq:E0-B3}
\mathbb E_f\big[H^2(f,\widehat f)\big]
\;\le\; \frac{4c_2}{c_1}\,\frac{J}{m}
+ \Big(\frac{4c_2}{c_1}+2m\Big)\,\mathbb P_{Z\sim f}[Z\ge J].
\end{equation}
By definition,
$$\mathbb P_{Z\sim f}[Z\ge J]=\mathbb P_{X\sim f_{\Gstar}}[X\ge J]$$ 
for every $1\le J\le T+1$, while $\mathbb P_{Z\sim f}[Z\ge T+2]=0$. Substituting into \eqref{eq:E0-B3} and noting
$H^2(f,\widehat f)=H^2(\widetilde f_{\Gstar}^{(T)},\widetilde f_{\widehat
G^{(T)}_{d,m}}^{(T)})$ gives \eqref{eq:E0-main}.
\end{proof}

Based on Lemma \ref{lem:E0}, we are now in the position of proving Lemma \ref{lem:EB_density-H2_coars}. Taking $J=T+2$ in \eqref{eq:E0-main} gives
\begin{equation}
\label{eq:E0-i}
\mathbb E\Big[H^2\big(\widetilde f_{\Gstar}^{(T)},\,\widetilde f_{\widehat
G^{(T)}_{d,m}}^{(T)}\big)\Big] \;\le\; \frac{4c_2}{c_1}\,\frac{T+2}{m}.
\end{equation}

(a)       By Lemma~\ref{lmm:bounded_prior_properties} with $a=2$ and 
$$K:=2(\hat h e^2+2)\frac{\log m}{\log\log m}$$ we obtain
$$\mathbb P_{X'\sim f_{\Gstar}}[X'>K]\le 2/m^2.$$
If $K>T+2$, then
\[
\min\left\{
\frac{4c_2}{c_1}(T+2),
\left(\frac{12c_2}{c_1}+4\right)K
\right\}
=
\frac{4c_2}{c_1}(T+2),
\]
and hence \eqref{eq:final-bounded} follows from \eqref{eq:E0-i}.

Suppose now that $K\le T+2$, and set $J=\min(\lfloor K\rfloor +1,T+2)$.

If $J\le T+1$, then $J=\lfloor K\rfloor+1$ and, since $X$ is integer-valued,
\[
\mathbb{P}_{X\sim f_{\Gstar}}(X\ge J)=\mathbb{P}_{X\sim f_{\Gstar}}(X>K)\le \frac{2}{m^2}.
\]
If $J=T+2$, then the second bound in \eqref{eq:E0-main} yields an even smaller upper bound. Hence, in either case,
\[
\begin{aligned}
\mathbb{E}\left[ H^2\left(\widetilde f^{(T)}_{\Gstar},\widetilde f^{(T)}_{\widetilde G^{(T)}_{d,m}}\right)\right]
&\le\frac{4c_2}{c_1}\frac{J}{m}+\left(\frac{4c_2}{c_1}+2m\right)\frac{2}{m^2}
\\
&\le\frac{4c_2}{c_1}\frac{K+1}{m}+\frac{8c_2}{c_1m^2}+\frac{4}{m}
\\
&=\frac{4c_2}{c_1}\frac{K}{m}+\frac{4c_2}{c_1m}\left(1+\frac{2}{m}\right)+\frac{4}{m}.
\end{aligned}
\]
Since $m\ge3$ and $K\ge1$,
\[
1+\frac{2}{m}\le\frac53\le2K,\qquad\frac4m\le\frac{4K}{m}.
\]
Therefore,
\[
\begin{aligned}
\mathbb{E}\left[H^2\left(\widetilde f^{(T)}_{\Gstar},\widetilde f^{(T)}_{\widetilde G^{(T)}_{d,m}}\right)\right]
&\le\left(\frac{12c_2}{c_1}+4\right)\frac{K}{m}
=\left(\frac{12c_2}{c_1}+4\right)(\bar h e^2+2)\frac{2\log m}{m\log\log m}.
\end{aligned}
\]
Together with \eqref{eq:E0-i}, this proves \eqref{eq:final-bounded}.

(b) By Lemma~\ref{lmm:subexpo_properties} with
$$K:=\frac{2\log m}{\log(1+1/(2s))}$$
we obtain
$$\mathbb P_{X'\sim
f_{\Gstar}}[X'\ge K]\le 3e^{-K\log(1+1/(2s))}\le 3/m^2.$$
We apply Lemma  \ref{lem:E0} with
\[
J:=\min\{\lfloor K\rfloor+1,T+2\}.
\]
If $J\le T+1$, then $J=\lfloor K\rfloor+1$ and, since $X$ is integer-valued,
\[
\mathbb{P}_{X\sim f_{\Gstar}}(X\ge J)=\mathbb{P}_{X\sim f_{\Gstar}}(X>K)\le\mathbb{P}_{X\sim f_{\Gstar}}(X\ge K)
\le \frac{3}{m^2}.
\]
If $J=T+2$, then the second bound in \eqref{eq:E0-main}
yields an even smaller upper bound. Hence, in either case,
\[
\begin{aligned}
\mathbb{E}\left[H^2\left(\widetilde f^{(T)}_{\Gstar},\widetilde f^{(T)}_{\widetilde G^{(T)}_{d,m}}\right)\right]
&\le\frac{4c_2}{c_1}\frac{J}{m}+\left(\frac{4c_2}{c_1}+2m\right)\frac{3}{m^2}\\
&\le\frac{4c_2}{c_1}\frac{K+1}{m}+\frac{12c_2}{c_1m^2}+\frac{6}{m}\\
&=\frac{4c_2}{c_1}\frac{K}{m}+\frac{4c_2}{c_1m}\left(1+\frac{3}{m}\right)+\frac{6}{m}.
\end{aligned}
\]
Since $m\ge2$ and $K\ge1$,
\[
1+\frac{3}{m}\le \frac{5}{2}\le 3K,\qquad\frac{6}{m}\le\frac{6K}{m}.
\]
Therefore,
\[
\begin{aligned}
\mathbb{E}\left[H^2\left(\widetilde f^{(T)}_{\Gstar},\widetilde f^{(T)}_{\widetilde G^{(T)}_{d,m}}\right)\right]
&\le\left(\frac{16c_2}{c_1}+6\right)\frac{K}{m}\\
&=\left(\frac{16c_2}{c_1}+6\right)\frac{2\log m}{m\log(1+1/(2s))}.
\end{aligned}
\]
Together with \eqref{eq:E0-i}, this proves \eqref{eq:final-subexp}, completing the proof.


\section{Proofs of Section \ref{sec2}}
To simplify notation, we write $k(x\mid\theta)$ for the Poisson probability mass function $\mathrm{Poisson}(x\mid\theta)$.

\subsection{Proof of Theorem \ref{th:consist1}}\label{app:proofconst1}
The proof of Theorem~\ref{th:consist1} relies on the following lemma, which immediate consequence of the definition of the class $\subexpo(s)$; for completeness, we include a short proof. 

\begin{lemma}\label{lemma:pre}
   The class $\subexpo(s)$ is tight and $\{f_G:G\in\subexpo(s)\}$
is uniformly dominated by a summable sequence. Furthermore, the mixing distributions in $\subexpo(s)$ are identifiable, i.e. 
$$
G,G'\in\subexpo(s) \mbox{ and } f_G(x)=f_{G'}(x) \,\forall x\in\integers_+\Rightarrow G=G'.
$$ 
\end{lemma}
\begin{proof}
  Denoting by $M_G$ the moment generating function of $G$, we can write for $a\in (1,1+\frac 1 s)$,
$$
f_G(x)=\int_0^\infty \frac{\exp{(-\theta)}\theta^x}{x!}G(d\theta)\leq \int \exp(-\theta)\frac{\exp(a\theta)}{a^x}G(d\theta)=\frac{M_G(a-1)}{a^x}\leq \frac{C}{a^x},
$$
where the first inequality holds using the expansion of $\exp(a\theta)$ for any $a$ and the last inequality follows from 
\begin{align*}
    M_G(a-1)=\int_0^\infty \exp((a-1)\theta)G(d\theta)&=1+(a-1)\int_0^\infty \exp((a-1)t)\overline G(t)dt\\&\leq 1+2(a-1)\int_0^\infty \exp\left(\left(a-1-\frac 1 s\right)t\right)dt<+\infty,
\end{align*}
with $\overline G(t)=G(t,\infty)$. Identifiability follows from the standard identifiability of Poisson mixtures \citep{Fel43,Tei61}.
\end{proof}

\begin{lemma}\label{lemma:convuG}
Under the assumptions of Theorem \ref{th:consist1} (a) or Theorem \ref{th:consist2} (a), let $u_G(x):=\EE_G\left[u(X_i,\theta_i)\mid X_i=x\right]$ for $G\in\mathcal C$ and $x\in\mathbb Z_+$. If $G_n\in\mathcal C$ converges weakly to $\Gstar$, then, for every $x$ such that $f_{\Gstar}(x)>0$, it holds  $$u_{G_n}(x)\rightarrow u_{\Gstar}(x).$$
\end{lemma}
\begin{proof}
    We can write, for every $G\in\mathcal C$,
    $$
  u_G(x):=\EE_G\left[u(X_i,\theta_i)\mid X_i=x\right]=\int_\Theta u(x,\theta)G(d\theta\mid x).
  $$
  By assumption, the function $\theta\mapsto u(x,\theta)$  is  $G$-almost surely continuous for every $x$ and every $G\in\mathcal C$. Let $N_x$ be the set of discontinuities of $u(x,\theta)$, By assumption, $G(N_x)=0$. Since $G(\cdot\mid x)$ is absolutely continuous with respect to $G$, then $G(N_x\mid x)=0$. In other words, the function $\theta\mapsto u(x,\theta)$ is continuous almost surely with respect to $G(\cdot\mid x)$ for every $G\in\mathcal C$. 
 By assumption \eqref{eq:condition2}, for every $x$ such that $f_{\Gstar}(x)>0$,
\[
\sup_{G\in\mathcal C}\int_\Theta |u(x,\theta)|^2\,G(d\theta\mid x)<\infty.
\]
Thus the function $\theta\mapsto u(x,\theta)$ is uniformly integrable with respect to the family of probability measures $\{G_n(\cdot\mid x)\}$.\\
We now prove that $G_n(\,\cdot\,\mid x)$ converges weakly to $\Gstar(\,\cdot\,\mid x)$.
Since $\theta\mapsto k(x\mid\theta)$ is continuous, then $$f_{G_n}(x)=\int_\Theta k(x\mid\theta)G_n(d\theta)\rightarrow f_{\Gstar}(x).$$ 
This, together with  $f_{\Gstar}(x)>0$, implies $f_{G_n}(x)$ is bounded away from zero.\\
Let $\varphi:\Theta\to\mathbb R$ be a bounded, continuous function. Then, 
\[
\int_\Theta \varphi(\theta)\,G_n(d\theta\mid x)=\frac{\int_\Theta \varphi(\theta)k(x\mid\theta)\,G_n(d\theta)}{
f_{G_n}(x)}.
\]
Since both $k(x\mid\theta)$ and $\theta\mapsto\varphi(\theta)k(x\mid\theta)$ are bounded and continuous, and since the denominator converges to $f_{\Gstar}(x)>0$ and is bounded away from zero, then the weak convergence $G_n$ implies
\[
\int_\Theta \varphi(\theta)\,G_n(d\theta\mid x)\rightarrow\int_\Theta \varphi(\theta)\,\Gstar(d\theta\mid x),
\]
and hence $G_n(\,\cdot\,\mid x)$ converges weakly to $\Gstar(\,\cdot\,\mid x)$.
The generalized version of Portmanteau theorem yields
\[
u_{G_n}(x)\rightarrow u_{\Gstar}(x).
\]
\end{proof}

Based on Lemmas \ref{lemma:pre}--\ref{lemma:convuG}, we are now in the position of proving Theorem \ref{th:consist1}. Since the random variables
$u(X_i,\theta_i)- u_{\Gstar}(X_i)$ $(i=1,\ldots,n)$ are independent, identically distributed, and have mean zero, the central limit theorem implies that
\[
(n\log n)^{-1/2}\bigl(S_n-\widehat S_n(\Gstar)\bigr)=(n\log n)^{-1/2}\sum_{i=1}^n (u(X_i,\theta_i)- u_{\Gstar}(X_i))\overset{\mathbb P}{\longrightarrow}0.
\]
Consequently, the assertions in parts (a) and (b) are equivalent to
\[
\frac{1}{n}\bigl(\widehat S_n(\Gstar)-\widehat S_n(\widehat G_{d,m_n})\bigr)=o_{\mathbb P}(1),
\]
and
\[
\frac{1}{n}\bigl(\widehat S_n(\Gstar)-\widehat S_n(\widehat G_{d,m_n})\bigr)=O_{\mathbb P}(n^{-1/2}
(\log n)^{1/2}),
\]
respectively. 
Denote
 $$
 Z_{n,i}=(u_{\Gstar}(X_i)-\overline u_{\Gstar})-(u_{\widehat G_{d,m_n}}(X_i)-\overline u_{\widehat G_{d,m_n}}),
 $$
  where
    $$
    \overline u_{\Gstar}:=\EE[u_{\Gstar}(X_i)],\quad\quad 
    \overline u_{\widehat G_{d,m_n}}:=\EE[u_{\widehat G_{d,m_n}}(X_i)\mid \widehat G_{d,m_n}].
    $$
    The random variables $\{Z_{n,i}:i=1,\dots,n\}$ are centered and conditionally i.i.d., given the sigma algebra generated by $(X'_{n})_{n\geq 1}$, where, for every $n$, $(X'_n)=(X'_{n,1},\dots,X'_{n,m_n})$ is the training sample. Moreover, 
\begin{align*}
     \EE[ (\sum_{i=1}^n Z_{n,i})^2\mid (X'_m)_{m\geq 1}]&=\operatorname{Var}(\sum_{i=1}^n (u_{\Gstar}(X_i)-\overline u_{\Gstar})-(u_{\widehat G_{d,m_n}}(X_i)-\overline u_{\widehat G_{d,m_n}})\mid (X'_m)_{m\geq 1})\\
     &\leq 4n \sup_{G\in\mathcal C}\operatorname{Var}(u_G(X_i))\\
     &\leq Cn,
\end{align*}
where $C=4 \sup_{G\in\mathcal C}
\sum_{x\ge0} u_G(x)^2 f_{\Gstar}(x) <\infty$ by Jensen's inequality and condition \eqref{eq:condition2}.
It follows that, for every $\epsilon >0$, 
  \begin{align*}
      \PP\left[ \left|\frac{1}{(n\log n)^{1/2}} \sum_{i=1}^n Z_{n,i}\right|>\epsilon
        \right]\leq \frac{\EE[ (\sum_{i=1}^n Z_{n,i})^2]}{n\log n\epsilon^2}\rightarrow 0,
  \end{align*}
  as $n\rightarrow\infty$. Thus,
  \begin{equation}\label{eq:convZ}
  \sum_{i=1}^n Z_{n,i}  =o_{\PP}((n\log n)^{1/2}),
  \end{equation}
  as $n\rightarrow\infty$.
  
  First, we deal with the proof of part (a) of Theorem \ref{th:consist1}. We can write
    \begin{align}\label{eq:decompgendist}
    \frac{\widehat S_n(\Gstar)-\widehat S_n(\widehat G_{d,m_n})}{n}&=\frac{\sum_{i=1}^n Z_{n,i}}{n}+(\overline u_{\Gstar}-\overline u_{\widehat G_{d,m_n}})=(\overline u_{\Gstar}-\overline u_{\widehat G_{d,m_n}})+o_{\PP}(1),
    \end{align}
   where the last equality comes from \eqref{eq:convZ}.
We now prove that $\overline u_{\widehat G_{d,m_n}}$ converges in probability to $\overline u_{\Gstar}$ as $n\rightarrow\infty$. 
For every $G\in\mathcal C$, 
  $$
  u_G(x):=\EE_G\left[u(X_i,\theta_i)\mid X_i=x\right]=\int_\Theta u(x,\theta)G(d\theta\mid x)=\frac{\int_\Theta u(x,\theta)k(x\mid\theta)G(d\theta)}{f_G(x)}.
  $$
 By the assumption \eqref{eq:condition2} the class $\{u_G(x):G\in\mathcal C\}$ is uniformly integrable with respect to $f_{\Gstar}(x)$ (see Lemma \ref{lemma:pre}).
Since
$$
\overline u_{\widehat G_{d,m_n}}=\sum_{x\in \mathbb Z_+ }u_{\widehat G_{d,m_n}}(x)f_{\Gstar}(x),\quad \overline u_{\Gstar}=\sum_{x\in \mathbb Z_+}u_{\Gstar}(x)f_{\Gstar}(x),
$$
and the class $\{u_G(x):G\in\mathcal C\}$ is uniformly integrable with respect to $f_{\Gstar}(x)$, then to prove that $\overline u_{\widehat G_{d,m_n}}$ converges in probability to $\overline u_{\Gstar}$ it is enough to establish that $u_{\widehat G_{d,m_n}}(x)$ converges in probability to $u_{\Gstar}(x)$ for every $x$ such that $f_{\Gstar}(x)>0$. In turn, by Lemma \ref{lemma:convuG}, a sufficient condition is that $\widehat G_{d,m_n}$ converges in probability to $\Gstar$.
By Lemma \ref{lem:EB_density-H2}, there exists a constant $c$ such that, for $n$ large enough
    $$
   \EE[H^2( f_{\widehat G_{d,m_{n}}}, f_{\Gstar})]\leq \frac{c\log m_{n}}{m_{n}},
    $$
which implies that $H( f_{\widehat G_{d,m_{n}}}, f_{\Gstar})$ converges to zero in probability. Since the total variation distance is dominated by the Hellinger distance, then 
$$
\sum_{x=0}^\infty|  f_{\widehat G_{d,m_{n}}}(x)- f_{\Gstar}(x)|\stackrel{\PP}{\rightarrow }0.$$
We now prove that $\hat G_{d,m_n}$ converges weakly to $\Gstar$ in probability. The space $\mathcal C$ is weakly closed and uniformly tight by Lemma \ref{lemma:pre}. Hence, it is weakly compact.
Consider the map
\[
G\mapsto f_G
\]
defined on $\mathcal C$ and with values in $\ell^1(\mathbb Z_+)$. By Lemma \ref{lemma:pre} and the continuity of $k(x\mid\theta)$, this map is continuous. Moreover, again by Lemma \ref{lemma:pre}, it is injective. Thus, it  has a continuous inverse on its image (see, e.g., \citet[Theorem 26.6]{Munkres}). Hence,
\[
\|f_{\widehat G_{d,m_n}}-f_{\Gstar}\|_1
\xrightarrow{\mathbb P}0
\]
implies that $\widehat G_{d,m_n}$ converges weakly to $\Gstar$ in probability. This concludes the proof of (a).

Now, we deal with the proof of part (b) of Theorem \ref{th:consist1}. We can write
    \begin{align}\label{eq:decompGrid}
    \frac{\widehat S_n(\Gstar)-\widehat S_n(\widehat G_{d,m_n})}{(n\log n)^{1/2}}&=\frac{\sum_{i=1}^n Z_{n,i}}{(n\log n)^{1/2}}+\frac{n^{1/2}}{(\log n)^{1/2}}(\overline u_{\Gstar}-\overline u_{\widehat G_{d,m_n}})=\frac{n^{1/2}}{(\log n)^{1/2}}(\overline u_{\Gstar}-\overline u_{\widehat G_{d,m_n}})+o_\PP(1),
    \end{align}
    where the last equality comes from \eqref{eq:convZ}.
    We now prove that 
    $$\frac{n^{1/2}}{(\log n)^{1/2}}(\overline u_{\widehat G_{d,m_n}}-\overline u_{\Gstar})=O_\PP(1),$$ as $n\rightarrow\infty$. At this aim, it is sufficient to prove that, for every $a_n\rightarrow\infty$,
     \begin{equation}\label{eq:conda}
     \frac{n^{1/2}}{a_n(\log n)^{1/2}}(\overline u_{\widehat G_{d,m_n}}-\overline u_{\Gstar})=o_\PP(1).
     \end{equation}
    Denote by $\hat g_{m_n}$ and $\Gstar$ the probability mass functions corresponding to $\widehat G_{d,m_n}$ and ${\Gstar}$, respectively. We can write
$$
\overline u_{\widehat G_{d,m_n}}-\overline u_{{\Gstar}}=\sum_{x=0}^\infty \left(\frac{A_n(x)}{f_{\widehat G_{d,m_n}}}(x)- \frac{A(x)}{f_{{\Gstar}}(x)}\right)f_{\Gstar}(x)
$$
where
$$
A_n(x)=\sum_{j=1}^J u(x,\theta_j)k(x\mid\theta_j)\widehat g_{m_n}(\theta_j),\quad A(x)=\sum_{j=1}^J u(x,\theta_j)k(x\mid\theta_j) \Gstar(\theta_j).
$$
Since ${\Gstar}\in\mathcal P(\Theta_J)$, then by Lemma \ref{lem:EB_density-H2} there exists a constant $c$ such that, for $n$ large enough
    $$
   \EE[H^2( f_{\widehat G_{d,m_{n}}}, f_{\Gstar})]\leq c\frac{\log m_{n}}{m_{n}}.
    $$
This, together with $m_n\asymp n$, implies that, for every $a_n\rightarrow\infty$, 
\begin{align}
\frac{n^{1/2}}{a_n(\log n)^{\frac 12}}\EE\qth{H( f_{\widehat G_{d,m_{n}}}, f_{\Gstar})}
	\leq 	\frac{n^{1/2}}{a_n(\log n)^{\frac 12}}\sqrt{\EE\qth{H^2( f_{\widehat G_{d,m_{n}}}, f_{\Gstar})}}
	\leq 	\frac{n^{1/2}}{a_n(\log n)^{\frac 12}}\sqrt{c{\log(m_n)\over m_n}}	{\rightarrow} 0,
\end{align}
as $n\rightarrow\infty$. It follows that
$$
\frac{n^{1/2}}{a_n(\log n)^{\frac 12}}{H( f_{\widehat G_{d,m_{n}}}, f_{\Gstar})}\stackrel{\PP}{\rightarrow }0,
$$
as $n\rightarrow\infty$. Since the total variation distance is dominated by the Hellinger distance, then 
$$
\frac{n^{1/2}}{a_n(\log n)^\frac12}\sum_{x=0}^\infty|  f_{\widehat G_{d,m_{n}}}(x)- f_{\Gstar}(x)|\stackrel{\PP}{\rightarrow }0,$$
which implies that, for every $x\in\mathbb Z_+$, $$\frac{n^{1/2}}{a_n(\log n)^{1/2}}|  f_{\widehat G_{d,m_{n}}}(x)- f_{\Gstar}(x)|\stackrel{\PP}{\rightarrow}0. $$
Define the vectors $v_n$ and $v$ with coordinates $v_n(x)=x!f_{\widehat G_{d,m_{n}}}(x)$ and $v(x)=x!f_{{\Gstar}}(x)$ for $x=0,\dots,J-1$. Then $v_n(x)=\sum_{j=1}^J   C(x,j)e^{-\theta_j}\hat g_{m_n}(\theta_j)$ and $v(x)=\sum_{j=1}^J   C(x,j)e^{-\theta_j}\Gstar(\theta_j)$ where $C$ is the matrix with entries $C(x,j)=\theta_j^x$ ($j=1,\dots, J, x=0,\dots,J-1$). Since the matrix $C$ is non singular, then for every $j=1,\dots,J$
$$
\frac{n^{1/2}}{a_n(\log n)^{1/2}}(\hat g_{m_n}(\theta_j)-\Gstar(\theta_j))
=\frac{n^{1/2}}{a_n(\log n)^{1/2}}e^{\theta_j}\sum_{x=0}^{J-1}C^{-1}(j,x)x!(f_{{\widehat G_{d,m_n}}}(x)-f_{{\Gstar}}(x))\stackrel{\PP}{\rightarrow}0.
$$
Thus, for every $x\in\mathbb Z_+$, $\frac{n^{1/2}}{a_n(\log n)^{1/2}}(A_n(x)-A(x))\stackrel{\PP}{\rightarrow}0$. Then 
$$
\frac{n^{1/2}}{a_n(\log n)^{1/2}}(\overline u_{\widehat G_{d,m_n}}-\overline u_{\Gstar})=\sum_{x=0}^M\frac{f_{\Gstar}(x)\frac{n^{1/2}}{a_n(\log n)^{1/2}}(A_n(x)-A(x))-A(x)\frac{n^{1/2}}{a_n(\log n)^{1/2}}(f_{\widehat G_{d,m_n}}(x)-f_{\Gstar}(x))}{f_{\widehat G_{d,m_n}}(x)}
$$
where the equality holds since $ \int_\Theta u(x,\theta)\widehat G_{d,m_n}(d\theta\mid x)=\int_\Theta u(x,\theta)\Gstar(d\theta\mid x) $ for $x>M$.  The numerator in the right-hand side converges to zero in probability and the denominator is bounded away from zero for $0\leq x\leq M$ uniformly with respect to $n$ by the assumption  $G(\{0\})\leq a_0$ for some $a_0<1$ and every $G\in \mathcal C$. Thus,
$$
\frac{n^{1/2}}{a_n(\log n)^{1/2}}(\overline u_{\widehat G_{d,m_n}}-\overline u_{\Gstar})\stackrel{\PP}{\rightarrow}0,
$$
which proves \eqref{eq:conda}. The claim follows by the arbitrariness of $a_n$.

\subsection{Proof of Theorem \ref{th:consist2}}\label{app:proofconsist2}
Since the random variables
\(
u(X_i,\theta_i)-\overline u_{\Gstar}(X_i),\qquad i=1,\ldots,n,
\)
are independent, identically distributed, and have mean zero, the central limit theorem implies that
\[
(n\log n)^{-1/2}\bigl(S_n-\widehat S_n(\Gstar)\bigr)
\overset{\mathbb P}{\longrightarrow}0.
\]
Consequently, the assertions in parts (a) and (b) are equivalent to
\[
\frac{1}{n}\bigl(\widehat S_n(\Gstar)-\widehat  S_n(\widetilde G_{d,m_n}^{(T_n)})\bigr)
=o_{\mathbb P}(1),
\]
and
\[
\frac{1}{n}\bigl(\widehat S_n(\Gstar)-\widehat S_n(\widetilde G_{d,m_n}^{(T_n)})\bigr)
=O_{\mathbb P}(n^{-1/2}(\log n)^{1/2}),
\]
respectively. 
Denote
 $$
 Z_{n,i}=(u_{\Gstar}(X_i)-\overline u_{\Gstar})-(u_{\widetilde G_{d,m_n}^{(T_n)}}(X_i)-\overline u_{\widetilde G_{d,m_n}^{(T_n)}}),
 $$
     where
    $$
    \overline u_{\Gstar}=\EE[u_{\Gstar}(X_i)],\quad\quad \overline u_{\widetilde G_{d,m_n}^{(T_n)}}=\EE[u_{\widetilde G_{d,m_n}^{(T_n)}}(X_i)\mid \widetilde G_{d,m_n}^{(T_n)}].
    $$
 The random variables $\{Z_{n,i}:i=1,\dots,n\}$ are centered and conditionally i.i.d., given the sigma algebra generated by $(X'_{n})_{n\geq 1}$, where, for every $n$, $(X'_n)=(X'_{n,1},\dots,X'_{n,m_n})$ is the training sample. Moreover, 
\begin{align*}
     \EE[ (\sum_{i=1}^n Z_{n,i})^2\mid (X'_m)_{m\geq 1}]&=\operatorname{Var}(\sum_{i=1}^n (u_{\Gstar}(X_i)-\overline u_{\Gstar})-(u_{\widetilde G_{d,m_n}^{(T_n)}}(X_i)-\overline u_{\widetilde G_{d,m_n}^{(T_n)}})\mid (X'_m)_{m\geq 1})\\
     &\leq 4n \sup_{G\in\mathcal C}\operatorname{Var}(u_G(X_i))\\
     &\leq Cn,
\end{align*}
where $C=4 \sup_{G\in\mathcal C}
\sum_{x\ge0} u_G(x)^2 f_{\Gstar}(x) <\infty$ by Jensen's inequality and condition \eqref{eq:condition2}.
It follows that, for every $\epsilon >0$, 
  \begin{align*}
      \PP\left[
      \left|\frac{1}{(n\log n)^{1/2}} \sum_{i=1}^n Z_{n,i}\right|>\epsilon
        \right]\leq \frac{\EE[ (\sum_{i=1}^n Z_{n,i})^2]}{n\log n\epsilon^2}\rightarrow 0,
  \end{align*}
  as $n\rightarrow\infty$. Thus,
  \begin{equation}\label{eq:convZ2}
  \sum_{i=1}^n Z_{n,i}  =o_{\PP}((n\log n)^{1/2}),
  \end{equation}
  as $n\rightarrow\infty$.

  First, we deal with the proof of part (a) of Theorem \ref{th:consist2}. We can write
    \begin{align}\label{eq:decompcoars}
    \frac{\widehat S_n(\Gstar)-\widehat S_n(\widetilde G_{d,m_n}^{(T_{n})})}{n}&
    =\frac{\sum_{i=1}^nZ_{n,i}}{n}+(\overline u_{\Gstar}-\overline u_{\widetilde G_{d,m_n}^{(T_n)}})=(\overline u_{\Gstar}-\overline u_{\widetilde G_{d,m_n}^{(T_n)}})+o_\PP(1),
    \end{align}
   where the last equality follows from \eqref{eq:convZ2}.\\
We now prove that $\overline u_{\widetilde G_{d,m_n}^{(T_n)}}$ converges  to $\overline u_{\Gstar}$ in probability as $n\rightarrow\infty$. 
 By Lemma \ref{lem:EB_density-H2_coars}, there exists a constant $c$ such that for $n$ large enough
    $$
   \EE[H^2(\tilde f_{\widetilde G_{d,m_{n}}^{(T_{n})}}^{(T_{n})},\tilde f_{\Gstar}^{(T_{n})})]\leq \frac{c\log m_{n}}{m_{n}},
    $$
which  implies that $H(\tilde f_{\widetilde G_{d,m_{n}}^{(T_{n})}}^{(T_{n})},\tilde f_{\Gstar}^{(T_{n})})$ converges to zero in probability. Since the total variation distance is dominated by the Hellinger distance, then 
\begin{equation}
    \label{eq:coarsconv}
\sum_{x=0}^\infty| \tilde f_{\widetilde G_{d,m_{n}}^{(T_{n})}}^{(T_{n})}(x)-\tilde f_{\Gstar}^{(T_{n})}(x)|\stackrel{\PP}{\rightarrow }0.
\end{equation}
We now prove that 
$
\sum_{x=0}^\infty| f_{\widetilde G_{d,m_{n}}^{(T_{n})}}(x)- f_{\Gstar}(x)|$ converges to zero in probability.
We can write
$$
\sum_{x=0}^\infty| f_{\widetilde G_{d,m_{n}}^{(T_{n})}}(x)- f_{\Gstar}(x)|\leq 
\sum_{x\leq T_n}^\infty| \tilde f_{\widetilde G_{d,m_{n}}^{(T_{n})}}^{(T_{n})}(x)- f_{\Gstar}^{(T_{n})}(x)|+
\sum_{x>T_n}^\infty| f_{\widetilde G_{d,m_{n}}^{(T_{n})}}(x)- f_{\Gstar}(x)|
$$
The first term converges to zero in probability by \eqref{eq:coarsconv}.
For the second term, we can write
$$
\sum_{x>T_n}^\infty| f_{\widetilde G_{d,m_{n}}^{(T_{n})}}(x)- f_{\Gstar}(x)|\leq \sum_{x>T_n}^\infty  f_{\widetilde G_{d,m_{n}}^{(T_{n})}}(x)+\sum_{x>T_n}^\infty f_{\Gstar}(x)\leq 2 \sup_{G\in\mathcal C}\sum_{x>T_n}^\infty f_{G}(x)\rightarrow 0,
$$
as $n\rightarrow\infty$, where the convergence holds  by Lemma \ref{lemma:pre}.
Thus,
$$
\sum_{x=0}^\infty| f_{\widetilde G_{d,m_{n}}^{(T_{n})}}(x)- f_{\Gstar}(x)|\stackrel{\PP}{\rightarrow} 0.$$
We now prove that $\widetilde G_{d,m_n}^{(T_n)}$ converges weakly to $\Gstar$ in probability. The space $\mathcal C$ is weakly closed and uniformly tight by Lemma \ref{lemma:pre}. Hence, it is weakly compact.
Consider the map
\[
G\mapsto f_G
\]
defined on $\mathcal C$ and with values in $\ell^1(\mathbb Z_+)$. By Lemma \ref{lemma:pre} and the continuity of $k(x\mid\theta)$, this map is continuous. Moreover, again by Lemma \ref{lemma:pre}, it is injective. Thus, it  has a continuous inverse on its image (see, e.g., \citet[Theorem 26.6]{Munkres}). Hence,
\[
\|f_{\widetilde G_{d,m_n}^{(T_n)}}-f_{\Gstar}\|_1
\xrightarrow{\mathbb P}0
\]
implies that $\widetilde G_{d,m_n}^{(T_n)}$ converges weakly to $\Gstar$ in probability.\\
By Lemma \ref{lemma:convuG}, for every $x\in\mathbb Z_+$,
\[
u_{\widetilde G_{d,m_n}^{(T_n)}}(x)
\xrightarrow{\mathbb P}
u_{\Gstar}(x).
\]
Moreover, by Jensen's inequality,
\[
|u_G(x)|^2
\le
\int_\Theta |u(x,\theta)|^2\,G(d\theta\mid x),
\qquad G\in\mathcal C.
\]
Hence, by condition \eqref{eq:condition2}
\[
\sum_{x\in\mathbb Z_+}f_{\Gstar}(x)\sup_{G\in\mathcal C}|u_G(x)|^2<\infty.
\]
Therefore,
\[
|u_{\widetilde G_{d,m_n}^{(T_n)}}(x)-u_{\Gstar}(x)|\le2\left(\sup_{G\in\mathcal C}\int_\Theta |u(x,\theta)|^2\,G(d\theta\mid x)\right)^{1/2},
\]
and the right-hand side is summable with respect to $f_{\Gstar}(x)$.
We now prove that
\begin{equation}
    \label{eq:convsum}
\sum_{x\in\mathbb Z_+} f_{\Gstar}(x)\left|u_{\widetilde G_{d,m_n}^{(T_n)}}(x)-u_{\Gstar}(x)\right|\xrightarrow{\mathbb P}0.
\end{equation}
Fix $\varepsilon>0$. By the summability of
$x\mapsto f_{\Gstar}(x)\left(\sup_{G\in\mathcal C}\int_\Theta |u(x,\theta)|^2\,G(d\theta\mid x)\right)^{1/2}$,
there exists $K\in\mathbb N$ such that, for every $n\geq 1$,
\[
\sum_{x>K} f_{\Gstar}(x)\left|u_{\widetilde G_{d,m_n}^{(T_n)}}(x)-u_{\Gstar}(x)\right|<\frac{\varepsilon}{2}.
\]
On the other hand, by Lemma \ref{lemma:convuG},
\[
\sum_{x=0}^{K} f_{\Gstar}(x)\left|u_{\widetilde G_{d,m_n}^{(T_n)}}(x)-u_{\Gstar}(x)\right|\xrightarrow{\mathbb P}0.
\]
Hence,
\[
\mathbb P\left(\sum_{x\in\mathbb Z_+} f_{\Gstar}(x)\left|u_{\widetilde G_{d,m_n}^{(T_n)}}(x)-u_{\Gstar}(x)\right|
>\varepsilon\right)\rightarrow0.
\]
Thus, \eqref{eq:convsum} holds, which implies that
\begin{align*}
|\overline u_{\widetilde G_{d,m_n}^{(T_n)}}-\overline u_G|&=\left|\sum_{x\in \mathbb Z_+}f_{\Gstar}(x)\left(u_{\widetilde G_{d,m_n}^{(T_n)}}(x)-u_{\Gstar}(x)\right)\right|\\
&\leq \sum_{x\in\mathbb Z_+}f_{\Gstar}(x)\left|u_{\widetilde G_{d,m_n}^{(T_n)}}(x)-u_{\Gstar}(x)\right|\xrightarrow{\mathbb P}0.
\end{align*}
This concludes the proof of (a).

Now, we deal with the proof of part (b) of Theorem \ref{th:consist2}. We can write
    \begin{align}\label{eq:decompGridcoars}
    \frac{\widehat S_n(\Gstar)-\widehat S_n(\widetilde G_{d,m_n}^{(T_{n})})}{(n\log n)^{1/2}}&
    =\frac{\sum_{i=1}^nZ_{n,i}}{(n\log n)^{1/2}}+\frac{n^{1/2}}{(\log n)^{1/2}}(\overline u_{\Gstar}-\overline u_{\widetilde G_{d,m_n}^{(T_n)}})=\frac{n^{1/2}}{(\log n)^{1/2}}(\overline u_{\Gstar}-\overline u_{\widetilde G_{d,m_n}^{(T_n)}})+o_\PP(1),
    \end{align}
    where the last equality follows from \eqref{eq:convZ2}.
  We now prove that 
  $$\frac{n^{1/2}}{(\log n)^{1/2}}(\overline u_{\widetilde G_{d,m_n}^{(T_n)}}-\overline u_{\Gstar})=O_\PP(1),$$ $n\rightarrow\infty$. At this aim, it is sufficient to prove that, for every sequence $a_n\rightarrow\infty$,
  \begin{equation}\label{eq:conda2}
      \frac{n^{1/2}}{a_n(\log n)^{1/2}}(\overline u_{\widetilde G_{d,m_n}^{(T_n)}}-\overline u_{\Gstar})=o_\PP(1)
  \end{equation}
Denote by $\widetilde g_{m_n}^{(T_n)}$ and $\Gstar$ the probability mass functions corresponding to $\widetilde G_{d,m_n}^{(T_n)}$ and ${\Gstar}$, respectively.
 We can write
$$
\overline u_{\widetilde G_{d,m_n}^{(T_n)}}-\overline u_{{\Gstar}}=\sum_{x=0}^\infty \left(\frac{A_n^{(T_n)}(x)}{f_{\widetilde G_{d,m_n}^{(T_n)}}(x)}- \frac{A(x)}{f_{{\Gstar}}(x)}\right)f_{\Gstar}(x)
$$
where
$$
A_n^{(T_n)}(x)=\sum_{j=1}^J u(x,\theta_j)k(x\mid\theta_j)\widetilde g_{m_n}^{(T_n)}(\theta_j),\quad A(x)=\sum_{j=1}^J u(x,\theta_j)k(x\mid\theta_j) \Gstar(\theta_j).
$$
Since ${\Gstar}\in \mathcal P(\Theta_J)$, then by Lemma \ref{lem:EB_density-H2_coars} there exists a constant $c$ such that for $m$ large enough
    $$
   \EE[H^2( \tilde f_{\widetilde G_{d,m_{n}}^{(T_n)}}^{(T_n)}, \tilde f_{\Gstar}^{(T_n)})]\leq c\frac{\log m_{n}}{m_{n}}.
    $$
    This implies
    \begin{align}
\frac{n^{1/2}}{a_n(\log n)^{\frac 12}}\EE\qth{H( \tilde f_{\widetilde G_{d,m_{n}}^{(T_n)}}^{(T_n)}, \tilde f_{\Gstar}^{(T_n)})}
	\leq 
	\frac{n^{1/2}}{a_n(\log n)^{\frac 12}}\sqrt{\EE\qth{H^2( \tilde f_{\widetilde G_{d,m_{n}}^{(T_n)}}^{(T_n)}, \tilde f_{\Gstar}^{(T_n)})}}
	\leq 
	\frac{n^{1/2}}{a_n(\log n)^{\frac 12}}\sqrt{c{\log(m_n)\over m_n}}
	{\rightarrow} 0,
\end{align}
This, together with $m_n\asymp n$, implies that 
$$\frac{n^{1/2}}{a_n(\log n)^{1/2}}H(\tilde  f_{\widetilde G_{d,m_{n}}^{(T_n)}}^{(T_n)},\tilde  f_{\Gstar}^{(T_n)})\stackrel{\PP}{\rightarrow} 0.$$ 
as $n\rightarrow\infty$. Since the total variation distance is dominated by the Hellinger distance, then 
$$
\frac{n^{1/2}}{a_n(\log n)^{1/2}}\sum_{x=0}^\infty| \tilde f_{\widetilde G_{d,m_{n}}^{(T_n)}}^{(T_n)}(x)-\tilde  f_{\Gstar}^{(T_n)}(x)|\stackrel{\PP}{\rightarrow }0.$$
For $x\leq T_n$, $\tilde f_{\Gstar}^{(T_n)}(x)=f_{\Gstar}(x)$ and $\tilde f_{\widetilde G_{d,m_{n}}^{(T_n)}}^{(T_n)}(x)=f_{\widetilde G_{d,m_{n}}^{(T_n)}}(x) $.
Therefore $$\frac{n^{1/2}}{a_n(\log n)^{1/2}}\left|  f_{\widetilde G_{d,m_{n}}^{(T_n)}}(x)- f_{\Gstar}(x)\right|\stackrel{\PP}{\rightarrow} 0\quad \mbox{ for every }x\leq T_n, $$ 
as $n\rightarrow\infty$.
Let $n$ be such that $T_n>J$.
Define the vectors $v_n$ and $v$ with coordinates $v_n(x)=x!f_{\widetilde G_{d,m_n}^{(T_n)}}(x)$ and $v(x)=x!f_{{\Gstar}}(x)$ for $x=0,\dots,J-1$. Then
$v_n(x)=\sum_{j=1}^J   C(x,j)e^{-\theta_j}\widetilde  g_{m_n}^{(T_n)}(\theta_j)$ and $v(x)=\sum_{j=1}^J   C(x,j)e^{-\theta_j}\Gstar(\theta_j)$ where $C$ is the matrix with entries $C(x,j)=\theta_j^x$ ($j=1,\dots, J, x=0,\dots,J-1$). Since the matrix $C$ is non singular, then for every $j=1,\dots,J$
$$
\frac{n^{1/2}}{a_n(\log n)^{1/2}}(\widetilde  g_{m_n}^{(T_n)}(\theta_j)-\Gstar(\theta_j))
=\frac{n^{1/2}}{a_n(\log n)^{1/2}}e^{\theta_j}\sum_{x=0}^{J-1}C^{-1}(j,x)x!(f_{\widetilde G_{d,m_n}^{(T_n)}}(x)-f_{{\Gstar}}(x))\stackrel{\PP}{\rightarrow}0.
$$
It follows that, for every $x\in\mathbb Z_+$, $\frac{n^{1/2}}{a_n(\log n)^{1/2}}(A_n^{(T_n)}(x)-A(x))\stackrel{\PP}{\rightarrow}0$.  Then,
\begin{align*}
    \frac{n^{1/2}}{a_n(\log n)^{1/2}}(\overline u_{\widetilde G_{d,m_n}^{(T_n)}}&-\overline u_{\Gstar})
    =\frac{n^{1/2}}{a_n(\log n)^{1/2}}
\sum_{x=0}^\infty f_{\Gstar}(x)\left(\int_\Theta u(x,\theta)\widetilde G_{d,m_n}^{(T_n)}(d\theta\mid x)-\int_\Theta u(x,\theta) {\Gstar}(d\theta\mid x)\right)
\\&=\sum_{x=0}^M\frac{f_{\Gstar}(x)\frac{n^{1/2}}{a_n(\log n)^{1/2}}(A_n^{(T_n)}(x)-A(x))-A(x)\frac{n^{1/2}}{a_n(\log n)^{1/2}}(f_{\widetilde G_{d,m_n}^{(T_n)}}(x)-f_{\Gstar}(x))}{f_{\widetilde G_{d,m_n}^{(T_n)}}(x)}
\stackrel{\PP}{\rightarrow}0,
\end{align*}
where the first equality holds since $ \int_\Theta u(x,\theta)\widetilde G_{d,m_n}^{(T_n)}(d\theta\mid x)=\int_\Theta u(x,\theta){\Gstar}(d\theta\mid x) $ for $x>M$ and the convergence holds since the numerator converges to zero in probability and the denominator converges and  is bounded away from zero uniformly over
$0\le x\le M$ and $G\in\mathcal C$. Indeed, the condition $G(\{0\})\le a_0<1$ implies that for every
fixed $x\in\mathbb Z_+$,
\[
\inf_{G\in\mathcal C} f_G(x)
\ge
(1-a_0)\min_{\theta_j>0}k(x\mid\theta_j)>0,
\]
and inequality holds uniformly in $x\in [0,M]\cap\mathbb Z_+$. Thus \eqref{eq:conda2} holds.

%
%
%


\section{Proofs of Section \ref{sec3}}

\subsection{Proof of Theorem \ref{thm:total_intensity_lb}}\label{proof41_lb}
	
The following constructions and  proofs follow the standard ideas underlying Le Cam’s two-point method; see, e.g., \citet[Chapter 5]{Rigollet}. We nevertheless provide the details of the proof because our target is a random oracle sum of the main sample, rather than a fixed parameter, and the usual reduction from estimation to testing therefore requires a slight adaptation.

Before proving the theorem, we state some preliminary lemmas. For any $a,b\in\mathbb R$ with $a<b$, we denote by $\operatorname{Unif}[a,b]$ the uniform distribution on the interval $[a,b]$. 
In the proofs of the lemmas we write $p_\theta(x)=\operatorname{Poisson}(x\mid\theta)$, $f_{j,\varepsilon}=f_{G_{j,\varepsilon}}$, $\PP_{j,\varepsilon}=\PP_{G_{j,\varepsilon}}$ $\EE_{j,\varepsilon}=\EE_{G_{j,\varepsilon}}$. 

\begin{lemma}\label{lemma:lowerBound1}
Let $c>0$, $a$, $b$ and $\varepsilon_0>0$ be fixed numbers satisfying $(1+\varepsilon_0)^2<a<b\leq 3$, and $c\varepsilon_0\leq 1/2$. Define
\[
G_{0,\varepsilon}  =(1-c\varepsilon)\operatorname{Unif}[a,b]    +c\varepsilon\,\operatorname{Unif}[1-\varepsilon,1],\qquad
G_{1,\varepsilon}  =(1-c\varepsilon)\operatorname{Unif}[a,b]    +c\varepsilon\,\operatorname{Unif}[1,1+\varepsilon],
\]
for $\varepsilon\in(0,\varepsilon_0]$. Then
\begin{itemize}
    \item[(a)]  There exists a constant $C_1>0$ such that $H^2(f_{G_{0,\varepsilon}},f_{G_{1,\varepsilon}})
\leq C_1\varepsilon^4.$
\item[(b)] There exists a constant $C_2>0$ such that 
 $
\min_{j\in\{0,1\}}f_{G_{j,\varepsilon}}(1)\geq C_2.
$
\end{itemize}
\end{lemma}
\begin{proof}
Denote $H=\operatorname{Unif}[a,b]$.

\noindent
    (a) We can write
\begin{align*}
    H^2(f_{0,\varepsilon},f_{1,\varepsilon})&=\sum_{t\geq 0}\left(\sqrt{f_{0,\varepsilon}(t)}-\sqrt{f_{1,\varepsilon}(t)}\right)^2\\
    &\leq \sum_{t\geq 0}\frac{\left(f_{0,\varepsilon}(t)-f_{1,\varepsilon}(t)\right)^2}{f_{0,\varepsilon}(t)+f_{1,\varepsilon}(t)}\\
   & \leq c^2 \sum_{t\geq 0}\frac{\left(\int_{1-\varepsilon}^1 p_\theta(t)d\theta-\int_{1}^{1+\varepsilon}p_\theta(t)d\theta\right)^2}{f_{0,\varepsilon}(t)+f_{1,\varepsilon}(t)}\\
   &\leq c^2 \varepsilon^2\sum_{t\geq 0}\frac{\left(
   \int_0^1 ( p_{1-\varepsilon+\varepsilon s}(t)-
    p_{1+\varepsilon s}(t)   d s\right)^2}{
    k_1  {a^t}/{t!}}\end{align*}
    for a constant $k_1$, where the first inequality comes from
    $$
    (\sqrt u -\sqrt v)^2=\frac{(u-v)^2}{(\sqrt u+\sqrt v)^2}\leq 
    \frac{(u-v)^2}{u+v}
    $$
for $u,v>0$, and the last inequality comes from
$$
f_{0,\varepsilon}(t)+f_{1,\varepsilon}(t)\geq \frac{2(1-c\varepsilon_0)}{b-a}\int_a^b \frac{e^{-\theta}\theta^t}{t!}d\theta\geq \frac{e^{-b}a^t}{t!}:=k_1\frac{a^t}{t!}.
$$
By the fundamental theorem of calculus
\begin{align*}
    \int_0^1 ( p_{1-\varepsilon+\varepsilon s}(t)-
    p_{1+\varepsilon s}(t) )  d s&=-
    \int_0^1\int_{1+\varepsilon s -\varepsilon}^{1+\varepsilon s}\frac{\partial p_\theta(t)}{\partial \theta}d\theta ds\\
    &=- \int_0^1\int_{1+\varepsilon s -\varepsilon}^{1+\varepsilon s}
    p_\theta(t)\left(\frac t \theta -1\right)d\theta ds.
\end{align*}
Thus,
\begin{align*}
     H^2(f_{0,\varepsilon},f_{1,\varepsilon})
     &\leq \frac{c^2}{k_1}  \varepsilon^2 \sum_{t\ge 0} 
     \frac{\left(\int_0^1\int_{1+\varepsilon s -\varepsilon}^{1+\varepsilon s}
    p_\theta(t)\left(\frac t \theta -1\right)d\theta ds\right)^2}{{a^t}/{t!}}\\
    &\leq \frac{c^2}{k_1}  \varepsilon^2 \sum_{t\ge 0} 
     \frac{\left(\int_0^1\int_{1+\varepsilon s -\varepsilon}^{1+\varepsilon s}
   \frac{e^{-\theta}\theta^t}{t!}\left(\frac t \theta -1\right)d\theta ds\right)^2}{{a^t}/{t!}}\\
   &\leq \frac{c^2}{k_1}  \varepsilon^2 \sum_{t\ge 0} 
     \frac{\left(\int_0^1\int_{1+\varepsilon s -\varepsilon}^{1+\varepsilon s}
   \left(\frac{e^{-\theta}\theta^t}{t!}+t\frac{e^{-\theta}\theta^{t-1}}{t!}\right)d\theta ds\right)^2}{{a^t}/{t!}}.
\end{align*}
For $s\in [0,1]$ and $\theta\in [1+\varepsilon s -\varepsilon,1+\varepsilon s]$, $\theta\in [1-\varepsilon_0,1+\varepsilon_0]$, which implies
\begin{align*}
H^2(f_{0,\varepsilon},f_{1,\varepsilon})
     &\leq \frac{c^2}{k_1^2}  e^{-2(1-\varepsilon_0)}\varepsilon^4 \left[  
    \sum_{t\ge 0} 
     \frac{\left(1+\varepsilon_0\right)^{2t}}{a^t}
     +\sum_{t\ge 0} \frac{t^2\left(1+\varepsilon_0\right)^{2(t-1)}}{a^t}\right]\\
     &\leq C_1\varepsilon^4,
\end{align*}
for a constant $C_1$.

\noindent (b)
We can write
\[
f_{j,\varepsilon}(1)\geq \frac{1-c\varepsilon}{b-a}\int_a^b \theta e^{-\theta}d\theta\geq \frac{1}{2(b-a)} \int_a^b \theta e^{-\theta}d\theta:=C_2>0.
\]
\end{proof}

\begin{lemma}\label{lemma:lowerBound2}
    Under the assumptions of Lemma \ref{lemma:lowerBound1}, let $P_{j,\varepsilon}^{(m,n)}$ denote the joint distribution of the training and main samples under $G_{j,\varepsilon}$, with $m+n\geq 1$. Then:
    \begin{itemize}
        \item[(a)]  
\[
H^2(P_{0,\varepsilon}^{(m,n)},P_{1,\varepsilon}^{(m,n)})
\leq C_1(m+n)\varepsilon^4,
\]
with $C_1$ as in Lemma \ref{lemma:lowerBound1}.
\item[(b)] If $\varepsilon=\gamma(m+n)^{-1/4}$ with $\gamma<\min(\varepsilon_0,(2\sqrt{C_1})^{-1/2})$, then every test $\phi\in\{0,1\}$  for $H_0:G=G_{0,\varepsilon}$ against $H_1:G=G_{1,\varepsilon}$, based on $(X_1,\dots,X_n,X'_1,\dots,X'_m)$ with $n+m\geq 1$, satisfies
\[
\PP_{0,\varepsilon}[\phi=1]+P_{1,\varepsilon}[\phi=0]\geq \frac12.
\]
    \end{itemize}
   \end{lemma}

\begin{proof}
(a) Since  $P_{j,\varepsilon}^{(m,n)}=f_{j,\varepsilon}^{\otimes(m+n)}$, then 
\begin{align*}
H^2(P_{0,\varepsilon}^{(m,n)},P_{1,\varepsilon}^{(m,n)})
&=2\left(1-\sum_{x_1,\ldots,x_{m+n}}\sqrt{\prod_{i=1}^{m+n}f_{0,\varepsilon}(x_i)
\prod_{i=1}^{m+n}f_{1,\varepsilon}(x_i)}\right)\\
&=2\left(1-\prod_{i=1}^{m+n}\sum_{x_i\geq0}\sqrt{f_{0,\varepsilon}(x_i)f_{1,\varepsilon}(x_i)}\right)\\
&=2\left(1-\left(1-\frac12H^2(f_{0,\varepsilon},f_{1,\varepsilon})\right)^{m+n}\right)\\
&\leq (m+n)H^2(f_{0,\varepsilon},f_{1,\varepsilon}),
\end{align*}
where the last inequality comes from  $1-(1-u)^{m+n}\leq(m+n)u$ for $u\in[0,1]$.
By Lemma \ref{lemma:lowerBound1}, $H^2(P_{0,\varepsilon}^{(m,n)},P_{1,\varepsilon}^{(m,n)})\leq
C_1(m+n)\varepsilon^4$.

\noindent (b) the assumption $\gamma<\varepsilon_0$ ensures that $\gamma (m+n)^{-1/4}<\varepsilon_0$.
By (a) and the inequality between total variation and Hellinger
distances,
\[
\operatorname{TV}(P_{0,\varepsilon}^{(m,n)},P_{1,\varepsilon}^{(m,n)})
\leq H(P_{0,\varepsilon}^{(m,n)},P_{1,\varepsilon}^{(m,n)})\leq C_1^{1/2}(m+n)^{1/2}\varepsilon^2
=C_1^{1/2}\gamma^2.
\]
Thus, for $\gamma<\frac{1}{\sqrt 2 C_1^{1/4}}$,
\[
\operatorname{TV}(P_{0,\varepsilon}^{(m,n)},P_{1,\varepsilon}^{(m,n)})\leq \frac12.
\]
It follows that, for every test $\phi\in\{0,1\}$,
\begin{align*}
\PP_{0,\varepsilon}(\phi=1)+\PP_{1,\varepsilon}(\phi=0)
\geq 1-\operatorname{TV}(P_{0,\varepsilon}^{(m,n)},P_{1,\varepsilon}^{(m,n)})\geq \frac12.
\end{align*}
Indeed, letting $A=\{\phi=1\}$,
\[
\PP_{0,\varepsilon}(\phi=1)+\PP_{1,\varepsilon}(\phi=0)
=\PP_{0,\varepsilon}(A)+1-\PP_{1,\varepsilon}(A)
\geq
1-\operatorname{TV}(P_{0,\varepsilon}^{(m,n)},P_{1,\varepsilon}^{(m,n)})\geq \frac12.
\]
\end{proof}

Based on Lemmas \ref{lemma:lowerBound1}--\ref{lemma:lowerBound2}, we are in the position of proving Theorem \ref{thm:total_intensity_lb}.

\medskip

\noindent (a)
For $x>\kappa$, $m_{j,\varepsilon}(x)=\EE_{j,\varepsilon}[\theta_i I(X_i\leq \kappa)\mid X_i=x]=0$. 
Denote $d=(b-a)^{-1}$ and 
 \[h_x(\theta)=e^{-\theta}\theta^x, \quad \dot h_x(\theta)=\frac{\partial h_x(\theta)}{\partial \theta},\quad \ddot{h}_x=\frac{\partial^2 h_x(\theta)}{\partial \theta^2}.\]
 For $x\leq\kappa$
 \begin{align*}
    &m_{0,\varepsilon}(x)-m_{1,\varepsilon}(x)\\
    &    = \frac{c \int_{1-\varepsilon}^1 e^{-\theta}\theta^{x+1} d\theta+(1-c\varepsilon)d\int_a^b e^{-\theta}\theta^{x+1} d\theta}{c \int_{1-\varepsilon}^1 e^{-\theta}\theta^x d\theta+(1-c\varepsilon)d\int_a^b e^{-\theta}\theta^x d\theta}
    - \frac{c \int_{1}^{1+\varepsilon} e^{-\theta}\theta^{x+1} d\theta+(1-c\varepsilon)d\int_a^b e^{-\theta}\theta^{x+1} d\theta}{c \int_{1}^{1+\varepsilon} e^{-\theta}\theta^x d\theta+(1-c\varepsilon)d\int_a^b e^{-\theta}\theta^x d\theta}\\
    &=  \frac{c \int_{1-\varepsilon}^1 h_{x+1}(\theta) d\theta+(1-c\varepsilon)d\int_a^b h_{x+1}(\theta)  d\theta}{c \int_{1-\varepsilon}^1 h_x(\theta) d\theta+(1-c\varepsilon)d\int_a^b h_x(\theta) d\theta}
    - \frac{c \int_{1}^{1+\varepsilon} h_{x+1}(\theta)  d\theta+(1-c\varepsilon)d\int_a^b h_{x+1}(\theta)  d\theta}{c \int_{1}^{1+\varepsilon} h_x(\theta) d\theta+(1-c\varepsilon)d\int_a^b h_x(\theta) d\theta}
   \\
    &=  \frac{c(h_{x+1}(1)\varepsilon-\frac 1 2 \dot h_{x+1}(1) \varepsilon^2+\frac 1 6 \ddot h_{x+1}(1-\varepsilon_{x+1}')\varepsilon^3)+(1-c\varepsilon)d\int_a^b h_{x+1}(\theta)  d\theta}
    {c(h_{x}(1)\varepsilon-\frac 1 2 \dot h_{x}(1) \varepsilon^2+\frac 1 6 \ddot h_{x}(1-\varepsilon_{x}')\varepsilon^3)+(1-c\varepsilon)d\int_a^b h_x(\theta) d\theta}. \\
    &\quad\quad 
    -\frac{c(h_{x+1}(1)\varepsilon+\frac 1 2 \dot h_{x+1}(1) \varepsilon^2+\frac 1 6 \ddot h_{x+1}(1+\varepsilon_{x+1}'')\varepsilon^3)+(1-c\varepsilon)d\int_a^b h_{x+1}(\theta)  d\theta}
    {c(h_{x}(1)\varepsilon+\frac 1 2 \dot h_{x}(1) \varepsilon^2+\frac 1 6 \ddot h_{x}(1+\varepsilon_{x}'')\varepsilon^3)+(1-c\varepsilon)d\int_a^b h_x(\theta) d\theta},
    \end{align*}
with
$ h_x(1)=e^{-1}$ and $\dot h_x(1)=e^{-1}(x-1)$.
The points $1-\varepsilon_{x}',1+\varepsilon_{x}'',1-\varepsilon_{x+1}',1+\varepsilon_{x+1}''$  belong to the interval $[1-\varepsilon_0,1+\varepsilon_0]$. Since
\[
 \max_{0\leq x\leq\kappa+1}
 \sup_{\theta\in[1-\varepsilon_0,1+\varepsilon_0]}
 |\ddot h_x(\theta)|<\infty,
\]
the Taylor remainders are $O(\varepsilon^3)$ uniformly over
$x\in\{0,\ldots,\kappa\}$. Denoting
\[
 H_x:=d\int_a^bh_x(\theta)\,d\theta,
\]
we can write
\begin{align*}
 &m_{0,\varepsilon}(x)- m_{1,\varepsilon}(x)\\
 &= \frac{H_{x+1} +c(e^{-1}-H_{x+1})\varepsilon  -\frac{1}{2}ce^{-1}x\varepsilon^2  +O(\varepsilon^3)} { H_x
 +c(e^{-1}-H_x)\varepsilon -\frac{1}{2}ce^{-1}(x-1)\varepsilon^2  +O(\varepsilon^3)}-\frac{ H_{x+1} +c(e^{-1}-H_{x+1})\varepsilon  +\frac{1}{2}ce^{-1}x\varepsilon^2 +O(\varepsilon^3)} { H_x +c(e^{-1}-H_x)\varepsilon  +\frac{1}{2}ce^{-1}(x-1)\varepsilon^2  +O(\varepsilon^3)}\\
 &= \frac{ce^{-1}\bigl((x-1)H_{x+1}-xH_x\bigr)\varepsilon^2 +O(\varepsilon^3)}{H_x^2+O(\varepsilon)}\\
 &= \frac{ce^{-1}}{H_x^2}\bigl((x-1)H_{x+1}-xH_x\bigr)\varepsilon^2 +O(\varepsilon^3).
\end{align*}
For $j=0,1$
\[
 f_{j,\varepsilon}(x)=\frac{H_x}{x!}+O(\varepsilon),
\] 
uniformly with respect to $x\leq \kappa$. Thus,
\begin{align*}
 &\EE_{j,\varepsilon} \left[m_{0,\varepsilon}(X_1)-m_{1,\varepsilon}(X_1)\right] = ce^{-1}\varepsilon^2 \sum_{x=0}^{\kappa} \frac{(x-1)H_{x+1}-xH_x}{x!H_x} +O(\varepsilon^3).
\end{align*}
We now show that the coefficient of $\varepsilon^2$ is strictly negative. The terms corresponding to $x=0$ and $x=1$ are
$-H_1/H_0\leq -a$ and $-1$, respectively. Since they are negative, the claim is true if $\kappa\leq 1$. To treat the case $\kappa>1$, notice that, for every $x\geq2$,
\[
 \frac{(x-1)H_{x+1}-xH_x}{x!H_x}
 =\frac{xH_{x+1}/{H_x}-x}{x!}\leq  \frac{(x-1)b-x}{x!}\leq \frac{3(x-1)-x}{x!}= \frac{2x-3}{x!}.
\]
It follows that
\begin{align*}
 \sum_{x=0}^{\kappa}
 \frac{(x-1)H_{x+1}-xH_x}{x!H_x}
& \leq  -1-a+\sum_{x=2}^{\kappa}\frac{x(b-1)-b}{x!}\leq  -1-a+\sum_{x=2}^{\infty}\frac{2x-3}{x!}\\
 &=-1-a+2(e-1)-3(e-2)= 3-a-e<0,
\end{align*}
where we have used $b\leq 3$ and $a>1$.
Therefore, for $\varepsilon_0$ sufficiently small, there exists a
constant $C_3>0$, independent of $\varepsilon$ and $j$, such that
\[
 \EE_{j,\varepsilon} \left[m_{0,\varepsilon}(X_1)-m_{1,\varepsilon}(X_1)\right] \leq -C_3\varepsilon^2,
\]
which proves (a).

\medskip

\noindent (b)
By point (a),
\[
\EE_{j,\varepsilon}[m_{1,\varepsilon}(X_1)-m_{0,\varepsilon}(X_1)]\geq C_3\varepsilon^2.
\]
On the other hand, the explicit expression for
$m_{j,\varepsilon}(x)$ imply that there exists a constant $k_2>0$ such that 
\[
\sup_{x\leq \kappa}|m_{0,\varepsilon}(x)-m_{1,\varepsilon}(x)|
\leq k_2\varepsilon^2.
\]
Thus, 
\[
\operatorname{Var}_{j,\varepsilon}
\{m_{0,\varepsilon}(X_1)-m_{1,\varepsilon}(X_1)\}
\leq k_2^2\varepsilon^4.
\]
Using Chebyshev inequality, we can write
\begin{align*}
     &\PP_{j,\varepsilon}\left[\widehat S_{n,\kappa}(G_{1,\varepsilon})
-\widehat S_{n,\kappa}(G_{0,\varepsilon})\leq \frac 1 2C_3n\varepsilon^2\right]\\
&\leq \PP_{j,\varepsilon}\left[\widehat S_{n,\kappa}(G_{1,\varepsilon})
-\widehat S_{n,\kappa}(G_{0,\varepsilon})\leq n\EE_{j,\varepsilon}[m_{1,\varepsilon}(X_1)-m_{0,\varepsilon}(X_1)]-\frac 1 2C_3n\varepsilon^2\right]\\
&\leq \PP_{j,\varepsilon}\left[\left|\widehat S_{n,\kappa}(G_{1,\varepsilon})
-\widehat S_{n,\kappa}(G_{0,\varepsilon})-n\EE_{j,\varepsilon}[m_{1,\varepsilon}(X_1)-m_{0,\varepsilon}(X_1)]\right|\geq \frac 1 2C_3n\varepsilon^2\right]\\
&\leq \frac{4 k_2^2}{C_3^2n}\rightarrow 0.
\end{align*}
Thus
\begin{equation}\label{eq:exampleinf}
    \operatorname{Inf}_{\{\epsilon\in (0,\epsilon_0],j=0,1\}}\PP_{j,\epsilon}\left[
\left|
\widehat S_{n,\kappa}(G_{1,\varepsilon})
-
\widehat S_{n,\kappa}(G_{0,\varepsilon})
\right|
\geq k_1 n\varepsilon^2\right]\rightarrow 1.
\end{equation}
where $k_1={C_3}/{2}.$
Given an estimator $\widehat S$, define a test $\phi$ for
$H_0:G=G_{0,\varepsilon}$ against $H_1:G=G_{1,\varepsilon}$ by
\[
\phi=
\begin{cases}
0,&\text{if }|\widehat S-\widehat S_{n,\kappa}(G_{0,\varepsilon})|<
|\widehat S-\widehat S_{n,\kappa}(G_{1,\varepsilon})|,
\\
1,&\text{otherwise}.
\end{cases}
\]
Let
\[
A_n=\left\{\left|\widehat S_{n,\kappa}(G_{0,\varepsilon})
-\widehat S_{n,\kappa}(G_{1,\varepsilon})\right|\geq k_1 n\varepsilon^2 \right\}.
\]
Then, $\PP_{j,\varepsilon}(A_n^c)\rightarrow0$ for $j=0,1$. Moreover, on $A_n\cap\{\phi=1\}$,
\[
\left|\widehat S-\widehat S_{n,\kappa}(G_{0,\varepsilon})
\right|\geq
\frac{k_1}{2}n\varepsilon^2,
\]
and on $A_n\cap\{\phi=0\}$,
\[
\left|\widehat S-\widehat S_{n,\kappa}(G_{1,\varepsilon})\right|
\geq\frac{k_1}{2}n\varepsilon^2.
\]
Hence
\begin{align*}
&\EE_{{0,\varepsilon}}\left[\left(
\frac{\widehat S-\widehat S_{n,\kappa}(G_{0,\varepsilon})}{\sqrt n}
\right)^2\right]\geq \frac{k_1^2}{4}n\varepsilon^4 \PP_{0,\varepsilon}
\left(A_n\cap\{\phi=1\}\right)
\end{align*}
and
\begin{align*}
&\EE_{{1,\varepsilon}}\left[\left(\frac{\widehat S-\widehat S_{n,\kappa}(G_{1,\varepsilon})}{\sqrt n}\right)^2\right]\geq \frac{k_1^2}{4}n\varepsilon^4 \PP_{1,\varepsilon}\left(A_n\cap\{\phi=0\}\right),
\end{align*}
with $\PP_{j,\varepsilon}$ denoting the probability measure under the mixing distribution $G_{j,\varepsilon}$ for $j=0,1$.
Moreover,
\begin{align*}
&\PP_{0,\varepsilon}\left(A_n\cap\{\phi=1\}\right)+\PP_{1,\varepsilon}
\left(A_n\cap\{\phi=0\}\right)\geq \frac12-o(1),
\end{align*}
where the last inequality follows from \eqref{eq:exampleinf} and Lemma \ref{lemma:lowerBound2}.
Hence, there exists a constant $k_3$  such that, for $n$ sufficiently large,
\[
\max_{j\in\{0,1\}}\EE_{{j,\varepsilon}}\left[\left(\frac{
\widehat S-\widehat S_{n,\kappa}(G_{j,\varepsilon})
}{\sqrt n}\right)^2\right]\geq k_3n\varepsilon^4.
\]
Taking
\[
\varepsilon=\gamma(m+n)^{-1/4}
\]
with $\gamma\in (0,\min(\varepsilon_0,(2\sqrt{C_1})^{-1/2})$ gives
\[
\max_{j\in\{0,1\}}\EE_{{j,\varepsilon}}\left[\left(\frac{
\widehat S-\widehat S_{n,\kappa}(G_{j,\varepsilon})}{\sqrt n}\right)^2\right]
\geq C_4\frac{n}{m+n},
\]
with $C_4=k_3\gamma^4$. Since $G_{j,\varepsilon}\in\mathcal C$ for $j=0,1$ the result follows.

	\subsection{Proof of \prettyref{thm:sum_regular}}\label{proof41}

The proof makes use of the following result, which is a modification of \citet[Lemma 4]{jana2025optimal}.
	\begin{lemma}\label{lmm:EB_general-regret-bound}
		Let $\Gstar$ be a distribution such that $\EE_{\Gstar}[\theta^4]\leq M$ for some constant $M$. Then given any distribution $\hat {Q}$ on $[0,\hat h]$, and $h>0$ with $\Gstar([0,h])>\frac 12$ and any $V\geq 1$, with $\epsilon_V(\Gstar)=\PP_{X\sim f_\Gstar}[X\geq V]$,
		\begin{align*}
			&\EE_{\Gstar}\left[\left(\hat{\theta}_{\hat Q}(x) - \hat{\theta}_{\Gstar}(x)\right)^2I(X\leq V-1)\right]
			\notag\\
			&\leq  
			\pth{6(h^2+\hat h^2)+24(h+\hat h)V}(\sum_{x=0}^{V}(\sqrt{f_{\Gstar}(x)}-\sqrt{f_{\hat{Q}}(x)})^2+4\Gstar((h,\infty)))
			\nonumber\\
			&\quad +2(h+\hat h)^2\epsilon_V(\Gstar) +{2(1+2\sqrt 2)\sqrt{(M+\hat h^4)\Gstar((h,\infty))}}.
		\end{align*}
	\end{lemma}
	Additionally, we will use the following result that provides a lower bound on the mixture density evaluated at $x\in \{0,1,\dots,\kappa\}$.
	\begin{lemma}
		\label{lmm:mixture-lower-bound}
Let $\mathcal{C}$ be a class of prior distributions on $[0, \infty)$ such that $\mathcal C\subset \subexpo(s)$ for some $s>0$.  Let $f_G(x)$ denote the Poisson mixture density defined as in \eqref{eq:mixture-density}.
		If, for some $\epsilon>0$ and $c_0<1$, $G([0,\epsilon))\leq c_0$ for every $G\in\mathcal C$,  then, for a given $\kappa > 0$, there exists a constant $b_\kappa > 0$ such that
		\[
		\min_{0 \le x \le \kappa} f_G(x) \ge b_\kappa \quad \text{for all } G \in \mathcal{C}.
		\]
	\end{lemma}
	
	\begin{proof}
    Since $\mathcal C\subset \subexpo(s)$, then $\mathcal C$ is tight. Fix $\delta\in (0,1-c_0)$, and let $M$ be such that $G(M,\infty)<\delta$ for every $G\in\mathcal C$. Then, for every $G\in\mathcal C$, $G([\epsilon,M]>1-c_0-\delta>0$.	Thus,    
    for any $x \in \{0,\dots,\kappa\}$, we have
          \begin{align*}
            \int_0^\infty \frac{e^{-\theta}\theta^x}{x!}G(d\theta)&\geq \frac{1}{\kappa!}\int_{\epsilon}^{M}e^{-\theta}\theta^xG(d\theta)\\
            &\geq \frac{1}{\kappa!}\min\{1,\epsilon^\kappa\}e^{-M}G([\epsilon,M])\\
            &\geq  \frac{1}{\kappa!}\min\{1,\epsilon^\kappa\}e^{-M}(1-c_0-\delta):=b_k>0.\
          \end{align*}
        
         	\end{proof}
	
	Based on Lemmas \ref{lmm:EB_general-regret-bound}--\ref{lmm:mixture-lower-bound}, we are now in the position of proving Theorem \ref{thm:sum_regular}
	
\subsubsection{Subexponential setting}\label{proof:th41subexpo}
	For a simplicity of notation, we denote for every $G$, $$\hat\theta_{G}(X_i):=\EE_{G}[\theta_i|X_i].$$ We first bound the error term $\mathbb{E}\!\Biggl[\left(
	\frac{\hat S_{n,\kappa}(\hatG_{d,m})-\hat S_{n,\kappa}(\Gstar)}{\sqrt n}
	\right)^2\Biggr]$ from above as
	\begin{align}\label{eq:m4}
		\mathbb{E}\!\Biggl[\left(
		\frac{\hat S_{n,\kappa}(\hatG_{d,m})-\hat S_{n,\kappa}(\Gstar)}{\sqrt n}
		\right)^2\Biggr]
		&\leq \frac 1n
		\EE\qth{\pth{\sum_{i=1}^n(\hat\theta_\Gstar(X_i)-\hat\theta_{\hatG_{d,m}}(X_i))I(X_i\leq \kappa)}^2}
		\notag\\
		&\leq 
		n\EE\qth{(\hat\theta_\Gstar(X)-\hat\theta_{\hatG_{d,m}}(X))^2I(X\leq \kappa)}
	\end{align}
	where $X$ is an independent representative sample point. We would like to bound the right hand side of the above display. To this end, consider $\hat h=C(s)\log m$ for a large constant $C(s)>0$ such that 
	$$
	\inf_{G\in \mathcal C} G([0,\hat h])
	\geq 1- \frac {\breve c}{2m^4}
	$$ 
	for some constant $\breve c>0$ which is satisfied due to the subexponential tail assumption on $G\in \calC$. Let $\hat G_{d,m}^{(\hat h)}$ be the restriction of $\hat G_{d,m}$ on $[0,\hat{h}]$, given as
	\[
	\hat G_{d,m}^{(\hat h)}(d\theta) = \frac{\hat G_{d,m}(d\theta)\,I(\theta \le \hat{h})}{\hat G_{d,m}([0,\hat{h}])}.
	\]
	Hence,
	\[
	\hat G_{d,m}(d\theta) = \hat G_{d,m}^{(\hat h)}(d\theta)\,\hat G_{d,m}([0,\hat{h}]) + \hat G_{d,m}(d\theta)I(\theta > \hat{h}).
	\]
	For an $x\leq \kappa$, define the following quantities.
	\[
	p = \hat G_{d,m}([0,\hat{h}]),
	\quad 
	q = \int_{\theta > \hat{h}} f_\theta(x)\, d\hat G_{d,m}(\theta), \quad
	q_1 = \int_{\theta > \hat{h}} f_\theta(x+1)\, d\hat G_{d,m}(\theta).
	\]
	In view of the above definition, we get
	\[
	f_{\hat G_{d,m}}(x) = p f_{\hat G_{d,m}^{(h)}}(x) + q, \quad
	f_{\hat G_{d,m}}(x+1) = p f_{\hat G_{d,m}^{(h)}}(x+1) + q_1.
	\]
	Define $a = f_{\hat G_{d,m}^{(\hat h)}}(x+1)$ and $b = f_{\hat G_{d,m}^{(\hat h)}}(x)$ and note that by our assumptions we have $b,\; pb+q\geq c$ for a constant $c>0$.
   Indeed, letting $c:=b_\kappa/2$, we obtain $pb+q=f_{\hat G_{d,m}}(x)>c$ and 
    $$
    b=f_{\hat G_{d,m}^{(\hat h)}}(x)\geq\frac{f_{\hat G_{d,m}}(x)-\hat G_{d,m}((\hat h,+\infty))}{\hat G_{d,m}([0,\hat h])}
    \geq \frac{b_\kappa}{2}=c.
    $$
    In view of the above definitions, for any $x\leq \kappa$, we have
	\begin{align}\label{eq:m1}
		\left| \hat{\theta}_{\hat G_{d,m}}(x) - \hat{\theta}_{\hat G_{d,m}^{(h)}}(x) \right|\nonumber
		&= (x+1)\left| \frac{f_{\hat G_{d,m}}(x+1)}{f_{\hat G_{d,m}}(x)} - \frac{f_{\hat G_{d,m}^{(\hat h)}}(x+1)}{f_{\hat G_{d,m}^{(h)}}(x)} \right|= (x+1)\left| \frac{p a + q_1}{p b + q} - \frac{a}{b} \right|\nonumber\\
		&= (x+1)\left| \frac{pab + bq_1 - pab - aq}{b(pb+q)} \right|
		= (x+1)\left| \frac{b q_1 - a q}{b(pb+q)} \right|\nonumber
		\\&\le (x+1)\frac{q_1}{pb+q} + (x+1)\frac{q}{c^{2}}.
	\end{align}
Hence, we have 
	\begin{align}\label{eq:m3}
		\EE\left[\left(\hat{\theta}_{\hat G_{d,m}}(X) - \hat{\theta}_{\Gstar}(X)\right)^2I(X\leq \kappa)\right]
		&\leq 
		2\EE\left[\left(\hat{\theta}_{\hatG_{d,m}^{(\hat h)}}(X) - \hat{\theta}_{\Gstar}(X)\right)^2I(X\leq \kappa)\right]
		\nonumber\\
        &+2\EE\left[\left(\hat{\theta}_{\hat G_{d,m}}(X) - \hat{\theta}_{\hatG_{d,m}^{(\hat h)}}(X)\right)^2I(X\leq \kappa)\right]
		\nonumber\\
		&
		\leq 
		2\EE\left[\left(\hat{\theta}_{\hatG_{d,m}^{(\hat h)}}(X) - \hat{\theta}_{\Gstar}(X)\right)^2I(X\leq \kappa)\right]
		+ \frac 4{c{m^4}},
	\end{align}
	where the last inequality used \prettyref{eq:m1} and the fact that 
	$$
	q,q_1\leq \hat G[(\hat h,\infty)]\leq \frac{c}{m^4}.
	$$
	 To bound $\EE\left[\left(\hat{\theta}_{\hatG_{d,m}^{(\hat h)}}(X) - \hat{\theta}_{\Gstar}(X)\right)^2I(X\leq \kappa)\right]$, we 
use Lemma \ref{lmm:EB_general-regret-bound} with $V,h,M = C_1(s)\log m, \hat Q=\hat G^{(\hat h)}$ for a large constant $C_1(s)$ that relies on $s$ and satisfies
	$$
	\Gstar([0,h])\geq 1-\frac 1{m^5},
	\quad 
	\epsilon_V(\Gstar)\leq \frac 1{m^5},
	$$
	where $\epsilon_V(\Gstar)=\PP_{X\sim f_\Gstar}[X\geq V]$. Note that the first inequality is guaranteed as $\Gstar\in \subexpo(s)$ and the second inequality is guaranteed in view of \cite[Lemma 12]{jana2025optimal}. Hence, we get, conditioned on $\hat G,\hat h$
	\begin{align}\label{eq:m2}
		\EE\left[\left(\hat{\theta}_{\hatG_{d,m}^{(\hat h)}}(X) - \hat{\theta}_{\Gstar}(X)\right)^2I(X\leq \kappa)\right]
		\leq C_2(s)(\log m)^2\pth{H^2(f_{{\Gstar}},f_{\hatG_{d,m}^{(\hat h)}})
			+\frac 1{m^{5/2}}}.
	\end{align}
	Next, we use the following inequality that
	\[
	H^2(f_{Q^{(\tilde h)}},f_{Q}) \leq 2 \TV(f_{Q^{(\tilde h)}},f_{Q}) \leq 2 \TV(Q^{(\tilde h)},Q) = 4 Q((\tilde h,\infty)),
	\]
	where $\TV$ denotes the total variation and the middle inequality applies the data-processing inequality \cite{csiszar1967information} and the last equality followed as
	\begin{align*}
		\TV(Q^{(\tilde h)},Q)
		&=\int_0^{\tilde h}\abs{dQ^{(\tilde h)}(\theta)-dQ(\theta)}+\int_{\tilde h}^\infty
		dQ(\theta)\\
		&=\pth{\frac 1{Q([0,\tilde h])}-1}\int_0^{\tilde h} dQ(\theta)+Q((\tilde h,\infty))\\
		&=2Q((\tilde h,\infty)).
	\end{align*}
	In view of the above, we use $Q=\hatG_{d,m},\tilde h=\hat h$ to get
	\begin{align*}\label{eq:D1Hineq}
		H^2(f_{\Gstar},f_{\hatG_{d,m}^{(\hat h)}})
		\leq 
		2H^2(f_{\Gstar},f_{\hatG_{d,m}})+2H^2(f_{\hatG_{d,m}},f_{\hatG_{d,m}^{(\hat h)}})
		\leq
		2H^2(f_{\Gstar},f_{\hatG_{d,m}})+ \frac 2{m^4}.
	\end{align*}
	Then our result follows from combining the last display with \eqref{eq:m2}, \eqref{eq:m3} and \eqref{eq:m4} in view of \prettyref{lem:EB_density-H2} \ref{pt:subexpo-regular-dens}.

	\subsubsection{Compact-support setting} The proof of the bounded supported prior follows similarly by using a constant $\hat h$, where $\hat h$ denotes the parameter involved in the distribution class $\calP[0,\hat h]$. The rest of the arguments are identical with \eqref{eq:m2} replaced by
	\begin{align}
		\EE\left[\left(\hat{\theta}_{\hatG_{d,m}^{(\hat h)}}(X) - \hat{\theta}_{\Gstar}(X)\right)^2I(X\leq \kappa)\right]
		\leq C(\hat h)(\hat h^2+h\kappa)\pth{H^2(f_{{\Gstar}},f_{\hatG_{d,m}^{(\hat h)}})
			+\frac 1{m^4}},
	\end{align}
	and noting that $\hatG_{d,m}^{(\hat h)}$ is identical to $\hat G$.

\subsection{Proof of \prettyref{thm:sum_coarse} }\label{proof:totalcoarse}

\subsubsection{Compact-support setting} We can write
\begin{align}
\frac{1}{n}
\mathbb{E}\!\left(\hat S_{n,\kappa}(\tilde G_{d,m}^{(T)})- \hat S_{n,\kappa}(\Gstar)\right)^2
&=
\frac{1}{n}\,
\mathbb{E}\!\left[
\left(
\sum_{i=1}^n
(\widehat\theta_{\tilde G^{(T)}_{d,m}}(X_i)
-\hat \theta_{\Gstar}(X_i))
I(X_i\le \kappa)
\right)^2
\right]
\notag\\
&\le
\sum_{i=1}^n
\mathbb{E}\!\left[
(\widehat\theta_{{\tilde G}^{(T)}_{d,m}}(X_i)
-\hat \theta_{\Gstar}(X_i))^2
I(X_i\le \kappa)
\right],\label{eq:firstinequality}
\end{align}
where we have used the Cauchy-Schwartz inequality. Moreover,
\begin{align}
	&\mathbb{E}\!\left[	(\widehat\theta_{{\tilde G_{d,m}^{(T)}}}(X_i)-\hat \theta_{\Gstar}(X_i))^2	I(X_i\le \kappa)	\right]	=	\sum_{x=0}^\kappa	(x+1)^2	\left(	\frac{f_{{\tilde G_{d,m}^{(T)}}}(x+1)}{f_{{\tilde G_{d,m}^{(T)}}}(x)}
	-	\frac{f_{\Gstar}(x+1)}{f_{\Gstar}(x)}	\right)^2	f_{\Gstar}(x)	\nonumber\\
	&\nonumber	\le	3\Biggl\{\sum_{x=0}^\kappa	(x+1)^2	\left(	\frac{f_{{\tilde G_{d,m}^{(T)}}}(x+1)}{f_{{\tilde G_{d,m}^{(T)}}}(x)}	-	\frac{2f_{\tilde G_{d,m}^{(T)}}(x+1)}{f_{{\tilde G_{d,m}^{(T)}}}(x)+f_{\Gstar}(x)}
	\right)^2	f_{\Gstar}(x)\\
	&+\sum_{x=0}^\kappa	(x+1)^2	\left(	\frac{2(f_{\tilde G_{d,m}^{(T)}}(x+1)-f_{\Gstar}(x+1))}{f_{{\tilde G_{d,m}^{(T)}}}(x)+f_{\Gstar}(x)}	\right)^2	f_{\Gstar}(x)	\nonumber\\
	&	+\sum_{x=0}^\kappa	(x+1)^2	\left(	\frac{2f_{\Gstar}(x+1)}{f_{{\tilde G_{d,m}^{(T)}}}(x)+f_{\Gstar}(x)}
	-	\frac{f_{\Gstar}(x+1)}{f_{{\tilde G_{d,m}^{(T)}}}(x)}	\right)^2 f_{\Gstar}(x)\Biggr\}	\nonumber\\
	&\le	3\Biggl\{\sum_{x=0}^\kappa (x+1)^2	\left(	\frac{f_{{\tilde G_{d,m}^{(T)}}}(x+1)}{f_{{\tilde G_{d,m}^{(T)}}}(x)}	\right)^2	\left(	\frac{f_{\Gstar}(x)-f_{{\tilde G_{d,m}^{(T)}}}(x)}{f_{\Gstar}(x)+f_{{\tilde G_{d,m}^{(T)}}}(x)}	\right)^2 f_{\Gstar}(x)	\nonumber\\
    &\nonumber +4\sum_{x=0}^\kappa (x+1)^2	\frac{\big(f_{{\tilde G_{d,m}^{(T)}}}(x+1)-f_{\Gstar}(x+1)\big)^2}	{f_{{\tilde G_{d,m}^{(T)}}}(x)+f_{\Gstar}(x)}	f_{\Gstar}(x)	\nonumber\\
	&+\sum_{x=0}^\kappa	\left\{	(x+1)\frac{f_{\Gstar}(x+1)}{f_{\Gstar}(x)}	\right\}^2	\left(	\frac{f_{\Gstar}(x)-f_{{\tilde G_{d,m}^{(T)}}}(x)}{f_{\Gstar}(x)+f_{{\tilde G_{d,m}^{(T)}}}(x)}	\right)^2f_{\Gstar}(x)\Biggr\}	\nonumber\\
	&\leq	3\Biggl\{\sum_{x=0}^\kappa	\left(\widehat\theta_{\tilde G_{d,m}^{(T)}}(x)\right)^2	\left(\sqrt{f_{\Gstar}(x)}-\sqrt{f_{{\tilde G_{d,m}^{(T)}}}(x)}\right)^2	\nonumber\\
    &\nonumber +4(x+1)^2	\frac{f_{{\tilde G_{d,m}^{(T)}}}(x+1)+f_{\Gstar}(x+1)}	{f_{\Gstar}(x)+f_{{\tilde G_{d,m}^{(T)}}}(x)}	\left(	\sqrt{f_{{\tilde G_{d,m}^{(T)}}}(x+1)}-\sqrt{f_{\Gstar}(x+1)}	\right)^2	\nonumber\\
	&+	\left(\widehat\theta_{\Gstar}(x)\right)^2	\left(\sqrt{f_{\Gstar}(x)}-\sqrt{f_{{\tilde G_{d,m}^{(T)}}}(x)}\right)^2\Biggr\}	\nonumber\\
	&	\leq 	3\sum_{x=0}^\kappa \Biggl\{	\left((\widehat\theta_\Gstar(x))^2+(\widehat\theta_{{\tilde G_{d,m}^{(T)}}}(x))^2\right)	\left(	\sqrt{f_{\Gstar}(x)}-\sqrt{f_{{\tilde G_{d,m}^{(T)}}}(x)}	\right)^2\nonumber\\
    &\nonumber +4(x+1)	\left(\widehat\theta_\Gstar(x)+\widehat\theta_{{\tilde G_{d,m}^{(T)}}}(x)\right)	\left(
	\sqrt{f_{\Gstar}(x+1)}-\sqrt{f_{{\tilde G_{d,m}^{(T)}}}(x+1)}	\right)^2\Biggr\}	\nonumber\\
	&\le	12((\kappa+1)(h+\hat h)+\widehat h^2+h^2)	H^2\!\left(	\tilde f_{\Gstar}^{(T)},	\tilde f_{{\tilde G_{d,m}^{(T)}}}^{(T)}	\right),\label{eq:inequalities}
\end{align}
where the last inequality follows from $T\geq \kappa+1$. 
By  Lemma~\ref{lem:EB_density-H2_coars} (a), we can write, for every  $m \ge 3$,
\begin{equation}\label{eq:expectation}
\mathbb{E}\Big[H^2\Big(\widetilde f^{(T)}_{\Gstar}, \widetilde f^{(T)}_{\widehat G^{(T)}_{d,m}}\Big)\Big]
\le \frac{1}{m}\min\left\{ \frac{4c_2}{c_1}(T+2),\ \left(\frac{12c_2}{c_1} + 4 \right)\frac{2(\hat h e^2+2)\log m}{\log\log m}\right\}.
\end{equation}
Combining \eqref{eq:inequalities} and \eqref{eq:expectation}, we obtain
\begin{align*}
&\mathbb{E}\Big[\big(\hat\theta_{\widehat G^{(T)}_{d,m}}(X_i) - \hat\theta_{\Gstar}(X_i)\big)^2 I(X_i \le \kappa)\Big]\\
&\quad\le \frac{C(\hat h,\kappa,h)}{m}\min\left\{ \frac{4c_2}{c_1}(T+2),\left( \frac{12c_2}{c_1} + 4\right) \frac{2(\hat h e^2+2)\log m}{\log\log m}\right\},
\end{align*}
where $C(\hat h,\kappa,h) := 12\big((\kappa+1)(h+\hat h) + \hat h^2 + h^2\big)$. Substituting in  \eqref{eq:firstinequality}, we obtain
\begin{equation*}
\frac{1}{n}\,\mathbb{E}\left[\Big(\widehat S_{n,\kappa}\big(\widehat G^{(T)}_{d,m}\big) - \widehat S_{n,\kappa}(\Gstar)\Big)^2\right]
\le n\, C_{\hat h}\, \frac{\hat h^2 + \hat h \kappa}{m}\, \min\left\{ T,\ \frac{\log m}{\log\log m} \right\},
\end{equation*}
for a constant $C_{\hat h}$ which only depends on  $\hat h,c_1,c_2$.
\subsubsection{Subexponential setting}
Following the same reasoning as in Section \ref{proof:th41subexpo}, we can write
	\begin{align*}
		\mathbb{E}\!\Biggl[\left(
		\frac{\hat S_{n,\kappa}(\tilde{G}_{d,m}^{(T)})-\hat S_{n,\kappa}(\Gstar)}{\sqrt n}
		\right)^2\Biggr]
		&\leq 
		n\EE\qth{(\hat\theta_\Gstar(X)-\hat\theta_{\tilde{G}_{d,m}^{(T)}}(X))^2I(X\leq \kappa)},
	\end{align*}
with $X$ denoting an independent copy of the $X_i$. Let $\hat h=C_1(s)\log m$ be such that 
$$
	\hat h \geq \kappa+1,\quad \inf_{G\in\mathcal C}G([0,\hat h])\geq 1-\frac {b_0}{2m^2},
$$
and let $\tilde G^{(T,\hat h)}_{d,m}$ denote the restriction of $\tilde G^{(T)}_{d,m}$ to $[0,\hat h]$.
Reasoning as in \eqref{eq:m3}, with $\tilde{G}_{d,m}^{(T)}$ in the place of $\hat G_{d,m}$, we can write 
    \begin{align*}
\EE\left[\left(\hat{\theta}_{\tilde{G}_{d,m}^{(T)}}(X) - \hat{\theta}_{\Gstar}(X)\right)^2I(X\leq \kappa)\right]
&\leq 2\EE\left[\left(\hat{\theta}_{\tilde G^{(T,\hat h)}_{d,m}}(X) - \hat{\theta}_{\Gstar}(X)\right)^2I(X\leq \kappa)\right]+ \frac 4{c{m^4}}.
	\end{align*}
By \prettyref{lmm:EB_general-regret-bound} with $V=h=\widehat h=M=C_1(s)\log m$,
for $C_1(s)$ sufficiently large,  $\hat Q=\tilde G^{(T,\hat h)}_{d,m}$, we can write
\begin{align*}
	&\EE_{\Gstar}\left[\left(\hat{\theta}_{\tilde G^{(T,\hat h)}_{d,m}}(X) - \hat{\theta}_{\Gstar}(X)\right)^2I(X\leq \kappa)\right]\\
	&\leq  
	C_1(s)((\log m)^2+\kappa\log m)\sum_{x=0}^{\kappa+1}(\sqrt{f_{\Gstar}(x)}-\sqrt{f_{\tilde G^{(T,\hat h)}_{d,m}}(x)})^2 + \frac {C_2(s)(\log m)^2}{m^2}\\
    &\leq  C_1(s)((\log m)^2+\kappa\log m)
    H^2(\tilde f_{\Gstar}^{(T)},f_{\tilde G^{(T,\hat h)}_{d,m}}^{(T)}) + \frac {C_2(s)(\log m)^2}{m^2},
\end{align*}
where the last inequality comes from $T\geq \kappa+1$.
Next, as in Section \ref{proof:th41subexpo}, it holds
\begin{align*}
		H^2(\tilde f_{\Gstar}^{(T)},\tilde f_{\tilde G^{(T,\hat h)}_{d,m}}^{(T)})
				\leq
		2H^2(\tilde f_{\Gstar}^{(T)},f_{\tilde G_{d,m}^{(T)}}^{(T)})+ \frac 2{m^4}.
	\end{align*}
Combining this inequality with Lemma~\ref{lem:EB_density-H2_coars} (b) yields
\[
E\!\left[
\left(
\frac{\widehat S_{n,\kappa}(\widetilde G^{(T)}_{d,m})-
\widehat S_{n,\kappa}(\Gstar)}{\sqrt n}
\right)^2
\right]
\le
C_1\frac{\{(\log m)^2+\kappa\log m\}\min\{T,\log m\}}{m}
+
C_2\frac{(\log m)^2}{m^2}.
\]
The claim follows since 
\[
\frac{(\log m)^2}{m^2}
\le
\frac{\{(\log m)^2+\kappa\log m\}\min\{T,\log m\}}{m},
\]
for every admissible $m$ and $T\ge1$.


\section{Proofs of Section \ref{sec4}}

\subsection{Proof of Theorem \ref{above_average_lb}}\label{proof51_lb}

Based Lemmas \ref{lemma:lowerBound1}--\ref{lemma:lowerBound2}, we are in the position of proving Theorem \ref{above_average_lb}. We write  $p_\theta(x)=\operatorname{Poisson}(x\mid\theta)$, $f_{j,\varepsilon}=f_{G_{j,\varepsilon}}$, $\PP_{j,\varepsilon}=\PP_{G_{j,\varepsilon}}$ $\EE_{j,\varepsilon}=\EE_{G_{j,\varepsilon}}$ as in Section \ref{proof41_lb}.

\medskip

\noindent    (a) 
By definition, $m_{1,\varepsilon}(1)=0$. On the other hand, since $\varepsilon\leq \varepsilon_0\leq \frac 1 4$ and $c\varepsilon\leq \frac 1 2$, then
\begin{align*}
    m_{0,\varepsilon}(1)&=\frac{c \int_{1-\varepsilon}^1 \theta e^{-\theta} d\theta}{c \int_{1-\varepsilon}^1 \theta e^{-\theta}d\theta+(1-c\varepsilon)\int_a^b \theta e^{-\theta}\theta d\theta/(b-a)}\\
    &\geq \frac{c \varepsilon \frac 3 4 e^{-5/4}}{e^{-1}}\\
    &\geq k_1\varepsilon,
\end{align*}
for a constant $k_1$, where we have used $\theta e^{-\theta}\leq e^{-1}$.
Thus, $m_{0,\varepsilon}(1)-m_{1,\varepsilon}(1)\geq k_1\varepsilon$. 

On the other hand, $m_{0,\varepsilon}(x)-m_{1,\varepsilon}(x)=0$ for $x=0$ and $x\geq 3$, and, denoting $d=(b-a)^{-1}$
\begin{align*}
 & |  m_{0,\varepsilon}(2)- m_{1,\varepsilon}(2)|\\
 &=  \left|\frac{c \int_{1-\varepsilon}^1 e^{-\theta}\theta^2 d\theta+(1-c\varepsilon)d\int_a^b I(\theta <2)e^{-\theta}\theta^2d\theta}{c \int_{1-\varepsilon}^1 e^{-\theta}\theta^2 d\theta+(1-c\varepsilon)d\int_a^b e^{-\theta}\theta^2 d\theta}
    - \frac{c \int_{1}^{1+\varepsilon} e^{-\theta}\theta^2 d\theta+(1-c\varepsilon)d\int_a^b I(\theta <2)e^{-\theta}\theta^2d\theta}{c \int_{1}^{1+\varepsilon} e^{-\theta}\theta^2 d\theta+(1-c\varepsilon)d\int_a^b e^{-\theta}\theta^2 d\theta}
    \right|\\
    &=\left|\frac{B_\varepsilon+u_{0,\varepsilon}}{D_\varepsilon +u_{0,\varepsilon}}-\frac{B_\varepsilon+u_{1,\varepsilon}}{D_\varepsilon +u_{1,\varepsilon}}\right|,
    \end{align*}
where
$$
B_\varepsilon=(1-c\varepsilon)d\int_a^b I(\theta<2)e^{-\theta}\theta^2d\theta,\quad 
D_\varepsilon=(1-c\varepsilon)d\int_a^{b}e^{-\theta}\theta^2d\theta,
$$
$$
u_{0,\varepsilon}=c\int_{1-\varepsilon}^1 e^{-\theta}\theta^2 d\theta,\quad u_{1,\varepsilon}=c\int_{1}^{1+\varepsilon} e^{-\theta}\theta^2 d\theta.
$$
Since $B_\varepsilon\leq D_\varepsilon$, with $D_\varepsilon$ bounded away from zero, uniformly with respect to $\varepsilon\leq \varepsilon_0$, then  the function $u\mapsto \frac{B_\varepsilon+u}{D_\varepsilon +u}$ is differentiable on $[0,\infty)$ and its derivative $\frac{D_\varepsilon-B_\varepsilon}{(D_\varepsilon+u)^2}$ is uniformly bounded with respect to $\varepsilon\leq \varepsilon_0$. It follows that there exists a constant $k_2>0$ such that
\begin{align*}
 |  m_{0,\varepsilon}(2)- m_{1,\varepsilon}(2)|&\leq k_2\left|\int_{1-\varepsilon}^1 e^{-\theta}\theta^2 d\theta-\int_{1}^{1+\varepsilon} e^{-\theta}\theta^2 d\theta\right|\\
   &=k_2\left|\varepsilon  \int_0^1e^{-(1-\varepsilon+\varepsilon s)}
    (1-\varepsilon+\varepsilon s)^2-e^{-(1+\varepsilon s)} (1+\varepsilon s)^2 ds\right|\\
     &=k_2\varepsilon \left| \int_0^1\int_{1+\varepsilon s-\varepsilon}^{1+\varepsilon s}(2-\theta)\theta e^{-\theta}d\theta
     ds\right|\\
     &\leq k_2\varepsilon^2 \left(\frac 5 4 \right)^2e^{-3/4}:=k_3\varepsilon^2.
\end{align*}
This, together with Lemma \ref{lemma:lowerBound1} yields
\begin{align*}
    \EE_{j,\varepsilon}[m_{0,\varepsilon}(X_1)-m_{1,\varepsilon}(X_1)]&
    =\sum_{t=0}^\infty (m_{0,\varepsilon}(t)-m_{1,\varepsilon}(t))f_{j,\varepsilon}(t)\\
    &\geq (m_{0,\varepsilon}(1)-m_{1,\varepsilon}(1))f_{j,\varepsilon}(1)-|m_{0,\varepsilon}(2)-m_{1,\varepsilon}(2)|f_{j,\varepsilon}(2)\\
    &\geq k_1 C_2\varepsilon- k_3\varepsilon^2\geq C_5\varepsilon,
\end{align*}
for a constant $C_5>0$ and $\varepsilon_0$ small enough.
\medskip

\noindent (b)
By point (a), its proof and the explicit expression of $m_{0,\varepsilon}(1)$,
\begin{align*}
    \EE_{j,\varepsilon}[m_{0,\varepsilon}(X_1)-m_{1,\varepsilon}(X_1)]    \geq C_5\varepsilon,\quad 
\sup_{t\geq0}|m_{0,\varepsilon}(t)-m_{1,\varepsilon}(t)|
\leq k_4 \varepsilon,
\end{align*}
for constants $C_5$, $k_4$ and $\varepsilon_0$ small enough
Thus, 
\[
\operatorname{Var}_{j,\varepsilon}
(m_{0,\varepsilon}(X_1)-m_{1,\varepsilon}(X_1))\leq k_4^2\varepsilon^2.
\]
Using Chebyshev inequality, we can write
\begin{align*}
     &\PP_{j,\varepsilon}\left[\widehat S_{n,>}(G_{0,\varepsilon})
-\widehat S_{n,>}(G_{1,\varepsilon})\leq \frac 1 2C_5n\varepsilon\right]\\
&\leq \PP_{j,\varepsilon}\left[\widehat S_{n,>}(G_{0,\varepsilon})
-\widehat S_{n,>}(G_{1,\varepsilon})\leq n\EE_{j,\varepsilon}[m_{0,\varepsilon}(X_1)-m_{1,\varepsilon}(X_1)]-\frac 1 2C_5n\varepsilon\right]\\
&\leq \PP_{j,\varepsilon}\left[\left|\widehat S_{n,>}(G_{0,\varepsilon})
-\widehat S_{n,>}(G_{1,\varepsilon})-n\EE_{j,\varepsilon}[m_{0,\varepsilon}(X_1)-m_{1,\varepsilon}(X_1)]\right|\geq \frac 1 2C_5n\varepsilon\right]\\
&\leq \frac{4 k_4^2}{C_5^2n}\rightarrow 0,
\end{align*}
as $n\rightarrow\infty$. Thus,
\begin{equation}\label{eq:aboveIneq}
        \PP_{j,\varepsilon}\left[\left|\widehat S_{n,>}(G_{0,\varepsilon})
-\widehat S_{n,>}(G_{1,\varepsilon})\right|\leq \frac 1 2C_5n\varepsilon\right]
\rightarrow 0,
\end{equation}
Given an estimator $\widehat S$, define a test $\phi$ for
$H_0:G=G_{0,\varepsilon}$ against $H_1:G=G_{1,\varepsilon}$ by
\[
\phi=
\begin{cases}
0,&\text{if }|\widehat S-\widehat S_{n,>}(G_{0,\varepsilon})|<
|\widehat S-\widehat S_{n,>}(G_{1,\varepsilon})|,
\\
1,&\text{otherwise}.
\end{cases}
\]
Let
\[
A_n=\left\{\left|\widehat S_{n,>}(G_{0,\varepsilon})
-\widehat S_{n,>}(G_{1,\varepsilon})\right|\geq \frac 1 2 C_5n\varepsilon \right\},
\]
with $C_5$ as in point (a).
By \eqref{eq:aboveIneq}, $\PP_{j,\varepsilon}(A_n^c)\rightarrow0$ for $j=0,1$. Moreover, on $A_n\cap\{\phi=1\}$,
\[
\left|\widehat S-\widehat S_{n,>}(G_{0,\varepsilon})
\right|\geq
\frac{C_5}{4}n\varepsilon,
\]
and on $A_n\cap\{\phi=0\}$,
\[
\left|\widehat S-\widehat S_{n,>}(G_{1,\varepsilon})\right|
\geq\frac{C_5}{4}n\varepsilon.
\]
Hence
\begin{align*}
&\EE_{{0,\varepsilon}}\left[\left(
\frac{\widehat S-\widehat S_{n,>}(G_{0,\varepsilon})}{\sqrt n}
\right)^2\right]\geq \frac{C_5^2}{16}n\varepsilon^2 \PP_{0,\varepsilon}
\left(A_n\cap\{\phi=1\}\right)
\end{align*}
and
\begin{align*}
&\EE_{{1,\varepsilon}}\left[\left(\frac{\widehat S-\widehat S_{n,>}(G_{1,\varepsilon})}{\sqrt n}\right)^2\right]\geq \frac{C_5^2}{16}n\varepsilon^2 \PP_{1,\varepsilon}\left(A_n\cap\{\phi=0\}\right).
\end{align*}
Moreover,
\begin{align*}
&\PP_{0,\varepsilon}\left(A_n\cap\{\phi=1\}\right)+\PP_{1,\varepsilon}
\left(A_n\cap\{\phi=0\}\right)\geq \frac12-o(1),
\end{align*}
where the last inequality follows from Lemma \ref{lemma:lowerBound2}. Hence, for 
$n$ sufficiently large,
\[
\max_{j\in\{0,1\}}\EE_{{j,\varepsilon}}\left[\left(\frac{
\widehat S-\widehat S_{n,>}(G_{j,\varepsilon})
}{\sqrt n}\right)^2\right]\geq \frac{C_5^2}{128}n\varepsilon^2.
\]
Taking
\[
\varepsilon=\gamma(m+n)^{-1/4}
\]
with $\gamma\in (0,\min(\varepsilon_0,(2\sqrt{C_1})^{-1/2})$ gives
\[
\max_{j\in\{0,1\}}\EE_{{j,\varepsilon}}\left[\left(\frac{
\widehat S-\widehat S_{n,>}(G_{j,\varepsilon})}{\sqrt n}\right)^2\right]
\geq \frac{\gamma^2C_5^2}{128}\frac{n}{\sqrt{m+n}}:=C_6 \frac{n}{\sqrt{m+n}}.
\]
Since $G_{j,\varepsilon}\in\mathcal C$ for $j=0,1$ the result follows.

\subsection{Proof of Theorem~\ref{thm:main-above-average}}\label{proof_teo_above_regular}

\subsubsection{Preliminary lemmas}

\begin{lemma}\label{lem:ratio-stab}
Let $A,B,\tilde A,\tilde B\ge 0$ with $B>0$, $\tilde B>0$, and $A\le B$, $\tilde A\le \tilde B$. Then
\begin{displaymath}
\frac{B\tilde B}{B+\tilde B}\Big(\frac{A}{B}-\frac{\tilde A}{\tilde B}\Big)^2\le 8\Big\{(\sqrt A-\sqrt{\tilde A})^2+(\sqrt B-\sqrt{\tilde B})^2\Big\}.
\end{displaymath}
\end{lemma}

\begin{proof}
Set $a=\sqrt A$, $b=\sqrt B$, $\tilde a=\sqrt{\tilde A}$, and $\tilde b=\sqrt{\tilde B}$.  Since $A\le B$ and $\tilde A\le \tilde B$, we have $0\le a\le b$ and $0\le \tilde a\le \tilde b$. 
Now
\begin{displaymath}
\frac{A}{B}-\frac{\tilde A}{\tilde B}=\Big(\frac{a}{b}\Big)^2-\Big(\frac{\tilde a}{\tilde b}\Big)^2=\Big(\frac{a}{b}-\frac{\tilde a}{\tilde b}\Big)\Big(\frac{a}{b}+\frac{\tilde a}{\tilde b}\Big).
\end{displaymath}
Because $a/b,\tilde a/\tilde b\in[0,1]$, we have $\frac{a}{b}+\frac{\tilde a}{\tilde b}\le 2$. Therefore $\Big(\frac{A}{B}-\frac{\tilde A}{\tilde B}\Big)^2\le4\Big(\frac{a}{b}-\frac{\tilde a}{\tilde b}\Big)^2$. Multiplying by $\frac{B\tilde B}{B+\tilde B}=\frac{b^2\tilde b^2}{b^2+\tilde b^2}$ gives
\begin{displaymath}
\frac{B\tilde B}{B+\tilde B}\Big(\frac{A}{B}-\frac{\tilde A}{\tilde B}\Big)^2\le4\cdot \frac{b^2\tilde b^2}{b^2+\tilde b^2}\Big(\frac{a}{b}-\frac{\tilde a}{\tilde b}\Big)^2=4\cdot \frac{(a\tilde b-\tilde a b)^2}{b^2+\tilde b^2}.
\end{displaymath}
Next, decompose $a\tilde b-\tilde a b=(a-\tilde a)\tilde b+\tilde a(\tilde b-b)$. Hence, by $(u+v)^2\le 2u^2+2v^2$, $(a\tilde b-\tilde a b)^2\le2(a-\tilde a)^2\tilde b^2+2\tilde a^2(\tilde b-b)^2$. Since $\tilde a\le \tilde b$, this yields
\begin{displaymath}
(a\tilde b-\tilde a b)^2\le2(a-\tilde a)^2\tilde b^2+2\tilde b^2(\tilde b-b)^2.
\end{displaymath}
Divide by $b^2+\tilde b^2\ge \tilde b^2$:
\begin{displaymath}
\frac{(a\tilde b-\tilde a b)^2}{b^2+\tilde b^2}\le2(a-\tilde a)^2+2(\tilde b-b)^2.
\end{displaymath}
Substituting back gives
\begin{displaymath}
\frac{B\tilde B}{B+\tilde B}\Big(\frac{A}{B}-\frac{\tilde A}{\tilde B}\Big)^2\le8(a-\tilde a)^2+8(\tilde b-b)^2=8\Big\{(\sqrt A-\sqrt{\tilde A})^2+(\sqrt B-\sqrt{\tilde B})^2\Big\}.
\end{displaymath}
This proves the lemma.
\end{proof}

\begin{lemma}\label{lem:quad-convert}
For any $a,b\ge 0$ (with the convention $0/0=0$),
\begin{displaymath}
a \le 4\,\frac{ab}{a+b} + 2\,\frac{(a-b)^2}{a+b}.
\end{displaymath}
Consequently, for any $u\in[0,1]$,
\begin{displaymath}
a\,u \le 4\,\frac{ab}{a+b}\,u + 2\,\frac{(a-b)^2}{a+b}.
\end{displaymath}
\end{lemma}

\begin{proof}
If $a+b=0$, then $a=b=0$ and the inequality is trivial. Assume now $a+b>0$. Multiplying the desired inequality by $(a+b)$, it is enough to prove that $a(a+b)\le 4ab+2(a-b)^2$. Expand the right-hand side: $4ab+2(a-b)^2=4ab+2(a^2-2ab+b^2)=2a^2+2b^2$. So it suffices to prove that $a(a+b)\le 2a^2+2b^2$. But
\begin{displaymath}
2a^2+2b^2-a(a+b)=a^2-ab+2b^2.
\end{displaymath}
Now $a^2-ab+2b^2=\frac12(a-b)^2+\frac12 a^2+\frac32 b^2\ge 0$. Hence $a(a+b)\le 2a^2+2b^2=4ab+2(a-b)^2$, which proves
\begin{displaymath}
a \le 4\,\frac{ab}{a+b}+2\,\frac{(a-b)^2}{a+b}.
\end{displaymath}
For the second statement, simply multiply both sides by $u\in[0,1]$ and use $u\le 1$:
\begin{displaymath}
a\,u\le4\,\frac{ab}{a+b}\,u+2\,\frac{(a-b)^2}{a+b}\,u\le4\,\frac{ab}{a+b}\,u+2\,\frac{(a-b)^2}{a+b}.
\end{displaymath}
This proves the lemma.
\end{proof}

\begin{lemma}[$\chi^2$ control by Hellinger]\label{lem:chi2-hell}
For any $a,b\ge 0$ (with $0/0=0$),
\begin{displaymath}
\frac{(a-b)^2}{a+b}\le 4(\sqrt a-\sqrt b)^2.
\end{displaymath}
Hence for pmfs $p,q$ on a countable alphabet,
\begin{displaymath}
\sum_t \frac{(p(t)-q(t))^2}{p(t)+q(t)}\le4\sum_t(\sqrt{p(t)}-\sqrt{q(t)})^2.
\end{displaymath}
\end{lemma}

\begin{proof}
If $a=b=0$, the claim is trivial. Otherwise, use $a-b=(\sqrt a-\sqrt b)(\sqrt a+\sqrt b)$, so $(a-b)^2=(\sqrt a-\sqrt b)^2(\sqrt a+\sqrt b)^2$. Since $(\sqrt a+\sqrt b)^2=a+b+2\sqrt{ab}\le 2(a+b)$, we obtain
\begin{displaymath}
(a-b)^2\le 2(a+b)(\sqrt a-\sqrt b)^2.
\end{displaymath}
In fact this already gives the sharper bound with constant \(2\); hence certainly $\frac{(a-b)^2}{a+b}\le 4(\sqrt a-\sqrt b)^2$. Summing the pointwise inequality with $a=p(t)$ and $b=q(t)$ yields
\begin{displaymath}
\sum_t \frac{(p(t)-q(t))^2}{p(t)+q(t)}\le4\sum_t(\sqrt{p(t)}-\sqrt{q(t)})^2.
\end{displaymath}
This proves the lemma.
\end{proof}

\subsubsection{Compact-support setting}

\paragraph{The smoothed numerator bound.}

For \(G_1,G_2\in\mathcal C\), let
\begin{displaymath}
\mathcal H(G_1,G_2):=H^2(f_{G_1},f_{G_2})=\sum_{t\ge0}\left(\sqrt{f_{G_1}(t)}-\sqrt{f_{G_2}(t)}\right)^2.
\end{displaymath}

\begin{lemma}\label{lem:sqrtA-control-route1-hhat}
Fix \(x\in\mathcal W_+\) and \(\delta\in(0,1]\). Assume \prettyref{asmp:aa-self-normal-regular} holds on \(\mathcal C\). Then, for all \(G_1,G_2\in\mathcal C\),
\begin{displaymath}
\big(\sqrt{A_{x,\delta}(G_1)}-\sqrt{A_{x,\delta}(G_2)}\big)^2\le 8C_A\,\mathcal H(G_1,G_2)+2\sqrt{2}\,\varepsilon_{x,\delta}.
\end{displaymath}
\end{lemma}

\begin{proof}
Write $A_j=A_{x,\delta}(G_j)$ and $L_j=\sum_{t\ge0}\alpha_{x,\delta}(t)f_{G_j}(t)$ for $j=1,2$. Set \(S=A_1+A_2\). If \(S=0\), then \(A_1=A_2=0\), and the claim is trivial. Assume \(S>0\). By the representation bound in Assumption~\ref{asmp:aa-self-normal-regular}, 
\begin{displaymath} 
|A_j-L_j|\le \varepsilon_{x,\delta}, \qquad j=1,2.
 \end{displaymath}
Hence $|A_1-A_2|\le|L_1-L_2|+2\varepsilon_{x,\delta}$. Using $(\sqrt u-\sqrt v)^2\le \frac{(u-v)^2}{u+v}$ for $u,v\ge0$, we obtain
\begin{displaymath}
(\sqrt{A_1}-\sqrt{A_2})^2\le\frac{2|L_1-L_2|^2}{S}+\frac{8\varepsilon_{x,\delta}^2}{S}.
\end{displaymath}
By Cauchy--Schwarz inequality
\begin{displaymath}
|L_1-L_2|^2\le\left\{\sum_{t\ge0}\alpha_{x,\delta}(t)^2(f_{G_1}(t)+f_{G_2}(t))\right\}\left\{\sum_{t\ge0}\frac{(f_{G_1}(t)-f_{G_2}(t))^2}{f_{G_1}(t)+f_{G_2}(t)}\right\}.
\end{displaymath}
By the self-normalization bound in Assumption~\ref{asmp:aa-self-normal-regular} and Lemma~\ref{lem:chi2-hell},
\begin{displaymath}
|L_1-L_2|^2\le 4C_A S\,\mathcal H(G_1,G_2).
\end{displaymath}
Therefore
\begin{displaymath}
(\sqrt{A_1}-\sqrt{A_2})^2\le 8C_A\,\mathcal H(G_1,G_2)+\frac{8\varepsilon_{x,\delta}^2}{S}.
\end{displaymath}
On the other hand,
\begin{displaymath}
(\sqrt{A_1}-\sqrt{A_2})^2\le A_1+A_2=S.
\end{displaymath}
Combining the two bounds gives
\begin{displaymath}
(\sqrt{A_1}-\sqrt{A_2})^2\le8C_A\,\mathcal H(G_1,G_2)+\min\left\{S,\frac{8\varepsilon_{x,\delta}^2}{S}\right\}.
\end{displaymath}
Finally, $\min\left\{S,\frac{8\varepsilon_{x,\delta}^2}{S}\right\}\le2\sqrt{2}\,\varepsilon_{x,\delta}$, which proves the result.
\end{proof}

We first note that the ratios appearing below are well defined without any margin assumption. Assumption~\ref{asmp:aa-local-mass}, applied at $x=0$, implies that $G(\{0\})=0$ for every $G\in\mathcal C$. Indeed, for every $\delta\in(0,1]$,
\begin{displaymath}
G(\{0\})\le G([-\delta,\delta]\cap\operatorname{supp}(G))\le L\delta,
\end{displaymath}
and the claim follows by letting $\delta\downarrow0$. Since $G$ is a probability measure supported on $[0,\hat h]$, it follows that, for every $x\geq1$,
\begin{displaymath}
B_x(G)=\int_{[0,\hat h]}e^{-\theta}\theta^x\,dG(\theta)>0,
\qquad G\in\mathcal C.
\end{displaymath}

If $\hat K=0$, then $m_G(0)=0$ and $m_G(x)=1$ for every $x\geq1$ and every $G\in\mathcal C$, so the result is immediate. We may therefore assume $\hat K\geq1$.

\paragraph{Regret with respect to the oracle.}

Define $\hat S_{n,>}(\hat G_{d,m}):=\sum_{i=1}^n m_{\hat G_{d,m}}(X_i)$. We start from the basic reduction
\begin{align}\label{eq:basic-reduction-hhat}
\frac{1}{n}\EE\!\left[\left(\hat S_{n,>}(\Gstar)-\hat S_{n,>}(\hat G_{d,m})\right)^2\mid \hat G_{d,m}\right]
&\le n\sum_{x=1}^{\hat K}(m_{\Gstar}(x)-m_{\hat G_{d,m}}(x))^2 f_{\Gstar}(x).
\end{align}
Since both $\Gstar$ and $\hat G_{d,m}$ are supported on $[0,\hat h]$, for every integer $x\ge \hat K+1$ we have $x>\hat h$, and therefore
\begin{displaymath}
I(\theta<x)=1\qquad\text{for every }\theta\in[0,\hat h].
\end{displaymath}
Hence $m_{\Gstar}(x)=m_{\hat G_{d,m}}(x)=1$ for every $x\ge \hat K+1$. Also, since $\theta\ge 0$ almost surely under both $\Gstar$ and $\hat G_{d,m}$,
\begin{displaymath}
I(\theta<0)=0,
\end{displaymath}
and therefore $m_{\Gstar}(0)=0$ and $m_{\hat G_{d,m}}(0)=0$. Thus the difference $m_{\Gstar}(x)-m_{\hat G_{d,m}}(x)$ vanishes outside $\{1,\dots,\hat K\}$. Let
\begin{displaymath}
\Delta(x):=(m_{\Gstar}(x)-m_{\hat G_{d,m}}(x))\,I(1\le x\le \hat K).
\end{displaymath}
Then, conditional on $\hat G_{d,m}$, the variables $\Delta(X_1),\dots,\Delta(X_n)$ are i.i.d.\ with common law $f_{\Gstar}$. Moreover, $0\le m_{\Gstar},m_{\hat G_{d,m}}\le 1$, hence $|\Delta(x)|\le 1$. Now
\begin{displaymath}
\hat S_{n,>}(\Gstar)-\hat S_{n,>}(\hat G_{d,m})=\sum_{i=1}^n \Delta(X_i),
\end{displaymath}
and therefore
\begin{displaymath}
\EE\!\left[\left(\sum_{i=1}^n \Delta(X_i)\right)^2\Bigm|\hat G_{d,m}\right]=\sum_{i=1}^n \EE[\Delta(X_i)^2\mid\hat G_{d,m}]+\sum_{i\neq j}\EE[\Delta(X_i)\Delta(X_j)\mid\hat G_{d,m}].
\end{displaymath}
By conditional independence, $\EE[\Delta(X_i)\Delta(X_j)\mid\hat G_{d,m}]=\EE[\Delta(X)\mid\hat G_{d,m}]^2$, for $i\neq j$. Thus
\begin{displaymath}
\EE\!\left[\left(\sum_{i=1}^n \Delta(X_i)\right)^2\Bigm|\hat G_{d,m}\right]=n\,\EE[\Delta(X)^2\mid\hat G_{d,m}]+n(n-1)\,\EE[\Delta(X)\mid\hat G_{d,m}]^2.
\end{displaymath}
Since $\EE[\Delta(X)\mid\hat G_{d,m}]^2\le \EE[\Delta(X)^2\mid\hat G_{d,m}]$, we obtain
\begin{displaymath}
\EE\!\left[\left(\sum_{i=1}^n \Delta(X_i)\right)^2\Bigm|\hat G_{d,m}\right]\le n^2\,\EE[\Delta(X)^2\mid\hat G_{d,m}].
\end{displaymath}
Dividing by \(n\) and writing out the expectation under \(X\sim f_{\Gstar}\) yields
\begin{displaymath}
\frac1n\EE\!\left[\left(\hat S_{n,>}(\Gstar)-\hat S_{n,>}(\hat G_{d,m})\right)^2\Bigm|\hat G_{d,m}\right]\le n\sum_{x=1}^{\hat K}(m_{\Gstar}(x)-m_{\hat G_{d,m}}(x))^2 f_{\Gstar}(x),
\end{displaymath}
which is \eqref{eq:basic-reduction-hhat}.

\paragraph{Smoothing decomposition.}

By $(a+b+c)^2\le 3(a^2+b^2+c^2)$,
\begin{align}\label{eq:main-sum-decomp-hhat}
&\sum_{x=1}^{\hat K}(m_{\Gstar}(x)-m_{\hat G_{d,m}}(x))^2 f_{\Gstar}(x)\\
&\notag\quad\le 3\sum_{x=1}^{\hat K}(m_{\Gstar}(x)-m_{\Gstar,\delta}(x))^2 f_{\Gstar}(x)\\
&\notag\quad\quad+3\sum_{x=1}^{\hat K}(m_{\Gstar,\delta}(x)-m_{\hat G_{d,m},\delta}(x))^2 f_{\Gstar}(x)+3\sum_{x=1}^{\hat K}(m_{\hat G_{d,m},\delta}(x)-m_{\hat G_{d,m}}(x))^2 f_{\Gstar}(x).
\end{align}

We now control the three terms on the right-hand side separately.

\paragraph{Bias term for $\Gstar$.}

We first establish a bound for a generic $G\in\mathcal C$. Fix $x\in\{1,\dots,\hat K\}$. Since
\begin{displaymath}
m_G(x)=\frac{A_x(G)}{B_x(G)},\qquad m_{G,\delta}(x)=\frac{A_{x,\delta}(G)}{B_x(G)},
\end{displaymath}
we have
\begin{displaymath}
m_{G,\delta}(x)-m_G(x)=\frac{A_{x,\delta}(G)-A_x(G)}{B_x(G)}.
\end{displaymath}
Hence, using \eqref{eq:phi-pointwise-error-main},
\begin{displaymath}
|m_{G,\delta}(x)-m_G(x)|\le\frac{1}{f_G(x)}\int_{\mathbb R_+}\frac{e^{-\theta}\theta^x}{x!}I(|\theta-x|<\delta)\,dG(\theta).
\end{displaymath}
Since
\begin{displaymath}
\int_{\mathbb R_+}\frac{e^{-\theta}\theta^x}{x!}I(|\theta-x|<\delta)\,dG(\theta)\le f_G(x),
\end{displaymath}
we obtain
\begin{align}
&\sum_{x=1}^{\hat K} f_G(x)\bigl(m_{G,\delta}(x)-m_G(x)\bigr)^2\notag\\
&\qquad\le\sum_{x=1}^{\hat K}\frac{1}{f_G(x)}\left\{\int_{\mathbb R_+}\frac{e^{-\theta}\theta^x}{x!}I(|\theta-x|<\delta)\,dG(\theta)\right\}^2\notag\\
&\qquad\le\sum_{x=1}^{\hat K}\int_{\mathbb R_+}\frac{e^{-\theta}\theta^x}{x!}I(|\theta-x|<\delta)\,dG(\theta)
\le\sum_{x=1}^{\hat K}G([x-\delta,x+\delta]\cap\operatorname{supp}(G))\le\hat K L\delta,\label{eq:bias-generic-hhat}
\end{align}
where the last inequality follows from
Assumption~\ref{asmp:aa-local-mass}. In particular, taking $G=\Gstar$ gives
\begin{equation}\label{eq:bias-G-hhat}
\sum_{x=1}^{\hat K} f_{\Gstar}(x)\bigl(m_{\Gstar,\delta}(x)-m_{\Gstar}(x)\bigr)^2\le\hat K L\delta.
\end{equation}

\paragraph{Bias term for $\hat G_{d,m}$.}

Applying \eqref{eq:bias-generic-hhat} with $G=\hat G_{d,m}$ gives 
\begin{equation}\label{eq:bias-Ghat-own-hhat}
\sum_{x=1}^{\hat K} f_{\hat G_{d,m}}(x)\bigl(m_{\hat G_{d,m},\delta}(x)-m_{\hat G_{d,m}}(x)\bigr)^2\le\hat K L\delta.
\end{equation}
Now define
\begin{displaymath}
u(x):=\bigl(m_{\hat G_{d,m},\delta}(x)-m_{\hat G_{d,m}}(x)\bigr)^2\in[0,1].
\end{displaymath}
Applying Lemma~\ref{lem:quad-convert} pointwise with $a=f_{\Gstar}(x)$ and $b=f_{\hat G_{d,m}}(x)$ gives
\begin{displaymath}
f_{\Gstar}(x)u(x)\le 4\,\frac{f_{\Gstar}(x)f_{\hat G_{d,m}}(x)}{f_{\Gstar}(x)+f_{\hat G_{d,m}}(x)}\,u(x)+2\,\frac{(f_{\Gstar}(x)-f_{\hat G_{d,m}}(x))^2}{f_{\Gstar}(x)+f_{\hat G_{d,m}}(x)}.
\end{displaymath}
Summing over $x=1,\dots,\hat K$ and using
\begin{displaymath}
\frac{f_{\Gstar}(x)f_{\hat G_{d,m}}(x)}{f_{\Gstar}(x)+f_{\hat G_{d,m}}(x)}\le f_{\hat G_{d,m}}(x),
\end{displaymath}
we obtain
\begin{displaymath}
\sum_{x=1}^{\hat K}f_{\Gstar}(x)u(x)\le 4\sum_{x=1}^{\hat K}f_{\hat G_{d,m}}(x)u(x)+2\sum_{x=1}^{\hat K}\frac{(f_{\Gstar}(x)-f_{\hat G_{d,m}}(x))^2}{f_{\Gstar}(x)+f_{\hat G_{d,m}}(x)}.
\end{displaymath}
By \eqref{eq:bias-Ghat-own-hhat} and Lemma~\ref{lem:chi2-hell}, it follows that
\begin{equation}\label{eq:bias-Gtilde-hhat}
\sum_{x=1}^{\hat K} f_{\Gstar}(x)\bigl(m_{\hat G_{d,m},\delta}(x)-m_{\hat G_{d,m}}(x)\bigr)^2\le 4\hat K L\delta+8\mathcal H(\Gstar,\hat G_{d,m}).
\end{equation}

\paragraph{Smoothed estimation error.}

Let $u(x):=(m_{\Gstar,\delta}(x)-m_{\hat G_{d,m},\delta}(x))^2\in[0,1]$. Applying Lemma~\ref{lem:quad-convert} with $a=f_{\Gstar}(x)$ and $b=f_{\hat G_{d,m}}(x)$ yields
\begin{displaymath}
\sum_{x=1}^{\hat K} f_{\Gstar}(x)u(x)\le 4\sum_{x=1}^{\hat K} w_x(\Gstar,\hat G_{d,m})u(x)+2\sum_{x=1}^{\hat K}\frac{(f_{\Gstar}(x)-f_{\hat G_{d,m}}(x))^2}{f_{\Gstar}(x)+f_{\hat G_{d,m}}(x)},
\end{displaymath}
where
\begin{displaymath}
w_x(\Gstar,\hat G_{d,m}):=\frac{f_{\Gstar}(x)f_{\hat G_{d,m}}(x)}{f_{\Gstar}(x)+f_{\hat G_{d,m}}(x)}.
\end{displaymath}
Using Lemma~\ref{lem:chi2-hell},
\begin{equation}\label{eq:smoothed-convert-hhat}
\sum_{x=1}^{\hat K} f_{\Gstar}(x)u(x)\le4\sum_{x=1}^{\hat K} w_x(\Gstar,\hat G_{d,m})u(x)+8\mathcal H(\Gstar,\hat G_{d,m}).
\end{equation}
Now fix \(x\le \hat K\). Apply Lemma~\ref{lem:ratio-stab} with $A=A_{x,\delta}(\Gstar)$, $B=B_x(\Gstar)$, $\tilde A=A_{x,\delta}(\hat G_{d,m})$ and $\tilde B=B_x(\hat G_{d,m})$. Since \(0\le \phi_\delta\le 1\), we have \(A_{x,\delta}(\Gstar)\le B_x(\Gstar)\) and \(A_{x,\delta}(\hat G_{d,m})\le B_x(\hat G_{d,m})\), so the lemma applies. Therefore
\begin{align*}
&\frac{B_x(\Gstar)B_x(\hat G_{d,m})}{B_x(\Gstar)+B_x(\hat G_{d,m})}\Big(m_{\Gstar,\delta}(x)-m_{\hat G_{d,m},\delta}(x)\Big)^2\\
&\quad\le 8\Big\{(\sqrt{A_{x,\delta}(\Gstar)}-\sqrt{A_{x,\delta}(\hat G_{d,m})})^2 +(\sqrt{B_x(\Gstar)}-\sqrt{B_x(\hat G_{d,m})})^2\Big\}.
\end{align*}
Now $B_x(\Gstar)=x!f_{\Gstar}(x)$ and $B_x(\hat G_{d,m})=x!f_{\hat G_{d,m}}(x)$, hence
\begin{displaymath}
\frac{B_x(\Gstar)B_x(\hat G_{d,m})}{B_x(\Gstar)+B_x(\hat G_{d,m})}=x!\,\frac{f_{\Gstar}(x)f_{\hat G_{d,m}}(x)}{f_{\Gstar}(x)+f_{\hat G_{d,m}}(x)}=x!\,w_x(\Gstar,\hat G_{d,m}).
\end{displaymath}
Thus
\begin{displaymath}
w_x(\Gstar,\hat G_{d,m})\,u(x)\le\frac{8}{x!}\Big\{(\sqrt{A_{x,\delta}(\Gstar)}-\sqrt{A_{x,\delta}(\hat G_{d,m})})^2+(\sqrt{B_x(\Gstar)}-\sqrt{B_x(\hat G_{d,m})})^2\Big\}.
\end{displaymath}
Since $\sqrt{B_x(\Gstar)}=\sqrt{x!}\,\sqrt{f_{\Gstar}(x)}$ and $\sqrt{B_x(\hat G_{d,m})}=\sqrt{x!}\,\sqrt{f_{\hat G_{d,m}}(x)}$, we get
\begin{displaymath}
\frac{1}{x!}(\sqrt{B_x(\Gstar)}-\sqrt{B_x(\hat G_{d,m})})^2=(\sqrt{f_{\Gstar}(x)}-\sqrt{f_{\hat G_{d,m}}(x)})^2.
\end{displaymath}
Hence
\begin{displaymath}
w_x(\Gstar,\hat G_{d,m})\,u(x)\le8(\sqrt{f_{\Gstar}(x)}-\sqrt{f_{\hat G_{d,m}}(x)})^2+\frac{8}{x!}(\sqrt{A_{x,\delta}(\Gstar)}-\sqrt{A_{x,\delta}(\hat G_{d,m})})^2.
\end{displaymath}
Summing over \(x=1,\dots,\hat K\) gives
\begin{equation}\label{eq:w-sum-split-hhat}
\sum_{x=1}^{\hat K} w_x(\Gstar,\hat G_{d,m})u(x)\le8\sum_{x=1}^{\hat K}(\sqrt{f_{\Gstar}(x)}-\sqrt{f_{\hat G_{d,m}}(x)})^2+8\sum_{x=1}^{\hat K}\frac{1}{x!}(\sqrt{A_{x,\delta}(\Gstar)}-\sqrt{A_{x,\delta}(\hat G_{d,m})})^2.
\end{equation}
The first term is bounded by \(8\mathcal H(\Gstar,\hat G_{d,m})\). By Lemma~\ref{lem:sqrtA-control-route1-hhat},
\begin{displaymath}
\Big(\sqrt{A_{x,\delta}(\Gstar)}-\sqrt{A_{x,\delta}(\hat G_{d,m})}\Big)^2\le 8C_A\,\mathcal H(\Gstar,\hat G_{d,m})+2\sqrt{2}\,\varepsilon_{x,\delta}.
\end{displaymath}
Summing over \(x=1,\dots,\hat K\) and using \(\sum_{x=1}^{\hat K}1/x!\le e\) yields
\begin{align}\label{eq:A-sum-control-route1-hhat}
\sum_{x=1}^{\hat K}\frac{1}{x!}\Big(\sqrt{A_{x,\delta}(\Gstar)}-\sqrt{A_{x,\delta}(\hat G_{d,m})}\Big)^2
&\le 8eC_A\,\mathcal H(\Gstar,\hat G_{d,m})+2\sqrt{2}\sum_{x=1}^{\hat K}\frac{\varepsilon_{x,\delta}}{x!}.
\end{align}
By Assumption~\ref{asmp:aa-self-normal-regular}, \(\varepsilon_{x,\delta}\le C_\varepsilon\delta\). Therefore
\begin{displaymath}
\sum_{x=1}^{\hat K}\frac{\varepsilon_{x,\delta}}{x!}\le C_\varepsilon\delta\sum_{x=1}^{\hat K}\frac1{x!}\le eC_\varepsilon\delta.
\end{displaymath}
Consequently,
\begin{equation}\label{eq:A-sum-control-route1-final-hhat}
\sum_{x=1}^{\hat K}\frac{1}{x!}\Big(\sqrt{A_{x,\delta}(\Gstar)}-\sqrt{A_{x,\delta}(\hat G_{d,m})}\Big)^2\le C_{A,1}\,\mathcal H(\Gstar,\hat G_{d,m})+ C_{A,2}\,\delta,
\end{equation}
for constants \(C_{A,1}\) and \(C_{A,2}\) depending only on the constants in
\prettyref{asmp:aa-self-normal-regular}. Combining \eqref{eq:w-sum-split-hhat} with \eqref{eq:A-sum-control-route1-final-hhat} gives
\begin{equation}\label{eq:w-sum-final-hhat}
\sum_{x=1}^{\hat K} w_x(\Gstar,\hat G_{d,m})u(x)\le C_{w,1}\,\mathcal H(\Gstar,\hat G_{d,m})+C_{w,2}\,\delta,
\end{equation}
for constants \(C_{w,1},C_{w,2}\) depending only on \(\hat h\) and the constants in \prettyref{asmp:aa-self-normal-regular}. Plugging \eqref{eq:w-sum-final-hhat} into \eqref{eq:smoothed-convert-hhat} yields
\begin{equation}\label{eq:smoothed-est-hhat}
\sum_{x=1}^{\hat K} f_{\Gstar}(x)\bigl(m_{\Gstar,\delta}(x)-m_{\hat G_{d,m},\delta}(x)\bigr)^2\le C_{\mathrm{sm},1}\,\mathcal H(\Gstar,\hat G_{d,m})+C_{\mathrm{sm},2}\,\delta,
\end{equation}
for constants \(C_{\mathrm{sm},1},C_{\mathrm{sm},2}\) depending only on \(\hat h\) and the constants in \prettyref{asmp:aa-self-normal-regular}.

\paragraph{Putting the three pieces together.}

Combine \eqref{eq:main-sum-decomp-hhat}, \eqref{eq:bias-G-hhat}, \eqref{eq:bias-Gtilde-hhat}, and \eqref{eq:smoothed-est-hhat}: there exist constants \(C_0,C_1\), depending only on \((\hat h,L)\) and the constants in Assumption~\ref{asmp:aa-self-normal-regular}, such that
\begin{equation}\label{eq:pointwise-main-hhat}
\sum_{x=1}^{\hat K}(m_{\Gstar}(x)-m_{\hat G_{d,m}}(x))^2 f_{\Gstar}(x)\le C_0\,\delta + C_1\,\mathcal H(\Gstar,\hat G_{d,m}).
\end{equation}
Using \eqref{eq:basic-reduction-hhat} and taking expectations over the training sample,
\begin{equation}\label{eq:pre-final-hhat}
\frac{1}{n}\EE\!\left[\left(\hat S_{n,>}(\Gstar)-\hat S_{n,>}(\hat G_{d,m})\right)^2\right]\le n\,C_0\,\delta+n\,C_1\,\EE\!\big[\mathcal H(\Gstar,\hat G_{d,m})\big].
\end{equation}

\paragraph{Final regret bound.}

By Lemma~\ref{lem:EB_density-H2}\textnormal{(a)}, with
\begin{displaymath}
R_{\hat h}(m):=\min\!\left\{\max\{1,\hat h\}\frac{\log m}{\log\log m},\hat h+\log m\right\},
\end{displaymath}
we have
\begin{displaymath}
\EE[\mathcal H(\Gstar,\hat G_{d,m})]=\EE\!\left[H^2(f_{\Gstar},f_{\hat G_{d,m}})\right]\le C_d(\hat h,c_1,c_2)\,\frac{R_{\hat h}(m)}{m}.
\end{displaymath}
Plug this into \eqref{eq:pre-final-hhat}:
\begin{displaymath}
\frac{1}{n}\EE\!\left[\left(\hat S_{n,>}(\Gstar)-\hat S_{n,>}(\hat G_{d,m})\right)^2\right]\le n\,C_0\,\delta+n\,C_1\,C_d(\hat h,c_1,c_2)\,\frac{R_{\hat h}(m)}{m}.
\end{displaymath}
That is,
\begin{align}\label{eq:final-regret-bound-hhat}
&\frac{1}{n}\EE\!\left[\left(\hat S_{n,>}(\Gstar)-\hat S_{n,>}(\hat G_{d,m})\right)^2\right]\\
&\notag\quad\le n\,C_0\,\delta+n\,C_1\,C_d(\hat h,c_1,c_2)\,\frac{1}{m}\min\!\left\{\max\{1,\hat h\}\frac{\log m}{\log\log m},\hat h+\log m\right\}.
\end{align}

\subsubsection{Subexponential setting}\label{app:proof-subexpo-new}


\begin{lemma}\label{lem:tail-X-exp-nocoarsen}
	Let $\Gstar\in\subexpo(s)$ and let $X\mid \theta\sim\Poi(\theta)$, $\theta\sim \Gstar$. Then there exist constants $c_s>0$ and $C_s<\infty$, depending only on $s$, such that
	\begin{displaymath}
		\PP_{\Gstar}(X>K)\le C_s e^{-c_s K},\qquad K\in\integers_+.
	\end{displaymath}
	One may take
	\begin{displaymath}
		c_s:=\log\Big(1+\frac{1}{2s}\Big),\qquad C_s:=3.
	\end{displaymath}
\end{lemma}

\begin{proof}
	For every $\lambda\in(0,1/s)$,
	\begin{displaymath}
		\EE[e^{\lambda \theta}]=1+\lambda\int_0^\infty e^{\lambda t}\,\Gstar((t,\infty))\,dt\le1+2\lambda\int_0^\infty e^{-(1/s-\lambda)t}\,dt=1+\frac{2\lambda s}{1-\lambda s}.
	\end{displaymath}
	Choose $\lambda=(2s)^{-1}$ to get
	\begin{displaymath}
		\EE[e^{\theta/(2s)}]\le 3.
	\end{displaymath}
	Now let $z_s:=1+\frac{1}{2s}$ and $c_s:=\log z_s$. Since $X\mid \theta\sim\Poi(\theta)$,
	\begin{displaymath}
		\EE[z_s^X\mid \theta]=\exp\big(\theta(z_s-1)\big)=e^{\theta/(2s)}.
	\end{displaymath}
	Hence
	\begin{displaymath}
		\EE[z_s^X]\le 3.
	\end{displaymath}
	By Markov's inequality,
	\begin{displaymath}
		\PP_{\Gstar}(X\ge K+1)\le z_s^{-(K+1)}\EE[z_s^X]\le3 z_s^{-(K+1)}\le3e^{-c_s K}.
	\end{displaymath}
	This proves the claim.
\end{proof}

\paragraph{Finite-window stability and completion of the proof.}

We first note that all ratios appearing below are well defined. Assumption~\ref{asmp:aa-local-mass}, applied at \(x=0\), implies that $G(\{0\})=0$ for every $G\in\mathcal C$. Indeed, for every \(\delta\in(0,1]\),
\begin{displaymath}
G(\{0\})\le G([-\delta,\delta]\cap\operatorname{supp}(G))\le L\delta,
\end{displaymath}
and the claim follows by letting \(\delta\downarrow0\). Since every \(G\in\mathcal C\) is a probability measure on \(\mathbb R_+\), it follows that
\begin{displaymath}
B_x(G)=\int_{\mathbb R_+}e^{-\theta}\theta^x\,dG(\theta)>0,\qquad x\ge1,\quad G\in\mathcal C.
\end{displaymath}

We now work on the finite window \(\{0,\dots,K\}\). If \(K=0\), then the positive part of the window is empty. Since $m_{\Gstar}(0)=m_{\hat G_{d,m}}(0)=0$ and $|m_{\Gstar}(x)-m_{\hat G_{d,m}}(x)|\le1$, the same reduction used below gives
\begin{displaymath}
\frac1n\EE\!\left[\left(\hat S_{n,>}(\Gstar)-\hat S_{n,>}(\hat G_{d,m})\right)^2\right]\le n\PP_{\Gstar}(X>0)\le nC_s,
\end{displaymath}
which is bounded by the right-hand side of \eqref{eq:main-sube-result}. We may therefore assume \(K\ge1\).

Since \(\Gstar\in\mathcal C\) and \(\hat G_{d,m}\in\mathcal C\) almost surely, Assumptions~\ref{asmp:aa-local-mass} and \ref{asmp:aa-self-normal-regular} hold for both
\(\Gstar\) and \(\hat G_{d,m}\) on the window \(\{0,\dots,K\}\).

We next record the truncated reduction. Conditional on \(\hat G_{d,m}\), set $\Delta(x):=m_{\Gstar}(x)-m_{\hat G_{d,m}}(x)$. Then \(|\Delta(x)|\le1\) and \(\Delta(0)=0\). As in the
compact-support case,
\begin{displaymath}
\frac1n\EE\!\left[\left(\hat S_{n,>}(\Gstar)-\hat S_{n,>}(\hat G_{d,m})\right)^2\Bigm|\hat G_{d,m}\right]\le n\sum_{x\ge0}\Delta(x)^2f_{\Gstar}(x).
\end{displaymath}
Splitting the last sum at \(K\) gives
\begin{equation}\label{eq:subexpo-new-basic-reduction-exp-short}
\frac1n\EE\!\left[\left(\hat S_{n,>}(\Gstar)-\hat S_{n,>}(\hat G_{d,m})\right)^2\Bigm|\hat G_{d,m}\right]\le n\sum_{x=1}^{K}\bigl(m_{\Gstar}(x)-m_{\hat G_{d,m}}(x)\bigr)^2f_{\Gstar}(x)+n\PP_{\Gstar}(X>K).
\end{equation}

\paragraph{Smoothing decomposition.}

By $(a+b+c)^2\le 3(a^2+b^2+c^2)$,
\begin{align}\label{eq:subexpo-new-main-sum-decomp-hhat}
&\sum_{x=1}^{K}\bigl(m_{\Gstar}(x)-m_{\hat G_{d,m}}(x)\bigr)^2f_{\Gstar}(x)\notag\\
&\quad\le 3\sum_{x=1}^{K}\bigl(m_{\Gstar}(x)-m_{\Gstar,\delta}(x)\bigr)^2f_{\Gstar}(x)\notag\\
&\qquad+3\sum_{x=1}^{K}\bigl(m_{\Gstar,\delta}(x)-m_{\hat G_{d,m},\delta}(x)\bigr)^2f_{\Gstar}(x)
+3\sum_{x=1}^{K}\bigl(m_{\hat G_{d,m},\delta}(x)-m_{\hat G_{d,m}}(x)\bigr)^2f_{\Gstar}(x).
\end{align}

We now control the three terms on the right-hand side separately.

\paragraph{Bias term for $\Gstar$.}
\label{app:bias}
We will first establish a bound on $\sum_{x=1}^K\bigl(m_G(x)-m_{G,\delta}(x)\bigr)^2f_G(x)$ for any generic $G$ satisfying \prettyref{asmp:aa-local-mass}. Fix $x\in\{1,\dots,K\}$. Since 
\begin{align*}
	m_{G}(x)&=\frac{A_x(G)}{B_x(G)}=
	{\int_{\reals_+}
	{e^{-\theta}\theta^x\over x!}\indc{\theta<x}dG(\theta)
	\over
	f_{G}(x)}\\
	m_{G,\delta}(x)&=\frac{A_{x,\delta}(G)}{B_x(G)}
	=
	{\int_{\reals_+}
		{e^{-\theta}\theta^x\over x!}\phi_\delta(x-\theta)dG(\theta)
		\over
		f_{G}(x)},
\end{align*} 
we have
\begin{displaymath}
	m_{G,\delta}(x)-m_{G}(x)=\frac{A_{x,\delta}(G)-A_x(G)}{B_x(G)}
	={\int_{\reals_+}
		{e^{-\theta}\theta^x\over x!}\bigl(\phi_\delta(x-\theta)-I(\theta<x)\bigr)dG(\theta)
		\over
		f_{G}(x)}.
\end{displaymath}
Hence, using \eqref{eq:phi-pointwise-error-main},
\begin{displaymath}
	|m_{G,\delta}(x)-m_{G}(x)|
	\le\frac 1{f_{G}(x)}\int_{\reals_+} {e^{-\theta}\theta^x\over x!}\,I(|\theta-x|<\delta)\,dG(\theta).
\end{displaymath}
Thus, using the inequality $\int_{\reals_+} {e^{-\theta}\theta^x\over x!} I(|\theta-x|<\delta)\,dG(\theta)\leq \int_{\reals_+}{e^{-\theta}\theta^x\over x!}\,dG(\theta)=f_G(x)$ we get
\begin{align}\label{eq:m5}
	\sum_{x=1}^{K} f_{G}(x)\,
	(m_{G,\delta}(x)-m_{G}(x))^2
	&\le\sum_{x=1}^{K} \frac 1{f_{G}(x)}
	\pth{\int_{\reals_+} {e^{-\theta}\theta^x\over x!}\,I(|\theta-x|<\delta)\,dG(\theta)}^2
	\notag\\
	&{\le}
	\sum_{x=1}^K \int_{\reals_+} {e^{-\theta}\theta^x\over x!} I(|\theta-x|<\delta)\,dG(\theta)
	\notag\\
	&{\le}
	\sum_{x=1}^K \int_{\reals_+} I(|\theta-x|<\delta)\,dG(\theta)
	=\sum_{x=1}^K G([x-\delta,x+\delta])
	{\le}  KL\delta,
\end{align}
where the last display followed from \prettyref{asmp:aa-local-mass}.
Hence, plugging in $G=\Gstar$ gives us
\begin{equation}\label{eq:subexpo-new-bias-G-hhat}
	\sum_{x=1}^{K} f_{\Gstar}(x)\,(m_{\Gstar,\delta}(x)-m_{\Gstar}(x))^2\le KL\delta.
\end{equation}

\paragraph{Bias term for $\hat G_{d,m}$.}

First we use \eqref{eq:m5} with $G=\hat G_{d,m}$ to get
\begin{align}\label{eq:m6}
	\sum_{x=1}^{K} f_{\hat G_{d,m}}(x)\,(m_{\hat G_{d,m},\delta}(x)-m_{\hat G_{d,m}}(x))^2\le KL\delta.
\end{align}
Now define $u(x):=(m_{\hat G_{d,m},\delta}(x)-m_{\hat G_{d,m}}(x))^2\in[0,1]$. Applying Lemma~\ref{lem:quad-convert} pointwise with $a=f_{\Gstar}(x)$ and $b=f_{\hat G_{d,m}}(x)$, we get
\begin{displaymath}
	f_{\Gstar}(x)u(x)\le 4\,\frac{f_{\Gstar}(x)f_{\hat G_{d,m}}(x)}{f_{\Gstar}(x)+f_{\hat G_{d,m}}(x)}\,u(x)+2\,\frac{(f_{\Gstar}(x)-f_{\hat G_{d,m}}(x))^2}{f_{\Gstar}(x)+f_{\hat G_{d,m}}(x)}.
\end{displaymath}
Summing over \(x\le K\) and using
\begin{displaymath}
	\frac{f_{\Gstar}(x)f_{\hat G_{d,m}}(x)}{f_{\Gstar}(x)+f_{\hat G_{d,m}}(x)}\le f_{\hat G_{d,m}}(x),
\end{displaymath}
we obtain
\begin{displaymath}
	\sum_{x=1}^{K} f_{\Gstar}(x)u(x)\le 4\sum_{x=1}^{K} f_{\hat G_{d,m}}(x)u(x)+2\sum_{x=1}^{K}\frac{(f_{\Gstar}(x)-f_{\hat G_{d,m}}(x))^2}{f_{\Gstar}(x)+f_{\hat G_{d,m}}(x)}.
\end{displaymath}
The first term is at most \(4KL\delta\) by \eqref{eq:m6}. By Lemma~\ref{lem:chi2-hell},
\begin{displaymath}
	\sum_{x=1}^{K}\frac{(f_{\Gstar}(x)-f_{\hat G_{d,m}}(x))^2}{f_{\Gstar}(x)+f_{\hat G_{d,m}}(x)}\le 4\sum_{x=1}^{K}(\sqrt{f_{\Gstar}(x)}-\sqrt{f_{\hat G_{d,m}}(x)})^2\le4\mathcal H({\Gstar},\hat G_{d,m}).
\end{displaymath}
Therefore
\begin{equation}\label{eq:subexpo-new-bias-Gtilde-hhat}
	\sum_{x=1}^{K} f_{\Gstar}(x)\,(m_{\hat G_{d,m},\delta}(x)-m_{\hat G_{d,m}}(x))^2\le 4KL\delta + 8\mathcal H(\Gstar,\hat G_{d,m}).
\end{equation}

\paragraph{Smoothed estimation error.}

Let $u(x):=(m_{\Gstar,\delta}(x)-m_{\hat G_{d,m},\delta}(x))^2\in[0,1]$. Applying Lemma~\ref{lem:quad-convert} with $a=f_{\Gstar}(x)$ and $b=f_{\hat G_{d,m}}(x)$ yields
\begin{displaymath}
	\sum_{x=1}^{K} f_{\Gstar}(x)u(x)\le 4\sum_{x=1}^{K} w_x(\Gstar,\hat G_{d,m})u(x)+2\sum_{x=1}^{K}\frac{(f_{\Gstar}(x)-f_{\hat G_{d,m}}(x))^2}{f_{\Gstar}(x)+f_{\hat G_{d,m}}(x)},
\end{displaymath}
where
\begin{displaymath}
	w_x(\Gstar,\hat G_{d,m}):=\frac{f_{\Gstar}(x)f_{\hat G_{d,m}}(x)}{f_{\Gstar}(x)+f_{\hat G_{d,m}}(x)}.
\end{displaymath}
Using Lemma~\ref{lem:chi2-hell},
\begin{equation}\label{eq:subexpo-new-smoothed-convert-hhat}
	\sum_{x=1}^{K} f_{\Gstar}(x)u(x)\le4\sum_{x=1}^{K} w_x(\Gstar,\hat G_{d,m})u(x)+8\mathcal H(\Gstar,\hat G_{d,m}).
\end{equation}
Now fix \(x\le K\). Apply Lemma~\ref{lem:ratio-stab} with $A=A_{x,\delta}(\Gstar)$, $B=B_x(\Gstar)$, $\tilde A=A_{x,\delta}(\hat G_{d,m})$ and $\tilde B=B_x(\hat G_{d,m})$. Since \(0\le \phi_\delta\le 1\), we have \(A_{x,\delta}(\Gstar)\le B_x(\Gstar)\) and \(A_{x,\delta}(\hat G_{d,m})\le B_x(\hat G_{d,m})\), so the lemma applies. Therefore
\begin{align*}
&\frac{B_x(\Gstar)B_x(\hat G_{d,m})}{B_x(\Gstar)+B_x(\hat G_{d,m})}\Big(m_{\Gstar,\delta}(x)-m_{\hat G_{d,m},\delta}(x)\Big)^2\\
&\quad\le 8\Big\{(\sqrt{A_{x,\delta}(\Gstar)}-\sqrt{A_{x,\delta}(\hat G_{d,m})})^2+(\sqrt{B_x(\Gstar)}-\sqrt{B_x(\hat G_{d,m})})^2\Big\}.
\end{align*}
Now $B_x(\Gstar)=x!f_{\Gstar}(x)$ and $B_x(\hat G_{d,m})=x!f_{\hat G_{d,m}}(x)$, hence
\begin{displaymath}
	\frac{B_x(\Gstar)B_x(\hat G_{d,m})}{B_x(\Gstar)+B_x(\hat G_{d,m})}=x!\,\frac{f_{\Gstar}(x)f_{\hat G_{d,m}}(x)}{f_{\Gstar}(x)+f_{\hat G_{d,m}}(x)}=x!\,w_x(\Gstar,\hat G_{d,m}).
\end{displaymath}
Thus
\begin{displaymath}
	w_x(\Gstar,\hat G_{d,m})\,u(x)\le\frac{8}{x!}\Big\{(\sqrt{A_{x,\delta}(\Gstar)}-\sqrt{A_{x,\delta}(\hat G_{d,m})})^2+(\sqrt{B_x(\Gstar)}-\sqrt{B_x(\hat G_{d,m})})^2\Big\}.
\end{displaymath}
Since $\sqrt{B_x(\Gstar)}=\sqrt{x!}\,\sqrt{f_{\Gstar}(x)}$ and $\sqrt{B_x(\hat G_{d,m})}=\sqrt{x!}\,\sqrt{f_{\hat G_{d,m}}(x)}$, we get
\begin{displaymath}
	\frac{1}{x!}(\sqrt{B_x(\Gstar)}-\sqrt{B_x(\hat G_{d,m})})^2=(\sqrt{f_{\Gstar}(x)}-\sqrt{f_{\hat G_{d,m}}(x)})^2.
\end{displaymath}
Hence
\begin{displaymath}
	w_x(\Gstar,\hat G_{d,m})\,u(x)\le8(\sqrt{f_{\Gstar}(x)}-\sqrt{f_{\hat G_{d,m}}(x)})^2+\frac{8}{x!}(\sqrt{A_{x,\delta}(\Gstar)}-\sqrt{A_{x,\delta}(\hat G_{d,m})})^2.
\end{displaymath}
Summing over \(x=1,\dots,K\) gives
\begin{equation}\label{eq:subexpo-new-w-sum-split-hhat}
	\sum_{x=1}^{K} w_x(\Gstar,\hat G_{d,m})u(x)\le8\sum_{x=1}^{K}(\sqrt{f_{\Gstar}(x)}-\sqrt{f_{\hat G_{d,m}}(x)})^2+8\sum_{x=1}^{K}\frac{1}{x!}(\sqrt{A_{x,\delta}(\Gstar)}-\sqrt{A_{x,\delta}(\hat G_{d,m})})^2.
\end{equation}
The first term is bounded by \(8\mathcal H(\Gstar,\hat G_{d,m})\). By Lemma~\ref{lem:sqrtA-control-route1-hhat},
\begin{displaymath}
	\Big(\sqrt{A_{x,\delta}(\Gstar)}-\sqrt{A_{x,\delta}(\hat G_{d,m})}\Big)^2\le 8C_A\,\mathcal H(\Gstar,\hat G_{d,m})+2\sqrt{2}\,\varepsilon_{x,\delta}.
\end{displaymath}
Summing over \(x=1,\dots,K\) and using \(\sum_{x=1}^{K}1/x!\le e\) yields
\begin{align}\label{eq:subexpo-new-A-sum-control-route1-hhat}
	\sum_{x=1}^{K}\frac{1}{x!}\Big(\sqrt{A_{x,\delta}(\Gstar)}-\sqrt{A_{x,\delta}(\hat G_{d,m})}\Big)^2
	&\le
	8eC_A\,\mathcal H(\Gstar,\hat G_{d,m})+2\sqrt{2}\sum_{x=1}^{K}\frac{\varepsilon_{x,\delta}}{x!}.
\end{align}
By Assumption~\ref{asmp:aa-self-normal-regular}, \(\varepsilon_{x,\delta}\le C_\varepsilon K \delta\). Therefore
\begin{align}
	\label{eq:m7}
	\sum_{x=1}^{K}\frac{\varepsilon_{x,\delta}}{x!}\le C_\varepsilon K \delta\sum_{x=1}^{K}\frac1{x!}\le eC_\varepsilon K \delta.
\end{align}
Consequently,
\begin{equation}\label{eq:subexpo-new-A-sum-control-route1-final-hhat}
	\sum_{x=1}^{K}\frac{1}{x!}\Big(\sqrt{A_{x,\delta}(\Gstar)}-\sqrt{A_{x,\delta}(\hat G_{d,m})}\Big)^2\le C_{A,1}\,\mathcal H(\Gstar,\hat G_{d,m})+ C_{A,2} \,K \,\delta,
\end{equation}
for constants \(C_{A,1}\) and \(C_{A,2}\) depending only on the constants in
\prettyref{asmp:aa-self-normal-regular}. Combining \eqref{eq:subexpo-new-w-sum-split-hhat} with \eqref{eq:subexpo-new-A-sum-control-route1-final-hhat} gives
\begin{equation}\label{eq:subexpo-new-w-sum-final-hhat}
	\sum_{x=1}^{K} w_x(\Gstar,\hat G_{d,m})u(x)\le C_{w,1}\,\mathcal H(\Gstar,\hat G_{d,m})+C_{w,2}\,K\,\delta,
\end{equation}
for constants \(C_{w,1},C_{w,2}\) depending only on the constants in \prettyref{asmp:aa-self-normal-regular}. Plugging \eqref{eq:subexpo-new-w-sum-final-hhat} into \eqref{eq:subexpo-new-smoothed-convert-hhat} yields
\begin{equation}\label{eq:subexpo-new-smoothed-est-hhat}
	\sum_{x=1}^{K} f_{\Gstar}(x)\bigl(m_{\Gstar,\delta}(x)-m_{\hat G_{d,m},\delta}(x)\bigr)^2\le C_{\mathrm{sm},1}\,\mathcal H(\Gstar,\hat G_{d,m})+C_{\mathrm{sm},2}\,K\,\delta,
\end{equation}
for constants \(C_{\mathrm{sm},1},C_{\mathrm{sm},2}\) depending only on the constants in \prettyref{asmp:aa-self-normal-regular}.

\paragraph{Putting the three pieces together.}

Combine \eqref{eq:subexpo-new-main-sum-decomp-hhat}, \eqref{eq:subexpo-new-bias-G-hhat}, \eqref{eq:subexpo-new-bias-Gtilde-hhat}, and \eqref{eq:subexpo-new-smoothed-est-hhat}: there exist constants \(C_0,C_1\), depending only on \(L\) and the constants in Assumption~\ref{asmp:aa-self-normal-regular}, such that
\begin{equation}\label{eq:subexpo-new-pointwise-main-hhat}
\sum_{x=1}^{K}\bigl(m_{\Gstar}(x)-m_{\hat G_{d,m}}(x)\bigr)^2f_{\Gstar}(x)\le C_0K\delta+C_1\mathcal H(\Gstar,\hat G_{d,m}).
\end{equation}
Using \eqref{eq:subexpo-new-basic-reduction-exp-short} and
\eqref{eq:subexpo-new-pointwise-main-hhat}, and then taking expectations
over the training sample, we obtain
\begin{equation}\label{eq:subexpo-new-pre-final-hhat}
\frac1n\EE\!\left[\left(\hat S_{n,>}(\Gstar)-\hat S_{n,>}(\hat G_{d,m})\right)^2\right]\le nC_0K\delta+nC_1\EE\!\left[H^2(f_{\Gstar},f_{\hat G_{d,m}})\right]+n\PP_{\Gstar}(X>K).
\end{equation}

By Lemma~\ref{lem:EB_density-H2}\textnormal{(b)},
\begin{displaymath}
\EE\!\left[H^2(f_{\Gstar},f_{\hat G_{d,m}})\right]\le C_d(s,c_1,c_2)\frac{\log m}{m}.
\end{displaymath}
Moreover, by Lemma~\ref{lem:tail-X-exp-nocoarsen},
\begin{displaymath}
\PP_{\Gstar}(X>K)\le C_s e^{-c_sK}.
\end{displaymath}
Therefore,
\begin{displaymath}
\frac1n\EE\!\left[\left(\hat S_{n,>}(\Gstar)-\hat S_{n,>}(\hat G_{d,m})\right)^2\right]\le nC_0K\delta+nC_1C_d(s,c_1,c_2)\frac{\log m}{m}+nC_s e^{-c_sK}.
\end{displaymath}
This proves Theorem~\ref{thm:main-above-average}\textnormal{(b)}.


\subsection{Proof of Theorem~\ref{thm:main-above-average-coarsened}}\label{proof:main-above-average-coarsened}

For notational convenience, throughout this appendix we write \(\mathcal H^{(T)}(G_1,G_2)\) for the squared Hellinger distance between the \(T\)-coarsened mixture probability mass functions, namely
\begin{displaymath}
\mathcal H^{(T)}(G_1,G_2):=H^2\!\left(\tilde f_{G_1}^{(T)},\tilde f_{G_2}^{(T)}\right)=\sum_{t=0}^{T+1}\left(\sqrt{\tilde f_{G_1}^{(T)}(t)}-\sqrt{\tilde f_{G_2}^{(T)}(t)}\right)^2.
\end{displaymath}

\subsubsection{Compact-support setting}


\paragraph{The coarsened smoothed numerator bound.}

\begin{lemma}\label{lem:sqrtA-control-coarse-compact}
Fix \(x\in\{1,\dots,\hat K\}\) and \(\delta\in(0,1]\). Assume \prettyref{asmp:aa-self-normal-coarse} holds on \(\mathcal C\). Then, for all \(G_1,G_2\in\mathcal C\),
\begin{displaymath}
\left(\sqrt{A_{x,\delta}(G_1)}-\sqrt{A_{x,\delta}(G_2)}\right)^2\le 8C_A\,\mathcal H^{(T)}(G_1,G_2)+2\sqrt{2}\,\varepsilon_{x,\delta}^{(T)}.
\end{displaymath}
\end{lemma}

\begin{proof}
Write $A_j=A_{x,\delta}(G_j)$, $L_j=\sum_{t=0}^{T+1}\alpha_{x,\delta}^{(T)}(t)\tilde f_{G_j}^{(T)}(t)$ for $j=1,2$. Set \(S=A_1+A_2\). If \(S=0\), then \(A_1=A_2=0\), and the claim is trivial. Therefore, assume \(S>0\). By the representation bound in Assumption~\ref{asmp:aa-self-normal-coarse},
\begin{displaymath}
|A_j-L_j|\le \varepsilon_{x,\delta}^{(T)},
\qquad j=1,2.
\end{displaymath}
Hence
\begin{displaymath}
|A_1-A_2|\le|L_1-L_2|+2\varepsilon_{x,\delta}^{(T)}.
\end{displaymath}
Using
\begin{displaymath}
(\sqrt u-\sqrt v)^2\le \frac{(u-v)^2}{u+v},\qquad u,v\ge0,
\end{displaymath}
we obtain
\begin{displaymath}
(\sqrt{A_1}-\sqrt{A_2})^2\le\frac{2|L_1-L_2|^2}{S}+\frac{8(\varepsilon_{x,\delta}^{(T)})^2}{S}.
\end{displaymath}
By Cauchy--Schwarz inequality
\begin{displaymath}
|L_1-L_2|^2\le\left\{\sum_{t=0}^{T+1}\left(\alpha_{x,\delta}^{(T)}(t)\right)^2\left(\tilde f_{G_1}^{(T)}(t)+\tilde f_{G_2}^{(T)}(t)\right)\right\}\left\{\sum_{t=0}^{T+1}\frac{\left(\tilde f_{G_1}^{(T)}(t)-\tilde f_{G_2}^{(T)}(t)\right)^2}{\tilde f_{G_1}^{(T)}(t)+\tilde f_{G_2}^{(T)}(t)}\right\}.
\end{displaymath}
By the self-normalization bound in Assumption~\ref{asmp:aa-self-normal-coarse} and Lemma~\ref{lem:chi2-hell}, $|L_1-L_2|^2\le4C_A S\,\mathcal H^{(T)}(G_1,G_2)$. Therefore
\begin{displaymath}
(\sqrt{A_1}-\sqrt{A_2})^2\le8C_A\,\mathcal H^{(T)}(G_1,G_2)+\frac{8(\varepsilon_{x,\delta}^{(T)})^2}{S}.
\end{displaymath}
On the other hand,
\begin{displaymath}
(\sqrt{A_1}-\sqrt{A_2})^2\le A_1+A_2=S.
\end{displaymath}
Combining the two bounds gives
\begin{displaymath}
(\sqrt{A_1}-\sqrt{A_2})^2\le8C_A\,\mathcal H^{(T)}(G_1,G_2)+\min\left\{S,\frac{8(\varepsilon_{x,\delta}^{(T)})^2}{S}\right\}.
\end{displaymath}
Finally, $\min\left\{S,\frac{8(\varepsilon_{x,\delta}^{(T)})^2}{S}\right\}\le2\sqrt{2}\,\varepsilon_{x,\delta}^{(T)}$. This proves the lemma.
\end{proof}

\paragraph{Regret with respect to the oracle.}

Define $\hat S_{n,>}(\tilde G^{(T)}_{d,m}):=\sum_{i=1}^n m_{\tilde G^{(T)}_{d,m}}(X_i)$. We start from the reduction
\begin{equation}\label{eq:basic-reduction-coarse-compact}
\frac1n\EE\!\left[\left(\hat S_{n,>}(\Gstar)-\hat S_{n,>}(\tilde G^{(T)}_{d,m})\right)^2\mid \tilde G^{(T)}_{d,m}\right]\le n\sum_{x=1}^{\hat K}\left(m_{\Gstar}(x)-m_{\tilde G^{(T)}_{d,m}}(x)\right)^2f_{\Gstar}(x).
\end{equation}
Indeed, since both \(\Gstar\) and \(\tilde G^{(T)}_{d,m}\) are supported on \([0,\hat h]\), for every integer \(x\ge \hat K+1\) we have \(x>\hat h\), and therefore
\begin{displaymath}
m_{\Gstar}(x)=m_{\tilde G^{(T)}_{d,m}}(x)=1.
\end{displaymath}
Moreover, $m_{\Gstar}(0)=m_{\tilde G^{(T)}_{d,m}}(0)=0$. Thus \(m_{\Gstar}(x)-m_{\tilde G^{(T)}_{d,m}}(x)\) vanishes outside \(\{1,\dots,\hat K\}\).
Conditional on \(\tilde G^{(T)}_{d,m}\), set
\begin{displaymath}
\Delta(x):=
\left(m_{\Gstar}(x)-m_{\tilde G^{(T)}_{d,m}}(x)\right)\mathbf 1\{1\le x\le \hat K\}.
\end{displaymath}
Then \(\Delta(X_1),\dots,\Delta(X_n)\) are i.i.d. under \(X_i\sim f_{\Gstar}\), and
\(|\Delta(x)|\le1\). Hence, as in the proof of
\eqref{eq:basic-reduction-hhat},
\begin{displaymath}
\EE\!\left[
\left(\sum_{i=1}^n\Delta(X_i)\right)^2
\Bigm|\tilde G^{(T)}_{d,m}
\right]
\le
n^2\EE[\Delta(X)^2\mid\tilde G^{(T)}_{d,m}],
\end{displaymath}
which gives \eqref{eq:basic-reduction-coarse-compact} after dividing by \(n\).

\paragraph{Smoothing decomposition.}

By \((a+b+c)^2\le3(a^2+b^2+c^2)\),
\begin{align}\label{eq:main-sum-decomp-coarse-compact}
&\sum_{x=1}^{\hat K}\left(m_{\Gstar}(x)-m_{\tilde G^{(T)}_{d,m}}(x)\right)^2 f_{\Gstar}(x)
\nonumber\\
&\qquad\le
3\sum_{x=1}^{\hat K}\left(m_{\Gstar}(x)-m_{\Gstar,\delta}(x)\right)^2f_{\Gstar}(x)+3\sum_{x=1}^{\hat K}\left(m_{\Gstar,\delta}(x)-m_{\tilde G^{(T)}_{d,m},\delta}(x)\right)^2f_{\Gstar}(x)
\nonumber\\
&\qquad\quad+
3\sum_{x=1}^{\hat K}\left(m_{\tilde G^{(T)}_{d,m},\delta}(x)-m_{\tilde G^{(T)}_{d,m}}(x)\right)^2f_{\Gstar}(x).
\end{align}

We now control the three terms on the right-hand side separately. The first term controls the bias in approximating $m_{\Gstar}(\cdot)$ using $m_{\Gstar,\delta}(\cdot)$. In view of the analysis in \prettyref{app:bias}, we have
\begin{equation}\label{eq:bias-G-coarse-compact}
\sum_{x=1}^{\hat K}f_{\Gstar}(x)\left(m_{\Gstar,\delta}(x)-m_{\Gstar}(x)\right)^2\le \hat K L\delta.
\end{equation}

Next, we control the third term in \eqref{eq:main-sum-decomp-coarse-compact}, i.e. the bias term in approximating $m_{\tilde G^{(T)}_{d,m}}(\cdot)$ using $m_{\tilde G^{(T)}_{d,m},\delta}(\cdot)$. We use \eqref{eq:m5} with $G=\tilde G^{(T)}_{d,m}$ and $K$ replaced by $\hat K$ to get
\begin{align}\label{eq:m6-truncated}
	\sum_{x=1}^{\hat K} f_{\tilde G^{(T)}_{d,m}}(x)\,(m_{\tilde G^{(T)}_{d,m},\delta}(x)-m_{\tilde G^{(T)}_{d,m}}(x))^2\le \hat KL\delta.
\end{align}

Now set $u(x):=\left(m_{\tilde G^{(T)}_{d,m},\delta}(x)-m_{\tilde G^{(T)}_{d,m}}(x)\right)^2\in[0,1]$. By Lemma~\ref{lem:quad-convert}, with \(a=f_{\Gstar}(x)\) and \(b=f_{\tilde G^{(T)}_{d,m}}(x)\), we obtain that
\begin{displaymath}
f_{\Gstar}(x)u(x)\le 4\frac{f_{\Gstar}(x)f_{\tilde G^{(T)}_{d,m}}(x)}{f_{\Gstar}(x)+f_{\tilde G^{(T)}_{d,m}}(x)}u(x)+2\frac{(f_{\Gstar}(x)-f_{\tilde G^{(T)}_{d,m}}(x))^2}{f_{\Gstar}(x)+f_{\tilde G^{(T)}_{d,m}}(x)}.
\end{displaymath}
Summing over \(x=1,\dots,\hat K\), using
\begin{displaymath}
\frac{f_{\Gstar}(x)f_{\tilde G^{(T)}_{d,m}}(x)}{f_{\Gstar}(x)+f_{\tilde G^{(T)}_{d,m}}(x)}\le f_{\tilde G^{(T)}_{d,m}}(x),
\end{displaymath}
and applying Lemma~\ref{lem:chi2-hell} together with \(T\ge\hat K\), we obtain
\begin{equation}\label{eq:bias-Gtilde-coarse-compact}
\sum_{x=1}^{\hat K}f_{\Gstar}(x)\left(m_{\tilde G^{(T)}_{d,m},\delta}(x)-m_{\tilde G^{(T)}_{d,m}}(x)\right)^2\le 4\hat KL\delta+8\mathcal H^{(T)}(\Gstar,\tilde G^{(T)}_{d,m}).
\end{equation}

\paragraph{Smoothed estimation error.}

Let $u(x):=\left(m_{\Gstar,\delta}(x)-m_{\tilde G^{(T)}_{d,m},\delta}(x)\right)^2\in[0,1]$. By Lemma~\ref{lem:quad-convert},
\begin{displaymath}
\sum_{x=1}^{\hat K}f_{\Gstar}(x)u(x)\le 4\sum_{x=1}^{\hat K}w_x(\Gstar,\tilde G^{(T)}_{d,m})u(x)+2\sum_{x=1}^{\hat K}
\frac{(f_{\Gstar}(x)-f_{\tilde G^{(T)}_{d,m}}(x))^2}{f_{\Gstar}(x)+f_{\tilde G^{(T)}_{d,m}}(x)},
\end{displaymath}
where $$w_x(\Gstar,\tilde G^{(T)}_{d,m}):=\frac{f_{\Gstar}(x)f_{\tilde G^{(T)}_{d,m}}(x)}{f_{\Gstar}(x)+f_{\tilde G^{(T)}_{d,m}}(x)}.$$
Using Lemma~\ref{lem:chi2-hell} and \(T\ge\hat K\),
\begin{equation}\label{eq:smoothed-convert-coarse-compact}
\sum_{x=1}^{\hat K}f_{\Gstar}(x)u(x)\le 4\sum_{x=1}^{\hat K}w_x(\Gstar,\tilde G^{(T)}_{d,m})u(x)+8\mathcal H^{(T)}(\Gstar,\tilde G^{(T)}_{d,m}).
\end{equation}
Fix \(x\le\hat K\). Apply Lemma~\ref{lem:ratio-stab} with $A=A_{x,\delta}(\Gstar)$, $B=B_x(\Gstar)$, $\tilde A=A_{x,\delta}(\tilde G^{(T)}_{d,m})$ and $\tilde B=B_x(\tilde G^{(T)}_{d,m})$. Since \(0\le\phi_\delta\le1\),
\begin{displaymath}
A_{x,\delta}(\Gstar)\le B_x(\Gstar),\qquad A_{x,\delta}(\tilde G^{(T)}_{d,m})\le B_x(\tilde G^{(T)}_{d,m}).
\end{displaymath}
Therefore
\begin{displaymath}
w_x(\Gstar,\tilde G^{(T)}_{d,m})u(x)\le\frac{8}{x!}\left\{\left(\sqrt{A_{x,\delta}(\Gstar)}-\sqrt{A_{x,\delta}(\tilde G^{(T)}_{d,m})}\right)^2+\left(\sqrt{B_x(\Gstar)}-\sqrt{B_x(\tilde G^{(T)}_{d,m})}\right)^2\right\}.
\end{displaymath}
Since \(B_x(G)=x!f_G(x)\), this becomes
\begin{displaymath}
w_x(\Gstar,\tilde G^{(T)}_{d,m})u(x)\le 8\left(\sqrt{f_{\Gstar}(x)}-\sqrt{f_{\tilde G^{(T)}_{d,m}}(x)}\right)^2+\frac8{x!}\left(\sqrt{A_{x,\delta}(\Gstar)}-\sqrt{A_{x,\delta}(\tilde G^{(T)}_{d,m})}\right)^2.
\end{displaymath}
Summing over \(x=1,\dots,\hat K\), and using \(T\ge\hat K\), gives
\begin{equation}\label{eq:w-sum-split-coarse-compact}
\sum_{x=1}^{\hat K}w_x(\Gstar,\tilde G^{(T)}_{d,m})u(x)\le 8\mathcal H^{(T)}(\Gstar,\tilde G^{(T)}_{d,m})+8\sum_{x=1}^{\hat K}\frac1{x!}\left(\sqrt{A_{x,\delta}(\Gstar)}-\sqrt{A_{x,\delta}(\tilde G^{(T)}_{d,m})}\right)^2.
\end{equation}
For the \(A\)-part, apply Lemma~\ref{lem:sqrtA-control-coarse-compact} with \((G_1,G_2)=(\Gstar,\tilde G^{(T)}_{d,m})\). Then
\begin{displaymath}
\left(\sqrt{A_{x,\delta}(\Gstar)}-\sqrt{A_{x,\delta}(\tilde G^{(T)}_{d,m})}\right)^2\le 8C_A\mathcal H^{(T)}(\Gstar,\tilde G^{(T)}_{d,m})+2\sqrt{2}\,\varepsilon_{x,\delta}^{(T)}.
\end{displaymath}
Summing over \(x=1,\dots,\hat K\), and using \(\sum_{x=1}^{\hat K}1/x!\le e\), yields
\begin{equation}\label{eq:A-sum-control-coarse-compact}
\sum_{x=1}^{\hat K}\frac1{x!}\left(\sqrt{A_{x,\delta}(\Gstar)}-\sqrt{A_{x,\delta}(\tilde G^{(T)}_{d,m})}\right)^2\le 8eC_A\mathcal H^{(T)}(\Gstar,\tilde G^{(T)}_{d,m})+2\sqrt{2}\sum_{x=1}^{\hat K}\frac{\varepsilon_{x,\delta}^{(T)}}{x!}.
\end{equation}
By Assumption~\ref{asmp:aa-self-normal-coarse}, \(\varepsilon_{x,\delta}^{(T)}\le C_\varepsilon|\calW_+|\delta\leq C_\epsilon K\delta\). Therefore
\begin{displaymath}
\sum_{x=1}^{\hat K}\frac{\varepsilon_{x,\delta}^{(T)}}{x!}\le C_\varepsilon\delta\sum_{x=1}^{\hat K}\frac1{x!}\le eC_\varepsilon\delta.
\end{displaymath}
Consequently,
\begin{equation}\label{eq:A-sum-control-coarse-final}
\sum_{x=1}^{\hat K}\frac1{x!}\left(\sqrt{A_{x,\delta}(\Gstar)}-\sqrt{A_{x,\delta}(\tilde G^{(T)}_{d,m})}\right)^2\le C_{A,1}\mathcal H^{(T)}(\Gstar,\tilde G^{(T)}_{d,m})+C_{A,2}\delta,
\end{equation}
for constants \(C_{A,1},C_{A,2}\) depending on the structural parameters of $\calC$ and the constants in \prettyref{asmp:aa-self-normal-coarse}. Combining \eqref{eq:w-sum-split-coarse-compact} with \eqref{eq:A-sum-control-coarse-final}, we obtain
\begin{equation}\label{eq:w-sum-final-coarse-compact}
\sum_{x=1}^{\hat K}w_x(\Gstar,\tilde G^{(T)}_{d,m})u(x)\le C_{w,1}\mathcal H^{(T)}(\Gstar,\tilde G^{(T)}_{d,m})+C_{w,2}\delta.
\end{equation}
Plugging this into \eqref{eq:smoothed-convert-coarse-compact} gives
\begin{equation}\label{eq:smoothed-est-coarse-compact}
\sum_{x=1}^{\hat K}f_{\Gstar}(x)\left(m_{\Gstar,\delta}(x)-m_{\tilde G^{(T)}_{d,m},\delta}(x)\right)^2\le C_{\mathrm{sm},1}\mathcal H^{(T)}(\Gstar,\tilde G^{(T)}_{d,m})+C_{\mathrm{sm},2}\delta.
\end{equation}

\paragraph{Putting the three pieces together.}

Combining \eqref{eq:main-sum-decomp-coarse-compact}, \eqref{eq:bias-G-coarse-compact}, \eqref{eq:bias-Gtilde-coarse-compact}, and \eqref{eq:smoothed-est-coarse-compact}, there exist constants \(C_0,C_1\), depending only on \((\hat h,T,L)\) and the constants in \prettyref{asmp:aa-self-normal-coarse}, such that

\begin{equation}\label{eq:pointwise-main-coarse-compact}
\sum_{x=1}^{\hat K}\left(m_{\Gstar}(x)-m_{\tilde G^{(T)}_{d,m}}(x)\right)^2 f_{\Gstar}(x)\le C_0\delta+C_1\mathcal H^{(T)}(\Gstar,\tilde G^{(T)}_{d,m}).
\end{equation}
Using \eqref{eq:basic-reduction-coarse-compact} and taking expectations over the training sample,
\begin{displaymath}
\frac1n\EE\!\left[\left(\hat S_{n,>}(\Gstar)-\hat S_{n,>}(\tilde G^{(T)}_{d,m})\right)^2\right]\le nC_0\delta+nC_1\EE\!\left[\mathcal H^{(T)}(\Gstar,\tilde G^{(T)}_{d,m})\right].
\end{displaymath}
Finally, by Lemma~\ref{lem:EB_density-H2_coars},
\begin{displaymath}
\EE\!\left[\mathcal H^{(T)}(\Gstar,\tilde G^{(T)}_{d,m})\right]\le C_d^{\mathrm{coarse}}(c_1,c_2)\frac{T+2}{m}.
\end{displaymath}
Therefore
\begin{displaymath}
\frac1n\EE\!\left[\left(\hat S_{n,>}(\Gstar)-\hat S_{n,>}(\tilde G^{(T)}_{d,m})\right)^2\right]\le n C_0\delta+nC_1C_d^{\mathrm{coarse}}(c_1,c_2)\frac{T+2}{m}.
\end{displaymath}
This proves Theorem~\ref{thm:main-above-average-coarsened} \textnormal{(a)}.

\subsubsection{Subexponential setting}


We now assume \(\Gstar\in\mathcal P_s\), fix integers \(K\ge0\) and \(T\ge K\), and let \(\tilde G^{(T)}_{d,m}\in\mathcal C\) denote the coarsened minimum-distance estimate
appearing in Theorem~\ref{thm:main-above-average-coarsened} \textnormal{(b)}. If \(K=0\), the finite window is empty. Since $m_{\Gstar}(0)=m_{\tilde G^{(T)}_{d,m}}(0)=0$  and \(|m_{\Gstar}(x)-m_{\tilde G^{(T)}_{d,m}}(x)|\le1\), the same reduction used below gives
\begin{displaymath}
\frac1n\EE\!\left[\left(\hat S_{n,>}(\Gstar)-\hat S_{n,>}(\tilde G^{(T)}_{d,m})\right)^2\right]\le n\PP_{\Gstar}(X>0)\le nC_s,
\end{displaymath}
which is upper bounded by the right-hand side of \eqref{eq:main-coarse-sube-result}. We may therefore assume \(K\ge1\). Conditional on \(\tilde G^{(T)}_{d,m}\), set
\begin{displaymath}
\Delta(x):=m_{\Gstar}(x)-m_{\tilde G^{(T)}_{d,m}}(x).
\end{displaymath}
Then \(|\Delta(x)|\le1\) and \(\Delta(0)=0\). As in the compact-support proof,
\begin{displaymath}
\frac1n\EE\!\left[\left(\hat S_{n,>}(\Gstar)-\hat S_{n,>}(\tilde G^{(T)}_{d,m})\right)^2\mid \tilde G^{(T)}_{d,m}\right]\le n\sum_{x\ge0}\Delta(x)^2 f_{\Gstar}(x).
\end{displaymath}
Splitting the last sum at \(K\) gives
\begin{equation}\label{eq:basic-reduction-coarse-sube}
\frac1n\EE\!\left[\left(\hat S_{n,>}(\Gstar)-\hat S_{n,>}(\tilde G^{(T)}_{d,m})\right)^2\mid \tilde G^{(T)}_{d,m}\right]\le n\sum_{x=1}^{K}\left(m_{\Gstar}(x)-m_{\tilde G^{(T)}_{d,m}}(x)\right)^2f_{\Gstar}(x)+n\PP_{\Gstar}(X>K).
\end{equation}

In view of \eqref{eq:basic-reduction-coarse-sube}, taking expectation over the training sample, an argument similar to that as in \prettyref{app:proof-subexpo-new} gives
\begin{displaymath}
\frac1n\EE\!\left[\left(\hat S_{n,>}(\Gstar)-\hat S_{n,>}(\tilde G^{(T)}_{d,m})\right)^2\right]\le nKC_0\delta+nC_1\EE\!\left[\mathcal H^{(T)}(\Gstar,\tilde G^{(T)}_{d,m})\right]+n\PP_{\Gstar}(X>K).
\end{displaymath}
Here, to bounds terms similar as in \eqref{eq:subexpo-new-A-sum-control-route1-hhat} and \eqref{eq:m7} we used $\epsilon_{x,\delta}\leq C_\epsilon h(T,|\calW_+|)\delta\leq C_\epsilon K\delta$.
By Lemma~\ref{lem:EB_density-H2_coars},
\begin{displaymath}
\EE\!\left[\mathcal H^{(T)}(\Gstar,\tilde G^{(T)}_{d,m})\right]\le C_d^{\mathrm{coarse}}(c_1,c_2)\frac{T+2}{m}.
\end{displaymath}
Moreover, by Lemma~\ref{lem:tail-X-exp-nocoarsen},
\begin{displaymath}
\PP_{\Gstar}(X>K)\le C_s e^{-c_sK},\qquad c_s=\log\left(1+\frac1{2s}\right),\qquad C_s=3.
\end{displaymath}
Therefore
\begin{displaymath}
\frac1n\EE\!\left[\left(\hat S_{n,>}(\Gstar)-\hat S_{n,>}(\tilde G^{(T)}_{d,m})\right)^2\right]\le nKC_0(T)\delta+nC_1(T)C_d^{\mathrm{coarse}}(c_1,c_2)\frac{T+2}{m}+nC_s e^{-c_sK}.
\end{displaymath}
This proves Theorem~\ref{thm:main-above-average-coarsened} \textnormal{(b)}.

\subsection{Proof of Corollary~\ref{corr:above-average-discrete-feasible}}\label{subsec:proof-above-average-discrete}

\subsubsection{Preliminary lemmas}

We first record two elementary facts for the finite-support model. For \(G=\sum_{j=1}^J w_j\delta_{\theta_j}\in\mathcal P(\Theta_J)\), set $v_G(j):=w_j e^{-\theta_j},
\qquad j=1,\dots,J$. Also write
\begin{displaymath}
q_G(t):=t!f_G(t)=\sum_{j=1}^J \theta_j^t v_G(j),\qquad t=0,\dots,J-1.
\end{displaymath}
Let \(V\) be the \(J\times J\) Vandermonde matrix with entries $V_{t,j}:=\theta_j^t$ for $t=0,\dots,J-1,\quad j=1,\dots,J$. Since the points \(\theta_1,\dots,\theta_J\) are distinct, \(V\) is invertible.

\begin{lemma}\label{lem:exact-A-representation-discrete-above}
For each \(x=1,\dots,K_J\), there exist constants \(\alpha_x(0),\dots,\alpha_x(J-1)\), depending only on \(\Theta_J\), such that
\begin{displaymath}
A_x(G)=\sum_{t=0}^{J-1}\alpha_x(t)f_G(t),\qquad G\in\mathcal P(\Theta_J).
\end{displaymath}
\end{lemma}

\begin{proof}
For \(G\in\mathcal P(\Theta_J)\), the vector \((q_G(0),\dots,q_G(J-1))^\top\) satisfies $q_G=Vv_G$. Hence \(v_G=V^{-1}q_G\). For fixed \(x=1,\dots,K_J\),
\begin{displaymath}
A_x(G)=\sum_{j:\theta_j<x}\theta_j^x v_G(j).
\end{displaymath}
Thus \(A_x(G)\) is a linear functional of \(q_G\). Since \(q_G(t)=t!f_G(t)\), there are constants \(\alpha_x(0),\dots,\alpha_x(J-1)\), depending only on \(\Theta_J\), such that
\begin{displaymath}
A_x(G)=\sum_{t=0}^{J-1}\alpha_x(t)f_G(t).
\end{displaymath}
This proves the claim.
\end{proof}

\begin{lemma}\label{lem:ratio-control-discrete-above}
For all \(G_1,G_2\in\mathcal C\), there exists a constant \(C_{\Theta_J}<\infty\), depending only on \(({\Theta_J},b_J)\), such that
\begin{displaymath}
\sum_{x=1}^{K_J} f_{G_1}(x)\bigl(m_{G_1}(x)-m_{G_2}(x)\bigr)^2\le C_{\Theta_J} H^2(f_{G_1},f_{G_2}).
\end{displaymath}
Moreover, if \(T\ge T_J\), then there exists a constant \(C_{{\Theta_J},T,b_J}<\infty\), depending only on \(({\Theta_J},T,b_J)\), such that
\begin{displaymath}
\sum_{x=1}^{K_J} f_{G_1}(x)\bigl(m_{G_1}(x)-m_{G_2}(x)\bigr)^2\le C_{{\Theta_J},T,b_J}H^2\!\left(\tilde f_{G_1}^{(T)},\tilde f_{G_2}^{(T)}\right).
\end{displaymath}
\end{lemma}

\begin{proof}
Fix \(x=1,\dots,K_J\). For \(r=1,2\), write $A_r=A_x(G_r)$ and $B_r=B_x(G_r)=x!f_{G_r}(x)$. Since \(G_1,G_2\in\mathcal C\),
\begin{displaymath}
B_r\ge x!b_J,\qquad r=1,2.
\end{displaymath}
Also \(0\le A_r\le B_r\). Hence
\begin{equation}\label{eq:discrete-ratio-direct}
\left|\frac{A_1}{B_1}-\frac{A_2}{B_2}\right|\le\frac{|A_1-A_2|}{B_1}+\frac{A_2|B_1-B_2|}{B_1B_2}\le\frac{|A_1-A_2|}{x!b_J}+\frac{|B_1-B_2|}{x!b_J}.
\end{equation}
By Lemma~\ref{lem:exact-A-representation-discrete-above},
\begin{displaymath}
A_1-A_2=\sum_{t=0}^{J-1}\alpha_x(t)\bigl(f_{G_1}(t)-f_{G_2}(t)\bigr).
\end{displaymath}
Therefore, by Cauchy--Schwarz and the elementary inequality \((a-b)^2\le4(\sqrt a-\sqrt b)^2\) for \(a,b\in[0,1]\),
\begin{displaymath}
|A_1-A_2|^2\le C_{{\Theta_J},x}H^2(f_{G_1},f_{G_2}),
\end{displaymath}
for a finite constant \(C_{{\Theta_J},x}\) depending only on \({\Theta_J}\). Similarly,
\begin{displaymath}
|B_1-B_2|^2=(x!)^2|f_{G_1}(x)-f_{G_2}(x)|^2\le 4(x!)^2H^2(f_{G_1},f_{G_2}).
\end{displaymath}
Combining these bounds with \eqref{eq:discrete-ratio-direct}, we obtain $\bigl(m_{G_1}(x)-m_{G_2}(x)\bigr)^2\le C_{{\Theta_J},b_J,x}H^2(f_{G_1},f_{G_2})$. Multiplying by \(f_{G_1}(x)\le1\) and summing over \(x=1,\dots,K_J\) gives
the first assertion.

For the coarsened assertion, assume \(T\ge T_J=\max\{K_J,J-1\}\). Then the coordinates
\(0,\dots,J-1\) and \(1,\dots,K_J\) all belong to the non-tail part of the
\(T\)-coarsened alphabet. Hence the preceding bounds remain valid with
\(H^2(f_{G_1},f_{G_2})\) replaced by
\(H^2(\tilde f_{G_1}^{(T)},\tilde f_{G_2}^{(T)})\). This proves the lemma.
\end{proof}

\subsubsection{Regular procedure}


If \(K_J=0\), then \(x>\theta_J\) for every integer \(x\ge1\). Hence, for every \(G\in\mathcal C\)
$m_G(0)=0$ and  $m_G(x)=1$ for $x\ge1$. Thus \(\hat S_{n,>}(G)=\hat S_{n,>}(\Gstar)\) for every \(G\in\mathcal C\).

First consider the regular estimator \(\hat G_{d,m}\). Conditional on \(\hat G_{d,m}\), define
\begin{displaymath}
\Delta(x):=m_{\Gstar}(x)-m_{\hat G_{d,m}}(x).
\end{displaymath}
Then \(|\Delta(x)|\le1\), \(\Delta(0)=0\), and \(\Delta(x)=0\) for all \(x>K_J\), because \(x>\theta_J\) implies \(\theta<x\) almost surely under both \(\Gstar\) and \(\hat G_{d,m}\). As in the proof of Theorem~\ref{thm:main-above-average},
\begin{displaymath}
\frac1n\EE\!\left[\left(\hat S_{n,>}(\hat G_{d,m})-\hat S_{n,>}(\Gstar)\right)^2\Bigm| \hat G_{d,m}\right]\le n\sum_{x=1}^{K_J}\Delta(x)^2 f_{\Gstar}(x).
\end{displaymath}
By Lemma~\ref{lem:ratio-control-discrete-above}, with \(G_1=\Gstar\) and \(G_2=\hat G_{d,m}\),
\begin{displaymath}
\sum_{x=1}^{K_J}\Delta(x)^2 f_{\Gstar}(x)\le C_{\Theta_J} H^2(f_{\Gstar},f_{\hat G_{d,m}}).
\end{displaymath}
Taking expectations gives
\begin{displaymath}
\frac1n\EE\!\left[\left(\hat S_{n,>}(\hat G_{d,m})-\hat S_{n,>}(\Gstar)\right)^2\right]\le nC_{\Theta_J}\EE\!\left[H^2(f_{\Gstar},f_{\hat G_{d,m}})\right].
\end{displaymath}
By Corollary~\ref{corr:hell-risk-discrete-feasible}, applied to the regular estimator in $\mathcal C$
\begin{displaymath}
\EE\!\left[H^2(f_{\Gstar},f_{\hat G_{d,m}})\right]\le C({\Theta_J},c_1,c_2)\frac{\log m}{m}.
\end{displaymath}
This proves \eqref{eq:above-average-discrete-regular}.

\subsubsection{Coarsened procedure}

Now consider the coarsened estimator \(\widetilde G_{d,m}^{(T)}\), with \(T\ge T_J\). The same reduction gives
\begin{displaymath}
\frac1n\EE\!\left[\left(\hat S_{n,>}(\widetilde G_{d,m}^{(T)})-\hat S_{n,>}(\Gstar)\right)^2\Bigm| \widetilde G_{d,m}^{(T)}\right]\le n\sum_{x=1}^{K_J}\left(m_{\Gstar}(x)-m_{\widetilde G_{d,m}^{(T)}}(x)\right)^2 f_{\Gstar}(x).
\end{displaymath}
By the coarsened part of Lemma~\ref{lem:ratio-control-discrete-above}, with \(G_1=\Gstar\) and \(G_2=\widetilde G_{d,m}^{(T)}\),
\begin{displaymath}
\sum_{x=1}^{K_J}\left(m_{\Gstar}(x)-m_{\widetilde G_{d,m}^{(T)}}(x)\right)^2 f_{\Gstar}(x)\le C_{\Theta_J,T,b_J}H^2\!\left(\tilde f_{\Gstar}^{(T)},\tilde f_{\widetilde G_{d,m}^{(T)}}^{(T)}\right).
\end{displaymath}
Taking expectations and applying Lemma~\ref{lem:EB_density-H2_coars} gives
\begin{displaymath}
\EE\!\left[H^2\!\left(\tilde f_{\Gstar}^{(T)},\tilde f_{\widetilde G_{d,m}^{(T)}}^{(T)}\right)\right]\le C_d^{\mathrm{coarse}}(c_1,c_2)\frac{T+2}{m}.
\end{displaymath}
Therefore
\begin{displaymath}
\frac1n\EE\!\left[\left(\hat S_{n,>}(\widetilde G_{d,m}^{(T)})-\hat S_{n,>}(\Gstar)\right)^2\right]\le C_{\Theta_J,T,b_J}C_d^{\mathrm{coarse}}(c_1,c_2)\frac{n(T+2)}{m}.
\end{displaymath}
Absorbing constants proves \eqref{eq:above-average-discrete-coarsened}.

\end{document}